\documentclass[11pt,a4paper,svgnames]{article}

\usepackage{tikz} 
\usepackage{amsmath,amsfonts,amssymb,amsthm}
\usepackage{amsthm}
\usepackage{graphicx,color}
\usepackage{colortbl}
\usepackage{boxedminipage}
\usepackage[ruled,boxed,linesnumbered]{algorithm2e}
\usepackage{framed}
\usepackage{mathtools}
\usepackage{float}
\usepackage{enumitem}

\usepackage{thmtools}
\usepackage{thm-restate}
\usepackage{xspace}
\usepackage{todonotes}
\usepackage{footnote}
\usepackage{xcolor}
\definecolor{anti-flashwhite}{rgb}{0.95, 0.95, 0.96}

\usepackage{vmargin}
\setmarginsrb{0.9in}{0.9in}{0.9in}{0.9in}{0pt}{0pt}{0pt}{6mm}

\usepackage{todonotes}
\usepackage{footnote}
\usepackage[pdftex, plainpages = false, pdfpagelabels, 
bookmarks=false,
bookmarksopen = true,
bookmarksnumbered = true,
breaklinks = true,
linktocpage,
pagebackref,
colorlinks = true,  
linkcolor = blue,
urlcolor  = blue,
citecolor = red,
anchorcolor = green,
hyperindex = true,
hyperfigures
]{hyperref} 
\usepackage{xifthen}
\usepackage{tabularx}
\usetikzlibrary{calc}

\usepackage[nameinlink]{cleveref}

\newtheorem{theorem}{Theorem}
\newtheorem{lemma}{Lemma}
\newtheorem{proposition}{Proposition}

\newtheorem{claim}{Claim}[section]
\newtheorem{corollary}{Corollary}
\theoremstyle{definition}
\newtheorem{definition}{Definition}
\newtheorem{observation}{Observation}

\newtheorem{remark}{Remark}

\theoremstyle{definition}
\newtheorem{reduction}{Reduction Rule}[section]

\newcommand{\blue}[1]{{\color{blue} #1}}
\newcommand{\lr}[1]{\left( #1\right)}
\newcommand{\LR}[1]{\left\{ #1\right\}}
\DeclareMathOperator*{\argmin}{arg\,min}

\newcommand{\runtimefortwo}{{$ 2^{\cO(\sqrt{k}\log{k})} + n\cdot k^{\cO(1)} $}\xspace}
\newcommand{\runtimeforthree}{{$ 2^{\cO({k}^{0.67})}+ n\cdot k^{\cO(1)}$}\xspace}
\newcommand{\kernelforthree}{{$   \cO(k^{8})$}\xspace}

\DeclareMathOperator{\operatorClassNP}{NP}
\newcommand{\classNP}{\ensuremath{\operatorClassNP}\xspace}

\DeclareMathOperator{\operatorClassCoNP}{coNP}
\newcommand{\classCoNP}{\ensuremath{\operatorClassCoNP}}
\DeclareMathOperator{\operatorClassFPT}{FPT\xspace}
\newcommand{\classFPT}{\ensuremath{\operatorClassFPT}\xspace}

\renewcommand{\diamondsuit}{\blacklozenge}
\newcommand{\ein}{extended instance\xspace}

\newcommand{\yin}{{\sc Yes}-instance\xspace}

\newcommand{\cC}{\mathcal{C}}
\newcommand{\cE}{\mathcal{E}}
\newcommand{\so}{\varsigma}
\newcommand{\wrt}{w.r.t.~}
\newcommand{\lf}{{\cal L}}
\newcommand{\mm}{{\cal M}}
\newcommand{\rt}{{\cal R}}

\newcommand{\cT}{\mathcal{T}}
\newcommand{\heavy}{{\sf Heavy}}
\newcommand{\light}{{\sf Light}}

\newcommand{\extreme}{{\sf Extreme}}
\newcommand{\surr}{{\sf Surround}}
\newcommand{\exceptional}{{\sf Exceptional}}
\newcommand{\nice}{{\sf nice}}

\newcommand{\slot}{{\sf slot}}

\newcommand{\types}{\mathsf{Types}}
\newcommand{\subslot}{\mathsf{subslot}}

\newcommand{\cG}{{\cal G}}

\newcommand{\low}{{\sf lower}}
\newcommand{\upr}{{\sf upper}}
\newcommand{\boundary}{{\sf BoundaryEdges}}
\newcommand{\regulare}{{\sf RegularEdges}}
\newcommand{\weightede}{{\sf WeightedEdges}}
\newcommand{\leftb}{{\sf leftbound}}
\newcommand{\rightb}{{\sf rightbound}}
\newcommand{\leftbone}{{\sf leftbound}_{12}}
\newcommand{\rightbone}{{\sf rightbound}_{12}}
\newcommand{\leftbtwo}{{\sf leftbound}_{23}}
\newcommand{\rightbtwo}{{\sf rightbound}_{23}}
\newcommand{\ept}{{\sf endpt}}
\newcommand{\epts}{{\sf Endpoints}}

\newcommand{\Blue}{{\sf Blue}}
\newcommand{\Red}{{\sf Red}}
\newcommand{\crs}{{\sf cross}}
\newcommand{\sep}{{\sf Sep}}
\newcommand{\lsep}{{\sf lsep}}
\newcommand{\rsep}{{\sf rsep}}
\newcommand{\lsepone}{{\sf lsep}_{12}}
\newcommand{\rsepone}{{\sf rsep}_{12}}
\newcommand{\lseptwo}{{\sf lsep}_{23}}
\newcommand{\rseptwo}{{\sf rsep}_{23}}

\newcommand{\md}{{\sf Mid}}
\newcommand{\mde}{{\sf MiddleEdges}}
\newcommand{\crosse}{{\sf CrossingEdges}}
\newcommand{\comps}{{\sf comps}}
\newcommand{\spl}{{\sf Special}}
\newcommand{\pendant}{{\sf Pendants}}
\newcommand{\updown}{{\sf UpDown}}
\newcommand{\pendantsl}{{\sf LeftPNum}}
\newcommand{\pendantsr}{{\sf RightPNum}}
\newcommand{\pendantsm}{{\sf MidPNum}}
\newcommand{\pendantst}{{\sf TotalPNum}}
\newcommand{\total}{{\sf TotalNum}}
\newcommand{\totall}{{\sf LeftNum}}
\newcommand{\totalr}{{\sf RightNum}}
\newcommand{\nicefamily}{{\sf NiceFamily}}
\newcommand{\pendantl}{{\sf LeftPendants}}
\newcommand{\pendantm}{{\sf MidPendants}}
\newcommand{\pendantr}{{\sf RightPendants}}
\newcommand{\vtwo}{V^{conn}_2}
\newcommand{\gapleft}{{\sf GapLeft}}
\newcommand{\gapright}{{\sf GapRight}}
\newcommand{\pecu}{{\sf Peculiar}}
\newcommand{\pecunp}{{\sf PeculiarNP}}

\newcommand{\DB}{{\sf DB}\xspace}
\newcommand{\DT}{{\sf DT}\xspace}
\newcommand{\SB}{{\sf SB}\xspace}
\newcommand{\ST}{{\sf ST}\xspace}
\newcommand{\llc}{{\sf LeftmostLC}}
\newcommand{\rlc}{{\sf RightmostLC}}
\newcommand{\leftstar}{{\sf LeftStar}}
\newcommand{\rightstar}{{\sf RightStar}}
\newcommand{\str}{{\sf st}}
\newcommand{\Star}{{\sf Star}}
\newif\ifstarr
\starrtrue

\newcommand{\loe}{{\sf m}_{12}}
\newcommand{\upe}{{\sf m}_{23}}
\newcommand{\loeo}{{\sf m_{12}^{orig}}}
\newcommand{\upeo}{{\sf m_{23}^{orig}}}
\newcommand{\leftpart}{{\sf LeftComps}}
\newcommand{\rightpart}{{\sf RightComps}}

\newcommand{\predef}{\pi\text{\sf-Defined}}
\newcommand{\opt}{{\sf opt}}

\newcommand{\optl}{{\sf optL}\xspace}
\newcommand{\optr}{{\sf optR}\xspace}

\newcommand{\nat}{\mathbb{N}}
\newcommand{\pbp}{{\sf PendantBlueprint}}
\newcommand{\pguess}{{\sf PendantGuess}}
\newcommand{\cbp}{{\sf ComponentBlueprint}}
\newcommand{\ecbp}{{\sf ExtComponentBlueprint}}
\newcommand{\ecguess}{{\sf ExtComponentGuess}}
\newcommand{\cguess}{{\sf ComponentGuess}}
\newcommand{\compedges}{{\sf CompEdges}}
\newcommand{\cedges}{{\sf ConnEdges}}
\newcommand{\nicepfamily}{{\sf NicePFamily}}
\newcommand{\sloted}{{\sf SlotEdges}}
\newcommand{\nicecfamily}{{\sf NiceCFamily}}

\newcommand{\bnd}{{\sf ExtBnd}}
\newcommand{\lbnd}{{\sf LeftExtBnd}}
\newcommand{\rbnd}{{\sf RightExtBnd}}

\newlength{\RoundedBoxWidth}
\newsavebox{\GrayRoundedBox}
\newenvironment{GrayBox}[1]%
{\setlength{\RoundedBoxWidth}{.93\textwidth}
	\def\boxheading{#1}
	\begin{lrbox}{\GrayRoundedBox}
		\begin{minipage}{\RoundedBoxWidth}}%
		{   \end{minipage}
	\end{lrbox}
	\begin{center}
		\begin{tikzpicture}%
			\node(Text)[draw=black!20,fill=white,rounded corners,%
			inner sep=2ex,text width=\RoundedBoxWidth]%
			{\usebox{\GrayRoundedBox}};
			\coordinate(x) at (current bounding box.north west);
			\node [draw=white,rectangle,inner sep=3pt,anchor=north west,fill=white] 
			at ($(x)+(6pt,.75em)$) {\boxheading};
		\end{tikzpicture}
\end{center}}

\newenvironment{defproblemx}[2][]{\noindent\ignorespaces%
	\FrameSep=6pt%
	\parindent=0pt%
	\vspace*{-1.5em}
	\ifthenelse{\isempty{#1}}{%
		\begin{GrayBox}{\textsc{#2}}%
		}{%
			\begin{GrayBox}{\textsc{#2} parameterized by~{#1}}%
			}
			\begin{tabular*}{\textwidth}{@{\hspace{.1em}} >{\itshape} p{1.8cm} p{0.8\textwidth} @{}}%
			}{
			\end{tabular*}%
		\end{GrayBox}%
		\ignorespacesafterend
	}

	\newcommand{\defproblema}[3]{
		\begin{defproblemx}{#1}
			Input:  & #2 \\
			Task: & #3
		\end{defproblemx}
	}%

	\newcommand{\cO}{\mathcal{O}}
	\newcommand{\Oh}{\cO}
	\newcommand{\OO}{\cO}

	\newcommand{\cI}{\mathcal{I}}

	\newcommand{\pname}{\textsc}
	\newcommand{\ProblemFormat}[1]{\pname{#1}}
	\newcommand{\ProblemIndex}[1]{\index{problem!\ProblemFormat{#1}}}
	\newcommand{\ProblemName}[1]{\ProblemFormat{#1}\ProblemIndex{#1}{}\xspace}

	\newcommand{\probHCr}{$h$-\ProblemName{Layer Crossing Minimization}}
	\newcommand{\probTwoCr}{$2$-\ProblemName{Layer Crossing Minimization}}
	\newcommand{\probThreeCr}{$3$-\ProblemName{Layer Crossing Minimization}}
	\newcommand{\probFourCr}{$4$-\ProblemName{Layer Crossing Minimization}}
	\newcommand{\probFiveCr}{$5$-\ProblemName{Layer Crossing Minimization}}
	\newcommand{\probDF}{\ProblemName{Disjoint Factors}}

\begin{document}	
		
		\title{Tight Parameterized (In)tractability of Layered Crossing Minimization: Subexponential Algorithms and Kernelization\thanks{This research was initiated during Dagstuhl Seminar 23162 ``New Frontiers of Parameterized Complexity in Graph Drawing''. 
				\\Fedor V.~Fomin acknowledges support from the Research Council of Norway via the  BWCA (grant no. 314528) project.
				\\Petr A.~Golovach acknowledges support from the Research Council of Norway via the BWCA (grant no. 314528) and Extreme-Algorithms (grant no. 355137) projects. 
				\\Tanmay Inamdar acknowledges support from Anusandhan National Research Foundation (ANRF), grant number ANRF/ECRG/2024/004144/PMS, and IIT Jodhpur Research Initiation Grant, grant number I/RIG/TNI/20240072. 
				\\Saket Saurabh acknowledges support from European Research Council (ERC) under the European Union’s Horizon 2020 research and innovation programme (grant agreement No. 819416). 
				\\Meirav Zehavi acknowledges support from the Israel Science Foundation, grant number 1470/24, and from the European Research Council, grant number 101039913 (PARAPATH).}}
	%
	\author{
		Fedor V. Fomin\thanks{
			Department of Informatics, University of Bergen, Norway.}
		\and
		Petr A. Golovach\addtocounter{footnote}{-1}\footnotemark{}
		\and
		Tanmay Inamdar\thanks{Department of Computer Science and Engineering, IIT Jodhpur, India}
		\and
		Saket Saurabh\addtocounter{footnote}{-2}\footnotemark{}\addtocounter{footnote}{1} \thanks{The Institute of Mathematical Sciences, HBNI, Chennai, India}
		\and 
		Meirav Zehavi\thanks{Ben-Gurion University of the Negev, Israel}
		}
	\date{}
	\maketitle
	\begin{abstract}
		
		The starting point of our work is the decade-old open question concerning the subexponential parameterized complexity of the \probTwoCr problem.
		In this problem, the input is an \(n\)-vertex graph \(G\) whose vertices are divided into two independent sets \(V_1, V_2\), and a non-negative integer \(k\). The question is whether \(G\) supports a \(2\)-layered drawing with at most \(k\) crossings. Here, a $2$-layered drawing refers to a drawing of \(G\) where each set \(V_i\) for \(i \in \{1,2\}\) is placed on a distinct straight line parallel to the \(x\)-axis, and all edges are drawn as straight lines connecting vertices.
		Our first theorem resolves the aforementioned question in the affirmative by providing a fixed-parameter tractable (FPT) subexponential algorithm with running time \runtimefortwo. 
		
		The existence of a subexponential fixed-parameter algorithm for two layers immediately raises the question of whether this phenomenon is specific to two layers or can be extended to more layers. (In this setting, vertices are divided into \(h\) independent sets \(V_1, \ldots, V_h\), and the question is whether \(G\) admits an \(h\)-layered drawing with at most \(k\) crossings.) Here, we delve into highly technical depths of the topic of layered drawings to answer this question almost completely, by providing a subexponential fixed-parameter algorithm for three layers with running time \runtimeforthree, and proving that there does not exist a $2^{\OO(k^{1-\epsilon})}\cdot n^{\OO(1)}$-time algorithm (for any fixed $\epsilon>0$) for five or more layers, under the Exponential-Time Hypothesis.
		
		
		Next to the question of subexponential-time algorithms, lies the question of the existence of polynomial kernels for $h$-layered crossing minimization. We completely resolve this question as well -- while a polynomial kernel was already known for \(h=2\), we derive a new polynomial kernel for \(h=3\). Complementarily, we rule out the existence of a polynomial kernel for any \(h \geq 4\), assuming that the polynomial hierarchy does not collapse. Thus, we establish a complete dichotomy regarding polynomial kernelization based on the number of layers \(h\).
		
		
		
	\end{abstract}
	\newpage
	\tableofcontents
	\newpage
	
\section{Introduction}\label{sec:intro}
A popular paradigm in Graph Drawing is that of layered graph drawings \cite{carpano1980automatic, sugiyama1981methods, catarci1995assignment, eades1994edge, junger20022}. In a layered graph drawing, the graph's vertices are allocated to parallel lines (layers), and the edges are drawn as lines between layers, aiming to minimize the number of crossings.
Drawing graphs in layers is a fundamental technique for representing hierarchical graphs. It is integrated into standard graph layout software such as yFiles, Graphviz, OGDF, or TiKz \cite{junger2012graph, tantau2013graph}. We recommend   \cite[Chapter 13]{tamassia2013handbook} for a comprehensive overview. We remark that in general, during the recent years, there is a significant interest in crossing minimization for various types of drawings \cite{HammH22,DBGanianMOPR24,LokshtanovP0S0Z25}.

 A popular model of a layered graph drawing is known as the  \probHCr problem. In this problem, for 
 $h\geq 2$,  the input is an $h$-partite $n$-vertex graph $G$ with $h$-partition $(V_1,\ldots,V_h)$ such that every edge has its endpoints in sets with consecutive indices, and a non-negative integer $k\in\mathbb{N}_0$.  The task is to decide whether $G$ admits an {\em $h$-layered drawing} with at most $k$  crossings. Here, an $h$-layered drawing is a plane drawing such that for every $i\in \{1,2,\ldots,h\}$, all vertices in $V_i$ belong to a single straight line parallel to the $x$-axis, and all edges are drawn as straight line segments between consecutive layers (see \Cref{fig:example}). The \probHCr  problem remains NP-hard even when 
  $k=0$~\cite{heath1992laying}. Furthermore, the question of whether a graph $G$ can be drawn in two layers with at most $k$  crossings, where $k$ is part of the input, is NP-complete~\cite{garey1983crossing}.  
  Thus, \probHCr is para-NP-hard when parameterized only by $k$ or $h$ (that is, we do not expect an algorithm with running time $f(h)\cdot n^{\OO(1)}$ or  $g(k)\cdot n^{\OO(1)}$, for any function $f$ or $g$, unless P=NP). Dujmovi{\'c} et al.~\cite{dujmovic2008parameterized} proved that the problem is 
FPT parameterized by $k+h$. While the running time of the algorithm of Dujmovi{\'c} et al. is linear in the number of vertices $n$, its dependence in the parameter is $(k+h)^{{\OO((k+h)^3)}}$.  

  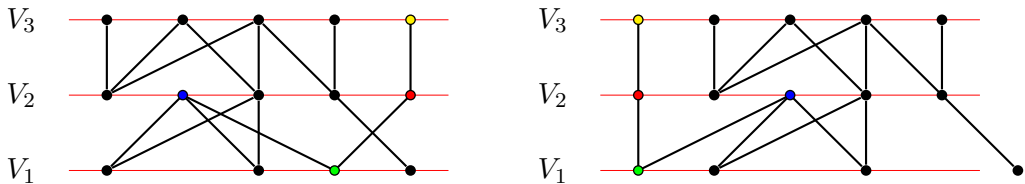
\begin{figure}[H]
	\begin{center}
		\begin{tikzpicture}
			\tikzset{vertex/.style={circle,fill=black,minimum size=4pt,inner sep=0pt},
				edge/.style={thick},
				line/.style={red, ultra thin},
				label/.style={black,left, xshift=-0.3cm}, 
				bluevertex/.style={circle,fill=blue,minimum size=3pt,inner sep=0pt},
				redvertex/.style={circle,fill=red,minimum size=3pt,inner sep=0pt},
				greenvertex/.style={circle,fill=green,minimum size=3pt,inner sep=0pt},
				yellowvertex/.style={circle,fill=yellow,minimum size=3pt,inner sep=0pt}}
			\draw[line] (0.5,0) -- (5.5,0) node [label, at start] {$V_1$};
			\draw[line] (0.5,1) -- (5.5,1) node [label, at start] {$V_2$};
			\draw[line] (0.5,2) -- (5.5,2) node [label, at start] {$V_3$};
			
			\foreach \pos/\name in {{(1,0)/a1},  {(3,0)/a3}, {(4,0)/a4}, {(5,0)/a5},
				{(1,1)/b1}, {(2,1)/b2}, 
				{(3,1)/b3}, 
				{(4,1)/b4}, {(5,1)/b5},
				{(1,2)/c1}, {(2,2)/c2}, {(3,2)/c3}, {(4,2)/c4}, {(5,2)/c5}}
			\node[vertex] (\name) at \pos {};
			\node[bluevertex] (b2) at (2,1) {};
			\node[redvertex] (b5) at (5,1) {};
			\node[greenvertex] (a4) at (4,0) {};
			\node[yellowvertex] (c5) at (5,2) {};
			
			\foreach \source/ \dest in {b1/c3, c3/b3,  a1/b2, a3/b3, a4/b5, b1/c1, b2/a3, b3/a1, c2/b3, c3/a5, c2/b1, a4/b2, c4/b4, c5/b5}
			\draw[edge] (\source) -- (\dest);

			\draw[line] (7.5,0) -- (12.5,0) node [label, at start] {$V_1$};
			\draw[line] (7.5,1) -- (12.5,1) node [label, at start] {$V_2$};
			\draw[line] (7.5,2) -- (12.5,2) node [label, at start] {$V_3$};
			
			\foreach \pos/\name in {{(9,0)/a11},  {(11,0)/a31},  {(8,0)/a41},
				{(13,0)/a51},
				{(9,1)/b11}, {(10,1)/b21},    {(11,1)/b31},  {(12,1)/b41},  {(8,1)/b51},
				{(9,2)/c11}, {(10,2)/c21}, {(11,2)/c31}, {(12,2)/c41}, {(8,2)/c51}}
			\node[vertex] (\name) at \pos {};
			\node[bluevertex] (b21) at (10,1) {};
			\node[redvertex] (b51) at (8,1) {};
			\node[greenvertex] (a41) at (8,0) {};
			\node[yellowvertex] (c51) at (8,2) {};
			\foreach \source/ \dest in {b11/c31, c31/b31,  a11/b21, a31/b31, a41/b51, b11/c11, b21/a31, b31/a11, c21/b31, c31/a51, c21/b11, a41/b21, c41/b41, c51/b51}
			\draw[edge] (\source) -- (\dest);
			
			%
			
		\end{tikzpicture}

		\caption{Example of $3$-layered drawings of the same graph with five and two crossings. } \label{fig:example}
	\end{center}
\end{figure}

In this work, we aim to fully delineate the tractability boundaries of \probHCr for all values of $h$, focusing on two fundamental questions in parameterized complexity: subexponential solvability and polynomial kernelization. We present nearly tight positive and negative results regarding subexponential running times in terms of the number of layers, with the case of four layers being the only remaining open case. For polynomial kernelization, we achieve a complete dichotomy: we provide the first polynomial kernel for three layers and, for all $h\ge4$, we rule out the existence of polynomial kernels unless the polynomial hierarchy collapses. A summary of our results is provided in \Cref{table:ourresults}. Interestingly, in the context of kernelization complexity, the tractability boundary appears to be a bit surprising -- typically the transition from tractability to intractability of problems happens from ``2 to 3'' (as seen in problems like colorability, matching, SAT); in this case, the transition may be better explained by considering the number of edge layers rather than vertex layers.


The simplest version of \probHCr is called {\sc One-Sided Crossing Minimization}. In this problem, the ordering (embedding) of vertices of $V_1$ on the line is fixed, and the task is to draw the vertices of $V_2$ on the line, minimizing the number of crossings. 
The last year's PACE challenge (PACE 2024) was entirely devoted to this problem.\footnote{\url{https://pacechallenge.org/2024/}}
The {\sc One-Sided Crossing Minimization} problem is known to admit a polynomial kernel with $\OO(k^2)$ vertices~\cite{dujmovic2008fixed} and is FPT parameterized by $k$  \cite{dujmovic2008fixed,dujmovic2004efficient,dujmovic2008parameterized}.  Moreover, 
several subexponential-time parameterized algorithms were introduced for this problem \cite{fernau2014social,DBLP:journals/algorithmica/KobayashiT15}. 
 The best running time for this problem is 
  $\OO(k2^{\sqrt{2k}}+ n)$ \cite{DBLP:journals/algorithmica/KobayashiT15} and 
 up to the ETH,  the problem cannot be solved in time $2^{o(\sqrt{k})}\cdot n^{\OO(1)}$.
Solving {\sc One-Sided Crossing Minimization} could be reduced to solving parameterized 
\textsc{Feedback Arc Set in Tournament} (FAST). For FAST, several subexponential algorithms are known 
 \cite{alon2009fast,Feige09,FominP13,KarpinskiS10}. However, we do not see any reduction from  \probHCr to FAST or any other known subexponential-FPT problem.

 The case $h=2$ has its own story. 
 \probTwoCr has a long history, dating back to the 1970s. The problem was introduced by  
 Harary  and Schwenk \cite{harary1971trees} as the bipartite crossing number. Finding a drawing on two lines minimizing the number of crossings is also the key procedure in the Sugiyama algorithm, the well-known layout framework for drawing graphs on layers \cite{sugiyama1981methods}. The problem is long known to be NP-complete~\cite{garey1983crossing}. 
 Kobayashi et al.~\cite{KobayashiMNT14} gave an algorithm of running time $2^{\OO(k\log{k})} +n^{\OO(1)}$. A faster algorithm of running time  $2^{\OO(k)}+n^{\OO(1)}$ was obtained by Kobayashi and Tamaki in \cite{kobayashi2016faster}. Kobayashi et al.~\cite{KobayashiMNT14} also gave
a polynomial kernel with $\OO(k)$ edges.  The starting point of our work on \probHCr is the question initially posed by Kobayashi, Maruta, Nakae, and Tamaki \cite{KobayashiMNT14}. This question was later reiterated by Kobayashi and Tamaki \cite{kobayashi2016faster} and Zehavi \cite{Zehavi22}. They ask whether \probTwoCr admits a subexponential parameterized algorithm, specifically, an algorithm with a running time of $2^{o(k)}\cdot n^{\OO(1)}$. Our first result answers this long-standing open question affirmatively. 
\begin{restatable}{theorem}{thmsubexp}\label{thm:subexp}
On an $n$-vertex graph, \probTwoCr is solvable in time  
 \runtimefortwo.
 \end{restatable}

\begin{table}[t]
	\begin{center}
		\begingroup
		\setlength{\tabcolsep}{7pt} 
		\renewcommand{\arraystretch}{1.5} 
			\begin{tabular}{ |c | c | c |}
				\hline
				\rowcolor{anti-flashwhite}
				\textbf{Number of layers $h$} &  Subexponential-time  & Kernelization   \\\hline\hline
				$h = 2$ & \runtimefortwo (Theorem~\ref{thm:subexp})&  $\Oh(k^2)$  \cite{KobayashiMNT14} \\\hline
				$h = 3$ & \runtimeforthree (Theorem~\ref{thm:subexpTHREE})& \kernelforthree (Theorem~\ref{thm:kernel})
				\\\hline $h=4$  & Open &  No poly kernel (Theorem~\ref{thm:nokern})
				\\\hline $h\geq 5$ & No $2^{o(k/\log k)}\cdot n^{\Oh(1)}$ (Theorem~\ref{thm:ETHlb}) &  No poly kernel  (Theorem~\ref{thm:nokern})
				\\\hline
			\end{tabular}
			\endgroup
		\end{center}
		\caption{Our results}
		\label{table:ourresults}
	\end{table}


In \Cref{sec:twolayer}, we give a short, elegant, and self-contained proof of this result, which already resolves the open question in the literature as mentioned earlier. However, \Cref{thm:subexp}  
naturally raises the question of whether a subexponential $2^{\OO(k^{1-\epsilon})}\cdot n^{\OO(1)}$-time algorithm can be extended to drawings with more layers. Our next two theorems address this question almost completely. They show that for $h=3$, the answer is yes, while for any $h \geq 5$, the answer is no (assuming ETH, the Exponential Time Hypothesis of Impagliazzo, Paturi, and Zane~\cite{ImpagliazzoP99,ImpagliazzoPZ01}). The status of the problem for $h=4$ remains open. More precisely, we prove the following theorems.

\begin{restatable}{theorem}{thmsubexpTHREE}\label{thm:subexpTHREE}
On an $n$-vertex graph, \probThreeCr is solvable in time  
   \runtimeforthree.
\end{restatable}

 \begin{restatable}{theorem}{thmeth}\label{thm:ETHlb}
For any $h\geq 5$, \probHCr cannot be solved by an algorithm running in $2^{o(k/\log k)}\cdot n^{\Oh(1)}$ time unless ETH fails.
\end{restatable}  

Essential components of the proofs of \Cref{thm:subexp,thm:subexpTHREE} are linear time algorithms constructing kernels of polynomial sizes. Informally, these are  algorithms constructing for an $n$-vertex $h$-partite graph $G$ and integer $k$, an equivalent instance on $k^{\cO(1)}$ vertices within $\cO(n)$ time. For \probTwoCr, such kernelization algorithm was given by  Kobayashi et al.~\cite{KobayashiMNT14}. 
However, for \probThreeCr, no polynomial kernel was previously known, making the following theorem of independent interest.

 \begin{restatable}{theorem}{thmkern}\label{thm:kernel}
\probThreeCr admits a kernel with   \kernelforthree vertices and edges. Furthermore, the kernelization algorithm can be implemented to run in $k^{\Oh(1)}\cdot n$ time.
\end{restatable}

Again, \Cref{thm:kernel} 
prompts the question of whether a polynomial kernel can also exist for drawings with more than three layers. 
The following theorem provides a negative answer to this question.

 \begin{restatable}{theorem}{thmnokern}\label{thm:nokern}
For any $h\geq 4$, \probHCr does not admit a polynomial kernel unless $\classNP\subseteq\classCoNP/{\rm poly}$.
\end{restatable} 

\Cref{thm:kernel,thm:nokern}   establish a  dichotomy for polynomial kernelization of  \probHCr  based on the number of layers \(h\). 
We summarize the results of this paper in Table \ref{table:ourresults}. We give a brief summary of our methods in the following subsection, and a detailed technical overview in \Cref{sec:overview}.

 \subsection{Our Methods}
 
At the core of our subexponential-time fixed-parameter algorithms is the divide-and-conquer method. Here, a naive attempt is to base the approach of divide-and-conquer on standard path (or tree) decompositions---indeed,  a graph admitting a  drawing with $k$ crossings can be proved to admit a path decomposition of width $\OO(\sqrt{k})$. However, while we can ``guess'' some {\em separator}  that is a bag in such a decomposition, it is completely unclear how the {\em separation} itself can be concluded from this separator in subexponential time. Indeed, we believe that this is the main reason why the problem of the existence of the subexponential-time fixed-parameter algorithms has remained opened for so long.

All the previous approaches for crossing minimization (for general drawings) bound the treewidth of the graph by using the irrelevant vertex technique, and then use dynamic programming or Courcelle's theorem on the tree decomposition~\cite{Kawarabayashi09}. We completely deviate from this approach. While we still use separators, they satisfy additional properties not required by tree decomposition. Indeed, we enrich the standard definition of a separator by including novel additional properties based on the drawing. Intuitively, this separator consists of four edges (or two in the case of two layers) that are crossed few times, and the edges that cross them. We also demand the separator to be ``balanced'' in some sense (so that the recursion will be efficient), but  this, on its own, is far from sufficient. We must enforce the separator, in order for it to be useful, to satisfy additional drawing-based properties so that, from the separator, we can deduce the separation. In the literature, there are a few examples where the existence of tree decompositions or just separators with certain properties is derived from the drawing, e.g.,~\cite{KleinM14,AshokFKSZ18}. However, typically, such approaches then completely work on the tree decomposition itself, completely ignoring the extra properties. However, as mentioned above, here the extra properties of the separator are critical for us. In the proof, for the sake of simplicity, we do not explicitly describe the tree structure. 

The deduction of the separation itself, particularly in the three-layered case, is a complicated process that requires: {\em (i)} multiple additional guesses, {\em (ii)} a propagation process, where, having guessed some vertices in one part of the separation, we propagate the information to derive additional  vertices in that part, and {\em (iii)} handling various components ``dangling'' from the separator vertices that are not captured by the guesses and the propagation process. Step {\em (iii)} is particularly challenging for the three-layered case, where we sometimes need to redraw the (unknown) optimal solution our separator is based on, based on novel arguments classifying the aforementioned components into types, and making use of knapsack-based arguments so as ensure that the separator remains a (balanced) separator.

 The main idea of the kernelization algorithm (which is a critical component in our subexponential algorithms) is that if $(G,V_1,V_2,V_3,k)$ is a yes-instance of \probThreeCr and $G$ is sufficiently big with respect to $k$, then $G$ should have large pieces that have to be drawn without crossings. Moreover, due to the rigidity of drawings without crossing, the drawing of these pieces is essentially unique (up to the reversals of orders and reordering of twins). Thus, our algorithm applies reduction rules that {\em (i)}~find the parts that are (almost) uniquely embeddable without crossings, and {\em (ii)}~reduce their size.
 
 The starting point of our lower bounds is, in fact, the arguments used for our subexponential and kernelization algorithms. Specifically, the cases where these arguments could not be applied revealed the ``hard instances'' of the problem. These hard instances are encoded using gadgets, and their combination (which yields  our hardness reductions) is quite technical on its own. 
 
We believe that the idea of {\em drawing based separators} is our main conceptual contribution. 

\subsection{Other Related Work}
%
\probHCr belongs to a broader category of graph-level drawing problems, which involve assigning vertices of a (di)graph to lines on a plane according to specific criteria. Historically, level drawings were some of the first styles explored systematically, following the introduction of the Sugiyama framework \cite{sugiyama1981methods}. These drawings are commonly used to visualize hierarchical data, and extensive research covers every aspect of the process, see e.g.,  \cite[Chapter 13]{tamassia2013handbook} and \cite{fulek2013hanani}. 

Besides heuristic and parameterized algorithms, the approximation of \probTwoCr  has been studied for several classes of bipartite graphs, as seen in, for example, \cite{Shahrokhi01}. However, obtaining non-trivial approximation algorithms with guaranteed performance on (general) bipartite graphs remains challenging.

Closely related to \probHCr, is the concept of \textsc{Level Planarity}. 
In level planarity, each vertex of the graph is assigned to a specific level (one of the parallel lines in the plane) based on predefined criteria, and edges (represented by curves) should not cross internally. A significant amount of work is devoted to efficient algorithms recognizing level planarity and some of its generalizations \cite{di1988hierarchies,junger1998level,bruckner2017partial,klemz2019ordered}.

Another close relative of \probHCr is the {\sc Crossing Number}  problem, which involves drawing a graph in a plane while minimizing the number of crossings. {\sc Crossing Number}, parameterized by $k$, was shown to be in  FPT by Grohe~\cite{grohe2004computing}. Intensive research has focused on approximation algorithms for this problem, leading to the work of   Chuzhoy and Tan~\cite{chuzhoy2022subpolynomial} who developed a randomized algorithm with approximation factor $2^{\OO(\log^{\frac{7}{8}}n\log\log n)}\Delta^{\OO(1)}$. For further information on {\sc Crossing Number}  and its variants, we recommend exploring the extensive literature starting with~\cite{schaefer2018crossing,schaefer2012graph,pach2000thirteen}.

One more problem related to  \probHCr  is called {\sc $h$-Layer Planarization}. Here, for a graph $G$ and a non-negative integer $k\in\mathbb{N}_0$, the task is to decide whether there exists a subset of at most $k$ edges that can be removed from $G$ so that it will admit an $h$-layered drawing with no crossings.
In particular, this problem and variants of it have received attention from the perspective of parameterized complexity~\cite{dujmovic2008parameterized,10.1007/s00453-005-1181-y,DBLP:journals/jgaa/Fernau05}.

\section{Technical Overview}\label{sec:overview}

%

\subsection{Overview of the Subexponential Algorithm for $h=2$}

We begin with the description of our subexponential fixed-parameter algorithm for two layers, because the one for three layers is based on all of its arguments as well as additional and much more technical ones. In fact, our algorithm for three layers, for some cases, uses the one for two layers directly as a subroutine. We suppose that the input graph is connected, else each component can be solved independently.

The basic approach that guides the design of our algorithm is that of divide-and-conquer. Specifically, we have a recursive algorithm. Then, at each recursive call, we would like to split the current instance of our problem into two (or more) smaller instances, so that the ``combination'' of their solutions  yields a solution to the current instance. Each of the smaller instances is solved by making one more recursive call. In particular, to obtain the desired subexponential dependency on the parameter in the time complexity, it is crucial to ensure that the parameter of each of the smaller instances is at most a linear fraction of the current parameter---that is, if the current parameter is $k$ and the new parameters are $k_1$ and $k_2$, then we must have that $\max(k_1,k_2)\leq \alpha k$ for some fixed constant $\alpha<1$. It is important to note that we will not produce just a single pair of smaller instances, but a large collection of such pairs, so that if we should find a solution to our current instance, then at least one of them will yield it. Since the number of these pairs will be upper-bounded by $2^{\OO(\sqrt{k}\log k)}$, the recurrence for the time complexity will resolve to the one that we claim.

\medskip
\noindent{\bf The Separator.} Having the above approach in mind, the key player in the proof is that of a {\em separator} for some (unknown)  solution  (that is, a $2$-layered drawing with at most $k$ crossings). For us, the definition of the separator will be {\em geometric}, relying on the drawing (or, equivalently, ordering of vertices). Ideally, we would have wanted the separator to be an edge of $G$ (where $G$ is the graph of the current instance) that satisfies two properties: {\em (i)} the number of edges that cross it at most  $c\sqrt{k}=\OO(\sqrt{k})$ for some constant $c$; {\em (ii)} the numbers of crossings to ``its left'' and to ``its right'' are at most $\alpha k$ for a fixed $\alpha<1$. Clearly, if all edges of $G$ are crossed more than $c\sqrt{k}$ times, then such  an ``elegant'' separator does not exist. However, in this case, we have only $\OO(\sqrt{k})$ many vertices in total, and hence can just solve the instance directly by using brute-force.
Unfortunately, even if this is not the case and, moreover, we have arbitrarily many edges that are not crossed at all, such an elegant separator might provably not exist---to see a simple situation where this is so, we consider precisely the case where all edges of $G$ are crossed more than $c\sqrt{k}$ times, and then we extend $G$ and the drawing by inserting an arbitrarily large drawing of a graph without any crossings to both the left and the right of the existing drawing (while we can keep the graph connected if so desired).

Thus, we work with a somewhat more complicated separator, consisting of two elements that are each either an edge or one of the two ``boundaries of the drawing'', \lsep\ and \rsep, so that the four following properties are satisfied (see Fig.~\ref{fig:2layerseparatorOverview}):\footnote{Figures that are relevant to the technical sections later will be repeated in the appropriate locations for convenience.} {\em (i)} \lsep\ lies to the left of \rsep\ (though they may share an endpoint);  {\em (ii)} the number of edges that cross any of them is upper-bounded by  $\OO(\sqrt{k})$; {\em (iii)} the numbers of crossings to the left of \lsep\ and to the right of \rsep\ are at most $\alpha k$ for some fixed $\alpha<1$ (here, we pick $\alpha=\frac{1}{2}$); (iv) the number of edges that are drawn between $\lsep$ and $\rsep$ is upper-bounded by  $\OO(\sqrt{k})$. 

\begin{figure}
	\centering	\includegraphics[page=3,scale=1.4]{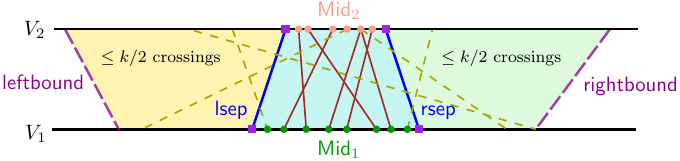}
	\caption{Illustration of a separator. The edges $\lsep$ and $\rsep$ are shown in \blue{blue}, and their endpoints are shown in \textcolor{BlueViolet}{purple}. It is possible that the endpoints on $V_1$ (or $V_2$) are the same. The edges between $\lsep$ and $\rsep$ are shown in \textcolor{Brown}{brown}, and the edges that cross $\lsep$ or $\rsep$ are shown in {\color{Olive}olive}. Vertices between the endpoints of $\lsep$ and $\rsep$ are shown in \textcolor{DarkGreen}{dark green} (denoted $\md_1$ for $V_1$) or \textcolor{LightCoral}{light pink} (denoted $\md_2$  for $V_2$).}
	\label{fig:2layerseparatorOverview}
\end{figure}

\medskip
\noindent{\bf Existence of the Separator.} First, having some (unknown) solution in mind, it is important to observe that among any set of $\sqrt{k}$ many edges, there will exist at least one that is crossed at most $\sqrt{k}$ many times. Hence, the proof that a separator as defined above {\em exists} follows from the following logic. First, we consider the leftmost edge (or the right boundary), denoted $\rsep$, that is crossed by at most $\OO(\sqrt{k})$ many edges and such that there are at least $\alpha k$ many crossings to its left (for some choice of $0<\alpha<1$). Then, neglecting edges that cross $\rsep$, we ``keep going'' to the left of $\rsep$, until we encounter---for the first time---another edge (or the left boundary), denoted $\lsep$, that is crossed at most $\OO(\sqrt{k})$ times.   As will not ``pass over'' more than $\OO(\sqrt{k})$ edges while seeking $\lsep$, we can show that this yields a separator as required.

\medskip
\noindent{\bf Computation of the Separation.} As we do not know a solution (but assume, for now, that one exists), the proof of the existence of a separator is non-algorithmic. Still, we can ``guess'' all of the required information. For this to be efficient, before the recursive process begins, we call a known kernelization algorithm by Kobayashi et al.~\cite{KobayashiMNT14}, so that the given instance can be supposed to have only polynomially in $k$ many vertices and edges. Clearly, we can then guess the separator (consisting of two edges) using just $\binom{m}{2}=k^{\OO(1)}$ many guesses, where $m$ is the number of edges in the (already kernelized) graph. The challenge is in guessing the separation (partition of the vertex set of the graph) that the separator induces into left, middle, and right parts. While in the two-layered case, that is not overly technical, the same will not hold for the three-layered case (discussed in the next overview). There, several new ideas and much more technical work are required on top of those for the two-layered case.

First, notice that the middle part of of the separation can be guessed easily, using only $k^{\OO(\sqrt{k})}$ many guesses, as it contains (by the definition of a separator) only $\OO(\sqrt{k})$ many edges and vertices. Additionally, we can similarly easily guess the edges that cross the two separator edges, and which of their endpoints belongs to which part of the partition. Here, we also guess the {\em ordering} of these endpoints in the (unknown) solution under consideration. This is necessary in order to count, already here, the crossings involving only the edges that cross the separator edges---they cannot be considered when inspecting each part of the partition independently. In turn, this means that the problem that will correspond to each of the parts will be more general than \probTwoCr\ as we will need the sought solutions to respect the guessed ordering. Another issue somewhat of this nature and concerning the crossing edges is that some of them still have to be encoded in the problem that will correspond to each of the parts since other edges of that part may or may not cross them, and we cannot know all information about this at this stage. We briefly remark that, to this end, we replace the crossings edges with some specific weighted (or, equivalently multi-)edges. 
For the simplicity of this overview, we will not continue to describe the way that these technical issues are handled, but proceed to the rest of the computation of the separation, which we view as more central for intuition.

Having guessed to where the endpoints of the crossing edges belong, we proceed to ``propagate'' this information using breadth-first search (BFS). Say, to be specific, that we have a vertex determined to belong to the left part. Then, for every edge incident to it that is not a crossing edge (w.r.t.~the separator edges), we know that the other endpoint must also belong to the left part (assuming that our guess is ``correct''). Afterwards, we can apply the same logic from these neighboring vertices to their own neighbours, and so on. Unfortunately, not all vertices of the graph (even though we suppose it to be connected) are marked left, middle, or right, by all the above guesses as well as the labelling propagation.

\begin{figure}
	\centering	\includegraphics[scale=1]{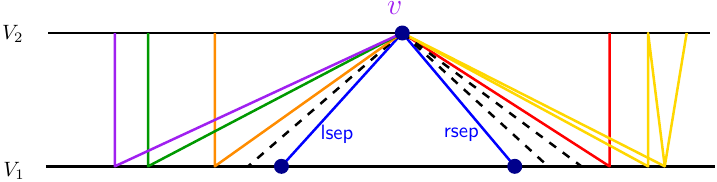}
	\caption{Illustration of the problematic components described in the overview. Here, the endpoint on $V_2$ of the two separator edges (colored blue) is shared. The edges of the three ``leaf components'' are denoted by dashed black lines. Each of the other five problematic components is denoted using a different color: four of them (colored purple, green, orange and red) have a single edge each, and the fifth (colored golden) has three edges---the edges connecting the component to the separator vertex on top are not part of the component, but are drawn using the same color for clarity.}	\label{fig:ProblematicCompsOverview}
\end{figure}

Specifically, consider the graph from which the endpoints of the separator edges are removed. Then, the unmarked vertices are precisely those that belong to the connected components whose only neighbours are the endpoints of the separator edges themselves, which do not cross the separator edges, and which do not belong to the middle part. In case none of the four endpoints of the two separator edges coincides, these unmarked vertices can be directly marked as left or right, as we already know which are all of the middle vertices: so, e.g., if an unmarked vertex belongs to a component having a neighbor that is an endpoint of the left edge of the separator, then we know that it should be marked left. However, if, say, the two endpoints of the two separator edges on $V_1$ coincide, then for a component whose only neighbor is that vertex in $V_1$, we do not know yet whether it should belong to the left part or to the right part (see Fig.~\ref{fig:ProblematicCompsOverview}).

Fortunately, we can show that the number of problematic components as described above, and which have at least one edge, is upper-bounded by $\OO(\sqrt{k})$. So, for these components, we can directly guess whether they go left or right. This leaves us only with the problematic components that consist just of a single vertex (which are, in the entire graph, simply pendants attached to one endpoint of a separator edge). Ideally, we could have just guessed how many these leaves (for each of at most four possible types) go left or right. However, due to some technicalities, we briefly remark that, here, our argument is slightly more technical as these leaves might be attached using weighted edges. For simplicity, we skip the description of this technicality. Overall, this completes the computation of a separation, which, in turn, completes the intuition guiding the design of our algorithm for two-layered drawings.

\subsection{Overview of the Subexponential Algorithm for $h=3$}

Similarly to the algorithm for $h=2$, we employ divide-and-conquer. Clearly, here, we will need to modify the definition of a separator to a more technical one. However, the main difficulty compared to the case of $h=2$ will be in the computation of the separation itself, where after the process of label propagation, much more complicated components will remain unmarked, and handling them will require several new ideas.

\medskip
\noindent{\bf The Separator.} With respect to an (unknown) solution, the separator consists of four elements: two elements between $V_1$ and $V_2$ that are each either a ``non-boundary'' edge or one of the two ``bottom boundaries of the drawing'' (e.g., $\leftbone$ or $\rightbtwo$), denoted by $\lsep_{12}$ and $\rsep_{12}$, and two elements between $V_2$ and $V_3$ that are each either an edge or one of the two ``top boundaries of the drawing'', denoted by $\lsep_{23}$ and $\rsep_{23}$. We will describe the case where we ``split the number of crossings'' between $V_1$ and $V_2$, but, symmetrically, we can use the definition for the case where we ``split the number of crossings'' between $V_2$ and $V_3$.  For simplicity, let us denote the number of crossings that lie between $V_1$ and $V_2$ by $k$, and suppose that the total number of crossings is at most $2k$. The four elements that compose a separator must satisfy the following properties  (see Fig.~\ref{fig:3layerseparatorOverview}):
\begin{enumerate}
	\item  $\lsep_{12}$ ($\lsep_{23}$) lies to the left of $\rsep_{12}$ ($\rsep_{23}$), though they may share an endpoint.
	\item The number of edges that cross any of the separator edges is upper-bounded by  $\OO(\sqrt{k})$.
	\item  The numbers of crossings to the left of $\lsep_{12}$ and to the right of $\rsep_{12}$ are at most $\alpha k$ for some fixed $\alpha<1$ (here, we pick $\alpha=\frac{1}{2}$).
	\item The number of edges  drawn between $\lsep_{12}$ and $\rsep_{12}$ (and which are, therefore, between $V_1$ and $V_2$) is upper-bounded by  $\OO(\sqrt{k})$.
	\item The number of vertices between the endpoints of $\lsep_{12}$ and $\lsep_{23}$ on $V_2$ and which have neighbors on $V_3$ is upper-bounded by $\OO(\sqrt{k})$. Symmetrically, the number of vertices between the endpoints of $\rsep_{12}$ and $\rsep_{23}$ on $V_2$ and which have neighbors on $V_3$ is upper-bounded by $\OO(\sqrt{k})$.
\end{enumerate}

We remark that the vertices whose number is upper-bounded in the last condition do not include all vertices between the respective endpoints, as some such vertices might have neighbors only on $V_1$.

\begin{figure}
	\centering	\includegraphics[page=3,scale=1.25]{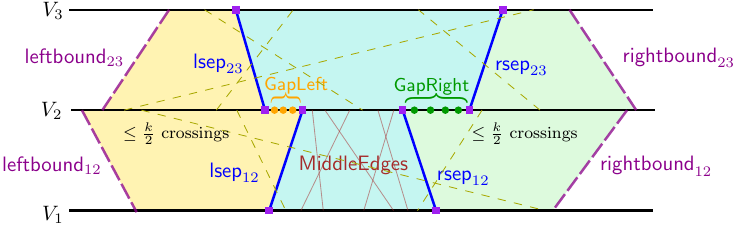}
	\caption{Illustration of a separator. The four edges of the separator are shown in \blue{blue}, and their endpoints are shown in \textcolor{BlueViolet}{purple}. It is possible that the some of the endpoints of the separator edges are the same. The edges between the left are right parts of the separator whose number is guaranteed to be upper-bounded (specifically, here, those between $V_1$ and $V_2$)  are shown in \textcolor{Brown}{brown}, and the edges that cross the separator edges are shown in {\color{Olive}olive}. Vertices on $V_2$ between the endpoints of the two right separator edges and which are incident to edges above them (these edges are not illustrated) are shown in \textcolor{DarkGreen}{dark green} (denoted \gapright), and vertices  on $V_2$ between the endpoints of the two left separator edges and which are incident to edges above them are shown in \textcolor{Orange}{orange} (denoted \gapleft).}
	\label{fig:3layerseparatorOverview}
\end{figure}

\medskip\noindent{\bf Existence of the Separator.} Recall (from the two-layered case) that, having some (unknown) solution in mind, among any set of $\sqrt{k}$ many edges, there will exist at least one that is crossed at most $\sqrt{k}$ many times. So, similarly to the two-layered case, to prove the existence of a separator, we start by considering the leftmost edge (or the right boundary) between $V_1$ and $V_2$, denoted $\rsep_{12}$, that is crossed by at most $\OO(\sqrt{k})$ many edges and such that there are at least $\alpha k$ many crossings to its left (for some choice of $0<\alpha<1$). Then, neglecting edges that cross $\rsep_{12}$, we ``keep going'' to the left of $\rsep_{12}$, until we encounter---for the first time---another edge (or the left boundary), denoted $\lsep_{12}$, that is crossed at most $\OO(\sqrt{k})$ times.  

To obtain the other two elements as required for the three-layered case, we proceed as follows. First, we consider the endpoint of $\lsep_{12}$ that belongs to $V_2$. Then, we start ``going left'' from that endpoint until we encounter a vertex incident to an edge between $V_2$ and $V_3$ that is crossed at most $\OO(\sqrt{k})$ times. That edge is precisely $\lsep_{23}$. In particular, while doing so, we will not ``pass over'' more than $\OO(\sqrt{k})$ many vertices that are neighbors to vertices on $V_3$. Second, symmetrically, we consider the endpoint of $\rsep_{12}$ that belongs to $V_2$. Then, we start ``going right'' from that endpoint until we encounter a vertex incident to an edge between $V_2$ and $V_3$ that is crossed at most $\OO(\sqrt{k})$ times. That edge is precisely $\rsep_{23}$.

\medskip\noindent{\bf Computation of the Separation: The Separator and the Propagated Vertices.} As in the two-layered case, the proof of the existence of a separator is non-algorithmic. Still, we can ``guess'' all of the required information. Again, for efficiency, as a preprocessing step preceding the divide-and-conquer approach, we need to apply a kernelization algorithm. Prior to our work, such an algorithm (for the three-layered case) was not known, and so we had to design it ourselves. This is a critical component of our algorithm.  Since it is a technical algorithm on its own right that is of independent interest, we describe it later in its own overview \Cref{sec:kern-overview}.

Now, having kernelized the instance, we can then directly guess the separator (consisting of four edges) using just $\binom{m}{4}=k^{\OO(1)}$ many guesses, where $m$ is the number of edges in the (already kernelized) graph. Similarly to the two-layered case, we guess all of the edges drawn between $\lsep_{12}$ and $\rsep_{12}$.  Additionally, similarly  to the two-layered case, we guess all the edges that cross the four separator edges, and which of their endpoints belongs to which part of the partition. Here, for the same purposes discussed for the two-layered case, we also guess the {\em ordering} of these endpoints in the (unknown) solution under consideration. This completes the description of our ``obvious'' set of guesses.

Here, unlike the two-layered case, the obvious set of guesses does not immediately yield the entire middle part, as the middle part contains edges between $\lsep_{23}$ and $\rsep_{23}$ as well. Thus, we further guess the vertices in \gapleft\ and \gapright, and the ordering between them (so to enforce that our solutions to the left and middle instances, as well as to our right and middle instances, will be consistent). From here, it is somewhat easy to see that by employing a propagation process similar to the one described for the two-layered case---particularly since we already know all edges in the bottom of the middle part---we can derive the entire middle part. To be more precise, we will not derive those vertices on $V_2$ between the endpoints of the two left separator edges (as well as those between the endpoints of the two right separator edges) that only have neighbors on $V_1$, but these are ``irrelevant'' for the middle part anyway, and can be ``inserted'' when recursing into the left and right parts. In particular, the internal ordering between the aforementioned ``derived'' vertices and the ``underived'' vertices of the middle part in either side is immaterial to the middle part. Since these considerations are  easy when compared to the substantially more complicated upcoming issues, we do not delve into further details here.

Still, before we proceed to analyze the left and right parts, one last important comment concerning the middle part is in place. The middle part might host a huge number of crossings on both its top and bottom, unlike the left and right parts, which can host only some $\alpha k$ crossings each on the bottom (by the definition of a separator). Thus, we cannot simply recursively solve the middle instance, as the parameter might not drop sufficiently to attain our desired time complexity overall. However, since the middle part has just $\OO(\sqrt{k})$ many vertices on $V_1$ as well as only $\OO(\sqrt{k})$ many edges between $V_1$ and $V_2$, we can guess some ordering concerning these $\OO(\sqrt{k})$ many elements so as to reduce the middle instance to the two-layered case. Then, the middle instance is solved using our algorithm for the two-layered case instead of direct recursion.

Next, let us proceed with the analysis of the left and right parts. By employing a propagation process similar to the two-layered case, we can conclude which vertices belong to the left or right parts with the exception of the following. Consider the graph from which we remove the endpoints of the separator edges (and the edges incident to them). Now, further consider the connected components that do not have edges crossing the separator edges as well as do not have vertices belonging the guessed vertices on \gapleft\ and \gapright\ or to the middle part. Then, the vertices of these components are precisely those that we have not yet classified. Still, by using arguments similar to the two-layered case (in particular, to handle components having vertices on $V_1$ or not having vertices on $V_3$), we can further classify some of these vertices, so that eventually we are left with a particular situation (or a situation symmetric to it), described in the next paragraph as ``the problematic situation''.

\medskip\noindent{\bf Computation of the Separation: Handling the Non-Propagated Vertices.} The ``problematic situation'' occurs when $\lsep_{12}$ and $\rsep_{12}$ share the endpoint on $V_1$, denoted by $v$. Then, the components whose vertices are not classified as left or right by any of the arguments considered above (or other arguments along the same lines that we omit from this overview) are the following.  Consider the graph from which we remove the endpoints of the separator edges. Now, further consider the connected components within it, such that the only neighbor from the endpoints of the separator edges that they have is $v$---they can be thought of as ``pendant components hanging from $v$'' (see, e.g., the red and blue components illustrated in Fig.~\ref{fig:3componentcrossingOverview}). Among these, the unclassified components further satisfy the following properties:
\begin{itemize}[noitemsep]
	\item They do not contain any edge crossing the separator edges.
	\item They neither belong to the middle part, nor have vertices on \gapleft\ or \gapright\ (that are adjacent to vertices on $V_3$).
	\item They do not contain vertices on $V_1$ (recall that $v$ cannot belong to the components themselves).
	\item They contain at least one vertex on $V_3$.
\end{itemize}

Now, we further guess the leftmost and rightmost  components among the problematic ones (see Fig.~\ref{fig:3starsOverview}), denoted respectively by $\llc$ and $\rlc$, that do not contain an edge crossed more than $\sqrt{k}$ many times. Having these components at hand, we further guess (using just $k^{\OO(\sqrt{k})}$ many guesses) which of the problematic components: {\em (i)} crosses at least one among $\llc$ or $\rlc$; {\em (ii)} lies to the left of $\llc$, or {\em (iii)} lies to the right of $\rlc$. We classify these components as left or right, and henceforth, exclude them from the set of problematic components under consideration. Furthermore, from now on, when we consider any element (e.g., vertex or edge), we suppose it to lie between $\llc$ or $\rlc$; intuitively, think about the area between $\llc$ or $\rlc$ as our ``workspace'' until the end of this overview.

\begin{figure}
	\centering
	\includegraphics[scale=1,page=6]{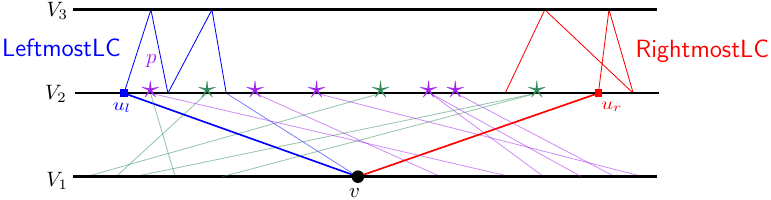}
	\caption{Illustration of components $\llc, \rlc$ and star vertices. Left-star vertices (shown in green) have at least one incident edge crossing $vu_l$, and right-star (purple) have at least one incident edge crossing $vu_r$. Note that a vertex can be both left- and right-star vertex, e.g., $p$.} \label{fig:3starsOverview}
\end{figure}

To handle the (remaining) problematic components, we define a notion of so-called {\em star vertices} (see Fig.~\ref{fig:3starsOverview}). Roughly speaking, a star-vertex is a vertex on $V_2$ that is incident to an edge crossing a bottom edge of $\llc$ (then it is a left-star), or a bottom edge of $\rlc$ (then it is a right-star), or both. A crucial observation in this context is that all bottom edges (of the problematic components) whose endpoints on $V_2$ lie between the same star vertices are crossed the same number of times, and, in fact, by exactly the same edges. Further, we note that there can exist at most $\OO(\sqrt{k})$ star-vertices.

We first handle the problematic components  that ``surround'' at least one star-vertex. We are able to bound their number by $\OO(k^{\frac{2}{3}})$ based on two observations. First, all problematic components surrounding the same star-vertex cross each other (see Fig.~\ref{fig:3componentcrossingOverview}). Second, as we ``go'' further left starting from the rightmost left-star and towards the leftmost left-star (and, symmetrically, w.r.t.~right-stars), the crossings of the components surrounding them (or even just lying in-between them) with the edges incident with these star-vertices that cross the bottom left boundary ``accumulate''. Having bounded the number of these components from above, we simply guess whether they should be marked left or right.

\begin{figure}
	\centering
	\includegraphics[scale=1,page=7]{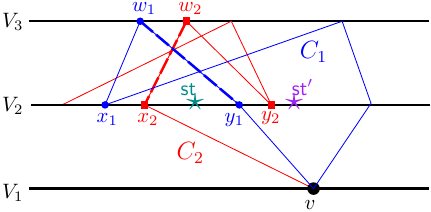}
	\caption{Problematic components $C_1$ (blue) and $C_2$ (red) surround the same star vertex $\str$ ($C_1$ happens to surround $\str'$ as well). In particular, their edges $y_1w_1$ and $x_2w_2$ cross.} \label{fig:3componentcrossingOverview}
\end{figure}

Finally, we are left with the most difficult problematic components: each of them lies entirely strictly in-between some two consecutive star-vertices (where different such components can lie between different consecutive pairs of star-vertices). We classify these components into $\Oh(k^{2/3})$ ``types'', where each type, essentially, encodes {\em (i)} the number of edges between $v$ and the component, and {\em (ii)} the number of edges in the component. We briefly remark that we do not consider all possible such types---for example, the components with ``large'' number of either of edges in {\em (i)} or {\em (ii)} above are bundled together into a single type (the number of such components can be directly upper-bounded). For the types that we do consider, we are able to prove that the (unknown) solution itself can be altered so that components of the same type can be redrawn in any order one next to another between the same two star-vertices. 

Unfortunately, when redrawing the problematic components as described above, a number of technicalities arise. Most importantly, we need to make sure that our separator remains a \emph{balanced} separator for the new drawing---in particular, in this context, we need to ensure that both the left and right parts still each have only some $\alpha k$ many crossings for a fixed $\alpha<1$. To handle this, while performing the ``redrawing'' mentioned in the previous paragraph, for each type, we need to draw a certain subset of components on the left side of the separator, and the remaining number of components on the right side. Note that, it is infeasible to guess all possible left-right partition for each type. Instead, we use knapsack-based arguments---based on the two numbers encoding the type, along with the number of ``optimal number of crossings'' to draw an individual component on its own---in order to limit the number of guesses to be subexponential. In order to keep the overview readable, we will not delve into further details in this context. Overall, this completes the entire classification process and thereby the description of our algorithm for three layers.



Lastly, we remind that for a five (or more)  layers, we cannot obtain an algorithm with running time $2^{\OO(k^{1-\epsilon})}\cdot n^{\OO(1)}$ for any fixed $\epsilon>0$ under the ETH. For intuition on why the approach described in this overview does not extend to yield this (presumably impossible) result, let us now address the two main components of this overview that are specific to three layers, and do not extend to four (or more) layers. First and foremost is the necessity of having a kernel to make all our guesses efficient---recall that even for $h=4$, we prove that a polynomial kernel is unlikely to exist. Second is our arguments for deducing the separation itself---in the case of four layers, much more complicated unclassified (as left or right) components arise, and it is unclear how to extend our arguments (while maintaining efficiency) to encompass them.

\subsection{Kernelization Overview}\label{sec:kern-overview}
The main idea of the kernelization algorithm in~\Cref{{thm:kernel}} is that if $(G,V_1,V_2,V_3,k)$ is a yes-instance of \probThreeCr and $G$ is sufficiently big with respect to $k$ then $G$ should have large pieces that have to be drawn without crossings. Moreover, due to the rigidity of drawings without crossing, the drawing of these pieces is essentially unique (up to the reversals of orders and reordering of twins). Thus, our algorithm applies reduction rules that (i)~find the parts that are (almost) uniquely embeddable without crossings and (ii)~reduce their size.

After the initial preprocessing allowing us to simplify the input instance, we apply crucial reduction rules to deal with pieces of the graph that admit a  drawing without crossings that are connected to the remaining graph via (not necessarily minimal) separators of size 4, 5, and 6. To explain these rules, we consider in more detail the case when we can cut off a subgraph by a separator of size 4.

 \begin{figure}[ht]
\centering
\scalebox{0.7}{
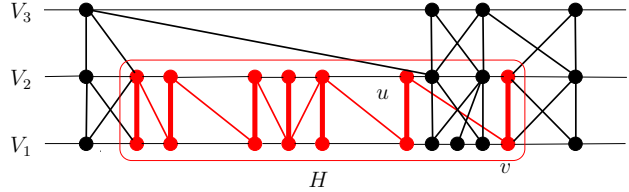}
\caption{Reducing $H$ attached by 4 vertices. The graph $H$ is shown in red and the part outside $H$ is shown in black; the edges of $M$ are highlighted by thick lines.}
\label{fig:overview-four}
\end{figure} 

Suppose that there is a connected induced subgraph $H$ of $G$ (see~\Cref{fig:overview-four}) such that (i)~$V(H)\subseteq V_1\cup V_2$, (ii)~$H$ has a matching $M$ of size $4k+3$, and 
(iii)~$H$ admits a $2$-layer drawing without crossings respecting $(V_1\cap V(H),V_2\cap V(H))$ such that the endpoints of two distinct edges of $M$ separate $H$ from the remaining graph. Then 
because of the big matching, the ordering of the vertices in the middle part of the $2$-layer drawing of $H$ is rigid in the sense  that for any $3$-layer drawing of $G$ with at most $k$ crossings such that the number of crossing is minimum, the vertices of the middle part of $H$ are ordered in the same way as in the $2$-layer drawing without crossings. Still, we may have crossings between the edges of the middle part of $H$ and some other edges of the part of $G$ outside of $H$. However, because $H$ admits a $2$-layer drawing without crossings, $H$ is a caterpillar. Then it could be assumed that all such crossings are confined to a specific single edge $uv$ of $H$ as shown in~\Cref{fig:overview-four}. This allows us to modify the part of $H$ between the three middle edges of $M$ and reduce its size. Notice that by symmetry, we can do the same for $H$ with $V(H)\subseteq V_2\cup V_3$.

 \begin{figure}[ht]
\centering
\scalebox{0.7}{
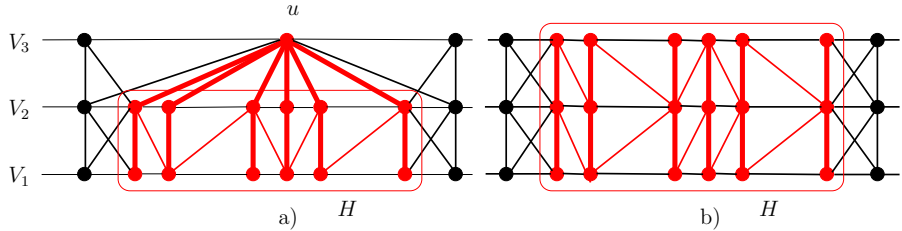}
\caption{Reducing $H$ attached by 5 (a) and 6 (b) vertices. The graph $H$ is shown in red and the part outside $H$ is shown in black; the paths of length $2$ in $H$ are highlighted by thick lines.}
\label{fig:overview-five}
\end{figure} 

For separators of sizes 5 and 6, the idea is the same. Suppose that there is a vertex $u\in V_3$ and a connected induced subgraph $H$ of $G$ (see~\Cref{fig:overview-five}(a)) such that  (i) $V(H)\subseteq V_1\cup V_2$, (ii)~$H$ has a matching $M$ of size $4k+3$ such that endpoints of the edges of $M$ in $V_2$ are adjacent to $u$, and 
(iii)~$H$ admits a $2$-layer drawing without crossings respecting $(V_1\cap V(H),V_2\cap V(H))$ such that the endpoints of two distinct edges of $M$ together with $u$ separate $H$ from the remaining graph. Then the paths formed by $u$ and $M$ together with connectivity of $H$ ensure that the middle part of $H$ should be drawn without crossings in any $3$-layer drawing of $G$. This makes it possible to reduce the middle part. The same arguments work for the symmetric case when $u\in V_1$. For the case of separators of size 6, we are looking for connected induced subgraphs $H$ of $G$ (see~\Cref{fig:overview-five}(b)) such that  (i)~$H$ has a family $\mathcal{P}$ of  $4k+3$ vertex disjoint paths of length two between $V_1$ and $V_3$ and 
(ii)~$H$ admits a $3$-layer drawing  such that the vertices of  two distinct paths  of $\mathcal{P}$  separate $H$ from the remaining graph, and repeat the trick with the middle part.

If after the exhaustive application of the reduction rules the graph remains sufficiently big, we can correctly report that $(G,V_1,V_2,V_3,k)$ is a no-instance of \probThreeCr. Otherwise, we have that the number of vertices and edges of the obtained graph is $\Oh(k^8)$ and the kernelization algorithm returns the graph.

With some extra work, we prove that the kernelization algorithm could be implemented to run $k^{\Oh(1)}\cdot n$ time.  This way, we ensure that the running time of the subexponential algorithm in  
\Cref{thm:subexpTHREE} is linear in the number of vertices. Finally, we remark that our reduction rules are robust in the sense that, given a solution for the instance obtained by the kernelization algorithm, we can lift it to the original instance.

\subsection{Lower Bounds}\label{sec:lb-overview}
We show that the complexity \probHCr is drastically different for $h\leq 3$ and $h>3$. To see the reason, return to the reduction rule of the kernelization algorithm where we consider $H$ attached via 4 vertices. We observed that  edges of the middle part of $H$ may cross some edges that are not edges of $H$ as shown in~\Cref{fig:overview-four}. Still, we can confine all such crossings to a unique edge of $H$. However, our arguments are particular to the case when $H$ admits a 2-layer drawing without crossings, and could not be scaled for three or more layers.   We use this in our reductions showing lower bounds. In particular, we construct two pairs of complementary gadgets $Z$ and $\widehat{Z}$, and $U$ and $\widehat{U}$ depicted in \Cref{fig:Z} and \Cref{fig:U}, respectively. The key property of these gadgets is that if they are forced to cross (and for this, we need additional layers) then they should cross in a very specific way shown in~\Cref{fig:UZ} to minimize the number of crossings.  Using these gadgets, we are able to encode the symbols of an alphabet and reduce the \probDF problem.

Given a string $s$ over an alphabet $\Sigma$, a \emph{factor} is a nontrivial substring of $s$ whose first and last symbols are the same. The task of  \probDF is to decide whether a string $s$ over an alphabet $\Sigma$ with  $k$ symbols has  $k$ disjoint factors with distinct first (and last) symbols. Bodlaender, Thomass{\'{e}}, and Yeo~\cite{BodlaenderTY11} proved that \probDF parameterized by $k$ does not admit a polynomial kernel unless  $\classNP\subseteq\classCoNP/{\rm poly}$. We show our kernelization lower bound in~\Cref{thm:nokern} by providing a polynomial parameter transformation from \probDF.  \Cref{thm:ETHlb} is proved in similar way. The difference is that in our reduction for four layers, the number of crossings is 
$\Oh(k^2\log k)$ and this is not good enough to exclude subexponential algorithms. Thus, we introduce an additional fifth layer and modify the reduction to decrease the number of crossings to $\Oh(k\log k)$. Because the  \classNP-hardness proof for \probDF~\cite{BodlaenderTY11} implies that the existence of an algorithm running in $2^{o(k)}|s|^{\Oh(1)}$ time would violate ETH, we obtain that for any $h\geq 5$, \probHCr cannot be solved by an algorithm running in $2^{o(k/\log k)}\cdot n^{\Oh(1)}$ time unless ETH fails.
	
	\section{Preliminaries and Problem Statement}\label{sec:prelim} 
	This section introduces the basic notation used throughout the paper and provides some auxiliary results.
	
	\paragraph{Graphs.} We use the standard graph-theoretic terminology compatible with the textbook of Diestel~\cite{Diestel12}. 
	We consider only finite undirected graphs. For a graph $G$,  $V(G)$ and $E(G)$ denote its vertex and edge sets. Throughout the paper, we use $n$ to denote the number of vertices if it does not create confusion.  
	For a graph $G$ and a subset $X\subseteq V(G)$ of vertices, we write $G[X]$ to denote the subgraph of $G$ induced by $X$. 
	We denote by  $G-X$ the graph obtained from $G$ by deleting every vertex of $X$ (together with incident edges); we write $G-v$ instead of $G-\{v\}$ for a single vertex set.  
	For $v\in V(G)$, we use $N_G(v)$ to denote the \emph{(open) neighborhood} of $v$, that is, the set of vertices of $G$ that are adjacent to  $v$; $N_G[v]=N_G(v)\cup\{v\}$ is the \emph{closed neighborhood} of $v$. We use $d_G(v)=|N_G(v)|$ to denote the \emph{degree} of $v$. We say that $v$ is a \emph{pendent} vertex if $d_G(v)=1$. 
	For $X\subseteq V(G)$, $N_G(X)=\big(\bigcup_{v\in X}N_G(v)\big)\setminus X$.
	A \emph{separation} of $G$ is a pair of vertex subsets $(A,B)$ such that $G$ has no edges with endpoints in both $A\setminus B$ and $B\setminus A$; $A\cap B$ is a \emph{separator}.
	A vertex $v$ is a \emph{cut vertex} if $G-v$ has more connected components than $G$. We use $P=v_0v_1\cdots  v_{\ell}$ to denote the path with the vertices $\{v_0,\dots,v_\ell\}$ and the edges $v_{i-1}v_i$ for $i\in\{1,\ldots,\ell\}$; $v_0$ and $v_\ell$ are \emph{end-vertices} and we say that $P$ is an $(v_0,v_\ell)$-path.
	
	\paragraph{\probHCr.}
	Let $G$ be an $h$-partite graph with an $h$-partition $(V_1,\ldots,V_h)$ of the vertex set for $h\geq 2$ with the property that for any edge $e\in E(G)$, there is $i\in \{2,\ldots,h\}$ such that $e$ has one endpoint in $V_{i-1}$ and the other in $V_i$. Throughout the paper, whenever we consider an $h$-partite graph with an $h$-partition $(V_1,\ldots,V_h)$, we assume the property that the edges have their endpoints in consecutive sets.
	An \emph{$h$-layer drawing of $G$ respecting $(V_1,\ldots,V_h)$}  is an embedding of $G$ in the plane such that the vertices of $G$ are drawn on $h$ consecutive parallel lines $L_1,\ldots,L_h$,  with the vertices of $V_i$ drawn on $L_i$ for each $i\in\{1,\ldots,h\}$,  and the edges are drawn as straight line segments.   It can be assumed that $L_i$ is defined by the equation $y=i$ in the Euclidean plane for each $i\in\{1,\ldots,h\}$.  Formally, an $h$-layer drawing is an $h$-tuple $\sigma=(\sigma_1,\ldots,\sigma_h)$ of linear orders of $V_1,\ldots,V_h$, respectively. Given such an $h$-tuple, the vertices of each $V_i$ are drawn on $L_i$ according to $\sigma_i$ with strictly increasing $x$-coordinates. 
	Two edges $u_1v_1$ and $u_2v_2$ with $u_1,u_2\in V_{i-1}$ and $v_1,v_2\in V_i$ for some $i\in\{2,\ldots,h\}$ are \emph{crossing} if the segments representing these edges on the plane are intersecting or, equivalently if either $\sigma_{i-1}(u_1)<\sigma_{i-1}(u_2)$ and $\sigma_{i}(v_1)>\sigma_{i}(v_2)$ or, symmetrically,  $\sigma_{i-1}(u_1)>\sigma_{i-1}(u_2)$ and $\sigma_{i}(v_1)<\sigma_{i}(v_2)$. Then, the number of crossings is the number of crossing pairs of edges. We consider the following problem.
	
	\defproblema{\probHCr}{An $h$-partite graph $G$ with an $h$-partition $(V_1,\ldots,V_h)$ for $h\geq 2$ and a positive integer $k$.}{Decide whether there is an $h$-layer drawing 
		$\sigma=(\sigma_1,\ldots,\sigma_h)$ of $G$ respecting $(V_1,\ldots,V_h)$ with at most $k$ crossings.}

	\paragraph{Parameterized Complexity.} We refer to the books of Cygan et al.~\cite{CyganFKLMPPS15} and Fomin et al.~\cite{fomin2019kernelization}
	for an introduction to the area. We remind that a \emph{parameterized problem} is a language $L\subseteq\Sigma^*\times\mathbb{N}$  where $\Sigma^*$ is a set of strings over a finite alphabet $\Sigma$. An input of a parameterized problem is a pair $(x,k)$ where $x$ is a string over $\Sigma$ and $k\in \mathbb{N}$ is a \emph{parameter}. 
	A parameterized problem is \emph{fixed-parameter tractable} (or \classFPT) if it can be solved in time $f(k)\cdot |x|^{\mathcal{O}(1)}$ for some computable function~$f$.  
	The complexity class \classFPT contains all fixed-parameter tractable parameterized problems.
	A \emph{kernelization algorithm} or \emph{kernel} for a parameterized problem $L$ is a polynomial-time algorithm that takes as its input an instance $(x,k)$ of $L$ and returns an instance $(x',k')$ of the same problem such that (i) $(x,k)\in L$ if and only if $(x',k')\in L$ and (ii) $|x'|+k'\leq f(k)$ for some computable function $f\colon \mathbb{N}\rightarrow \mathbb{N}$. The function $f$ is the \emph{size} of the kernel; a kernel is \emph{polynomial} if $f$ is a polynomial. It is common to wright down a kernelization algorithm as a series of \emph{reduction rules}, that is, polynomial-time algorithms that take as the input an instance of a parameterized problem and output a ``smaller'' instance of the problem. A rule is \emph{safe} if it outputs an equivalent instance.

	While every decidable parameterized problem is \classFPT if and only if the problem admits a kernel, all \classFPT problems are unlikely to have polynomial kernels.  In particular, to rule out a polynomial kernel, one can use a \emph{polynomial parameter transformation}. Given two parameterized problems $L$ and $L'$, a polynomial parameter transformation is a polynomial-time algorithm that transforms an instance $(x,k)$ of $L$ into an instance $(x',k')$ of $L'$ such that  (i) $(x,k)\in L$ if and only if $(x',k')\in L'$ and (ii) $k'\leq p(k)$ for a polynomial $p$. Then if $L$ does not admit a polynomial kernel (up to some complexity assumptions), the same holds for $L'$.
	
	The \emph{Exponential Time Hypothesis} (ETH) of Impagliazzo, Paturi, and Zane~\cite{ImpagliazzoP99,ImpagliazzoPZ01} is an important tool for obtaining fine-grained complexity lower bounds for parameterized problems. We remind that ETH is the assumption that $\delta_k>0$ for $k\geq 3$, where $\delta_k$ is the infimum of the values  $\delta$ such that
	$k$\textsc{-SAT} can be solved in $\Oh(2^{\delta n})$ time on formulas with $n$ variables.  In particular, this means that $3$-\textsc{SAT} cannot be solved in $2^{o(n)}$-time. Moreover, by the Sparsification Lemma~\cite{ImpagliazzoPZ01},  $3$-\textsc{SAT} cannot be solved in $2^{o(m)}\cdot (n+m)^{\Oh(1)}$ time on formulas with $n$ variables and 
	$m$ clauses.  
	
	\medskip
	
	%
	%
	
	For $h=2$, testing whether a graph could be drawn without crossing is accomplished by using the following observation, instead of using the heavy machinery of the linear-time FPT algorithm of \cite{dujmovic2008parameterized}. Recall that a \emph{caterpillar} is a tree such that every vertex is adjacent to a vertex of a \emph{central} path (or a \emph{backbone}).
	
	\begin{observation}[\cite{EadesMCW85}]\label{obs:caterpillar}
		A bipartite graph $G$ with a bipartition $(V_1,V_2)$ of the vertex set admits a $2$-layer drawing respecting $(V_1,V_2)$ without any crossings if and only if each connected component of $G$ is a caterpillar.  Moreover, a $2$-layer drawing of a caterpillar is unique in that it reversals the orders and permutations of pendent vertices in the orders, and the vertices of a backbone are ordered in the path order. 
	\end{observation} 
	
	Also, to verify the existence of an $h$-layered drawing without crossings, we can use the algorithm of J{\"{u}}nger,  Leipert, and Mutzel~\cite{JungerLM98}. We are using this result in the following special form.
	
	\begin{proposition}[\cite{JungerLM98}]\label{prop:linear-time} 
		It can be decided in $\Oh(n)$ time for a connected $h$-partite graph $G$ with an $h $-partition $(V_1,\ldots,V_h)$ and $h$ pairs of distinct vertices $s_i,t_i\in V_i$ for $i\in\{1,\ldots,h\}$, whether $G$ admits an $h$-layered drawing $\sigma=(\sigma_1,\ldots,\sigma_h)$ without crossings such that for each $i\in\{1,\ldots,h\}$, $\sigma_i(s_i)\leq \sigma_i(v)\leq \sigma_i(t_i)$ for all $v\in V_i$. Furthermore, if such a drawing exists, it can be found in linear time. 
	\end{proposition} 
	
	\begin{proof}
		By the results of~\cite{JungerLM98}, it can be decided in $\Oh(n)$ time whether $G$ admits an $h$-layered drawing $\sigma=(\sigma_1,\ldots,\sigma_h)$ respecting $(V_1,\ldots,V_h)$ that has no crossings.
		It remains to show that we can check the existence of a drawing with the additional constraints on the first and last vertices in each layer.
		
		Consider a connected $h$-partite graph $G$ with an $h$-partition $(V_1,\ldots,V_h)$ given together with $h$ pairs of distinct vertices $s_i,t_i\in V_i$ for $i\in\{1,\ldots,h\}$. Notice that if there is $i\in\{2,\ldots,h\}$ 
		such that $s_{i-1}$ has a neighbor in $V_i$ distinct from $s_i$ and $s_i$ has a neighbor in $V_{i-1}$ distinct from $s_{i-1}$ then $G$ has no  $h$-layered drawing without crossings such that $s_{i-1}$ and $s_{i}$ are the first vertices in the corresponding orders. Symmetrically, $G$ has no required $h$-layered drawing if $t_{i-1}$ has a neighbor distinct from $t_i$ in $V_{i}$ and $t_i$ is adjacent to a vertex of $V_{i-1}$ distinct from $t_{i-1}$. We can check whether $G$ has such a structure in linear time. Then we stop and return a no-answer. From now on, we assume that this is not the case.
		
		We construct the graph $G'$ from $G$ as follows.
		\begin{itemize}
			\item For each $i\in\{2,\ldots,h\}$, make $s_{i-1}$ adjacent to $s_i$ and make $t_{i-1}$ adjacent to $t_i$ (unless these edges are already in $G$).
			\item Construct $h$ vertices $s_1',\ldots,s_h'$ and create the path $s_1'\cdots s_h'$. Symmetrically, construct $h$ vertices $t_1',\cdots,t_h'$ and construct the path $t_1'\cdots t_h'$. Put $s_i'$ and $t_i'$ in $V_i$ for all $i\in\{1,\ldots,h\}$. 
			\item For every $i\in\{2,\ldots,h\}$, make $s_{i-1}'$ adjacent to $s_{i}$ if $s_{i-1}$ has a neighbor in $V_i$ distinct from $s_i$ and make $s_i'$ adjacent to $s_{i-1}$, otherwise.
			\item For every $i\in\{2,\ldots,h\}$, make $t_{i-1}'$ adjacent to $t_{i}$ if $t_{i-1}$ has a neighbor in $V_i$ distinct from $t_i$ and make $t_i'$ adjacent to $t_{i-1}$, otherwise.
		\end{itemize}
		
		Observe that if $G$ has an $h$-layered drawing $\sigma=(\sigma_1,\ldots,\sigma_h)$ without crossings such that for each $i\in\{1,\ldots,h\}$, $\sigma_i(s_i)\leq \sigma_i(v)\leq \sigma_i(t_i)$ for all $v\in V_i$ then 
		$G'$ admits an $h$-layered drawing that has no crossings. Such a drawing $\sigma'=(\sigma_1',\ldots,\sigma_h')$ can be obtained by constructing $\sigma_i'$ from $\sigma_i$ by adding $s_i'$ in the beginning of the order and making $t_i'$ the last vertex for all $i\in \{1,\ldots,h\}$. We claim that if $G'$ has  an $h$-layered drawing  $\sigma'=(\sigma_1',\ldots,\sigma_h')$ without crossings then $G$ has an $h$-layered drawing $\sigma=(\sigma_1,\ldots,\sigma_h)$ without crossings such that for each $i\in\{1,\ldots,h\}$, $\sigma_i(s_i)\leq \sigma_i(v)\leq \sigma_i(t_i)$ for all $v\in V_i$. Moreover, such a drawing is induced by $\sigma'$.
		
		To see this, notice that $G'$ has paths $P=s_1\cdots s_h$, $P'=s_1'\cdots s_h'$, $Q=t_1\cdots t_h$, and $Q'=t_1'\cdots t_h'$. Because $\sigma'$ has no crossings, we have that $P$ and $P'$ do not cross and either $\sigma_i'(s_i)<\sigma_i'(t_i)$ for every $i\in\{1,\ldots,h\}$ or  $\sigma_i'(t_i)<\sigma_i'(s_i)$ for every $i\in\{1,\ldots,h\}$. By symmetry, we assume without loss of generality that $\sigma_i'(s_i)<\sigma_i'(t_i)$ for every $i\in\{1,\ldots,h\}$. Then because $H$ is connected, we have that for every $i\in\{1,\ldots,h\}$, $\sigma_i'(s_i')<\sigma_i'(s_i)<\sigma_i'(t_i)<\sigma_i'(t_i')$ as $P$, $P'$, $Q$, and $Q'$ have no crossings and $P'$ and $Q'$ do not cross the edges of $G$. Consider $i\in\{2,\ldots,h\}$ and let $v\in V_{i}$ be a neighbor of $s_{i-1}$ in $G$ distinct from $s_i$. Then $\sigma_i(v)>\sigma_{i}(s_i)$ because, otherwise, $s_{i-1}v$ would cross the edge $s_{i-1}'s_i$ of $G'$. Similarly, if $v\in V_{i-1}$ is a neighbor of $s_i$ in $G$ distinct from $s_{i-1}$ then $\sigma_{i-1}(v)>\sigma_{i-1}(s_{i-1})$. We use the same argument for the vertices $t_1.\ldots,t_h$ and obtain that 
		for each $i\in\{1,\ldots,h\}$, $\sigma_i'(s_i')<\sigma_i'(s_i)\leq \sigma_i'(v)\leq \sigma_i'(t_i)<\sigma_i'(t_i')$ for all $v\in V_i$. We define $\sigma_i$ to be the order obtained from $\sigma_i'$ by the deletion of $s_i'$ and $t_i'$ for all $i\in\{1,\ldots,h\}$ and conclude that $\sigma=(\sigma_1,\ldots,\sigma_h)$ is an $h$-layered drawing of $G$ without crossings such that for each $i\in\{1,\ldots,h\}$, $\sigma_i(s_i)\leq \sigma_i(v)\leq \sigma_i(t_i)$ for all $v\in V_i$.
		
		Observe that $G'$ can be constructed from $G$ in linear time. Then in linear time, we can verify whether $G'$ admits an $h$-layered drawing without crossings and find such a drawing
		$\sigma'=(\sigma_1',\ldots,\sigma_h')$ if exists using the results of~\cite{JungerLM98}. Finally, we construct the drawing of $G$ without crossings in linear time. This concludes the proof. 
	\end{proof}
	
	We will also require the following proposition due to Bla\v{z}ej et al. \cite{bla2024constrained}.
	\begin{proposition}[\cite{bla2024constrained}] \label{prop:partialorderdrawing}
		For $h \le 3$, there exists a polynomial-time algorithm, that takes as an input an $h$-layered graph $G = (V, E)$, where $V = \bigcup_{i \in [h]} V_i$ and partial orders $\pi_i$ for $i \in [h]$, and outputs an $h$-layered planar drawing of $G$ that also respects the partial orders $\pi_i$ for each $i \in [h]$; or correctly outputs that there exists no such drawing.  
	\end{proposition}
	
	We conclude this section by discussing the kernelization result of Kobayashi et al.~\cite{KobayashiMNT14} for \probTwoCr showing that the problem admits a kernel with $\Oh(k^2)$ vertices and edges for connected graphs. We note that the kernelization can be done in linear time.
	
	\begin{proposition}[\cite{KobayashiMNT14}]\label{prop:kern-two}
		\probTwoCr admits a kernel with $\Oh(k^2)$ vertices and edges for instances restricted to connected graphs. Moreover, the kernelization algorithm can be implemented to run in $\Oh(n)$ time.
	\end{proposition}
	
	\begin{proof}[Proof of \Cref{prop:kern-two}: Sketch of the running time evaluation] 
		Notice that if $|E(G)|\geq |V(G)|+k$ then $(G,V_1,V_2,k)$ is a trivial no-instance of  \probTwoCr by~\Cref{obs:caterpillar}. 
		Also, if $|V(G)|\leq k^2$, we can return the original instance and stop.  Thus, we can assume that $|E(G)|\leq n+k$ and $n>k$. In particular, this means that any linear in the input size algorithm runs in $\Oh(n)$ time.
		By its first step, the kernelization algorithm of  Kobayashi et al. finds the blocks of $G$, and this can be done in linear time using standard tools~\cite{HopcroftT73} that simultaneously produce the block-cut tree of $G$. Given the set $\mathcal{B}$ of nontrivial blocks, it is shown that
		if $\ell=\frac{1}{3}\sum_{B\in \mathcal{B}}(E(B)-1)>k$ then  $(G,V_1,V_2,k)$ is a no-instance. Therefore, it can be assumed that $\ell\leq k$. Then, the algorithm boils down to contracting certain bridges.
		
		A bridge $e$ is called  \emph{order-inducing} if each of the two connected components of $G-e$ has more than $k-\ell$ edges. Further, 
		an order-inducing bridge $e$ is \emph{contractable} if each endpoint of $e$ (i)~is incident to an order-inducing bridge $e'$ distinct from $e$ 
		and (ii)~is not adjacent to any other non-leaf edge $e''\neq e,e'$ (an edge is a \emph{leaf-edge} if it is incident to a pendent vertex). Then, the algorithm contracts all contractible edges. To see that the algorithm can be implemented to run in linear time, it is sufficient to observe that given the block-cut tree, one can compute all order-inducing bridges and all contractible bridges in time $\Oh(n)$. Finally, we contract these bridges simultaneously.  
	\end{proof}
	

\section{Subexponential Algorithm for \probTwoCr} \label{sec:twolayer}


Recall that an instance of \probTwoCr is given by $(H, U_1, U_2, k^\star)$, where $H = (U_1 \uplus U_2, E)$ is a two-layered (i.e., bipartite) graph. Throughout the algorithm, we will refer to $k^\star$ as the \emph{original} parameter. We will use $n$ and $m$ to denote the number of vertices and edges in the original instance, respectively.\footnote{Note that if we start with the kernelized instance, then $n$ and $m$ are polynomial in the original parameter $k^\star$; however we will only use the kernelization to derive our final theorem.} Before delving into our algorithm, in this subsection we set up some notation and define some notions that will be used to design our subexponential algorithm. We start with the following definition.

\begin{definition} \label{def:2endpts}
	Consider a two-layered graph $(G, V_1, V_2)$.
	For an edge $e$ and $i \in [2]$, we refer its endpoint in $V_i$ by $\ept(e, i)$. The set of both endpoints of $e$ is denoted by $\epts(e)$. We also extend the notations $\ept(\cdot, i)$ and $\epts(\cdot)$ to subsets of edges. 
\end{definition}

Next, we define the notion of extended instance. Here, we will define two kinds of extended instances -- (i) \emph{normal} extended instances, which are the simpler kind of instances encountered while designing a recursive algorithm \probTwoCr; and (ii) \emph{elaborate} extended instances, which have additional constraints that arise when we wish to use \probTwoCr as a subroutine in the recursive algorithm for \probThreeCr. For the most part, normal extended instances are simpler to handle, as compared to the elaborate extended instances, due to the additional constraints on the latter. Next, we define these two instances. Here, the properties that are only applicable for \emph{normal} extended instances are marked with $\diamondsuit$, whereas those that are only applicable for \emph{elaborate} ones are marked with $\clubsuit$. The properties that are common to both are not marked. We refer to Figure~\ref{figure:extendedInstance} for an illustration. 

\begin{definition}[Extended Instance for $h = 2$] \label{def:extendedinstance2}
	An \emph{\ein} of \probTwoCr is given by $(G, V_1, V_2, \bnd, \pi, k)$ where,
	\begin{enumerate}
		\item $G$, with $V(G) = V_1 \uplus V_2$, and $E(G) = \regulare \uplus \weightede \uplus \boundary$. Here, $E(G)$ could be a multiset. 
		\item $\boundary  = \LR{\leftb, \rightb}$ is a set of two  \emph{boundary edges} of the instance. We refer to $\leftb$ and $\rightb$ as the left and right boundary edges of the instance, respectively.
		\item $\diamondsuit$ $G$ is connected. 
		Further, $\bnd = \emptyset$ and it is not relevant.
		\\$\clubsuit$ $\bnd = \lbnd \uplus \rbnd \subseteq V_1$ is a set of at most $3 \sqrt{k^\star}$ vertices, 
		where $\ept(\leftb, 1) \in \lbnd$ and $\ept(\rightb, 1) \in \rbnd$. Further, each vertex $v$ is reachable from some vertex of $\bnd\cup \epts(\boundary)$. \label{item:2extdconn}

		\item $\diamondsuit$ $\pi$ is not relevant.
		\\$\clubsuit$ $\pi = (\pi_1, \pi_2)$, where for each $i \in [2]$, $\pi_i$ is a partial order on $V_i$ such that, $\pi_i$ is the union of $\lambda \ge 0$ different relations $\pi_i^1, \pi_i^2, \ldots, \pi_i^\lambda$, where for each $1 \le j \le \lambda$, $\pi_i^j$ is a total order on $V_i^j \subseteq V_i$, and the sets $\LR{V_i^j: 1 \le j \le \lambda}$ are not necessarily disjoint. Here, $\lambda$ always remains bounded by $4\log(k^\star)$.
		\\Among these partial orders, $\pi_i^1$ is a total order on $\bnd$ where the vertices appear in the following order from left to right: $\ept(\leftb, 1)$, vertices of $\lbnd \setminus \ept(\leftb, 1)$ in some order, vertices of $\rbnd \setminus \ept(\rightb, 1)$ in some order, $\ept(\rightb, 1)$. \label{item:2partialorder}
		
		\item Each $e \in \weightede$ is referred to as a \emph{multi-edge}, and has an associated weight/multiplicity $\mu(e)$ (which is a positive integer). 
		\item Each multi-edge $e \in \weightede$ is incident to exactly one vertex of $\epts(\boundary)$. 
		\\Further, for each vertex $v \in \epts(\boundary)$, the total sum of multiplicities of multi-edges incident to $v$ is at most $2\sqrt{k^\star}$. \label{item:2multiplicity-sum}
	\end{enumerate}
\end{definition}

\begin{figure}
	\centering
	\includegraphics[page=1,scale=1]{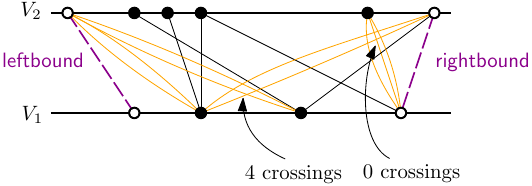}
	\caption{Example of an extended instance of \probTwoCr with a corresponding drawing. $\boundary = \LR{\leftb, \rightb}$ are shown as dashed purple edges. All multi-edges in $\weightede$ are shown in orange, and parallel edges indicate multiplicity. $\regulare$ are shown in black. Crossings are counted with multiplicity in accordance with  \cref{item:drawing} in \Cref{def:2extdrawing}. }
	\label{figure:extendedInstance}
\end{figure}
Next, we define the notion of drawing of (normal/elaborate) extended instances. 
\begin{definition}[Drawing of an extended instance.] \label{def:2extdrawing}
	Let $\cI= (G, V_1, V_2, \bnd, \pi, k)$ be an extended instance of \probTwoCr. We say that $\sigma = (\sigma_1, \sigma_2)$ is a drawing of $\cI$ if the following properties are satisfied for each $i \in [2]$ (see \Cref{figure:extendedInstance}).
	\begin{enumerate}
		\item $\sigma_i$ is a permutation of $V_i$.
		\item For $i \in [2]$, and $u\in V_i$,   $\sigma_i(\ept(\leftb, i)) \le \sigma_i(u) \le \sigma_i(\ept(\rightb, i))$. That is, the endpoints of $\leftb, \rightb$ are placed leftmost and rightmost on the respective layers.
		\item If for some $u, v \in V_i$, if $\pi_i(u) < \pi_i(v)$, then $\sigma_i(u) < \sigma_i(v)$, i.e., $\sigma_i$ is compatible with $\pi_i$.
		\item $\clubsuit$ In the order $\sigma_1$, the vertices of $\lbnd$ must be placed consecutively (i.e., no vertex outside $\lbnd$ can have two vertices of $\lbnd$ placed on its either sides), and analogously for $\rbnd$.\footnote{Note that, this condition, along with the compatibility $\pi_i^1$, implies that a prefix (resp.~suffix) of $\sigma_1$ is given by $\pi_i^1$.}
	\end{enumerate}
	Further, we count the crossings in the drawing $\sigma$ in the following way:
	\begin{enumerate}
		\item Constraints on $\sigma$ imply that no edge of $\boundary$ is crossed by any other edge.
		\item For any pair of edges $e_1, e_2 \in \regulare$, if they cross according to $\sigma$ (in the regular sense), then we count it as a single crossing.
		\item If an edge $e_r \in \regulare$ and an edge $e_m \in \weightede$ with multiplicity $\mu(e_m)$ cross (in the regular sense), then we count it as $\mu(e_m)$ crossings.
		\item Let $e_1, e_2 \in \weightede$ be two multi-edges.
		\begin{itemize}[leftmargin=10pt]
			\item If $e_1, e_2$ are incident to different endpoints of the same edge $e \in \boundary$, then they will necessarily cross in any drawing $\sigma$, that satisfies the constraints above. However we count zero crossings in $\sigma$.\footnote{These crossings will be already accounted for while creating the extended instance.} \label{item:drawing}
			\item Otherwise, if $e_1, e_2$ cross (in the regular sense), then we count as $\mu(e_1) \cdot \mu(e_2)$ crossings.
		\end{itemize}
	\end{enumerate}
\end{definition}


\paragraph{From original instance to extended instances.} 
Consider an original instance $\cI^o= (H, U_1, U_2, k^\star)$ of \probTwoCr. Let us suppose that $H$ is connected -- otherwise we can handle each connected component separately. First, we try each value of $0 \le k \le k^\star$ (these are $k^\star+1$ guesses), such that if $\cI$ is a \yin, then it admits a drawing $\sigma$ with $k$ crossings, but not $k-1$ crossings. Next, we guess the leftmost vertex $u_1^l$ and the rightmost vertex $u_1^r$ in $U_1$, according to some hypothetical unknown drawing $\sigma = (\sigma_1, \sigma_2)$. Note that there are at most $n^2$ choices for these vertices. Then, for a fixed choice of $u_1^l, u_1^r$ we add two new vertices $u_2^l, u_2^r \in U_2$ to the graph, and add two new edges $(u_1^l, u_2^l)$, and $(u_1^r, u_2^r)$. Let $H'$ denote the resulting graph. It is straightforward to check, that assuming that the guess for $u_1^l, u_1^r$ is correct w.r.t. $\sigma$, one can obtain a new drawing $\sigma' = (\sigma'_1, \sigma'_2)$ of the new graph $H'$, where $\sigma'_1 = \sigma_1$, and $\sigma'_2$ places the new vertices $u_2^l, u_2^r$ at leftmost and rightmost positions, respectively. Further, this drawing does not create any additional crossings. We let $\leftb = (u_1^l, u_2^l), \rightb = (u_1^r, u_2^r)$, and $\boundary = \LR{\leftb, \rightb}$. We let $\pi = (\pi_1, \pi_2)$ to be empty partial orders, since they are not relevant in normal instances. All the edges of $E(H') \setminus \boundary$ are $\regulare$. Finally, we define $\bnd = \emptyset$, since again it is irrelevant for normal instances. Note that the graph $H'$ is connected, since $H$ was connected. From the preceding discussion, the following observation is immediate.

\begin{observation} \label{obs:2normal2extended}
	$\cI^o= (H, U_1, U_2, k^\star)$
	is a \yin  of \probTwoCr if and only if at least one of the at most $n^{\Oh(1)}$ extended instances $\cI$ obtained after the  guessing above is an extended \yin. 
\end{observation}
Henceforth, we may drop the qualifier \emph{extended} from \emph{extended instance}, and simply say \emph{instance}. Further, the graph $G$ in an instance $\cI = (G, V_1, V_2, \bnd, \pi, k)$ is said to be the \emph{underlying graph} of the instance $\cI$, and $k$ is said to be its parameter.  


\begin{definition} \label{def:cross}
	Let $\cI$ be an instance with underlying graph $G$ and let $\sigma$ be a drawing of $\cI$. Then, for any edge $e \in E$, we define $\crs(e, \sigma)$ as the set of edges that cross $e$ in the drawing $\sigma$. If the drawing $\sigma$ is clear from the context, then we may drop $\sigma$ from the notation.
\end{definition}

\begin{definition}[Blue and red edges] \label{def:redblue}
	Let $\cI$ be a \yin with underlying graph $G$ and parameter $k$, and let $\sigma$ be a corresponding drawing with at most $k$ crossings. We say that an edge $e \in E(G)$ is a \emph{blue edge} in $\sigma$ if $|\crs(e, \sigma)| \le \sqrt{k}$, and we say that $e$ is a \emph{red edge}, otherwise. We denote the set of blue edges by $\Blue(\sigma)$ and the set of red edges by $\Red(\sigma)$. If the drawing is clear from the context, then we may simply say \emph{blue} and \emph{red} edge. See \Cref{fig:redblue}.
\end{definition}

We have the following observation.
\begin{observation} \label{obs:atleastoneblue}
	Let $\cI$ be a \yin with underlying graph $G$, and consider any drawing $\sigma$ with $k$ crossings. Then, for any $t \ge 0$, any subset of $E(G)$ of size larger than $2\sqrt{k} + t$ (where multi-edges are counted with multiplicities) contains at least $t+1$ blue edges w.r.t. $\sigma$.
\end{observation}
\begin{proof}
	Let $S \subseteq E(G)$ be a subset greater than $2\sqrt{k}+t$, and suppose for contradiction, that $S$ contains at most $t$ blue edges. Therefore, the number of red edges in $S$ is larger than $2\sqrt{k}$. Since each red edge is crossed by more than $\sqrt{k}$ edges in $\sigma$, and every crossing is counted at most twice in this manner, the total number of crossings is strictly larger than $k$, which contradicts $\sigma$ has at most $k$ crossings. Note that this argument is oblivious to the convention adopted in \Cref{def:2extdrawing} regarding the crossing of edges of $\weightede$.
\end{proof}

\begin{figure}[t]
	\centering
	\includegraphics[scale=0.8,page=2]{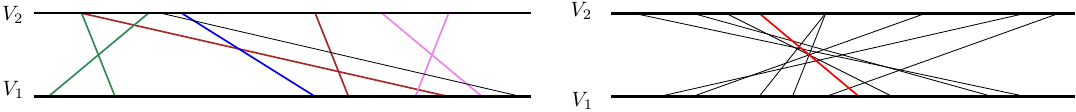}
	\caption{Examples of a blue edge (above) and red edge (below) w.r.t. a drawing from \Cref{def:redblue}. 
		\\This figure also illustrates the notion of left and right from \Cref{def:crossingorder}. The crossing between the two green edges is to the left of the blue edge, and that between two pink edges is to its right. The crossing between the  two brown edges is neither to its left nor to its right, because (at least) one of the two brown edges crosses the blue edge.} \label{fig:redblue}
\end{figure}

\subsection{Algorithm}
We use $n$ and $m$ to denote the number of vertices and edges (with multiplicities), respectively.\footnote{Note that if we start with the kernelized instance, then $n$ and $m$ are polynomial in the original parameter $k^\star$; however we will only use the kernelization to derive our final theorem.} Next, we perform $n^{\Oh(1)}$ guesses in the original instance in order to transform it into an extended instance. This increases the number of vertices and edges by at most $2$, and from \Cref{obs:2normal2extended} it follows that the original instance is a \yin iff for at least one of the guesses, we have an \yin. For each extended instance corresponding to each of the guesses, we will run the following algorithm. 

Let $\cI= (G, V_1, V_2, \bnd, \pi, k)$ be the \ein given as an input. Let $c$ be a sufficiently large constant. 
When $k \le c \sqrt{k^\star}$, our algorithm directly solves the given instance by a careful enumeration. 
Otherwise, when
$k > c \sqrt{k^\star}$, we are in the recursive case, explained later.

\paragraph{Base case.} At a high level, we want to guess the subset of at most $2k$ crossing edges and the ordering of their endpoints, and check whether the rest of the graph (specifically, each connected component in the graph) can be drawn without crossings, such that it respects the given partial orders $\pi$ and other constraints. However, due to the way we count the crossings between the multi-edges incident to the different endpoints of a boundary edge (\Cref{item:drawing} in \Cref{def:2extdrawing}), it may happen that a drawing contains crossings; but we do not want to count such drawings. To handle this subtlety, we proceed as follows.

Suppose $\cI$ is a \yin, then let $\sigma$ be a drawing of $\cI$ with $k$ crossings. Let $X \subseteq E(G)$ be the set of at most $2k$ edges involved in crossings in $\sigma$ -- note that the crossings are counted according to \Cref{def:2extdrawing}. First, we guess $X$, for which there are at most $\binom{m}{k} = n^{\Oh(\sqrt{k^\star})}$ choices, since $k = \Oh(\sqrt{k^\star})$. Let $Y^a \coloneqq \epts(X)$, and note that $|Y^a| = \Oh(\sqrt{k^\star})$. Next, let $Y^b \coloneqq \epts(\weightede)$. Note that due to \Cref{item:2multiplicity-sum} in \Cref{def:extendedinstance2}, it follows that $|Y^b| = \Oh(\sqrt{k^\star})$. Let $Y \coloneqq Y^a \cup Y^b \cup \epts(\boundary)$. Let $Y_i = Y \cap V_i$ for $i \in [2]$. We guess a permutation $\zeta_i$ of $Y_i$ that is compatible with the given ordering $\pi$, and let $\zeta = (\zeta_1, \zeta_2)$. Note that the number of choices for $\zeta$ are bounded by $(\Oh(\sqrt{k^\star}))! = 2^{\Oh(\sqrt{k^\star} \log(k^\star))}$. Henceforth, we focus on the guess where the guessed permutation $\zeta$ is compatible with $\sigma$. 

Then, we use the algorithm of \Cref{prop:partialorderdrawing} 
(\cite{bla2024constrained}) to find, for each connected component $C$ in $G-(X \cup \weightede)$ a drawing $\sigma'(C)$ without any crossings that respects the given partial order $\pi$ restricted to $C$, if such a drawing exists. This can be done in polynomial time due to \Cref{prop:partialorderdrawing}. If such a drawing does not exist for some component $C$, then we terminate this guess. Otherwise, we obtain a combined drawing $\sigma'$ of $G$ by gluing together the orderings of components $C$ and $X$ in an appropriate order, that respects the given partial ordering $\pi$. 

On the other hand, if for all guesses $X$ and the ordering of the  endpoints of edges in $X$, if either the number of crossings is larger than $k$, or if at least one connected component of $G-X$ cannot be drawn without crossings, then we return that $\cI$ is a no-instance. From the above discussion, the following lemma easily follows.

This completes the description of how we handle the first base case when $k < c \sqrt{k^\star}$.

\begin{lemma} \label{lem:2layerbasecase}
	Let $\cI$ be an extended instance of \probTwoCr with underlying graph $G$, and parameter $k  \le c \sqrt{k^\star}$, where $c$ is a sufficiently large constant. Then, the algorithm returns a drawing $\sigma'$ of $\cI$ with $k$ crossings in time $n^{\Oh(\sqrt{k^\star})}$ iff $\cI$ is an extended \yin.
\end{lemma}
Otherwise, we proceed to the recursive case, explained next. 

\paragraph{Recursive case.} Since the base case is not applicable, we know that $k > c \sqrt{k^\star}$.
At a high level, our recursive algorithm consists of the following steps.
\begin{enumerate}[leftmargin=1em]
	\item \textbf{Divide.} First, we prove some combinatorial properties of a ``balanced separator'' in this step, which is defined by a set of two edges, called \lsep, \rsep. This is a balanced separator, in that the number of crossings to the left of \lsep\ and that to the right of \rsep, is at most $k/2$; in addition it has certain additional nice properties. This lets us divide the instance into three parts -- left, middle, and right. Of these, the left and the right parts become sub-instances (in fact, extended instances) with parameters bounded by $k/2$, whereas the middle part has certain special properties, which lets us solve it efficiently. 
	
	\item \textbf{Guess the interfaces.} In this step, we ``guess'' the two separator edges, as well as other auxiliary sets of vertices and edges, whose existence is showed in the previous step. For this, we have $n^{\Oh(\sqrt{k})}$ choices. 
	
	\item \textbf{Creating sub-instances.} For each of the guesses made in the previous step, we formally create the two sub-instances, which takes polynomial time.
	\item \textbf{Conquer left and right.} For each such pair of sub-instances corresponding to $n^{\Oh(\sqrt{k}))}$ guesses, we recursively call the algorithm on the left and right sub-instances. Since the parameter is bounded by $k/2$ in each sub-instance, this inductively takes $$n^{\Oh(\sqrt{k})} \cdot T(k/2) \le n^{\Oh(\sqrt{k}))} \cdot n^{\Oh(\sqrt{k/2})}$$ time.
	\item \textbf{Obtaining a combined drawing.} Assuming that for some guess, both the recursive calls corresponding to the left and right instances output the corresponding drawings, then we combine the drawings---along with the guesses made for the middle subproblem---in order to obtain a combined drawing for the current instance. This takes polynomial time.
\end{enumerate}
Let $T_2(n, k)$ denote an upper bound on the running time of our algorithm on an extended instance on a graph with at most $n$ vertices and parameter bounded by $k$. From the above description, it follows that $T_2(n, k)$ satisfies the following recurrence.\footnote{We will give a formal proof of this after giving a formal description of the algorithm.}
\begin{align}
	T_2(n, k) &= \begin{cases}
		n^{\Oh(\sqrt{k^\star})} &\text{ if } k \le c \sqrt{k^\star} 
		\\n^{\Oh(\sqrt{k})} \cdot T_2(n, \frac{k}{2}) + n^{\Oh(1)} &\text{ if } k > c \sqrt{k^\star} 
	\end{cases}
	\label{eqn:2recurrence}
\end{align}

It can be shown that this recurrence solves to $T_2(n, k^{\star}) = n^{\Oh(\sqrt{k^\star})}$.

\subsection{Step 1: Divide}

\begin{definition}[Ordering between edges] \label{def:edgeorder}
	See \Cref{fig:redblue}. 
	Let $e_a, e_b$ be two distinct edges in $E$. Further, let $\pi = (\pi_1, \pi_2)$, where each $\pi_i$ is a (partial) order on $V_i$, for $i \in [2]$. We say that $e_a$ is \emph{to the left of $e_b$ \wrt $\pi$} if for each $i \in [2]$, we have that $\ept(e_a, i) \le \ept(e_b, i)$, with at least one inequality being strict (since $e_a \neq e_b$). In this case, we also say that $e_b$ is \emph{to the right of} $e_a$.  
\end{definition}
Note that ``left of'' and ``right of'' relations are partial orders on the edges. 
\begin{observation} \label{obs:bluecrossings}
	Let $\tau = (\tau_1, \tau_2)$, where each $\tau_i$ is permutation of $V_i$ for $i \in [2]$. Then, for any pair of distinct edges $e_a, e_b \in E(G)$, either $e_a$ and $e_b$ cross, or one of the two edges is to the left of the other \wrt $\tau$. In particular, for any blue edge $b \in E(G)$, except a subset of at most $\sqrt{k}$ edges of $E$  that cross $b$ in $\pi$, each edge in $E_{12}$ is either to the left of $b$ or to the right of $b$ \wrt $\pi$.
\end{observation}

\begin{definition}[Ordering between an edge and a crossing] \label{def:crossingorder}
	Let $\sigma = (\sigma_1, \sigma_2)$ be a drawing of an instance $\cI$ with an underlying graph $G$. Consider an edge $e \in E(G)$. We say that a crossing between edges $e_a, e_b \in E(G)$ (possibly $e = e_a$ or $e = e_b$ but not both) \emph{is to the left (resp.~right) of $e$ \wrt $\sigma$} if (i) $e_a$ and $e_b$ cross in $\sigma$, and (ii) both $e_a$ and $e_b$ are to the left (resp.~right) of $e$ \wrt $\sigma$. See \Cref{fig:redblue}.
\end{definition}

Note that when $\sigma$ is clear from the context, we will omit ``according to $\sigma$'' when using left/right relations. For the rest of this section, we will fix an instance $\cI$ of \probTwoCr that is an \yin for $k$ crossings, but is a {\sc No}-instance for at most $k-1$ crossings. We refer to this type of instance as a \emph{$k$-minimal \yin}. We now introduce the following crucial definition of a \emph{separator}.  

\begin{definition} \label{def:2separator}
	Let $\cI$ be a $k$-minimal \yin, and let $\sigma$ be any drawing of $\cI$ with $k$ crossings. We say that $\sep = \LR{\lsep, \rsep}$ is a \emph{separator} w.r.t. $\sigma$, if the following properties hold (see \Cref{fig:2layerseparator}).
	\begin{enumerate}
		\item $\lsep, \rsep \in E(G)$ are distinct blue edges \wrt $\sigma$ and they do not cross each other. 
		\item The number of crossings to the left of $\lsep$, and that to the right of $\rsep$ are each at most $\frac{k}{2}$. 
		\item $\diamondsuit$ Define $$\md_i \coloneqq \LR{u \in V_i: \sigma_i(\ept(\lsep, i)) \le \sigma_i(u) \le \sigma_i(\ept(\rsep, i))}.$$ That is, $\md_i$ is the set of vertices that are drawn between the respective endpoints of the two edges $\lsep$ and $\rsep$. 
		\\ For $i \in [2]$, it holds that $|\md_i| < 4\sqrt{k} + 2$.
		\item $|\mde(\sep, \sigma)| \le 2\sqrt{k}$, where
		$$\mde(\sep, \sigma) \coloneqq \LR{uv \in E: u \in \md_1, v \in \md_2} \setminus \LR{\lsep, \rsep}.$$ 
		\item $|\crosse(\sep, \sigma)| \le 2\sqrt{k}$, where $\crosse(\sep, \sigma) \coloneqq \crs(\lsep, \sigma) \cup \crs(\rsep, \sigma)$
	\end{enumerate} 
\end{definition}

\begin{figure}
	\centering
	\includegraphics[page=3,scale=1.4]{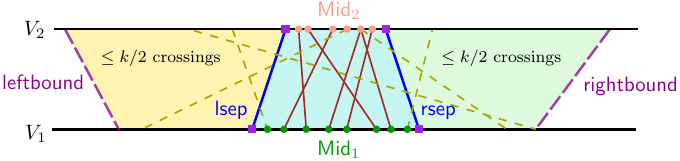}
	\caption{Illustration of a separator from \Cref{def:2separator}. $\sep = \LR{\lsep, \rsep}$ shown in \blue{blue}. Edges of \textcolor{Brown}{$M \coloneqq \mde(\sep)$} are shown in \textcolor{Brown}{brown}, and those of \textcolor{Olive}{$\crosse(\sep)$} shown in {\color{Olive}olive} (these are the edges that cross $\lsep$ or $\rsep$). Vertices of \textcolor{DarkGreen}{$\md_1$} are shown in \textcolor{DarkGreen}{dark green} and those of \textcolor{LightCoral}{$\md_2$} in \textcolor{LightCoral}{light pink}. Vertices of \textcolor{BlueViolet}{$\epts(\sep)$} are shown in \textcolor{BlueViolet}{purple}.}
	\label{fig:2layerseparator}
\end{figure}
\begin{lemma} \label{lem:2separatorlemma}
	Let $\cI$ be a $k$-minimal \yin, and let $\sigma$ be any drawing of $\cI$ with $k$ crossings. Then, there exists a separator $\sep$ \wrt $\sigma$.
\end{lemma}

\begin{proof}
	
	Let us define a lexicographic order $\prec$ on the edges of $E$ using the order of their endpoints in $V_1$ using $\sigma_1$, with ties broken using the endpoints in $V_2$ using $\sigma_2$. More formally, for two distinct edges $e_a, e_b \in E$, we have $e_a \prec e_b$  iff $\sigma_1(\ept(e_a, 1)) < \sigma_1(\ept(e_b, 1))$, or $\sigma_1(\ept(e_a, 1)) = \sigma_1(\ept(e_b, 1))$ and $\sigma_2(\ept(e_a, 2)) < \sigma_2(\ept(e_b, 2))$ (multi-edges are treated as a single edge for this order). Let $\rsep \in E(G)$ be the first blue edge, according to $\prec$, that has more than $\frac{k}{2}$ crossings to its left. First we claim that $\rsep$ exists, since $\rightb$---which is blue, and has all ($k > \frac{k}{2}$) crossings to its left--- satisfies the stated property. 
	
	By the definition of $\rsep$, the number of crossings to its left is larger than $\frac{k}{2}$, which implies the following: (1) the number of crossings to its right is at most $\frac{k}{2}$ (note that some crossings may not be either to the left or to the right of $\rsep$), and (2) the number of edges to its left is larger than $\frac{k}{2}$. Further, (2) implies that $\leftb$ is one of the edges that is to the left of $\rsep$, since $\leftb$ is to the left of every other edge. In particular, this implies that $\rsep \neq \leftb$. See \Cref{fig:2layerseparator} for an illustration.

	Now, we define $\lsep$ to be the last edge according to $\prec$ such that: (i) $\lsep$ is blue, (ii) $\lsep \prec \rsep$, and (iii) $\lsep$ does not cross $\rsep$ according to $\sigma$. First, $\lsep$ exists because $\leftb$ satisfies conditions (i) - (iii) -- indeed $\leftb$ is not equal to $\rsep$ as shown above. Further, since $\lsep \prec \rsep$, $\lsep$ is to the left of $\rsep$. Now, if the number of crossings to the left of $\lsep$ is larger than $\frac{k}{2}$, then this contradicts the choice of $\rsep$ as the lexicographically first blue edge with more than $\frac{k}{2}$ crossings to its left. Therefore, the number of crossings to the left of $\lsep$ is at most $\frac{k}{2}$. This shows that $\sep \coloneqq \LR{\lsep, \rsep}$ satisfies the first two properties from \Cref{def:2separator}. It remains to show that it also satisfies the third property from \Cref{def:2separator}. Before this, we will show that $|\mde(\sep)| \le 2\sqrt{k}$, as claimed in the statement of the lemma.
	
	Let us use $M \coloneqq \mde(\sep)$ for brevity, where the latter is defined in \Cref{def:2separator}. Suppose for the contradiction $|M| > 2\sqrt{k}$. Note that the definition of $M$, for each edge $e \in M$, $\lsep \prec e \prec \rsep$, and $e$ does not cross $\rsep$. Then, by \Cref{obs:atleastoneblue}, $M$ contains at least one blue edge, say $b$. However, note that $b$ satisfies conditions (i)--(iii) above, which contradicts the choice of $\lsep$, since $\lsep$ is the lexicographically last such edge. This shows that $|\mde(\sep)| \le 2\sqrt{k}$.
	
	$\diamondsuit$ Now we will show the bound on the number of vertices in $\md_1, \md_2$. Note that this part is relevant only for normal instances. Note that $|\epts(\sep, 1)| \le 2$ and $\epts(\sep, 1) \subseteq \md_1$, which accounts for the $+2$ factor in the bound. So now, consider any vertex $x \in \md_1 \setminus \epts(\sep, 1)$. Since $G$ is a connected graph, there is at least one edge incident to $x$, say it is $xy$ with $y \in V_2$. Consider one such edge. If $y \in \md_2$, then $xy \in M$, and $|M| \le 2 \sqrt{k}$. Therefore, the number of vertices in $\md_1$ with at least one endpoint in $\md_2$ is bounded by $2 \sqrt{k}$. Otherwise, $y$ satisfies either $\sigma_2(y) < \sigma_2(\ept(\lsep, 2))$ or $\sigma_2(y) > \sigma_2(\ept(\rsep, 2)))$. Therefore, it follows that the edge $xy$ crosses either $\lsep$ or $\rsep$. However, since $\lsep$ and $\rsep$ are blue edges, the number of edges crossing either of them is bounded by $\sqrt{k}$, which implies that the number of vertices in $\md_1$ without a neighbor in $\md_2$ is bounded by $2\sqrt{k}$. Therefore, $|\md_1| \le 2 + 2\sqrt{k} + 2\sqrt{k} = 2 + 4\sqrt{k}$. An analogous argument shows the same bound on $|\md_2|$. 
\end{proof}

We introduce the following definition that will be later useful for proving the correctness of the algorithm.

\begin{definition} \label{def:2division}
	Let $\cI$ be a \yin with underlying graph $G$, and let $\sigma$ be a drawing of $\cI$ with $k$ crossings. Consider a separator $\sep$ w.r.t. $\sigma$ along with the associated objects as in \Cref{def:2separator} (see \Cref{fig:2layerseparator}). Then, we define the following division of vertices on each layer into left (\lf), middle (\mm), and right (\rt) parts, as follows.
	\begin{align*}
		V_1(\sigma, \sep, \lf) &\coloneqq \LR{u \in V_1: \sigma_1(u) \le \sigma_1(\ept(\lsep, 1))}
		\\V_1(\sigma, \sep, \mm) &\coloneqq \LR{u \in V_1 : \sigma_1(\ept(\lsep, 1) \le \sigma_1(u) \le \sigma_1(\ept(\rsep, 1))}
		\\V_1(\sigma, \sep, \rt) &\coloneqq \LR{u \in V_1: \sigma_1(\ept(\rsep, 1)) \le \sigma_1(u)}
		\\V_2(\sigma,\sep, \lf) &\coloneqq \LR{u \in V_2: \sigma_2(u) \le \sigma_2(\ept(\lsep, 2))} 
		\\V_2(\sigma, \sep, \mm) &\coloneqq \LR{u \in V_2: \sigma_2(\ept(\lsep, 2)) \le \sigma_2(u) \le \sigma_2(\ept(\rsep, 2))}
		\\V_2(\sigma, \sep, \rt) &\coloneqq \LR{u \in V_2: \sigma_2(\ept(\rsep, 2)) \le \sigma_2(u)} 
	\end{align*}
	Finally, for each ${\cal T} \in \LR{\lf, \mm, \rt}$, define $V(\sigma,\sep, {\cal T}) \coloneqq V_1(\sigma, \sep,{\cal T}) \cup V_2(\sigma,\sep, {\cal T})$.
	\\Specifically, in \Cref{fig:2layerseparator}, these sets correspond to the vertices belonging to the yellow, cyan, and light green shaded regions (along with the shared vertices of $\epts(\sep)$), respectively.
\end{definition}

Henceforth, until the beginning of \Cref{subsubsec:2SBcase}, we fix a {\em hypothetical} drawing $\sigma=(\sigma_1,\sigma_2)$ and the corresponding separator $\sep$ corresponding to a $k$-minimal \yin $\cI$, as guaranteed by \Cref{lem:2separatorlemma}, and the associated objects with $\sep$ from \Cref{def:2separator}. These qualifiers/properties of $\sigma$ and $\sep$ will not be repeated in the definitions and lemma statements that follow, until the beginning of \Cref{subsubsec:2SBcase}.

In the following definition, we classify the components attached to a vertex in $\epts(\sep)$ into multiple classes. Components of each class satisfy certain properties, which will be subsequently used to derive more information on the number of such components, or the way in which they must be drawn.

\begin{definition} \label{def:2componenttypes}
	Let $\cI= (G, V_1, V_2, \bnd, \pi, k)$ be an extended instance of \probTwoCr and $\sigma$ be a hypothetical drawing. Furthermore, let $\sep$ be the separator corresponding to $\cI$. 
	Let $v \in \epts(\sep)$ be a vertex, and let $\comps_G(v)$ denote the set of connected components of $C_v-v$, where $C_v$ is the connected component of $G$ that contains $v$. We partition these components into multiple classes (subsets), called $\spl(v), \predef(v), \pendant(v)$, and $\updown(v)$ (we may omit the subscript $G$, if it is clear from the context). A component $C \in \comps_G(v)$ belongs to:
	\begin{itemize}[leftmargin=*]
		\item $\spl(v)$, if it contains at least one edge of $\crosse(\sep) \cup \mde(\sep)$, or a vertex of $\epts(\sep)$.
		\item $\predef(v)$, if $C \not\in \spl(v)$, and $C$ contains a vertex $u \in \bigcup_{i \in [2]} \bigcup_{j \in [\lambda]} V^j_i$ \footnote{Recall that $V^j_i$'s are the $\lambda = \Oh(\log k^\star)$ subsets over which the total orders $\pi^j_i$ are defined, and these $\lambda$ total orders constitute $\pi_i$. Furthermore, $ \bigcup_{i \in [2]} \bigcup_{j \in [\lambda]} V^j_i$ may not be equal to $V(G)$}
		\item $\pendant(v)$, if it does not belong to $\spl(v)$ or $\predef(v)$, and $C$ consists of only a single \emph{pendant vertex} $u$ attached to $v$. Note that the edge $uv$ could belong to $\weightede$.
		\item $\updown(v)$, if it does not belong to $\spl(v)$ or $\predef(v)$, and contains vertices of both $V_1$ and $V_2$.
	\end{itemize}
	See \Cref{fig:compclass} for an illustration.
\end{definition}

\begin{figure}
	\centering
	\includegraphics[scale=1,page=4]{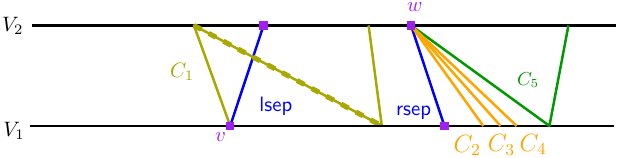}
	\caption{Illustration of classification of components from \Cref{def:2componenttypes}. \textcolor{BlueViolet}{$v, w$} are two vertices from $\epts(\sep)$. \textcolor{Olive}{$C_1$} $\in \comps(v)$ belongs to $\spl(v)$ since it contains an edge crossing $\lsep$. Three components \textcolor{DarkOrange}{$C_2, C_3, C_4$} $\in \comps(w)$ consist of pendant vertices attached to $w$. Finally, \textcolor{DarkGreen}{$C_5$} $\in \comps(w)$ belongs to $\updown(w)$ since it contains a vertex in the same layer as $w$. } 
	\label{fig:compclass}
\end{figure}

First, we have the following observation.

\begin{observation} \label{obs:2layerclass1}
	For any $v \in \epts(\sep)$, it holds that $|\spl(v)| \le |\crosse(\sep)| + |\mde(\sep)| + 2 \le 4 \sqrt{k} +2$.
\end{observation}

Next, we bound the number of components in $\updown(v)$ in the following lemma.

\begin{lemma} \label{lem:2type3components}
	For any $v \in \epts(\sep)$, it holds that $|\updown(v)| \le 4\sqrt{k}$.
\end{lemma}
\begin{proof}
	Fix any $v \in \epts(\sep)$. Suppose w.l.o.g. that $v \in V_1$ since $V_2$ case is symmetric. Suppose for the contradiction that $|\updown(v)| > 4\sqrt{k}$. 
	
	Consider some component $C \in \updown(v)$. Note that $C$ must contain a vertex, say $\upr(C) \in V_2$, that is adjacent to $v$, as well as a vertex $\low(C)$, such that $\low(C) \in V_1$ and is adjacent to $\upr(C)$ (since $C$ does not belong to $\pendant(v)$). Define $\upr(C)$ and $\low(C)$ for each $C \in \updown(v)$ in this manner (in case of multiple choices for such vertices, select arbitrarily).
	
	For any $C \in \updown(v)$, since $\low(C) \in V_1$ and $v \in V_1$, either $\low(C)$ is placed to the left of $v$ or to the right of $v$ according to $\sigma_1$. By pigeonhole principle, either the number of $\low(C)$'s to the left or right of $u$ is larger than $2\sqrt{k}$. W.l.o.g., suppose the number of components such that $\low(C)$ is to the right of $v$ is $r > 2\sqrt{k}$. Let us order the vertices $\low(C)$ corresponding to such components, as $\low_1, \low_2, \ldots, \low_r$ according to their left-to-right order in $\sigma_1$, and let $\upr_1, \upr_2, \ldots, \upr_r$ be the corresponding $\upr(C)$ vertices (note that this need not be their left-to-right order according to $\sigma_2$). Note that $\low_j$'s and $\upr_j$'s are distinct, since they belong to different connected components in $G-v$. We will show that the number of crossings between the edges $E_r \coloneqq \LR{(v, \upr_i), (\low_i,\upr_i): 1 \le i \le r}$ is at least $\frac{r(r-1)}{2} \ge \frac{(r-1)^2}{2} > k$, which will yield the desired contradiction.
	
	We prove the claim using induction, where we show that for any $i \ge 1$, the number of crossings among the edges of $E_i \coloneqq \LR{(v, \upr_j), (\low_j, \upr_j): 1 \le j \le i}$ is at least $\frac{i(i-1)}{2}$. The base case for $i = 1$ is trivial, since there are exactly $0$ crossings among the edges of $E_1$ -- since both the edges in $E_1$ are incident to $\upr_1$. Now, suppose the claim is true for some $i \ge 1$. Let $U_i \coloneqq \LR{\upr_j: 1 \le j \le i}$, and let $U_i^l, U^r_i$ denote the set of vertices from $U$ that are to the left and right of $\upr_{i+1}$, respectively. Notice that, for each $\upr_j \in U^l_i$, $\sigma_2(\upr_j) < \sigma_2(\upr_{i+1})$ and $\sigma_1(v) < \sigma_1(\upr_j)$, hence the edge $(v, \upr_{i+1})$ crosses the edge $(\low_j,\upr_j)$ for each $\upr_j \in U^l_i$. On the other hand, for each $\upr_j \in U^r_i$, $\sigma_1(v) < \sigma_1(\low_j) < \sigma_1(\low_{i+1})$ and $\sigma_2(\upr_{i+1}) < \sigma_2(\upr_j)$, hence the edge $(\low_{i+1}, \upr_{i+1})$ crosses both $(v, \upr_j)$ as well $(\low_j,\upr_j)$, for each $\upr_j \in U^r_i$. Therefore, the number of crossings, where one edge is from $E_{i}$ and the other edge is from $\LR{(v,\upr_{i+1}), (\upr_{i+1},\low_{i+1})}$ is at least $l + 2r \ge l + r = i$ (see \Cref{fig:2updown}). Then, by using the inductive hypothesis, the number of crossings among the edges of $E_{i+1}$ is at least $\frac{i(i-1)}{2} + i = \frac{i(i+1)}{2}$. This finishes the proof.
	
	\begin{figure}
		\centering
		\includegraphics[scale=1,page=6]{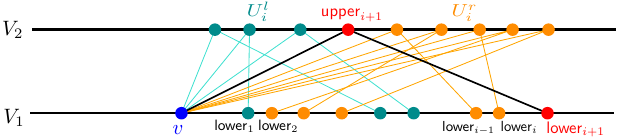}
		\caption{Illustration for the inductive proof in \Cref{lem:2type3components}. Here, the vertices of $U^l_i$---which are to the left of ${\sf upper}^{i+1}$---are shown in \textcolor{Teal}{teal} and that of $U_i^r$---which are to the right of ${\sf upper}^{i+1}$---are shown in \textcolor{DarkOrange}{orange} color.} \label{fig:2updown}
	\end{figure}
	
\end{proof}

At this point, we are left with analyzing the behavior of components in $\pendant(v)$. This depends on the structure of the separator $\sep(\sigma)$ guaranteed by \Cref{lem:2separatorlemma}. In the following definition, we study three different cases.

\begin{figure}
	\centering
	\includegraphics[scale=1,page=13]{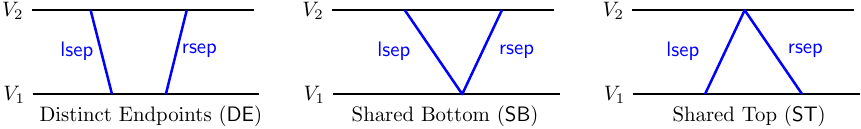}
	\caption{Illustration of three different cases from \Cref{def:2layerendptcases}} \label{fig:2layerendptcases}
\end{figure}

\begin{definition} \label{def:2layerendptcases}
	Let $\sep(\sigma) = \LR{\lsep, \rsep}$ be the separator as guaranteed by \Cref{lem:2separatorlemma}. We have the following
	two different (non-exclusive) cases.
	\begin{description}[leftmargin=0pt]
		\item[(i) {\sf Distinct Top (DT)}.] When $\ept(\lsep, 2) \neq \ept(\rsep, 2)$.
		\item[(ii) {\sf Distinct Bottom (DB)}.] When $\ept(\lsep, 1) \neq \ept(\rsep, 1)$.
	\end{description}
	Based on this, we now have three different mutually exclusive cases (see \Cref{fig:2layerendptcases}).
	\begin{description}
		\item[1. {\sf Distinct Endpoints (DE) case}.] When both {\sf DT} and {\sf DB} hold.
		\item[2. {\sf Shared Bottom (SB) case}.] When {\sf DT} holds but {\sf DB} does not. That is, $\ept(\lsep, 2) \neq \ept(\rsep, 2)$ and $\ept(\lsep, 1) = \ept(\rsep, 1)$. 
		\item[3. {\sf Shared Top (ST) case}.] When {\sf DB} holds but {\sf DT} does not. That is, $\ept(\lsep, 1) \neq \ept(\rsep, 1)$ and $\ept(\lsep, 2) = \ept(\rsep, 2)$. 
	\end{description} 
\end{definition}

\begin{lemma} \label{lem:pendantcases}
	The following properties hold in the drawing $\sigma$. 
	\begin{enumerate}
		\item In {\sf DT case}, 
		\begin{enumerate}
			\item All the vertices of $\pendant(\ept(\lsep, 2))$  are placed to the left of $\ept(\lsep, 1)$.
			\item All the vertices of $\pendant(\ept(\rsep, 2))$ are placed to the right of $\ept(\rsep, 1)$.
		\end{enumerate}
		\item In {\sf DB case},
		\begin{enumerate}[start=3]
			\item All the vertices of $\pendant(\ept(\lsep, 1))$  are placed to the left of $\ept(\lsep, 2)$.
			\item All the vertices of $\pendant(\ept(\rsep, 1))$ are placed to the right of $\ept(\rsep, 2)$.
		\end{enumerate}
	\end{enumerate}
	Thus, in {\sf DE case}, all (a), (b), (c), (d) hold. In {\sf SB case}, (a), and (b) hold. In {\sf ST case}, (c) and (d) hold.
\end{lemma}
\begin{proof}
	
	First let us consider {\sf DT case}. Suppose that there exists a pendant vertex $w \in \pendant(u_2)$ attached to $u_2 \coloneqq \ept(\lsep, 2))$ that lies to the right of $u_1 \coloneqq \ept(\lsep, 1)$. That is, $\sigma_1(u_1) < \sigma_1(w)$. Let $u'_1 \coloneqq \ept(\rsep, 1)$, and $u'_2 \coloneqq \ept(\rsep, 2)$. 
	Now, either (i) $\sigma_1(w) < \sigma_1(u'_1)$, or (ii) $\sigma_1(u'_1) < \sigma_1(w)$. In case (i), $w \in \md_1$ and thus $u_2w \in \mde(\sep)$, contradicting that $w \not\in \spl(u_2)$. In case (ii), the edge $u_2w$ crosses $\rsep = (u'_1,u'_2)$, since $\sigma_1(u_1) < \sigma_1(u'_1)$ and $\sigma_1(u'_1) < \sigma_1(w)$. This implies that $u_2w \in \crosse(\sep)$, again contradicting that $w \not\in \spl(u_2)$. Similarly, one can prove (b) in {\sf DT case}, and (c) and (d) in {\sf DB} case.
	
	Now in {\sf DE case}, both {\sf DT} and {\sf DB} hold, which implies properties (a)--(d). In {\sf SB} case, only {\sf DT} holds, which implies properties (a) and (b). Finally, in {\sf ST case}, only {\sf DB} holds, which implies properties (c) and (d).
\end{proof}

Thus, we are left with analyzing the behavior of some of the pendant components in {\sf SB} and {\sf ST} cases. Specifically, in {\sf SB} case, when $\ept(\lsep, 1) = \ept(\rsep, 1)$, we need to analyze which components in $\pendant(\ept(\lsep, 1))$ (which is the same as $\pendant(\ept(\rsep, 1))$) are placed to the left of $\ept(\lsep, 2)$, and which of them are placed to the right of $\ept(\rsep, 2)$. Note that there cannot be any that are placed in between the two vertices, since otherwise they would belong to $\spl(\ept(\lsep, 1))$. Analogously, in {\sf ST} case, when $\ept(\lsep, 2) = \ept(\rsep, 2)$, we need to analyze which components in $\pendant(\ept(\lsep, 2))$ (which is the same as $\pendant(\ept(\rsep, 2)$) are placed to the left of $\ept(\lsep, 1)$, and which of them are placed to the right of $\ept(\rsep, 1)$. In the next subsection, we will focus on {\sf SB} case, and the {\sf ST} case will be handled analogously. 
These cases are not equivalent due to weight, and, also, we need to take care of the placement of these pendant components to preserve the number of crossings on both sides.  In short, we are going to show that all that matters is the number of these components on both sides and not the components itself. 

\subsubsection{Handling {\sf SB case}} \label{subsubsec:2SBcase}
In this subsection, we will make some simplifying assumptions. First, we are in {\sf SB case}, when $\ept(\lsep, 1) = \ept(\rsep, 1)$, and for simplicity, we denote this shared endpoint by $v$. Similarly, let $u^l \coloneqq \ept(\lsep, 2)$ and $u^r \coloneqq \ept(\rsep, 2)$. If a component $C \in \pendant(v)$, then it contains a pendant vertex $w$ attached to $v$. In this case, we will use the vertex $w$ and the component $C$ interchangeably. Note here, that the edge $vw$ in a pendant component may be a multi-edge with weight/multiplicity $\ell \ge 0$, in which case we also say that $\ell$ is the multiplicity of $C$ or $w$. If $vw$ is not a multi-edge, then we define the multiplicity to be $1$. 

In the following definition, we introduce the notion of nice families of drawings.

\begin{definition} \label{def:2nicefamily} Let the shared endpoint be $v$ and  let $u^l \coloneqq \ept(\lsep, 2)$ and $u^r \coloneqq \ept(\rsep, 2)$.
	\begin{itemize}[leftmargin=3pt]
		\setlength{\itemsep}{-1pt}
		\item 	A \emph{pendant-blueprint} is a tuple $\pbp = (\pendantsl, \pendantsr)$, where $\pendantsl + \pendantsr = \pendantst$, where $\pendantst$ denotes the total sum of multiplicities of all pendants in $\pendant(v)$. 
		
		\item A \emph{pendant-guess} is a partition $\pguess = (\pendantl, \pendantr)$ of $\pendant(v)$. That is, the partition of components of 
		$\pendant(v)$ into left and right side, respectively. See \Cref{fig:2pendantsillustration} for an illustration.
		\item Let $\pguess = (\pendantl, \pendantr)$ be a pendant guess, and let $\pbp = (\pendantl, \pendantr)$ be a pendant blueprint. We say that $\pguess$ is \emph{$\pbp$-compliant} if it satisfies that, the total sum of multiplicities of vertices in $\pendantl$ is equal to $\pendantsl$, and the total sum of multiplicities of vertices in $\pendantr$ is equal to $\pendantsr$.
		
		\item Let $\pbp = (\pendantsl, \pendantsr)$ be a pendant-blueprint, and let $\pguess = (\pendantl, \pendantr)$ be a $\pbp$-compliant pendant-guess. Then, we define a $\nicefamily(\pbp, \pguess, v, \sigma)$ to be a family of drawings $\sigma'=(\sigma_1',\sigma_2')$ satisfying the following properties.
		\begin{enumerate}
			\item $\sigma_1(w) = \sigma'_1(w)$ for all $w \in V_1$, i.e., the order of the vertices on the first layer in the drawing $\sigma'$ is exactly the same as that of $\sigma$,
			\item $\sigma'_2$ and $\sigma_2$, when restricted to the vertices of $V_2 \setminus \pendant(v)$ are the same,
			\item All the pendants in $\pendantl(v)$ are placed to the left of $u^l$ in $\sigma'_2$.
			\item All the pendants in $\pendantr(v)$ are placed to the right of $u^r$ in $\sigma'_2$.
		\end{enumerate}
	\end{itemize}
	%
\end{definition}

\begin{figure}
	\centering
	\includegraphics[scale=1,page=14]{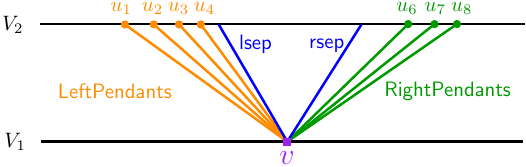}
	\caption{A separator in {\sf SB case}, where $v = \ept(\lsep, 1) = \ept(\rsep, 1)$. A partition $(\pendantl, \pendantr)$ of $\pendant(v)$, which can be thought of as a \pguess. It is compliant with $\pbp = (4, 3)$, assuming that the multiplicity of each edge is $1$. (Note that the other two endpoints of $\lsep, \rsep$ are not in $\pendant(v)$, since the corresponding components belong to $\spl(v)$ according to \Cref{def:2componenttypes}.)} \label{fig:2pendantsillustration}
\end{figure}

We also have the following definition.
\begin{definition} \label{def:compatibledrawing}
	Let $\cI$ be a $k$-minimal \yin and $\sigma, \sigma'$ be drawings with $k$ crossings. Let $\sep$ be a separator w.r.t. $\sigma$. We say that $\sigma'$ is $(\sep, \sigma)$-compatible, if it satisfies the following properties.
	\begin{itemize}
		\item $\sep$ is also a separator w.r.t. $\sigma'$.
		\item $\mde(\sep, \sigma) = \mde(\sep, \sigma')$. 
		\item $\crosse(\sep, \sigma) = \crosse(\sep, \sigma')$.
	\end{itemize}
\end{definition}

Next, we state the following  lemma. 
\begin{lemma} \label{lem:2permutation}
	Let $\cI$ be a $k$-minimal \yin, let $\sigma$ be a drawing with $k$ crossings, and let $\sep$ be a separator w.r.t. $\sigma$ with {\sf SB case}. Let $v = \ept(\lsep, 1) = \ept(\rsep, 1)$.  
	
	Then, there exists a $\pbp$ such that, for any $\pguess$ that is $\pbp$-compliant, there exists a drawing $\sigma' \in \nicefamily(\pbp, \pguess, v, \sigma)$ such that $\sigma'$ is $(\sep, \sigma)$-compatible.
	
\end{lemma}
\paragraph{Intuitive explanation of the statement of the lemma.} Before proving the lemma, let us illustrate what the lemma is trying to say with the help of an example of \Cref{fig:2pendantsillustration}. Suppose the drawing $\sigma$ distributes $\pendant(v)$ into left and right as illustrated in the example. As mentioned in the caption, $\pguess = (\LR{v_1, v_2, v_3, v_4}, \LR{v_5, v_6, v_7})$ is $\pbp$-compliant, where $\pbp = (4, 3)$. 
Now consider another $\pguess' = (\LR{v_2, v_3, v_6, v_7},$ $ \LR{v_1, v_4, v_5})$, which is also $\pbp$-compliant. Then, there exists a drawing $\sigma'$ such that, (i) $\sigma'$ places $\LR{v_2, v_3, v_6, v_7}$ on the left of $\ept(\lsep, 2)$, $\LR{v_1, v_4, v_5}$ to the right of $\ept(\rsep, 2)$, and does not change the ordering of the rest of the vertices (this is captured by the condition that $\sigma' \in \nicefamily(\pbp, \pguess', v, \sigma)$); and (ii) $\sep$ and its associated objects remain unchanged in $\sigma$ and $\sigma'$ (this is captured by $\sigma'$ being $(\sep, \sigma)$-compatible). Now let us prove the lemma formally.

\begin{proof}[Proof of \Cref{lem:2permutation}]
	Let $w^\star \in \pendant(v)$, such that the edge $vw^\star$ is crossed $q$ times in $\sigma$. We claim that each edge $vw$ for any $w \in \pendant(v)$ is {\em also} crossed exactly $q$ times in $\sigma$. Suppose for the contradiction that there exists some $w' \in \pendant(v)$ such that the edge $vw'$ is crossed $q' < q$ times in $\sigma$ (if $q' > q$, then swap the roles of $w$ and $w'$ in the following). Then, the drawing $\sigma'$ is obtained by placing $w$ just after $w'$. If $\mu(vw) > 0$ is the multiplicity of the edge $vw$, then the number of crossings in $\sigma'$ is now $k-\mu(vw) \cdot q + \mu(vw) \cdot q' < k$. Thus, $\sigma'$ is a drawing of $\cI$ with fewer than $k$ crossings, contradicting the $k$-minimality of $\cI$. 
	
	Let $\pbp \coloneqq (\totall, \totalr)$. Here,  $\totall$ ($\totalr)$ denotes the total sum of multiplicities of all pendants in $\pendant(v)$ that are placed to the left of (resp.~right of) $u^l = \ept(\lsep, 2)$ (resp.~$u^r = \ept(\rsep, 2)$) in $\sigma_2$. 
	

	Consider any pendant $w^l \in \pendant(v)$ that is placed to the left of $u^l$ according $\sigma_2$, i.e., $\sigma_2(w^l) < \sigma_2(u^l)$. Similarly, let $w^r \in \pendant(v)$ that is placed to the right of $u^r$ according to $\sigma_2$, i.e., $\sigma_2(w^r) > \sigma_2(u^r)$. We will use $w^l$ and $w^r$ in defining $\sigma'$ in the next paragraph. 
	
	Now consider any $\pguess = (\pendantl, \pendantr)$ that is $\pbp$-compliant. We claim that, the drawing $\sigma'$ obtained by modifying $\sigma$ in the following way, satisfies the claimed properties: (1) relocate all the vertices in $\pendantl$ in the same place as $w^l$\footnote{More formally, if $a^l, b^l \in V_2 \setminus \pendant(v)$ are the vertices that are immediately to the left and right of $w^l$ in $\sigma_2$, respectively, then we place all the vertices of $\pendantl$ between $a^l$ and $b^l$ in an arbitrary order.}, (2) place all the vertices in $\pendantr$ in the same place as $w^r$, and (3) the order of all other vertices remains unchanged. First, we observe that $\sigma'$ thus obtained indeed belongs to $\nicefamily(\pbp, \pguess, v,  \sigma)$. Furthermore, since each edge $vw$ for $w \in \pendant(v)$ is crossed exactly the same number of times in $\sigma$, it follows that the relocation of vertices used to obtain $\sigma'$ does not change the number of times each edge is crossed in $\sigma$ or $\sigma'$. Furthermore, note that $\sigma$ already places $\totall$ (resp.~$\totalr$) pendants (counting multiplicities) from $\pendant(v)$ to the left (resp.~right) of $u^l$ (resp.~$u^r$), and the resulting drawing $\sigma'$ does not alter this either. Therefore, $\sep$ also satisfies all properties from \Cref{def:2separator} w.r.t. $\sigma'$ as well. This shows that $\sigma'$ is $(\sep, \sigma)$-compatible.
\end{proof}

\subsection{Step 2. Guessing the Interfaces}

In this step, we will ``guess'' various subsets of edges and vertices that can be used in the creation of two sub-instances in the next step. Each of the guesses is made by assuming that $\cI$ is a $k$-minimal \yin with hypothetical solution  drawing $\sigma$. These guesses will be made in multiple steps, with each step dependent on all possible choices of the previous steps. These guesses can be thought of as \emph{branching} for each of the choices specified in all steps. Furthermore, if any of the guesses in a subsequent step contradicts the constraints imposed on a drawing (as in \Cref{def:2extdrawing}), or a guess made in a prior step, then such a guess is considered invalid and terminated (i.e., that branch is not explored further).  If all the guesses are terminated, or the corresponding recursive calls report that the sub-instances are {\sc No}-instances, then we will conclude that our assumption that $\cI$ was a \yin was wrong, and thus we must have a {\sc No}-instance. 

In each step, we specify at a high level the kind of object that we are trying to guess. In this algorithm, we label vertices using one or more labels from the following set $\LR{\lf, \mm, \rt}$, that are indicative of the sub-instance(s) that the vertex will belong to.\footnote{We do not explicitly define a middle sub-instance in the algorithm, but the label $\mm$ is defined for convenience.} Whenever we are trying to guess a subset of edges, we follow the following convention: all the edges of a multi-edge are treated in the same manner, such as being chosen in a guess. 

\begin{description}[leftmargin=10pt]
	\item[Step A: Separator.] We start by guessing $\sep = \LR{\lsep, \rsep} \subseteq E$ that is a separator w.r.t. $\sigma$. Note that there are at most $m^2 = n^{\Oh(1)}$ choices. Furthermore, based on our guess, we can also conclude whether we are in {\sf DE, SB,} or {\sf ST case}, which will be important later. Next, we guess $\md_i \subseteq V_i \setminus \epts(\sep, i)$ for $i \in \LR{1, 2}$, such that each is of size $4\sqrt{k}+2$. There are at most $n^{\Oh(\sqrt{k})}$ choices for this. Overall, we have at most $n^{\Oh(\sqrt{k})}$ choices.

	\item[Step B: Labeling $\epts(\sep)$.] 
	\begin{description}[leftmargin=0pt]
		\item[{\sf DE case.}] Vertices in $\epts(\lsep)$ receive label $\LR{\lf, \mm}$, and that in $\epts(\rsep)$ receive $\LR{\mm, \rt}$.
		\item[{\sf SB case.}] $\ept(\lsep, 1) = \ept(\rsep, 1)$ receives labels $\LR{\lf, \mm, \rt}$, $\ept(\lsep, 2)$ receives labels $\LR{\lf, \mm}$, and $\ept(\rsep, 2)$ receives labels $\LR{\mm, \rt}$. 
		\item[{\sf ST case.}] $\ept(\lsep, 2) = \ept(\rsep, 2)$ receives labels $\LR{\lf, \mm, \rt}$, $\ept(\lsep, 1)$ receives labels $\LR{\lf, \mm}$, and $\ept(\rsep, 1)$ receives labels $\LR{\mm, \rt}$. 
	\end{description}

	\item[Step C: Labeling $\mde$ and $\crosse$.] First, we label all endpoints of all the edges of $\mde$ with the label $\mm$. 
	
	Next, we guess two sets of edges $\crs(\lsep), \crs(\rsep)$ of size at most $\sqrt{k}$ that cross $\lsep, \rsep$, respectively. Then, we define $\crosse \coloneqq \crs(\lsep) \cup \crs(\rsep)$. Note that $$|\crosse| \le 2 \sqrt{k}$$ and the number of guesses is bounded by $n^{\Oh(\sqrt{k})}$. 
	
	We now describe feasible pairs for labels of the endpoints of $\crosse$; and each guess corresponds to using labeling the endpoint of each edge in $\crosse$ by one of the feasible (ordered) pairs. However, if the guess for the label is \emph{incompatible} with a previously assigned label, then we consider it as an invalid guess and proceed to the next one. Note that, if a vertex has previously been assigned multiple labels, then the incompatible label(s) is a label that has not been assigned to the said vertex, if any.
	
	\begin{itemize}
		\item For $e \in \crs(\lsep)$, 
		the feasible label pairs are $(\lf, \mm), (\mm, \lf), (\lf, \rt), (\rt, \lf)$.
		
		\item For $e \in \crs(\rsep)$, 
		the feasible label pairs are $(\rt, \mm), (\mm, \rt), (\lf, \rt), (\rt, \lf)$.
	\end{itemize}
	Overall, we have at most $\binom{m}{2\sqrt{k}}^4 \cdot 2^{\Oh(\sqrt{k})} = n^{\Oh(\sqrt{k})}$ choices. Let ${\sf CrossLeft, CrossRight} \subseteq \epts(\crosse)$ denote the subsets of vertices that receive labels $\lf, \rt$, respectively.
	
	\item[Step D: Labeling $\comps(v)$ for $v \in \epts(\sep)$]
	Recall that for each $v \in \epts(\sep)$, each component $C \in \comps(v)$ can be classified into $\spl(v), \predef(v), \updown(v),$ and $\pendant(v)$ according to \Cref{def:2componenttypes}. Now we describe how to handle each of the components separately.
	
	\begin{description}[leftmargin=2pt]
		\item[$\blacktriangleright$ $\spl(v)$.] Some of the vertices in such components are already labeled. We will return to the remaining vertices in these components at a later step.
		
		\item[$\blacktriangleright$ $\predef(v)$.] Recall that each $\pi_i$ for $i \in [2]$ is comprised of $\lambda = \Oh(\log k^\star)$ total orders $\pi_i^1, \pi_i^2, \ldots, \pi_i^\lambda$ for $V_i^1, V_i^2, \ldots, V_i^\lambda$, respectively. For each $i \in [2]$, and for each $j \in [\lambda]$, we do the following. 
		\begin{itemize}
			\item We guess two vertices $u(j, l), w(j, r) \in V^j_i$ with $\pi^j_i(u(j, l)) < \pi^j_i(w(j, r))$.
			\item Any $C \in \predef(v)$ that contains a vertex $u \in V^j_i$ with $\pi^j_i(u) \le \pi^j_i(u(j, l))$, we label all the vertices of $C$ using $\lf$. 
			\item Any $C \in \predef(v)$ that contains a vertex $w \in V^j_i$ with $\pi^j_i(w) \ge \pi^j_i(w(j, r))$, we label all the vertices of $C$ using $\rt$.
		\end{itemize}  
		If any of the guesses made in this step is incompatible with a previous step, or with a prior guess made in this step, then we terminate this guess. Note that the number of guesses made in this step is at most $n^{\Oh(\lambda)} = n^{\Oh(\log k^\star)} = n^{\Oh(\sqrt{k})}$, since $k = \Omega(\sqrt{k^\star})$. 
		
		\item[$\blacktriangleright$ $\updown(v)$.] By \Cref{lem:2type3components}, we know that $|\updown(v)| \le 8\sqrt{k}$ for each $v \in \epts(\sep)$. Thus, if the number of such components is more than $8\sqrt{k}$, then we terminate the guess. Otherwise, we guess the label of each of them. The number of guesses for this is bounded by $\binom{n}{8\sqrt{k}} \cdot 2^{8\sqrt{k}} = n^{\Oh(\sqrt{k})}$.
		
		\item[$\blacktriangleright$ $\pendant(v)$.] We distinguish between {\sf DE, SB, ST cases} and use \Cref{lem:pendantcases} to label components as applicable. In {\sf DE} case, we label all the components in $\bigcup_{v \in \epts(\sep)} \pendant(v)$. However, in {\sf SB case} (resp.~{\sf ST case}), we are left with labeling $\pendant(\ept(\lsep, 1))$ (resp. $\pendant(\ept(\lsep, 2))$), which we do next. 
		
		We focus on the {\sf SB} case and show how in accordance with \Cref{def:2nicefamily} and \Cref{lem:2permutation}. The {\sf ST} case is handled analogously. Let $v \coloneqq \ept(\lsep, 1) = \ept(\lsep, 2)$, and let $\total$ denote the total sum of the multiplicities of $\pendant(v)$. First, we guess a $\pbp = (\pendantsl, \pendantsr)$, for which there are at most $n$ possibilities. Then, for each choice of $\pbp$, we use a knapsack-style dynamic programming to find a $\pbp$-compatible $\pguess = (\pendantl, \pendantr)$ in polynomial time, if one exists. If such a partition does not exist, then we move to the next guess. Otherwise, we label the vertices of $\pendantl$ using $\lf$ and that of $\pendantr$ using $\rt$. 
	\end{description}


	\item[Step E: Label propagation.]  Let $A$ denote the subset of vertices that have already received a label, and let $A^\lf, A^\rt \subseteq A$ denote the subsets of $A$ that have received the respective label. Further, let $O \subseteq A$ denote the subset of vertices that have received a unique label, and again let $O^\lf, O^\rt \subseteq O$ denote the partition of $O$ into two disjoint subsets. At this point, some of the vertices of the graph may still be unlabeled, which we label during this phase. To this end, we consider normal ($\diamondsuit$) and elaborate ($\clubsuit$) instances separately. 
	
	$\diamondsuit$ Define the following sets:
	\begin{align*}
		{\sf Origin}^\lf &= \epts(\LR{\leftb, \lsep}) \cup {\sf CrossLeft}
		\\{\sf Origin}^\rt &= \epts(\LR{\rightb, \rsep}) \cup {\sf CrossRight}
	\end{align*}

	$\clubsuit$ Define the following sets:
	\begin{align*}
		{\sf Origin}^\lf &= (\lbnd \cap A^\lf) \cup \epts(\LR{\leftb, \lsep}) \cup {\sf CrossLeft}
		\\{\sf Origin}^\rt &= (\rbnd \cap A^\rt) \cup \epts(\LR{\rightb, \rsep}) \cup {\sf CrossRight}
	\end{align*}
	
	Now, regardless of $\diamondsuit$ or $\clubsuit$, we proceed as follows. 
	\begin{enumerate}
		\item For each $u \in {\sf Origin}^\lf$, and each unlabeled $w \not\in A$ that is reachable from $u$ in the graph $G - (O^\mm \cup O^\rt)$, we label it using $\lf$.
		\item For each $u \in {\sf Origin}^\rt$, and each unlabeled $w \not\in A$ that is reachable from $u$ in the graph $G - (O^\lf \cup O^\mm)$, we label it using $\rt$.
		\item For each $u \in {\sf Origin}^\mm$, and each unlabeled $w \not\in A$ that is reachable from $u$ in the graph $G - (O^\lf \cup O^\rt)$, we label it using $\mm$.
	\end{enumerate}
	During this process, if a previously unlabeled vertex, i.e., $w \notin A$ receives multiple labels, then we terminate the guess. 
	
	Otherwise, for ${\cal T} \in \LR{\lf, \mm, \rt}$, let $V^{\cal T}$ denote the subset of vertices that have received the respective label. 

	\item[Step F: Partial Order.] Recall that we are given a partial orders $\pi_1, \pi_2$. For each $i \in [2]$, we will guess permutation between certain subsets of vertices, which will be used to extend $\pi_i$ to a new partial order $\pi'_i$. 
	\\Let $S \coloneqq \md_1 \cup \md_2 \cup \epts(\crosse)$. The bounds on the respective sets imply that $|S| = \Oh(\sqrt{k})$. For each $i \in [2]$, let $S_i \coloneqq V_i \cap S$. We guess permutations $\pi^{\lambda+1}_i$ of $S_i$ for each $i \in [2]$. Since $|S_i| = \Oh(\sqrt{k})$, it follows that the total number of guesses is $|S_1|! \cdot |S_2|! = (\Oh(\sqrt{k})!)^2 = 2^{\Oh(\sqrt{k}  \log k)}$. Let $\pi^{\lambda+1} = (\pi^{\lambda+1}_1, \pi^{\lambda+1}_2)$. We say that $\pi^{\lambda+1}$ is \emph{valid} if the following condition is met.
	\begin{itemize}
		\item $\clubsuit$ For each $i \in [2]$, for any distinct $u, v \in S_i$, if $\pi^{\lambda+1}_i(u) < \pi^{\lambda+1}_i(v)$, then, either $u$ and $v$ are incomparable in $\pi_i$, or $\pi_i(u) < \pi_i(v)$.
		\item $\pi^{\lambda+1}$ satisfies the following properties (that must be satisfied by the separator):
		\begin{itemize}
			\item $\lsep, \rsep$ do not cross
			\item For each $e \in \sep$, the set of edges crossing $e$ is exactly equal to $\crs(e)$,
			\item $\mde$ is exactly the set of edges that lie between $\lsep, \rsep$  
		\end{itemize}
	\end{itemize}
	Note that the validity of $\pi^{\lambda+1}$ can be checked in polynomial time. If a guess for $\pi^{\lambda+1}_i$ is invalid, then we terminate it and move to the next guess. 
	
	$\clubsuit$ We define $\pi^{\lambda+2}$ as an ordering over $\bnd \cup \epts(\sep, 1)$ by combining parts of $\pi^{\lambda+1}_1$, and one of the constituent total orders of $\pi^2$ that is a permutation of the set $\bnd$. Then, we define $\pi'_1 = \pi_1 \cup \pi^{\lambda+1}_1 \cup \pi^{\lambda+2}_1$, and $\pi'_2 \coloneqq \pi_2 \cup \pi^{\lambda+1}_2$. Note that $\pi'_i$ is a valid extension of $\pi_i$, since $\pi^{\lambda+1}$ is valid. Henceforth, we assume that we are working with such a $\pi' = (\pi'_1, \pi'_2)$ for some guess for $\pi^{\lambda+1}$. 
	
	\item[Step G: Parameter division.]
	
	Let $k_{\sf current}$ denote the number of crossings among the edges of $\crosse \cup \mde \cup \sep$ according to $\zeta$. If $k_{\sf current} > k$, then we terminate the current guess. Otherwise, let $k' \coloneqq k - k_{\sf current}$. We guess two non-negative integers $0 \le k^\lf, k^\rt \le k/2$ such that $k' = k^\lf + k^\rt$. Note that the number of such guesses is bounded by $k^{2}$.

\end{description}

The following lemma follows from the description of Step 2.
\begin{lemma} \label{lem:2guessbound}
	The total number of guesses made in steps A through G, for (i) separator $\sep$ and its associated sets of vertices and edges, (ii) labels of vertices of $V(G)$, and (iii) division of $k$ into smaller parameters, is bounded by $n^{\Oh(\sqrt{k})}$.
\end{lemma}

\subsection{Step 3: Creating sub-instances.} 
For each set of guesses made in the previous step (as stated in \Cref{lem:2guessbound}), we create a pair of sub-instances, a left sub-instance $\cI^\lf$, and a right sub-instance $\cI^\rt$. In this step, we construct each such pair of sub-instances corresponding to a particular guess. Then, in the next step, we will recursively call our algorithm to solve each of the two sub-instances. 

\paragraph{Creating left sub-instance.} 
Let $G^{\lf} = G[V^{\lf}]$, where recall that $V^\lf$ is a set of all the vertices with label $\lf$ (including a vertex with multiple labels). Next, we add certain multi-edges in the graph.

\emph{Creating mutli-edges.} For each edge $uv \in \crs(\lsep)$, we identify a multi-edge $e_{add}(uv)$ to be added in each case. After the case analysis, we will discuss the multiplicity of this edge, and actually add it to the graph.
\begin{itemize}
	\item $u \in V_1$ has label $\lf$ and $v \in V_2$ has label has label $\mm$ or $\rt$.  
	\\\textbf{Operation.} Define $e_{add}(uv) \coloneqq (u, \ept(\lsep, 2))$.
	
	\item $v \in V_2$ has label $\lf$, and $u \in V_1$ has label $\mm$ or  $\rt$.
	\\\textbf{Operation.} Define $e_{add}(uv) \coloneqq (\ept(\lsep, 1), v)$.
\end{itemize}
Now, we proceed as follows:
\begin{itemize}
	\item If an edge is defined as $e_{add}$ of multiple edges, say $e_1, e_2, \ldots$, then its multiplicity is defined as the sum of the respective multiplicities of $e_1, e_2, \ldots$ -- here we adopt the convention that, if an edge $e_i$ belongs to $\regulare$, then its multiplicity is $1$.
	\item Finally, we add $e_{add}$ with the combined multiplicity to $\weightede^\lf$. Here, if $e_{add}$ is already in $\regulare$, we do not remove it.
\end{itemize}

Let $G^{\lf}$ denote the resulting graph. The left sub-instance is an extended instance  $\cI^{\lf} = (G^{\lf}, V_1^{\lf}, V_2^{\lf}, \bnd^\lf, \pi^\lf, k^\lf)$, where:
\begin{itemize}[leftmargin=*]
	\item For each $i \in [2]$, $V_i^\lf \coloneqq V_i \cap V^\lf$, and $\pi_i^\lf = \pi'_i$ projected to $V^\lf$. 
	\item $\boundary^\lf = \LR{\leftb^\lf, \rightb^\lf}$, where $\leftb^\lf = \leftb$, $\rightb^\lf = \lsep$. 
	\item $\diamondsuit$ Let $\bnd^\lf = \lbnd^\lf = \rbnd^\lf = \emptyset$, since these sets are irrelevant for normal extended instances.
	\\$\clubsuit$ For elaborate extended instances, let $\lbnd^\lf = \lbnd \cap V^\lf$ and $\rbnd^\lf = \ept(\lsep, 1)$, and $\bnd^\lf = \lbnd^\lf \cup \rbnd^\lf$.
\end{itemize}
At this point, we perform the following check based on whether $\cI$ was a normal ($\diamondsuit$), or an elaborate ($\clubsuit$) extended instance.
\\$\diamondsuit$ If $G^\lf$ is not connected, then we terminate the guess.
\\$\clubsuit$ If each $u \in V^\lf$ is reachable from some vertex of $v \in \bnd^\lf$ in $G^\lf$, then we terminate the guess.

The right sub-instance $\cI^\rt = (G^{\rt}, V_1^{\rt}, V_2^{\rt}, \bnd^\rt, \pi^\rt, k^\rt)$ is created analogously, and we perform the connectivity check as described above; we omit the description. 

Further, we note the following auxiliary observations, that satisfy the constraints on the extended instances, as mentioned in \Cref{def:extendedinstance2}.
\begin{itemize}
	\item $\clubsuit$ The parameter in each recursive call reduces by at least a factor of $1/2$ (and the base case is triggered when $k \le c \sqrt{k^\star}$), and each time we add two total orders to the relation $\pi$, it follows that the total number of partial orders remains bounded by $4\log k^\star$ at each step as required in \Cref{item:2partialorder}.
	\item $\clubsuit$ The size of $\bnd^\lf, \bnd^\rt$ remains bounded by $\bnd$, which was originally $3 \sqrt{k^\star}$, satisfying \Cref{item:2extdconn}.
	\item Note that, in the original instance $\weightede = \emptyset$, and in each call of the algorithm, the total sum of multiplicities of $\weightede$ incident to a vertex $u = \ept(e, i)$ for some $e \in \sep$, and some $i \in [2]$ added in that call, is at most the number of crossings the corresponding separator edge $e$ is involved in--which is bounded by $\sqrt{k}$. Again, since the parameter decreases by a constant factor in each recursive call, it holds that the total sum of multiplicities of all edges incident to a vertex, remains bounded by $\sum_{\ell \ge 0} \sqrt{(1/2)^\ell \cdot k^\star} = 2\sqrt{k^\star}$. \Cref{item:2multiplicity-sum} 
\end{itemize}

From the previous discussion, the following claim follows.

\begin{claim} \label{cl:2validinstances}\ 
	If $\cI$ was a valid normal (resp.~ elaborate) extended instance, then $\cI^\lf, \cI^\rt$ are valid normal (resp.~elaborate) extended instances.
\end{claim}

\subsection{Step 4: Conquer left and right.} For each guess in Step 2, we obtain a pair of sub-instances $\cI^\lf$ and $\cI^\rt$, with each having a parameter at most $k/2$. We focus on one pair of instances that we solve recursively in time $2\cdot n^{\Oh(\sqrt{k})}$. There are two possibilities. The first one is that at least one of the recursive calls returns that the corresponding sub-instance is a \textsc{No}-instance, in which case we proceed to the next guess. Hence, assume that both recursive calls return that $\cI^{\lf}$ and $\cI^{\rt}$ are \textsc{Yes}-instances, and in addition they also output an ordering of the corresponding vertices. 


\subsection{Step 5: Obtaining the Final Drawing.} 
If, for any of the non-terminated guesses, all three recursive calls reporting that the respective sub-instances are  \textsc{Yes}-instances, along with the corresponding orderings, then we can combine these orderings to obtain orderings of $V_1, V_2$ that has at most $k$ crossings, and respect the given partial order $\pi$. 

Specifically, suppose $\sigma^\lf, \sigma^\rt$ be the drawings of $\cI^\lf, \cI^\rt$ returned by the corresponding recursive calls. Then, we obtain the final drawing of $\cI$ by defining $\so'_i \coloneqq \sigma_i^\lf \circ \pi'_i(\md_i) \circ \sigma_i^\rt$ for $i \in [2]$. We return $\sigma' = (\sigma'_1, \sigma'_2)$ as our final drawing. Otherwise, if for each guess, either (i) it is terminated due to incompatibility, or (ii) one of the three sub-instances corresponding to the guess is found to be a {\sc No}-instance, then we output that $\cI$ is a {\sc No}-instance.

\begin{lemma} \label{lem:2drawingcomb}
	Let $\cI^\lf, \cI^\rt$ be a pair of sub-instances (with parameters $k^\lf, k^\rt$) corresponding to some guess made in Step 2. Further, let  $\sigma^\lf, \sigma^\rt$ be valid drawings of the $\cI^\lf, \cI^\rt$. Then, if we obtain a combined drawing $\sigma'$ as in Step 5, then $\sigma'$ is a valid drawing of $\cI$.
\end{lemma}
\begin{proof}
	Recall that, in step F we guess the permutations $\pi^{\lambda+1}_i$ of $S_i$, and we use it in step G to determine the number of crossings $k_{\sf current}$ among the edges of $\crosse \cup \mde \cup \sep$. Then, we define $k = k' + k_{\sf current}$, and divide it into $k^\lf + k^\rt$, which are used as parameters for the respective sub-instances. 
	
	Next, we claim that all the $k_{\sf current}$ crossings among the edges are not counted in any of the sub-instances. To this end, first observe that the edges of $\mde$ are not involved in any of the subproblems. Next, consider any pair of edges $e_1, e_2 \in \crosse \cup \sep$, such that the crossing between $e_1$ and $e_2$ is counted in $k_{\sf current}$. Note that both $e_1, e_2$ cannot be in $\sep$, since the guessed permutation must satisfy that these edges do not cross. Then, suppose $e_1 \in \sep$ and $e_2 \in \crosse$. Then in each of the (at most) two sub-instances, the edge $e_2$ is replaced by a multi-edge $e'_2$ incident to one of the endpoints of $e$. Thus, in any drawing of the respective sub-instance, the edge $e_1$ and $e'_2$ cannot cross. \footnote{Note that in different sub-instances, the multi-edge $e'_2$ that replaces $e_2$ is different; however, in each of the sub-instances, it holds that $e_1$ cannot cross with any of the replacing $e'_2$.} 
	Finally, consider the case when $e_1, e_2 \in \crosse$. Since $e_1, e_2$ cross, we have two cases: (a) $e_1 \in \crs(e'_1)$ and $e_2 \in \crs(e'_2)$ for distinct edges $e'_1, e'_2 \in \sep$, or (b) $e_1, e_2 \in \crs(e)$ for some $e \in \sep$. In case (a),  the multi-edges replacing $e_1$ and $e_2$ belong to completely different sub-instances (note that we do not have a middle sub-instance). Thus, we are left with case (b), when they cross the same edge in the separator.
	Then, in each of the corresponding sub-instance, there is a multi-edge $e'_1$ that replaces $e_1$, and a multi-edge $e'_2$ that replaces $e_2$. If $e'_1$ and $e'_2$ are incident to the same endpoint of $e$, then they cannot cross. Alternatively, if $e'_1$ and $e'_2$ are incident to different endpoints of $e$, then due to \Cref{item:drawing} in \Cref{def:2extdrawing}, we do not count such a crossing. In any case, we have that the crossing between $e, e' \in \crosse \cup \mde \cup \sep$ is not counted in any of the sub-instances.
	
	Next, we note that in $\cI^\lf$, the left boundary edge $\leftb$ (and left extended boundary $\lbnd$, in the case of $\clubsuit$ elaborate instance) is derived from the current instance $\cI$, and thus are placed at the appropriate locations. An analogous property holds about $\cI^\rt$. Next, we note that the orderings $\pi^\lf,\pi^\rt$ is obtained by the projecting $\pi'$ (which is an extension of $\pi$) to the respective vertex subsets. Since the respective drawings are compatible with these orders, it follows that the combined drawing is compatible with $\pi'$, and hence with $\pi$. Furthermore, since $\pi'$ also contains the orderings $\pi^{\lambda+1}$, the respective drawings also satisfy the ordering on the vertices of $ \lbnd \cup \epts(\sep, 1)$. 
	
	Finally, since $\sigma^\lf, \sigma^\rt$ are valid drawings for $\cI^\lf, \cI^\rt$ respectively, when we obtain the combined drawing $\sigma'$ by gluing the three drawings, it is straightforward to verify that \Cref{def:2extdrawing} are satisfied. This shows that $\sigma'$ is a valid drawing for $\cI$ with $k$ crossings.
\end{proof}

\subsection{Proof of Correctness} \label{subsec:twolayerproof}

\begin{lemma} \label{lem:2layerlemma}
	$\cI$ is an \yin for parameter $k$ iff the algorithm returns a drawing $\sigma'$ of $\cI$ with at most $k$ crossings.
\end{lemma}
\begin{proof}
	In the base case when $k \le c \sqrt{k^\star}$, 
	the correctness of the algorithm follows via \Cref{lem:2layerbasecase}. Thus, we consider the recursive case. Suppose inductively that the statement is true for all values of parameter strictly smaller than $k > c \sqrt{k^\star}$, and we want to show the statement holds for $k$.  
	
	\textbf{Forward direction.} Suppose $\cI$ is a $k$-minimal \yin, and let $\sigma$ be a drawing of $\cI$ with $k$ crossings. Consider the guess where the following objects are correctly guessed in the algorithm:
	\begin{enumerate}
		\item Separator $\sep(\sigma)$ as guaranteed by \Cref{lem:2separatorlemma},
		\item Sets $\crosse(\sep)$ and $\mde(\sep)$ and the relative ordering between the corresponding vertices is guessed according to $\sigma$. 
		\item For each $v \in \epts(\sep)$, by \Cref{lem:2type3components}, $|\updown(v)| \le 4\sqrt{k}$. Then, for each component $C \in \updown(v)$, if all the vertices of $C \cap V_i$ belong to the left of $\ept(\lsep, i)$ (resp. to the right of $\ept(\rsep, i)$) for each $i \in [2]$, then consider the guess where it is labeled $\lf$ (resp.~$\rt$).
	\end{enumerate}
	
	Now we distinguish between three cases, namely {\sf DT, SB, ST}. 
	\begin{description}[leftmargin=1pt]
		\item[$\blacktriangleright$ {\sf DE case.}] Here, all components of $\pendant(v), v \in \epts(\sep)$ behave according to \Cref{lem:pendantcases}, which is how the algorithm labels the respective components.
		\item[$\blacktriangleright$ {\sf SB case.}] In this case, all components of $\pendant(\ept(\lsep, 2)) \cup \pendant(\ept(\rsep, 2))$ behave according to \Cref{lem:pendantcases}, and the algorithm also labels such components accordingly.
		\\Now, we are left with labeling the components of $\pendant(\ept(\lsep, 1))$. In this case, let $(\totall, \totalr)$ be the pair of integers guaranteed by \Cref{lem:2permutation}, and consider the guess corresponding to $(\totall, \totalr)$. In this case, \Cref{lem:2permutation} implies that a $(\totall, \totalr)$-partition of $\pendant(v)$ exists, and one such partition, say $(\pendantl, \pendantr)$ will be found by the algorithm. \Cref{lem:2permutation} implies that there exists a drawing $\sigma'$ in the family $\nicefamily(\sigma,\ept(\lsep, 1), \totall, \totalr, \pendantl, \pendantr)$, w.r.t. which $\sep(\sigma)$ is still a separator, and furthermore, $\sigma'$ is compatible with all the guesses made so far. Henceforth, we consider this guess, and use $\sigma'$ instead of $\sigma$ in the subsequent analysis. For ease of notation, let us rename $\sigma \gets \sigma'$. 
		\item[$\blacktriangleright$ {\sf ST case}.] This case is handled in a symmetric manner as {\sf SB} case, and we omit the description.
	\end{description}
	Let $\sigma$ be the drawing obtained at the end after appropriate redefintion (if any), as described above. In the following claim, we show that the labeling process correctly labels all the vertices w.r.t. the (new) drawing $\sigma$.
	
	\begin{claim} \label{cl:2cllabeling}
		Let $\sigma$ be the (new) drawing as described above, and consider the guess as defined above. Then, after the vertex labeling process of step H corresponding to this particular guess, we have that $V^{\cal T} = V(\sigma, \sep, {\cal T})$, for each ${\cal T} \in \LR{\lf, \mm, \rt}$. Note that the definition of the sets $V(\sigma, \sep, {\cal T})$ can be found in \Cref{def:2division}.
	\end{claim}
	\begin{proof}[Proof of \Cref{cl:2cllabeling}]
		We consider $\clubsuit$ elaborate instances first, since the proof for the $\diamondsuit$ normal instances is only simpler. 
		
		Let us consider a vertex $u \in V(\sigma,\sep, \lf)$. By  of \Cref{def:extendedinstance3}, $u$ is reachable from some vertex $v \in \bnd \cup \epts(\boundary)$, say via a path $P(uv)$. 
		We consider different cases based on $v$ and a path $P(u \leadsto v)$. Let $P(uv) = \langle u = w_0, w_1, w_2, \ldots, w_\ell = v \rangle$. Let $w_{i+1}$ be the first vertex (if any) along $P(u \leadsto v)$, such $w_0, w_1, \ldots, w_i \in V(\sigma, \lf)$, and $w_{i+1} \not\in V(\sigma, \lf)$. Otherwise, let $w_i = w_\ell = v$. We consider these two cases separately; see \Cref{fig:2pathfigure}.
		
		\begin{figure}[h]
			\centering
			\includegraphics[scale=0.8,page=15]{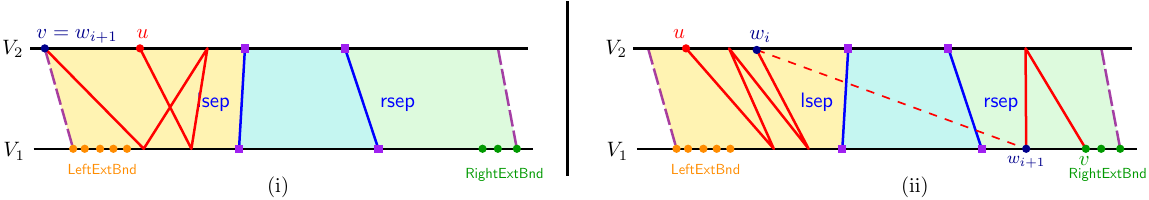}
			\caption{Illustration of the two cases in \Cref{cl:2cllabeling}. In case (i) on the left, the entire path contains vertices of $V(\sigma, \lf)$, and the path ends at vertex $v \in {\sf Origin}^\lf$. In case (ii) on the right, the path goes outside $V(\sigma,\sep, \lf)$; however to do so, it must contain a crossing edge (shown as dashed red), in which case its left endpoint $w_i \in {\sf CrossLeft} \subseteq {\sf Origin}^\lf$.}\label{fig:2pathfigure}
		\end{figure}
		
		(i) If $w_i = v$, then since $w_i \in V(\sigma, \lf)$, it cannot be the case that $w_i \in (\rbnd \cup \epts(\rightb)) \setminus \epts(\sep)$. Thus, it implies that $w_i \in \lbnd \cup \epts(\leftb) \cup {\sf CrossLeft} = {\sf Origin}^\lf$. 
		
		(ii) If $w_i \neq v$, then we know that $w_{i+1}$, the vertex immediately after $w_i$ in $P(uv)$ satisfies that $w_{i+1} \notin V(\sigma, \lf)$. Then, it follows that the edge $(w_i, w_{i+1})$ crosses some edge of $\sep$--implying that $w_i \in {\sf CrossLeft}$.
		
		
		

		This implies that, for any vertex $u \in V(\sigma, \sep,\lf) \setminus {\sf Origin}^\lf$, and any path $P(u \leadsto v)$ to a vertex $v \in \bnd \cup \epts(\boundary)$, the last vertex $w$ with label $\lf$ must belong to the set ${\sf Origin}^\lf$. Further, the subpath $P(u \leadsto w)$ does not contain any vertex with a label $\mm$ or $\rt$. This also implies that, all the vertices of $V(\sigma, \lf)$ are reachable from some vertex $w' \in {\sf Origin}^\lf$ in the graph $G - (O^{\mm} \cup O^{\rt})$. Thus, all the unlabeled vertices in $V(\sigma, \lf)$ get labeled $\lf$ in the label propagation phase. This shows that $V(\sigma, \lf) \subseteq V^\lf$. Similar proofs also show that $V(\sigma, \sep,\mm) \subseteq V^\mm$, and $V(\sigma, \sep,\rt) \subseteq V^\rt$. Now, consider a vertex of $u' \in (V(\sigma, \sep,\mm) \cup V(\sigma, \sep,\rt)) \setminus {\sf Origin}^\lf$. However, such a vertex $u'$ belongs to $O^\mm \cup O^\rt$, and hence is not part of the graph $G - (O^\mm \cup O^\rt)$. Thus, it cannot receive label $\lf$, which shows that $V^\lf \subseteq V(\sigma, \sep,\lf)$. By combining both containments, we obtain that $V^\lf = V(\sigma, \sep,\lf)$. Similar proofs also show that $V^\mm = V(\sigma, \sep,\mm)$, and $V^\rt = V(\sigma, \sep,\rt)$, and are therefore omitted.
		
		The proof for $\diamondsuit$ normal instances is simpler, since the graph $G$ is connected. Therefore, each vertex $u$ is reachable from any other vertex, and we can select an arbitrary vertex to play the part of $v$ in the proof above.
	\end{proof}
	
	After accounting for the $k_{\sf current}$ crossings among the endpoints of $\crosse \cup \mde \cup \sep$ in $\sigma$, the rest of the $k'$ crossings are either to the left of $\lsep$ or to the right of $\rsep$ (this argument is similar to \Cref{lem:2drawingcomb}). Let the number of these crossings be $k^\lf, k^\mm$, respectively, and consider the guess where these numbers are correctly guessed. Further, since we label all the vertices correctly, it follows that all the vertices of $V^\lf$ (resp.~$V^\rt$) are reachable from ${\sf Origin}^\lf$ (resp.~${\sf Origin}^\rt$). 
	
	$\diamondsuit$ For normal extended instances, we note that any $u \leadsto v$ path, where $u, v \in V^\lf$, must either be completely contained inside $G'^\lf = G[V^\lf]$, use an edge of $\crs(\lsep) \cup \lsep$. In any case, due to the way we add extra edges to $\weightede^\lf$, it follows that $u$ and $v$ are also reachable in $G^\lf$. 
	
	$\clubsuit$ Now, consider elaborate extended instances. Then, we observe that vertices of ${\sf CrossLeft}$ (resp.~${\sf CrossRight}$) are directly connected to one of the vertices of $\epts(\LR{\lsep})$ (resp.~$\epts(\LR{\rsep})$) in the graph $G^\lf$ (resp.~$G^\rt$).
	
	In either case, we have shown that these guesses are not terminated by the connectivity check done in Step 3 after creating the instances $\cI^\lf, \cI^\rt$, respectively. Then, \Cref{cl:2validinstances} implies that $\cI^\lf, \cI^\rt$ are valid extended instances. 
	Further, $\sigma$ restricted to the respective vertex sets $V^\lf, V^\rt$ gives a valid drawing of the respective instances. This implies that $\cI^\lf, \cI^\rt$ are {\sc Yes}-instances for the parameters $k^\lf, k^\rt$, respectively. 
	
	Next, we claim that, since $\cI$ is a $k$-minimal instance, both of these instances must also be minimal instances for the respective parameters. suppose for contradiction that $\cI^\lf$ is a {\sc Yes}-instance, but is not $k^\lf$-minimal {\sc Yes}-instance, then it admits a drawing $(\sigma')^\lf$ with $(k')^\lf < k^\lf$ crossings. Then, we can apply an argument from \Cref{lem:2drawingcomb}, in order to obtain a drawing for $\cI$ with fewer than $k$ crossings, contradicting the $k$-minimality of $\cI$. A similar argument also holds for $\cI^\rt$.
	
	Then, by induction (since $k^\lf, k^\rt \le k/2$) the respective recursive calls return valid drawings $\sigma^\lf, \sigma^\rt$ of the respective instances. Then, \Cref{lem:2drawingcomb} implies that that in Step 5, we return a combined drawing $\sigma'$ corresponding to this guess.

	\textbf{Reverse direction.} The reverse direction is trivial, since the valid drawing $\sigma'$ of $\cI$ witnesses that $\cI$ is a \yin.
	
\end{proof}

In the following lemma, we argue that the algorithm runs in subexponential time.

\begin{lemma} \label{lem:2runtime}
	For any input \ein $\cI$ with parameter $k$, the algorithm runs in time $T_2(k) \le n^{\Oh(\sqrt{k})}$.
\end{lemma}
\begin{proof}
	The base case corresponds to when $k \le c \sqrt{k^\star}$,
	for some constant $c$. In this case, \Cref{lem:2layerbasecase} implies that the algorithm runs in time $n^{\Oh(\sqrt{k})}$.
	
	Otherwise, consider the recursive case, and suppose the claim is true for all $k' \le k-1$, where $k > c \sqrt{k^\star}$ is the parameter of the given instance $\cI$. \Cref{lem:2guessbound} implies that the total number of guesses is bounded by $n^{\Oh(\sqrt{k})}$, and corresponding to each guess that is not terminated, we make two recursive calls to the algorithm with parameters at most $k/2$ each. Thus, $T_2(n,k)$ satisfies the recurrence given in \Cref{eqn:2recurrence}. Then, the claim follows via induction.
\end{proof}

By combining \Cref{obs:2normal2extended}, \Cref{lem:2layerlemma}, \Cref{lem:2runtime}, and the fact that due to linear-time kernelization algorithm of \Cref{prop:kern-two}, we have $n = (k^\star)^{\Oh(1)}$, we obtain the following theorem (where we rename $k^\star$ as $k$).

\begin{theorem} \label{thm:2layertheorem}
	On an $n$ vertex graph, \probTwoCr can be solved in time $2^{\Oh(\sqrt{k} \log k)} + n \cdot k^{\Oh(1)}$, where $k$ denotes the number of crossings.
\end{theorem}

\section{Subexponential Algorithm for \probThreeCr}\label{sec:subexp}
In this section, we will design our subexponential FPT algorithm for \probThreeCr running in time $2^{\Oh(k^{2/3} \log k)} + n k^{\Oh(1)}$. The overall structure of this algorithm is similar to that for \probTwoCr -- the algorithm is recursive, and is based on dividing the instance into multiple sub-instances with the help of a separator (consisting of at most 4 edges, and two special subsets of vertices) with at most $\alpha k$ crossings its either side. However, the details are different in the following two crucial aspects: (i) the structure of the separator is more complex, which necessitates handling multiple cases (like {\sc DE, SB, ST cases} for \probTwoCr), and (ii) the analysis of connected components in $G-v$, where $v$ belongs to one of the endpoints of the separator edges, is more involved. Recall that, in the case of two layers, we needed to handle the set of pendant vertices attached to $v$, where we showed the existence of another, more structured, separator-compatible drawing in \Cref{lem:2permutation}. For three layers, this requires a much more sophisticated analysis. Another notable aspect we would like to highlight is that, at various places in the algorithm for three layers, we will use the two layer algorithm as a subroutine. 

\subsection{Setting Up}

First, we adopt the same notation set up in \Cref{def:2endpts} for $\ept(\cdot), \epts(\cdot)$ and extend it to three layers. Further, in any (extended) instance, we partition the set of edges into $E_{12} \cup E_{23}$, where $E_{12}$ (resp.~$E_{23}$) is the set of edges with one endpoint in $V_1$ and another in $V_2$ (resp.~one endpoint in $V_2$ and another in $V_3$). 
Next, we describe the notion of an extended instance.

We denote the \emph{original} instance of \probThreeCr as $(H, U_1, U_2, U_3, k^\star)$. In the course of the recursive algorithm, our algorithm will produce \emph{extended instances} of \probThreeCr, defined as follows. 
  
\begin{definition}[Extended Instance for $h = 3$] \label{def:extendedinstance3}
	An \emph{\ein} of \probThreeCr is given by $(G, V_1, V_2, V_3, \bnd, \pi, k_{12}, k_{23}, k)$ where,
	\begin{enumerate}
		\item $V(G) = V_1 \uplus V_2 \uplus V_3$, \item $E(G) = \regulare \uplus \weightede \uplus \boundary$.
		\\The set of edges between $V_1$ and $V_2$ respectively are denoted by $E_{12}$, and those between $V_2$ and $V_3$ are denoted $E_{23}$. 
		\item $\bnd = (\lbnd, \rbnd)$, where \label{item:3extdconn}
		\begin{itemize}
			\item $\boundary  = \LR{\leftbone, \rightbone, \leftbtwo, \rightbtwo}$ is a set of at most four \emph{boundary edges} of the instance, where $\leftbone, \rightbone \in E_{12}$ (both may be equal), and $\leftbtwo, \rightbtwo \in E_{23}$ (both may be equal). 
			\item 
			$\lbnd, \rbnd \subseteq V_2$ is a set of at most $3\sqrt{k^\star}$ vertices, where 
			\\(i) $\ept(\leftbone, 2), \ept(\leftbtwo, 2) \in \lbnd $, and
			\\(ii) $\ept(\rightbone, 2), \ept(\rightbtwo, 2) \in \rbnd$.
			\item Further, each $v \in V(G)$ is reachable from some vertex in $\epts(\boundary) \cup \lbnd \cup \rbnd$.
		\end{itemize} \label{item:3extendedconn}
		\item $\pi = (\pi_1, \pi_2, \pi_3)$, where for each $i \in [3]$, $\pi_i$ is a partial order on $V_i$ such that, $\pi_i$ is the union of $\lambda \ge 0$ different relations $\pi_i^1, \pi_i^2, \ldots, \pi_i^\lambda$, where for each $1 \le j \le \lambda$, $\pi_i^j$ is a total order on $V_i^j \subseteq V_i$, and the sets $\LR{V_i^j: 1 \le j \le \lambda}$ are not necessarily disjoint.\footnote{The fact that each $\pi_i$ is a union of $\lambda$ total orders is crucial only at a few specific steps in the algorithm, which is where we will make this notation explicit.} Here, $\lambda$ always remains bounded by $4\log(k^\star)$. \label{item:3partialorder}
		\\One of the orderings $\pi^j_2$ has ground set $V^j_2 = \bnd$, and it orders these vertices in the following way from left to right: 
		\\$u^l_1 < $ vertices of $\lbnd \setminus \LR{u^l_1, u^l_2}$ in some order $ < u^l_2 < u^r_1 < $ vertices of $\rbnd \setminus \LR{u^r_1, u^r_2}$ in some order $< u^r_2$. Here, $\LR{u^l_1, u^l_2} = \LR{\ept(\leftbone, 2), \ept(\leftbtwo, 2)}$, and $\LR{u^r_1, u^r_2} = \LR{\ept(\rightbone, 2), \ept(\rightbtwo, 2)}$.  
		
		\item Each $e \in \weightede$ is referred to as a \emph{multi-edge}, and has an associated weight/multiplicity $\mu(e)$ (which is a non-negative integer). 
		\item Each multi-edge $e \in \weightede$ is incident to exactly one vertex of $\epts(\boundary)$. Further, for each vertex $v \in \epts(\boundary)$, the total sum of multiplicities of multi-edges incident to $v$ is at most $2\sqrt{k^\star}$. \label{item:3multiplicity-sum}
		\item $k = k_{12} + k_{23}$, where $k, k_{12}, k_{23}$ are non-negative integers.
	\end{enumerate}
\end{definition}
\begin{figure}
	\centering
	\includegraphics[scale=0.85,page=1]{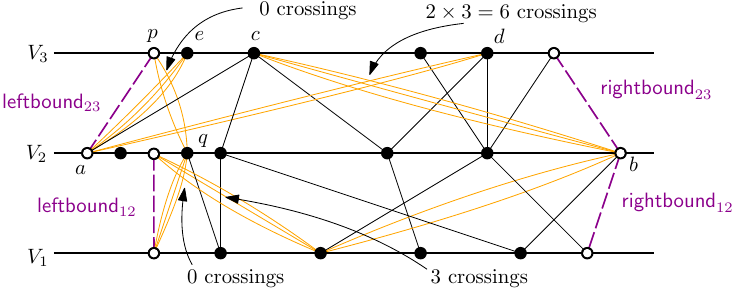}
	\caption{\small Left: Drawing of an extended instance for $h = 3$. $\boundary$ are shown as dashed magenta edges. Note that the top right and top right boundary edges share an endpoint. \weightede (shown in orange) are incident only to $\epts(\boundary)$ (shown as unfilled nodes). As explained in \Cref{def:3extdrawing}, crossings between multi-edges $ad$ and $bc$ are counted with multiplicities ($2 \times 3 = 6$), however, those between $ae$ and $pq$ are counted as $0$, since these multi-edges are incident to different endpoints of the same boundary edge, namely $\leftbtwo$.} \label{fig:extendedinstance}
\end{figure}

Next, we define the notion of a drawing of an extended instance, along the same lines as \Cref{def:2extdrawing}.

\begin{definition}[Drawing of an extended instance.] \label{def:3extdrawing}
	Let $\cI= (G, V_1, V_2, V_3, \pi, k_{12}, k_{23}, k)$ be an extended instance of \probThreeCr. We say that $\sigma = (\sigma_1, \sigma_2, \sigma_3)$ is a drawing of $\cI$ if the following properties are satisfied.
	\begin{enumerate}
		\item $\sigma_i$ is a permutation of $V_i$ for each $i \in [3]$.
		\item Ordering of vertices on respective lines:
		\begin{enumerate}
			\item $\sigma_1(\ept(\leftbone, 1)) \le \sigma_1(u_1) \le \sigma_1(\ept(\rightbone, 1))$ for all $u_1 \in V_1$.
			\item Let $\ell_{12} = \ept(\leftbone, 2), \ell_{23} = \ept(\leftbtwo, 2)$, \\$r_{12} = \ept(\rightbone, 2)$, and $r_{23} = \ept(\rightbtwo, 2)$. Then,
			\begin{itemize}
				\item $\min\LR{\sigma_2(\ell_{12}), \sigma_2(\ell_{23})} \le \sigma_2(u^l) \le \max\LR{\sigma_2(\ell_{12}), \sigma_2(\ell_{23})}$ for all $u^l \in \lbnd$,
				\item $\min\LR{\sigma_2(r_{12}), \sigma_2(r_{23})} \le \sigma_2(u^r) \le \max\LR{\sigma_2(r_{12}), \sigma_2(r_{23})}$ for all $u^r \in \rbnd$,
				\item $\min\LR{\sigma_2(\ell_{12}), \sigma_2(\ell_{23})} \le \sigma_2(u) \le \max\LR{\sigma_2(r_{12}), \sigma_2(r_{23})}$ for all $u \in V_2$.
			\end{itemize}
			\item $\sigma_3(\ept(\leftbtwo, 3)) \le \sigma_3(u_3) \le \sigma_3(\ept(\rightbtwo, 3))$ for all $u_3 \in V_3$.
			\item For each $e \in E_{12}$, for each $i \in [2]$, $\sigma_i(\ept(\leftbone, i)) \le \sigma_i(\ept(e, i)) \le \sigma_i(\ept(\rightbone, i))$.
			\item For each $e \in E_{23}$, for each $j \in \LR{2, 3}$, $\sigma_j(\ept(\leftbtwo, j)) \le \sigma_j(\ept(e, j)) \le \sigma_j(\ept(\rightbtwo, j))$.
		\end{enumerate}
		\item For each $u, v \in V_i$, if $\pi_i(u) < \pi_i(v)$, then $\sigma_i(u) < \sigma_i(v)$, i.e., $\sigma_i$ is compatible with $\pi_i$.
	\end{enumerate}
	Further, we count the crossings in the drawing $\sigma$ in the following way:
	\begin{itemize}[leftmargin=3pt]
		\item Constraints on $\sigma$ imply that no edge of $\boundary$ is crossed by any other edge.
		\item For any pair of edges $e_1, e_2 \in \regulare$, if they cross according to $\sigma$ (in the usual sense), then we count it as a single crossing.
		\item If an edge $e_r \in \regulare$ and an edge $e_m \in \weightede$ with multiplicity $\mu(e_m)$ cross in $\sigma$ (in the usual sense), then we count it as $\mu(e_m)$ crossings.
		\item Let $e_1, e_2 \in \regulare$ be two multi-edges.
		\begin{itemize}[leftmargin=10pt]
			\item If $e_1, e_2$ are incident to different endpoints of the same edge $e \in \boundary$, then they will necessarily cross in any drawing $\sigma$, that satisfies the constraints above. However we count zero crossings in $\sigma$. \footnote{These crossings will be already accounted for while creating the extended instance.}
			\item Otherwise, if $e_1, e_2$ cross in $\sigma$ (in the usual sense), then we count as $\mu(e_1) \cdot \mu(e_2)$ crossings.
		\end{itemize}
	\end{itemize}
\end{definition}


\paragraph{From the original instance to extended instance.}  Given the original instance $\cI^o= (H, U_1, U_2, U_3, k^\star)$ of \probThreeCr. First, we try each value of $0 \le k \le k^\star$ (these are $k+1$ guesses). Next, we also guess the partition $k = k_{12} + k_{23}$, such that if $\cI$ is a \yin, then it admits a drawing $\sigma$ with $k$ crossings, then there are $k_{12}$ crossings among the edges of $E_{12}$ and $k_{23}$ crossings among the edges of $E_{23}$. There are $k \le k^\star$ guesses for this. Next, we guess the leftmost and rightmost vertices in $U_1$ and $U_3$ in $\sigma$, for which there are at most $4$ choices, and hence the number of guesses is at most $n^4$. Fix one such guess, and let us call the leftmost (resp.~rightmost) vertices in $U_1$ and $U_3$ as $u^l_1, u^l_3$ (resp.~$u^r_1, u^r_3$), respectively. Now, we add two new vertices $u^l_2, u^r_2$, and add edges $(u^l_1, u^l_2), (u^l_2, u^l_3), (u^r_1, u^r_2), (u^r_2, u^r_3)$, respectively, and denote them as $\leftbone, \leftbtwo, \rightbone, \rightbtwo$, respectively. By slightly abusing the notation, we continue to use $U_i$ for the new sets of vertices. Note that we have added $2$ additional vertices and $4$ additional edges, which only increases the size of the instance by $\Oh(1)$. 
Now, it can be easily observed that, given a drawing for the original instance $\cI^o$ that satisfies $\pi_i(u^l_i) \le \pi_i(u) \le \pi_i(u^r_i)$ for $i \in \LR{1, 3}$ (if the guess is correct, then $\sigma$ is such a drawing), one can obtain a new drawing of the new instance by placing $u^l_2$ at the extreme left on $V_2$ and $u^r_2$ on the extreme right on $V_2$. Further, this does not create any additional crossings. Thus, the following observation is immediate. 
\begin{observation} \label{obs:3normal2extended}
	$\cI^o$ is a \yin iff at least one of the at most $n^{\Oh(1)}$ extended instances $\cI$ obtained after this guessing is an extended \yin. 
\end{observation}
Henceforth, we may drop the qualifier \emph{extended} from \emph{extended instance}, and simply say \emph{instance}/\emph{\yin}, respectively. Further, the graph $G$ in an instance \\$\cI = (G, V_1, V_2, V_3, \pi, k, k_{12}, k_{23})$ is said to be the \emph{underlying graph} of the instance $\cI$, and $k$ is said to be its parameter.  

Next, we continue to adopt the notation of $\crs(\cdot)$ and that of red and blue edges from \Cref{def:cross} and \Cref{def:redblue}. Further, note that \Cref{obs:atleastoneblue} continues to hold.





\subsection{Algorithm}

We use $n$ and $m$ to denote the number of vertices and edges, respectively. \footnote{Note that if we start with the kernelized instance, then $n$ and $m$ are polynomial in the original parameter $k^\star$; however we will only use the kernelization to derive our final theorem.} Next, we transform it into an extended instance, as described above. This increases the number of vertices and edges by at most $4$ each, and from \Cref{obs:3normal2extended} it follows that the kernelized instance is a \yin iff for at least one of the guesses, we have an \yin. For each extended instance corresponding to each of the guesses, we will run the following algorithm.

Let $\cI= (G, V_1, V_2, V_3, \pi, k, k_{12}, k_{23})$ be the \ein given as an input. Let $c$ be a sufficiently large constant.  In two special cases, our algorithm directly solves the given instance by a careful enumeration. These two cases correspond to, when (i) $k \le c \sqrt{k^\star}$, or (ii) $\min\LR{|E_{12}|, |E_{23}|} \le c \sqrt{k^\star}$, respectively. We first explain these two base cases. Otherwise, when $\min\LR{|E_{12}|, |E_{23}|} > c \sqrt{k^\star}$, and $k > c \sqrt{k^\star}$, we are in the recursive case, explained later.


\paragraph{Base case.} We have the following sub-cases.

\paragraph{Case 1. $k < c \sqrt{(k^\star)}$.} This base case is handled exactly as in \probTwoCr via \Cref{prop:partialorderdrawing} (\cite{bla2024constrained}), which also holds for three layers. 

\paragraph{Case 2. $\min\LR{|E_{12}|, |E_{23}|} < c \sqrt{k^\star}$.} \footnote{This base case very closely resembles the way we solve the ``middle'' subproblem in the recursive case. In fact, we could create an extended instance corresponding to the middle subproblem and make a recursive call, which will immediately trigger this base case. However, we explicitly give the reduction to the two-layer case in the middle subproblem for the ease of presentation.} Note here that we count the edges without multiplicity similar to the analogous base case for \probTwoCr. Now, suppose w.l.o.g.~ that $|E_{23}| < c \sqrt{k^\star}$ (the other case is exactly analogous). Let $V'_2 \subseteq V_2$ be the set of all the vertices defined as follows:
$$V'_2 \coloneqq \LR{u \in V_2: u \text{ has a neighbor in } V_3} \cup \bnd.$$
We know that $|V_3| \le c \sqrt{k^\star}$ (note that each vertex of $V_3$ must have at least one edge incident to it), and $|V'_2| \le c \sqrt{k^\star} + \alpha \sqrt{k^\star}$. We guess permutations $\zeta_2, \zeta_3$ of $V'_2, V_3$ for which there are $(k^\star)^{\Oh(\sqrt{k^\star})}$ choices. We terminate the guess for any choice that violates the conditions of a drawing of $\cI$, given in \Cref{def:3extdrawing}, specifically, this includes the case when the number of crossings among the edges of $E_{23}$ according to $\zeta$ is not equal to $k_{23}$. Otherwise, for each non-terminated guess, we create an instance $\cI' = (G', W_1, W_2, \bnd', \pi', k')$ of \probTwoCr as follows.
\begin{itemize}
	\item $W_1 = V_1$ and $W_2 \coloneqq \LR{u \in V_2: u \text{ has a neighbor in } V_1} \cup \bnd$.
	\item $\bnd' = \bnd = \lbnd' \cup \rbnd'$, where $\lbnd' = \lbnd, \rbnd' = \rbnd$.
	\item $\boundary' = \LR{\leftb' \cup \rightb'}$, where $\leftb' = \leftbtwo, \rightb' = \rightbtwo$.
	\item $k' = k_{12}$
	\item $\pi' = (\pi'_1, \pi'_2)$, where $\pi'_1 = \pi_1$, and $\pi'_2 = \pi_2 \cup \zeta_2$.
\end{itemize}
It is straightforward to check that the instance $\cI'$ is a valid elaborate extended instance of \probTwoCr as per \Cref{def:extendedinstance2}. Then, we solve this instance using \Cref{thm:2layertheorem} (without kernelization) and find a drawing if it is a {\sc Yes}-instance. Each such call takes time $n^{\Oh(\sqrt{k})} = n^{\Oh(\sqrt{k^\star})}$, and there are at most $(k^{\star})^{\Oh(\sqrt{k^\star})}$ recursive calls. Thus, the overall time taken in this base case is at most $n^{\Oh(\sqrt{k^\star})}$.

If for some instance $\cI'$, the algorithm returns a drawing $\so' = (\so'_1, \so'_2)$ of $\cI'$, then we can obtain a combined drawing $\sigma' = (\sigma'_1, \sigma'_2, \sigma'_3)$ of $\cI$ as follows: $\sigma'_1 = \so'_1$, $\sigma'_2 = \zeta_2 \cup \so'_2$, $\sigma'_3 = \zeta_3$. It is straightforward to check that $\sigma'$ is a valid drawing of $\cI$ with $k$ crossings. Otherwise, if all guesses are terminated, or the two-layer algorithm returns {\sc No}, then we conclude that $\cI$ is a {\sc No}-instance. This finishes the description of the second base case.

\begin{lemma} \label{lem:3layerbasecase}
	Let $\cI$ be an extended instance of \probThreeCr with underlying graph $G$ and parameter $k$. If either (i) $k \le c \sqrt{(k^\star)}$, or (ii) $\min\LR{|E_{12}|,| E_{23}|} \le c \sqrt{k^\star}$. Then, the base case of the algorithm as described above returns a drawing $\sigma$ of $\cI$ with $k$ crossings in time $n^{\Oh(\sqrt{k^\star})}$ iff $\cI$ is an extended \yin.
\end{lemma}
Otherwise, we proceed to the recursive case, explained next. 


\paragraph{Recursive case.}


Since the base case is not applicable, we know that $k > c\sqrt{k^\star}$, and $\min\LR{|E_{12}|, |E_{23}|} > c \sqrt{k^\star}$ for some sufficiently large constant $c$. In this case, at a high level, our recursive algorithm consists of the following steps.

\begin{enumerate}[leftmargin=1em]
	\item \textbf{Divide.} First, we prove some combinatorial properties of a ``nice separator'' in this step, which is defined by a set of at most $4$ edges, and $\Oh(\sqrt{k})$ additional relevant vertices and edges. This separator divides the instance into at most three instances having certain properties. We refer to these sub-instances as \emph{left, middle}, and \emph{right} sub-instances. Note that, of these, the left and right sub-instances will also be extended instances.
	\item \textbf{Guess the interfaces.} In this step, we actually guess various components of the separator (whose existence is guaranteed in the previous step, assuming that we are dealing with a \yin), other relevant information about the separator, and the rest of the interface between the sub-instances, i.e., the vertices and edges that span across multiple sub-instances. 
	\\Due to the properties of the separator, this will constitute $n^{\Oh(k^{2/3})}$ guesses. 
	\item \textbf{Creating sub-instances.} Formally defining the sub-instances based on the choices guessed so far.
	\\The construction of each sub-instance takes $k^{\Oh(1)}$ time.
	\item \textbf{Conquer left and right.} Recursively solving the left and right instances. 
	\\Inductively, this takes time $n^{\Oh(( 3k/4)^{2/3})}$.
	\item \textbf{Conquer middle.} If, for a particular guess, both sub-instances are found to be {\sc Yes}-instances, then we will use the solutions for left and right sub-problems, plus some $n^{\Oh(\sqrt{k})}$ guesses to reduce the problem to \probTwoCr, which can be solved in time $n^{\Oh(\sqrt{k})}$ using the algorithm from the previous section (\Cref{thm:2layertheorem}).
	\item \textbf{Obtaining a combined drawing.} Assuming that for some guess, both the recursive calls corresponding to left and right instances output the corresponding drawings, and we find a ``compatible'' drawing for the middle subproblem via \Cref{thm:2layertheorem}, then we combine the drawings along with the guesses made for the middle subproblem, in order to obtain a combined drawing for the current sub-instance. This step takes polynomial time for each guess.
\end{enumerate}
Let $T_3(n, k)$ denote an upper bound on the running time of our algorithm on an extended instance with a graph with at most $n$ vertices and with parameter bounded by $k$. From the above description, it follows that $T_3(n, k)$ satisfies the following recurrence. 
\begin{align}
	T_3(n, k) &\le \begin{cases}
		 n^{\Oh(k^{2/3})} \cdot (T_3(n, \frac{3k}{4}) + T_2(n, k)) + n^{\Oh(1)} & \text{ if } k > c \sqrt{k^\star} \text{ and } |E| > c\sqrt{k^\star} 
		\\n^{\Oh(\sqrt{k^\star})} & \text{ if } k \le c \sqrt{k^\star} \text{ or } |E| \le c \sqrt{k^\star}
	\end{cases}
	\label{eqn:3recurrence}
\end{align}
Here, $T_2(n, k) = n^{\Oh(\sqrt{k})}$ is the running time of the algorithm for \probTwoCr when we directly use the algorithm of \Cref{thm:2layertheorem} without the kernelization step. Then, it can be shown that this recurrence solves to $T_3(n, k^{\star}) = n^{\Oh((k^\star)^{2/3})}$.

\subsection{Step 1: Divide.} \label{subsec:3divide}
This step is subdivided into two steps. In \Cref{subsubsec:3separator}, we will show the existence of a ``balanced separator'' and explore its various properties. Then, in \Cref{subsubsec:3components}, we will show a technical claim that lets us obtain a more structured drawing (for a \yin) that is still compatible with the separator from the previous step. Now, we explore each of these steps in more detail.

\subsubsection{Existence and Properties of a Separator} \label{subsubsec:3separator}
\begin{definition}[Bottom- and Top-heavy instances] \label{def:top-bottom-heavy}
	If, in the given instance we have that $k_{12} \ge k_{23}$, i.e., $k_{12} \ge \frac{k}{2}$, then say that such an instance is \emph{bottom-heavy}; and otherwise if $k_{23} < k_{12}$, then we say that an instance is \emph{top-heavy}.
\end{definition}
For the rest of this section while describing the algorithm, we will assume that we are dealing with a bottom-heavy instance, without repeating this assumption at each step. Specifically, this assumption is crucial while defining the structure of the separator defined in this subsection, which follows. Otherwise, the top-heavy instances can be dealt with interchanging the roles of the first and third lines in the definition of separator (\Cref{def:3separator}) and the rest of the algorithm. 

We adopt and naturally extend the notions of edge orderings from \Cref{def:edgeorder} as well as that of an ordering between an edge and a crossing from \Cref{def:crossingorder} to \probThreeCr. We note that an analogous version of \Cref{obs:bluecrossings} holds in \probThreeCr as well.


\begin{definition} \label{def:3separator}
	Let $\cI$ be a $k$-minimal bottom-heavy \yin, and let $\sigma$ be any drawing of $\cI$ with $k$ crossings. We say that $\sep = \LR{\lsepone, \rsepone, \lseptwo, \rseptwo}$ is a \emph{separator} w.r.t. $\sigma$, if the following properties hold (see \Cref{fig:3layerseparator} for an example of $\sep$ and associated objects.).
	\begin{enumerate}
		\item $\lsepone, \rsepone \in E_{12}$ are two distinct blue edges \wrt $\sigma$ and they do not cross each other
		\\$\lseptwo, \rseptwo \in E_{23}$ are two distinct blue edges \wrt $\sigma$ and they do not cross each other.
		\item $\sigma_1(\ept(\lsepone), 1) \le \sigma_1(\ept(\rsepone), 1)$,
		\\$\sigma_2(\ept(\lseptwo), 2) \le \sigma_2(\ept(\lsepone, 2)) \le \sigma_2(\ept(\rsepone, 2)) \le \sigma_2(\ept(\rsepone, 2))$,
		\\$\sigma_3(\ept(\lseptwo), 3) \le \sigma_3(\ept(\rseptwo), 3)$.
		\item The number of crossings to the left of $\lsepone$, and that to the right of $\rsepone$ is at most $\frac{k_{12}}{2} $.
		\item Let $\mde(\sep, \sigma)$ denote the set of all edges $u_1u_2 \in E_{12}$ other than $\lsepone, \rsepone$, such that $$\sigma_i(\ept(\lsepone), i) \le \sigma_i(u_i) \le \sigma_i(\ept(\rsepone), i) \text{ for } i \in [2].$$
		Then $|\mde(\sep, \sigma)| \le 2\sqrt{k}$.
		\item Let $\vtwo \coloneqq \LR{u \in V_2 : N(u) \cap V_3 \neq \emptyset}$, and define
		\begin{align*}
			\gapleft(\sep, \sigma) &\coloneqq \LR{v \in \vtwo : \sigma_2(\ept(\lseptwo), 2) < \sigma_2(v) < \sigma_2(\ept(\lsepone), 2)}
			\\\gapright(\sep, \sigma) &\coloneqq \LR{v \in \vtwo : \sigma_2(\ept(\rsepone), 2) < \sigma_2(v) < \sigma_2(\ept(\rseptwo), 2)}
		\end{align*}
		Then, $|\gapleft(\sep, \sigma)|, |\gapright(\sep, \sigma)| \le 3 \sqrt{k}$.
		\item $|\crosse(\sep, \sigma)| \le 4\sqrt{k}$, where $\crosse(\sep, \sigma) \coloneqq \bigcup_{e \in \sep}\crs(e, \sigma)$. 
	\end{enumerate} 
	When the separator $\sep$ and/or the drawing $\sigma$ are obvious from the context, we may simply use $\gapleft, \gapright,$ $\mde$, and $\crosse$.
\end{definition}

We state an immediate consequence of the preceding definition.

\begin{observation} \label{obs:3balancedseparator}
	Let $\cI$ be a $k$-minimal bottom heavy extended \yin, let $\sigma$ be a drawing of $\cI$ with $k$ crossings, and let $\sep$ be a separator w.r.t. $\sigma$. Let $k^{l}_{12}, k^{r}_{12}$ denote the number of crossings that are to the left of $\lsepone$ and to the right of $\rsepone$, respectively, and $k^{l}_{23}, k^{r}_{23}$ denote the number of crossings to the left of $\lseptwo$ and to the right of $\rseptwo$, respectively. Then, $k^{l}, k^{r} \le 3k/4,$ where $k^l \coloneqq k^{l}_{12} + k^{l}_{23}$ and $k^r \coloneqq k^r_{12} + k^r_{23}$. 
\end{observation}
\begin{proof}
	Let $k_{12} = \alpha k$ and $k_{23} = (1-\alpha)k$ for some $\alpha \ge 1/2$, since $k_{12} \ge k/2 \ge k_{23}$. Then, since $\sep$ is a separator, $k^l_{12} \le k_{12}/2 = \alpha k/2$, and $k^r_{12} \le k_{23} = (1-\alpha)k$. Therefore, $k^l_{23} + k^r_{23} \le k(\alpha/2 + 1 - \alpha) = (1-\alpha/2)k \le 3k/4$, where the last inequality follows from $\alpha \ge 1/2$. Analogous argument works for $k^r$.
\end{proof}


\begin{figure}
	\centering
	\includegraphics[page=13,scale=1.4]{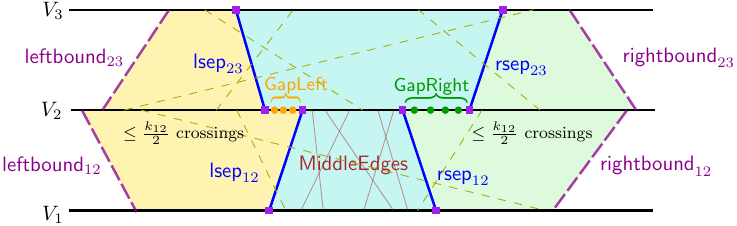}
	\caption{Illustration of a separator from \Cref{def:3separator}. $\sep = \LR{\lsepone, \rsepone, \lseptwo, \rseptwo}$ are shown in \blue{blue}, and $\epts(\sep)$ are shown in \textcolor{BlueViolet}{purple}. Edges of \textcolor{Brown}{$M \coloneqq \mde(\sep)$} shown in \textcolor{Brown}{brown} and that of \textcolor{Olive}{$\crosse(\sep)$} are shown in dashed {\color{Olive}olive} -- these are the edges that cross at least one edge of $\sep$. Vertices of \textcolor{Orange}{\gapleft} are shown in \textcolor{Orange}{orange} and that of \textcolor{DarkGreen}{\gapright} in \textcolor{DarkGreen}{green}.}
	\label{fig:3layerseparator}
\end{figure}

Next, we state the following important lemma that shows the existence of a separator, which forms the basis of our divide-and-conquer algorithm.

\begin{lemma}[Separator Lemma for \probThreeCr] \label{lem:3separatorlemma}
	Let $\cI$ be a bottom-heavy $k$-minimal \yin and let $\sigma$ be any drawing of $\cI$ with $k$ crossings. Then, there exists a separator $\sep$ w.r.t. $\sigma$.
\end{lemma}
	
\begin{proof}
	Note that due to $k$-minimality and bottom-heaviness of $\cI$, we know that the number of crossings among the edges of $E_{12}$ is larger than $k_{12}/2$; note here that $k_{12}/2 < k_{12}$, since $k_{12} \ge k/2 > c/2$ for some absolute constant $c$. We divide the proof into four parts: (I ) constructing $\lsepone, \rsepone$, (II) bounding the size of $\mde$, (III) constructing $\rseptwo$ and analyzing $\gapright$, and (IV) constructing $\lseptwo$ and analyzing  $\gapleft$. The first two parts are exactly identical to the proof of \Cref{lem:2separatorlemma}; nevertheless we include them here for the sake of completeness.
	
	\paragraph{Part (I): Constructing $\lsepone, \rsepone$.} We define a lexicographic order $\prec$ on the edges of $E$ using the order of their endpoints in $V_1$ using $\sigma_1$, with ties broken using the endpoints in $V_2$ using $\sigma_2$. More formally, for two distinct edges $e_a, e_b \in E$, we have $e_a \prec e_b$  iff $\sigma_1(\ept(e_a, 1)) < \sigma_1(\ept(e_b, 1))$, or $\sigma_1(\ept(e_a, 1)) = \sigma_1(\ept(e_b, 1))$ and $\sigma_2(\ept(e_a, 2)) < \sigma_2(\ept(e_b, 2))$ (multi-edges are treated as a single edge for this order). Let $\rsepone \in E(G)$ be the first blue edge, according to $\prec$, that has more than $\frac{k_{12}}{2}$ crossings to its left. First we claim that $\rsepone$ exists, since $\rightbone$---which is blue, and has all ($k_{12} > \frac{k_{12}}{2}$) crossings to its left--- satisfies the stated property.
	
	By the definition of $\rsepone$, the number of crossings to its left is larger than $\frac{k_{12}}{2}$, which implies the following: (a) the number of crossings to its right is at most $\frac{k_{12}}{2} \le \frac{k}{2}$ (note that some crossings may not be either to the left or to the right of $\rsepone$). (b) the number of edges to its left is at least $\frac{k_{12}}{2} > 0$. See \Cref{fig:3layerseparator} for an illustration. Further, (b) implies that $\leftbone$ is one of the $\frac{k_{12}}{2}$ edges that is to the left of $\rsepone$, which specifically implies that $\leftbone \neq \rsepone$.

	Now, we define $\lsepone$ to be the last edge according to $\prec$ such that: (i) $\lsepone$ is blue, (ii) $\lsepone \prec \rsepone$, and (iii) $\lsepone$ does not cross $\rsepone$ according to $\sigma$. First, $\lsepone$ exists because $\leftbone$ satisfies conditions (i) - (iii) -- indeed $\leftbone \neq \rsepone$ as shown above. Further, since $\lsepone \prec \rsepone$, $\lsepone$ is to the left of $\rsepone$ and in particular the two edges do not cross each other. Now, if the number of crossings to the left of $\lsepone$ is larger than $\frac{k_{12}}{2}$, then this contradicts the choice of $\rsepone$ as the lexicographically first blue edge with more than $\frac{k_{12}}{2}$ crossings to its left. Therefore, the number of crossings to the left of $\lsepone$ is at most $\frac{k_{12}}{2}$. This finishes the first part.
	
	\paragraph{Part (II): Bounding the size of $\mde$.} Suppose for the contradiction $|\mde| > 2\sqrt{k}$. Note that the definition of $\mde$, for each edge $e \in \mde$, $\lsepone \prec e \prec \rsepone$, and $e$ does not cross $\rsepone$. Then, by \Cref{obs:atleastoneblue}, $\mde$ contains at least one blue edge, say $b$. However, note that $b$ satisfies conditions (i)--(iii) above, which contradicts the choice of $\lsepone$, since $\lsepone$ is the lexicographically last such edge. This shows that $|\mde(\sep)| \le 2\sqrt{k}$.
	
	\paragraph{Part (III): Constructing $\rseptwo$ and analyzing $\gapright$.} Now we define a similar lexicographic order $\prec'$ on the edges of $E_{23}$, except first we consider the order of the $V_2$ endpoints, and then the ties are broken using $V_3$ endpoints (multi-edges are treated as a single edge for this order). 
	
	
	Let $\rseptwo$ be a lexicographically first blue edge according to $\prec'$ such that (i) $\rseptwo \neq \leftbtwo$, and (ii) $\sigma_2(\ept(\rsepone), 2) \le \sigma_2(\ept(\rseptwo), 2)$, if it exists. \footnote{Note that, if $\sigma_2(\ept(\rsepone, 2)) \le \sigma_2(\rightbtwo, 2))$, then $\rightbtwo$ satisfies both the conditions (indeed $\rightbtwo \neq \leftbtwo$ since $|E_{23}| > c \sqrt{k^\star}$). In case $\rseptwo$ does not exist, then we define $\rseptwo$ to be a ``phantom (placeholder) edge'', and $\gapleft = \emptyset$. When such an edge will be chosen for dividing the problem into sub-problems, the right sub-problem will be a degenerate instance with zero edges in $E_{23}$, which can be then solved by the second base case. We omit the details.}
	Then, consider the set of vertices $\gapright$ as defined above. First, if $\ept(\rseptwo, 2) = \ept(\rsepone, 2)$, then $\gapright = \emptyset$. Now, consider $\ept(\rseptwo, 2) \neq \ept(\rsepone, 2)$, and suppose for the sake of contradiction that $|\gapright| > 2\sqrt{k}$. Let $E^r = \LR{wz \in E_{23}: w \in \gapright, z \in V_3}$. Note that $|E^r| \ge |\gapright| > 2\sqrt{k}$, since each $w \in \gapright$ contributes at least one edge to $E^r$, and these edges are distinct. Then, by \Cref{obs:atleastoneblue}, $E^r$ contains at least one blue edge $b$. Note that $b \prec' \rseptwo$, since $\sigma_2(\ept(b, 2)) < \sigma_2(\ept(\rseptwo, 2))$ by construction, which contradicts the choice of $\rseptwo$ as the lexicographically first such edge. Thus, $|\gapright| \le 2 \sqrt{k} \le 3 \sqrt{k}$.
	
	\paragraph{Part (IV). Constructing $\lseptwo$ and analyzing $\gapleft$.}
	We define $\lseptwo$ to be the lexicographically last blue edge (if it exists) such that (i) $\sigma_2(\ept(\lseptwo), 2) \ge \sigma_2(\ept(\lsepone), 2)$, and (ii) $\lseptwo$ does not cross $\rseptwo$. \footnote{A degenerate corner case arises when $\leftbtwo$ does not satisfy the conditions, and hence $\lseptwo$ does not exist. This can be handled analogous to the previous case. We omit the details.}
	Now consider $\gapleft$ as defined above. First, if $\ept(\lseptwo, 2) = \ept(\lsepone, 2)$, then $\gapleft = \emptyset$. Now, suppose $\ept(\rseptwo, 2) \neq \ept(\rsepone, 2)$, and suppose for the sake of contradiction that, $|\gapleft| > 2\sqrt{k}$. Then, let $E^l = \LR{wz \in E_{23}: w \in \gapleft, z \in V_3}$. Note that by a similar argument, $|E^l| \ge |V^l| > 3\sqrt{k}$, since each $w \in V^l$ contributes at least one edge to $E^r$. Let $E'^r \subseteq E^r$ be the set of edges that do not cross $\rseptwo$. Note that $|E'^r| \ge |E^r| - \sqrt{k} > 2\sqrt{k}$, since $\rseptwo$ is a blue edge. Then, by \Cref{obs:atleastoneblue}, $E'^r$ contains at least one blue edge, say $b'$. Further, $\lseptwo \prec' b$ since $\sigma_2(\ept(\lseptwo, 2)) < \sigma_2(\ept(b', 2))$. This contradicts the choice of $\lseptwo$ as the lexicographically last edge. This shows that $|\gapleft| \le 3\sqrt{k}$.

\end{proof}

Note that a symmetric version of \Cref{def:3separator} and \Cref{lem:3separatorlemma} holds for top-heavy instances. We omit the definition and the proof. 

In the following definition, we divide (not necessarily in disjoint sets) the vertex sets into three parts based on a drawing of an extended instance. These definitions will be used later on during the analysis of the algorithm.

\begin{definition} \label{def:3division}
	Let $\cI$ be a \yin with underlying graph $G$, and let $\sigma$ be a drawing of $\cI$ with $k$ crossings. Consider a separator $\sep$ w.r.t. $\sigma$ along with the associated objects as in \Cref{def:3separator}. Then, we define the following division of vertices on each layer into left (\lf), middle (\mm), and right (\rt) parts, as follows.
	\begin{align*}
		V_1(\sigma, \sep,  \lf) &\coloneqq \LR{u \in V_1: \sigma_1(u) \le \sigma_1(\ept(\lsepone, 1))}
		\\V_1(\sigma, \sep,  \mm) &\coloneqq \LR{u \in V_1 : \sigma_1(\ept(\lsepone, 1) \le \sigma_1(u) \le \sigma_1(\ept(\rsepone, 1))}
		\\V_1(\sigma, \sep,  \rt) &\coloneqq \LR{u \in V_1: \sigma_1(\ept(\rsepone, 1)) \le \sigma_1(u)}
		\\V_2(\sigma, \sep, \lf) &\coloneqq \LR{u \in V_2: \sigma_2(u) \le \sigma_2(\lseptwo, 2)} \cup \gapleft \cup \ept(\lsepone, 2)
		\\V_2(\sigma, \sep, \mm) &\coloneqq \LR{u \in V^{conn}_2: \sigma_2(\ept(\lseptwo, 2)) \le \sigma_2(u) \le \sigma_2(\ept(\lsepone, 2))}  \\&\qquad \cup \LR{u \in V_2: \sigma_2(\ept(\lsepone, 2)) \le \sigma_2(u) \le \sigma_2(\ept(\rsepone, 2))} 
		\\&\qquad \cup \LR{u \in V^{conn}_2: \sigma_2(\ept(\rsepone, 2)) \le \sigma_2(u) \le \sigma_2(\ept(\rseptwo, 2))}
		\\V_2(\sigma,\sep, \rt) &\coloneqq \LR{u \in V_2: \sigma_2(\ept(\rseptwo, 2)) \le \sigma_2(u)} \cup \gapright \cup \ept(\rsepone, 2)
		\\V_3(\sigma,\sep, \lf) &\coloneqq \LR{u \in V_3: \sigma_3(u) \le \sigma_3(\ept(\lseptwo, 3))}
		\\V_3(\sigma,\sep, \mm) &\coloneqq \LR{u \in V_3 : \sigma_3(\ept(\lseptwo, 1) \le \sigma_3(u) \le \sigma_3(\ept(\rseptwo, 1))}
		\\V_3(\sigma, \sep, \rt) &\coloneqq \LR{u \in V_3: \sigma_3(\ept(\rseptwo, 3)) \le \sigma_3(u)}
	\end{align*}
	Finally, for each ${\cal T} \in \LR{\lf, \mm, \rt}$, define $V(\sigma, {\cal T}) \coloneqq V_1(\sigma, \sep,  {\cal T}) \cup V_2(\sigma, \sep, {\cal T}) \cup V_3(\sigma, \sep, {\cal T})$.
\end{definition}

Now, fix a $k$-minimal \yin, a corresponding drawing $\sigma$ of $\cI$ with $k$ crossings, and a separator $\sep$ (and the associated objects from \Cref{def:3separator}) as guaranteed by \Cref{lem:3separatorlemma}. In the next part, we classify the components in $G-v$ for $v \in \epts(\sep)$ into multiple types, and then analyze their behavior in $\sigma$. This is largely similar to the two layer case; except that the analogue of pendant components for the three layer case is technically much more challenging to analyze, which is the focus of the next subsection.


\begin{definition} \label{def:3componenttypes}
	Let $v \in \epts(\sep)$ be a vertex, and let $\comps(v)$ denote the set of connected components of $C_v-v$, where $C_v$ is the component of $G$ that contains $v$. We partition these components into different classes (subsets), called $\spl(v), \predef(v), \updown(v),$ and $\pecu(v)$ (we may omit the subscript $G$). A component $C \in \comps_G(v)$ belongs to:
	\begin{itemize}[leftmargin=*]
		\item $\spl(v)$, if $E(G[C \cup \LR{v}])$ contains an edge of $\crosse(\sep) \cup \mde(\sep)$, or if $C$ contains a vertex of $\gapleft \cup \gapright \cup \epts(\sep) \setminus \LR{v}$.
		\item $\predef(v)$, if $C \not\in \spl(v)$, and $C$ contains a vertex $u \in \bigcup_{i \in [3]} \bigcup_{j \in [\lambda]} V^j_i$ \footnote{Recall that $V^j_i$'s are the $\lambda$ subsets over which the total orders $\pi^j_i$ are defined, and these $\lambda$ total orders constitute $\pi_i$.}
		\item $\updown(v)$, if it does not belong to $\spl(v) \cup \predef(v)$, and $C$ contains at least one vertex on the same line as $v$.
		\item $\pecu(v)$, if it does not belong to $\spl(v) \cup \predef(v)$, and $C$ does not contain a vertex on the same layer as $v$.
		\\We also define a subset $\pendant(v) \subseteq \pecu(v)$ consisting of components that consist of a single pendant vertex adjacent to $v$ (possibly with a multi-edge). 
		\\Note that if $v \in V_2$, then $\pendant(v) = \pecu(v)$. However, if $v \in V_1 \cup V_3$, then we define $\pecunp(v) \coloneqq \pecu(v) \setminus \pendant(v)$ (peculiar non-pendants), and these are the components that contain vertices of the remaining two layers other than that of $v$.
	\end{itemize}
\end{definition}

First, the following observation follows from the bounds on $\crosse(\sep), \mde(\sep)$, $ \gapleft$, and $\gapright$.

\begin{observation} \label{obs:3specialcomps}
	For any $v \in \epts(\sep)$, it holds that $|\spl(v)| \le 12 \sqrt{k} + 8 = \Oh(\sqrt{k})$.
\end{observation} 

Next, we bound the size of $\updown$ in the following lemma, which is done via arguments almost identical to that in \Cref{lem:2type3components}. 

\begin{lemma} \label{lem:3type3components}
	For any $v \in \epts(\sep)$, it holds that $|\updown(v)| \le 8\sqrt{k}$.
\end{lemma}
\begin{proof}
	Consider any $v \in \epts(\sep) \cap (V_1 \cap V_3)$. Then, one can use arguments exactly as in \Cref{lem:2type3components} to show that, in fact, $|\updown(v)| \le 4 \sqrt{k} \le 8 \sqrt{k}$. Now, consider $v \in \epts(\sep) \cap V_2 = \epts(\sep, 2)$, and suppose that $|\updown(v)| > 8 \sqrt{k}$. Then, we classify the components of $\updown(v)$ into two types, say (12) and (32), where type 12 components contain a neighbor of $v$ in $V_1$, and its neighbor in $V_2$ (it may possibly contain vertices of $V_3$); whereas type 32 components contain a neighbor of $v$ in $V_3$ and its neighbor in $V_2$. If a component satisfies both conditions, then arbitrarily place it in one of them. At least the number of components of either of the two types of components must be larger than $4\sqrt{k}$. Then, one can use an argument similar to that in \Cref{lem:2type3components} to arrive at a contradiction. This proves that for any $v \in \epts(\sep, 2)$, it holds that $|\updown(v)| \le 8\sqrt{k}$, which finishes the proof.
\end{proof}

To analyze the peculiar components, we again need to consider different cases based on the structure of the separator. 

\begin{definition} \label{def:3separatorcases}
	Let $\sep = \LR{\lsepone, \rsepone, \lseptwo, \rseptwo}$ be the separator as guaranteed by \Cref{lem:3separatorlemma}. We have the following different (non-exclusive) cases.
	\begin{description}[leftmargin=0pt]
		\item[(i) {\sf Distinct Bottom (\DB)}:] if $\ept(\lsepone, 1) \neq \ept(\rsepone, 1)$. 
		\\{\sf Shared Bottom (\SB):} if $\ept(\lsepone, 1) = \ept(\rsepone, 1)$. 
		\item[(ii) {\sf Distinct Top (\DT)}:] if $\ept(\lsepone, 1) \neq \ept(\rsepone, 1)$. 
		\\{\sf Shared Top (\ST):} if $\ept(\lsepone, 1) = \ept(\rsepone, 1)$. 
	\end{description}
	Based on this, we have the following disjoint cases: (\DB, \DT), (\DB, \ST), (\SB, \DT), (\SB, \ST).
\end{definition}

Next, we introduce an additional definition that further classifies the components of .

\begin{definition} \label{def:3peculiartypes}
	We say that a vertex subset $C \in V_1 \cup V_2$ (resp.$C \in V_2 \cup V_3)$ is to the left of $e \in E_{12}$ (resp.~$E_{23}$), if for each $i \in \LR{1,2}$ (resp.~$i \in \LR{2,3}$) it holds that for all $u \in C \cap V_i$, $\sigma_i(u) \le \sigma_i(\ept(e, i))$. Analogously, we define ``right of'' relation if the inequality is flipped.
	\begin{itemize} 
		\item For each $v \in \LR{\ept(\lsepone, 1), \ept(\rsepone, 1)}$, let $\pecu^l(v) \subseteq \pecu(v)$ \\(resp.~$\pecunp^l(v) \subseteq \pecunp(v)$) be the set of components $C$ that are to the left of $\lseptwo$. 
		\item For each $v \in \LR{\ept(\lsepone, 1), \ept(\rsepone, 1)}$, let $\pecu^r(v) \subseteq \pecu(v)$ \\(resp.~$\pecunp^r(v) \subseteq \pecunp(v)$) be the set of components $C$ that are to the right of $\rseptwo$.
		\item For each $v \in \LR{\ept(\lseptwo, 1), \ept(\rseptwo, 1)}$, let $\pecu^l(v) \subseteq \pecu(v)$ \\(resp.~$\pecunp^l(v) \subseteq \pecunp(v)$) be the set of components $C$ that are to the left of $\lseptwo$. 
		\item For each $v \in \LR{\ept(\lseptwo, 1), \ept(\rseptwo, 1)}$, let $\pecu^r(v) \subseteq \pecu(v)$ \\(resp.~$\pecunp^r(v) \subseteq \pecunp(v)$) be the set of components $C$ that are to the left of $\rseptwo$. 
	\end{itemize}
\end{definition}

We next prove the following lemma.

\begin{lemma} \label{lem:3diffcases}
	The following bounds hold: 
	\begin{enumerate}
		\item $\min\LR{|\pecunp^l(\ept(\lsepone, 1)), \pecunp^l(\ept(\lseptwo, 3))|} \le \sqrt{k}$.
		\item $\min\LR{|\pecunp^r(\ept(\rsepone, 1)), \pecunp^r(\ept(\rseptwo, 3))|} \le \sqrt{k}$.
	\end{enumerate} 
\end{lemma}
\begin{proof}
	We will show the first item, and the proof of the second one is analogous. Let us simplify the notation and say 
	$x_1 \coloneqq \ept(\lsepone, 1), x_2 \coloneqq \ept(\lseptwo, 3)$.
	Suppose for the contradiction that $|\pecunp^l(x_1)| > \sqrt{k}$ and $|\pecunp^l(x_1)| > \sqrt{k}$. Note that $\pecunp^l(x_1) \cap \pecunp^l(x_2) = \emptyset$, since a component $C \in \pecunp^l(x_1)$ contains a neighbor of $x_1$ and hence in $G-\ept(x_2)$, $V(C) \cup \LR{x_2}$ is part of one connected component whose vertex set is a strict superset of that of $C$. In the next paragraph, we will show that each pair of components $C_1 \in \pecunp^l(x_1)$ and $C_2 \in \pecunp^l(x_2)$ must account for a distinct crossing. This implies that the number of crossings between the components of $\pecunp^l(x_1)$ and $\pecunp^l(x_2)$ is larger than $k$, a contradiction.
	
	Let $u_1 \in V(C_1) \cap V_2$ be the rightmost neighbor of $x_1$ in $C$, and let $w_1 \in V(C_1) \cap V_3$ be the rightmost vertex in $V(C_1) \cap V_3$. Similarly, define $u_2 \in V(C_2) \cap V_2$ as the rightmost neighbor of $x_2$ in $C_2$ and let $w_2 \in V_1 \cap V(C_2)$ be the rightmost vertex of $V(C_2) \cap V_1$. If $w_2$ has a neighbor $q_2$ that is to the right of $u_1$, then $w_2q_2$ crosses $x_1u_1$ (see \Cref{fig:3type2crossing}, case (a)). Otherwise, if all neighbors of $w_2$ are to the left of $u_1$. Let $q_2$ be one such neighbor of $w_2$. There exists a path between $u_2$ and $q_2$ that is entirely within $C_1$. Since $u_1$ is between $q_2$ and $u_2$, such a path must cross either the edge $x_1u_1$ or the edge $u_1w_1$ (see \Cref{fig:3type2crossing}, case (b)). 
	
	
	\begin{figure}
		\centering
		\includegraphics[scale=1,page=9]{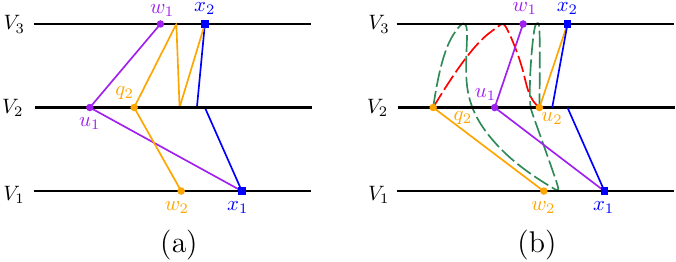}
		\caption{Illustration for the proof of \Cref{lem:3diffcases}. In case (b), we show two possible paths from $q_2$ to $u_2$ -- the green path crosses $x_1u_1$; whereas the red path crosses $u_1w_1$.} \label{fig:3type2crossing}
	\end{figure}
	
\end{proof}

In the next observation, we consolidate some properties regarding the components $\comps(v)$ for $v\in \epts(\sep)$ in different cases.

\begin{lemma} \label{lem:3pecunp-cases}\ 
	\begin{enumerate}[leftmargin=2pt]
		\item In {\sf \DB} case, 
		\begin{enumerate}[leftmargin=2pt]
			\item $\pecunp^l(\ept(\lsepone, 1) = \pecunp(\ept(\lsepone, 1))$. Further, each $C \in \pendant(\lsepone, 1)$ lies to the left of $\ept(\lsepone, 1)$.
			\item $\pecunp^r(\ept(\rsepone, 1) = \pecunp(\ept(\rsepone, 1))$. Further, each $C \in \pendant(\rsepone, 1)$ lies to the left of $\ept(\rsepone, 1)$.
		\end{enumerate}
		\item In {\sf \DT} case,
		\begin{enumerate}[leftmargin=2pt]
			\item $\pecunp^l(\ept(\lseptwo, 3) = \pecunp(\ept(\lseptwo, 3))$.
			\item $\pecunp^r(\ept(\rseptwo, 1) = \pecunp(\ept(\rseptwo, 3))$.
		\end{enumerate}
		\item In {\sf \SB} case, when $\ept(\lsepone, 1) = \ept(\rsepone, 1) = v_b$, say, then $(\pecu^l(v_b), \pecu^r(v_b))$ is a partition of $\pecu(v_b)$.
		\item In {\sf \ST} case, when $\ept(\lseptwo, 3) = \ept(\rseptwo) = v_t$, say, then $(\pecunp^l(v_t), \pecunp^r(v_t))$ is a partition of $\pecunp(v_t)$.
	\end{enumerate}
\end{lemma}
\begin{proof} \ 
	\begin{enumerate}[leftmargin=2pt]
		\item {\sf \DB case.} We give the proof of item 1 and the proof of 2 is analogous. Let $v = \ept(\lsepone, 1)$, and we will prove that $\pecunp^l(v) = \pecunp(v)$. Consider a component $C \in \pecunp(v)$, and suppose for the contradiction that there exists some edge $e = uw \in E(C)$  that is not to the left of $\lseptwo$. If $e = \lseptwo$, or if $e$ crosses $\lseptwo$, then $C$ would belong to $\spl(v)$, which is a contradiction. 
		Now, suppose that $e = uw$ is to the right of $\lseptwo$, where $u \in V_2$ and $w \in V_3$. Since $C$ is a connected component in $G-v$, there is a path $P = (v = v_0, v_1, \ldots, v_\ell = u)$ such that $vv_1 \in E(G)$, and $(v_i, v_{i+1}) \in E(C)$ for $1 \le i \le \ell-1$. Let us also observe that, except for $v \in V_1$, all other vertices in $P$ alternate between vertices of $V_2$ and $V_3$,  with odd indexed vertices belonging to $V_2$ and even indexed vertices belonging to $V_3$, -- if $P$ contained a vertex of $V_1$, then $C \in \updown(v)$, a contradiction.

		If $\sigma_3(\ept(\rseptwo, 3)) < \sigma_3(w)$, then we say that this is \textbf{case 1}, otherwise if $\sigma_3(\ept(\lseptwo, 3)) < \sigma_(w) < \sigma_3(\ept(\rseptwo, 3))$, then we say that this is \textbf{case 2}.
		
		In \textbf{case 1}, note that there exists an edge of the path $P$ must either cross $\lseptwo$, $\rseptwo$, or pass through a vertex of $\gapright$. In each of the cases, this implies that $C \in \spl(v)$, which is a contradiction.

		Now, let us consider \textbf{case 2}. Let $v_j \in C \cap V_2$, $j > 0$ be the first vertex along the path that lies to the right of $\ept(\lseptwo, 2)$. Note that $v_j$ exists, since the last vertex, i.e., $u$ satisfies the condition. We will show that in this case, $C$ must belong to $\spl(v)$, a contradiction.
		
		First, if $j = 1$, then let us consider two cases. (i) If $\sigma_2(\ept(\lseptwo, 2)) < \sigma_2(v_1) < \sigma_2(\ept(\lsepone, 2))$, then $v_1 \in \gapleft$. (ii) If $\sigma_2(v_1) < \sigma_2(\ept(\lsepone, 2)) < \sigma_2(v_1) < \sigma_2(\ept(\rsepone, 2))$, then $vv_1 \in \mde(\sep)$. (iii) Finally, if $\sigma_2(v_1) > \sigma_2(\ept(\rsepone, 2))$, then $vv_1$ crosses $\rsepone$. Thus, in each of the three sub-cases,  we obtain that $C \in \spl(v)$, which is a contradiction. Note that this same argument also shows the second claim about components from $\pendant(v)$.
		
		Next, suppose $j > 1$. First, note that $j \ge 3$, and $v_{j-1} \in V_3$. (i) If $\sigma_3(v_{j-1}) < \sigma_3(\ept(\lseptwo, 3))$, then the edge $v_{j-1} v_j$ crosses $\lsepone$. (ii) Otherwise, $\sigma_3(\ept(\lseptwo, 3)) < \sigma_3(v_{j-1})$. Note that $v_{j-2} \in V_2$ satisfies that $\sigma_2(v_{j-2}) < \sigma_2(\ept(\lseptwo, 2))$. However, this implies that the edge $v_{j-2} v_{j-1}$ crosses $\lseptwo$. Thus, again in each of the two sub-cases, we obtain that $C \in \spl(v)$, a contradiction. 
	\item {\sf \DT case.} Again, we prove the first item, and the proof of the second item is analogous. Again, for convenience, let $v = \ept(\lseptwo, 3)$, and consider any $C \in \pecunp(v)$. Suppose for contradiction that $C \not\in \pecunp^l(v)$, which implies that $C$ contains an edge $uw$ with $u \in V_2, w \in V_1$ that is not to the left of $\lsepone$. Let us again consider a path $P = (v = v_0, v_1, v_2, \ldots, v_{\ell} = u)$ where $vv_1 \in E(G)$ and $(v_i, v_{i+1}) \in E(C)$. Again, the vertices must alternate between $V_2$ and $V_1$ respectively. Let $v_j \in V_2$ be the first vertex along the path that is to the right of $\ept(\lsepone, 2)$. 
	
	If $j = 1$, then consider the following cases. (i) $\sigma_1(v_2) < \sigma_1(\ept(\lsepone, 1))$, then the edge $v_1v_2$ crosses $\lsepone$. (ii) If $\sigma_1(\ept(\lsepone, 1)) < \sigma_1(v_2) < \sigma_1(\ept(\rsepone, 1))$, then $v_1v_2 \in \mde(\sep)$. (iii) If $\sigma_1(\ept(\rsepone, 1)) < \sigma_2(v_2)$, then $v_1v_2$ crosses $\rsepone$. Thus, in each case $C \in \spl(v)$, a contradiction.
	
	If $j > 1$, then this proof is analogous to a similar case from the previous part, and hence we omit it. 
	
	\item {\sf SB case.} In this case, suppose for contradiction that there exists some $C \in \pecu(v_b)$ such that $C \not\in \pecu^l(v_b) \cup \pecu^r(v_b)$. We show that $C \in \spl(v_b)$, a contradiction. First, if $C$ is a pendant vertex $u$, then using an argument similar to {\sf DB case} above, we can show that either $u \in \gapleft \cup \gapright$, or the edge $v_bu \in \mde(\sep)$, a contradiction. 
	Note that for some $i \in \LR{2, 3}$, there exists a vertex $w \in V(C) \cap V_i$ such that $\sigma_i(\ept(\lseptwo, i)) < \sigma_i(w) < \sigma_i(\ept(\rseptwo, i))$. Suppose w.l.o.g. that $w \in V_2$. Then, let $w' \in V_3$ be a neighbor of $w$ in $C$. Note that $w'$ must also satisfy $\sigma_3(\ept(\lseptwo, 3)) < \sigma_3(w') < \sigma_3(\ept(\rseptwo, 3))$ -- otherwise $ww'$ edge crosses either $\lseptwo$ or $\rseptwo$, a contradiction. However, if $w'$ lies between the top endpoints of $\lseptwo$ and $\rseptwo$, then we can show that $C \in \spl(v_b)$ using an argument similar to the {\sf DB} case above, and arrive at a contradiction. 
	
	\item {\sf ST case.} Here, the argument is exactly similar to the {\sf SB case}, and is therefore omitted. Note that in this case, the argument regarding the pendant vertices does not hold, since the edges of $\mde(\sep)$ are only defined in $E_{23}$.  
	\end{enumerate}
\end{proof}

\paragraph{Handling Pendant Components.} In this step, we will handle the pendant components attached to $\epts(\sep)$. At a high level, the idea is similar to that used for the 2 layer case, except that we may have at most 8 vertices in $\epts(\sep)$. \footnote{By a more careful case analysis, it may be shown that we do not have to handle each of the 8 vertices, which further restricts the number of guesses; however since this does not substantially affect the running time, we handle all the vertices in $\epts(\sep)$ in a unified fashion.}

\begin{definition} \label{def:3pendantblueprint}\ 
	\begin{itemize}[leftmargin=3pt]
		\item A \emph{pendant-blueprint} is a collection of tuples $$\pbp = \LR{(\pendantsl(v), \pendantsm(v),  \pendantsr(v)): v \in \epts(\sep)},$$ such that, (i) for each $v \in \epts(\sep)$, $\pendantsl(v) + \pendantsm(v) + \pendantsr(v) = \pendantst(v)$, which denotes the total sum of multiplicities of $\pendant(v)$, and (ii) for each $v' \in \epts(\sep, 1)$, $\pendantsm(v) = 0$. 
		\item A \emph{pendant-guess} is a collection of tuples $$\pguess = \LR{(\pendantl(v), \pendantm(v), \pendantr(v)): v \in \epts(\sep)}$$ where (i) for each $v \in \epts(\sep)$, $(\pendantl(v), \pendantm(v), \pendantr(v))$ is a partition of $\pendant(v)$, and (ii) for each $v' \in \epts(\sep, 1)$, $\pendantm(v') = \emptyset$.
		\item We say that a $\pguess = \LR{(\pendantl(v), \pendantsm(v), \pendantr(v)): v \in \epts(\sep)}$ is \emph{$\pbp$-compliant} if, for each $v \in \epts(v)$ the total sum of multiplicities of vertices in $\pendantl(v)$ (resp.~$\pendantsm(v), \pendantsr(v)$) is equal to $\pendantl(v)$ (resp.~$\pendantm(v), \pendantr(v)$).
		
		\item Let $\pbp$ be a pendant-blueprint, and let $\pguess$ be a $\pbp$-compliant pendant-guess. Then, we define a $\nicepfamily(\pbp, \pguess, \sep, \sigma)$ to be a family of drawings $\sigma'$ satisfying the following properties.
		\begin{itemize}
			\item For each $i \in [3]$, $\sigma_i$ and $\sigma'_i$, when restricted to
			$\displaystyle V_i \setminus \bigcup_{v \in \epts(\sep)} \pendant(v)$, are the same.
			\item For each $v \in \epts(\sep, 1) \cup \epts(\sep, 3)$, 
			$\sigma'_2$ places (i) all the vertices in $\pendantl(v)$ to the left of $\ept(\lseptwo, 2)$, (ii) all the vertices of $\pendantm(v)$ between $\ept(\lseptwo, 2)$ and $\ept(\rseptwo, 2)$, and (iii) all the vertices in $\pendantr(v)$ to the right of $\ept(\rseptwo, 2)$.
			\item For each $v_2 \in \epts(\sep, 2)$, 
			\begin{itemize}
				\item $\sigma'_1$ places (i) all the vertices in $\pendantl(v)$ to the left of $\ept(\lsepone, 1)$, (ii) all the vertices in $\pendantm(v)$ between $\ept(\lseptwo, 2)$ and $\ept(\rseptwo, 2)$, and (iii) all the vertices in $\pendantr(v)$ to the right of $\ept(\rsepone, 1)$,
				\item $\sigma'_3$ places (i) all the vertices in $\pendantl(v)$ to the left of $\ept(\lseptwo, 3)$, (ii) all the vertices in $\pendantm(v)$ between $\ept(\lseptwo, 3)$ and $\ept(\rseptwo, 3)$, and (iii) all the vertices in $\pendantr(v)$ to the right of $\ept(\rseptwo, 3)$. 
			\end{itemize} 
%
		\end{itemize}
	\end{itemize}
\end{definition}

Next, we introduce the notion of separator-compatible drawings.
\begin{definition} \label{def:3sepcompatible}
	Let $\cI$ be a $k$-minimal \yin and $\sigma, \sigma'$ be drawings with $k$ crossings. Further, let $\sep$ be a separator w.r.t. $\sigma$. 
	We say that $\sigma'$ is a \emph{$(\sep, \sigma)$-compatible drawing} if it satisfies the following properties.
	\begin{itemize}
		\item $\sep$ is also a separator w.r.t. $\sigma'$, such that the number of crossings on each side of each edge in $\sep$ is same in $\sigma$ and $\sigma'$, 
		\item $\gapleft(\sep, \sigma) = \gapleft(\sep, \sigma')$, and $\gapright(\sep, \sigma) = \gapright(\sep, \sigma')$.
		\item $\mde(\sep, \sigma) = \mde(\sep, \sigma')$, and $\crosse(\sep, \sigma) = \crosse(\sep, \sigma')$.
	\end{itemize}
\end{definition}

Next, we state a transitive property of separator-compatibility.
\begin{observation} \label{obs:3compatibletransitive}
	Let $\cI$ be a $k$-minimal extended instance, and let $\sigma_1, \sigma_2, \sigma_3$ be drawings of $\cI$ with $k$ crossings. Let $\sep$ be a separator w.r.t. $\sigma_1$. If $\sigma_2$ is $(\sep, \sigma_1)$-compatible, and $\sigma_3$ is $(\sep, \sigma_2)$-compatible, then $\sigma_3$ is $(\sep, \sigma_1)$-compatible.
\end{observation}

The proof of the following lemma is also via the same kinds of arguments used in the proof of \Cref{lem:2permutation}, and thus is omitted.
\begin{lemma} \label{lem:3permutation}
	Let $\cI$ be a $k$-minimal \yin, let $\sigma$ be a drawing with $k$ crossings, and let $\sep$ be a separator guaranteed by \Cref{lem:2separatorlemma}. 
	
	Then, there exists a $\pbp$ such that, for any $\pguess$ that is $\pbp$-compliant, there exists a drawing $\sigma' \in \nicepfamily(\pbp, \pguess, \sep, \sigma)$ such that $\sigma'$ is $(\sep, \sigma)$-compatible.
\end{lemma}
To summarize, in (\DB, \DT) case, we can conclude the fate of all the components (i.e., which side of the separator edges they belong to) for $\bigcup_{v \in \epts(\sep)} \pecunp(v)$. On the other hand, in (\SB, \ST) case, we know that, either (i) all except a $\Oh(\sqrt{k})$ components in $\pecunp(\ept(\lsepone, 1))$ (resp.~$\pecunp(\ept(\lseptwo, 3))$) lie to the left of $\lseptwo$ (resp.~ to the right of $\rsepone$), or (ii) vice versa. Thus, we can guess which of the two cases we are in, and in each case, two subsets of size $\Oh(\sqrt{k})$, respectively. This leaves us with (\SB, \DT) and (\ST, \DB) case, in which case, we know that the $\pecunp(v)$ for the shared endpoint $v$ are partitioned into $(\pecunp^l(v), \pecunp^r(v))$; however, we do not know how to find such a partition. The goal of the next subsection is to show that such a partition can be found in a structured way.

\subsubsection{Reordering Peculiar Components in (\SB, \DT)/(\ST, \DB) Case.} \label{subsubsec:3components}

Throughout the subsubsection, we will focus on (\SB, \DT) case, wherein $\ept(\lsepone, 1) = \ept(\rsepone, 1)$. The (\ST, \DB) case is handled symmetrically by focusing on the corresponidng shared endpoint $\ept(\lseptwo, 3) = \ept(\rseptwo, 3)$, and by flipping the roles of $V_1$ and $V_3$. We omit the repetition.

We note down the assumptions that will be made throughout this subsubsection. We have a bottom-heavy $k$-minimal instance $\cI$ with a drawing $\sigma$ with $k$ crossings, and a $\sep$ w.r.t. $\sigma$, which is of the case (\SB, \DT) case, with $\ept(\lsepone, 1) = \ept(\rsepone, 1)$, which we denote by $v$ throughout the subsection. Further, the vertices $u \in V_2$ with $\sigma_2(\ept(\lseptwo, 2)) < \sigma_2(u) < \sigma_2(\ept(\rseptwo, 2))$ and $w \in V_3$ with $\sigma_3(\ept(\lseptwo, 3)) \le \sigma_3(w) \le \sigma_3(\ept(\rseptwo, 3))$ will not be relevant for any of the arguments/discussion of this subsubsection, and we will tacitly ignore cases involving such vertices.

Further, let us use $\cC, \cC^l, \cC^r$ instead of $\pecunp, \pecunp^l$,  $\pecunp^r$, respectively. Note that these sets are defined w.r.t. a drawing $\sigma$ and its corresponding separator $\sep$. First, we will analyze the properties of the components of $\cC$, then subsequently we will modify $\sigma$ to obtain another separator-compatible drawing $\sigma'$.

\subsubsection*{Setup.}

\begin{definition}[light and heavy components] \label{def:3lightheavy}
	We say that a component $C \in \cC$ is \emph{$\sigma$-light} if the total number of crossings involving at least one edge of $E(C) \cup \LR{vu: u \in C}$, is at most $\sqrt{k}$; otherwise we say that $C$ is \emph{$\sigma$-heavy}. We denote the subsets of light and heavy components as $\heavy(v, \cC, \sigma)$ and $\light(v, \cC, \sigma)$, respectively.
\end{definition}

The following remark is due regarding the simplification of the notation.
\begin{remark}
	In the notions introduced in \Cref{def:3lightheavy} and later, we initially make the dependence on $v, \cC$, and $\sigma$ explicit in the notation/terminology, and subsequently simplify it by omitting $v, \cC, \sigma$ from the notation, if the respective objects are obvious from the context (indeed, they are fixed for the rest of the subsection, except for $\sigma$ which will be modified at the end). For example, henceforth $\heavy(v, \cC, \sigma)$ is abbreviated to simply $\heavy$. 
	\\
\end{remark}

We state the following simple observation.
\begin{observation} \label{obs:3heavy}
	The number of heavy components is bounded by $2\sqrt{k}$, i.e., $|\heavy| \le 2 \sqrt{k}$.
\end{observation}
\begin{proof}
	Suppose for contradiction that $|\heavy| > 2\sqrt{k}$, then total number of crossings in $\sigma$ is more than $2\sqrt{k} \cdot \sqrt{k}/ 2$, since each $C \in \heavy$ accounts for more than $\sqrt{k}$ crossings, and since each crossing is counted at most twice in this manner (since, for example, the pair of crossing edges may belong to a pair of two different heavy components). This contradicts the fact that $\sigma$ is a drawing with $k$ crossings.
\end{proof}


Leftmost and rightmost light components.
\begin{definition}[leftmost and rightmost light components] \label{def:3leftrightmostcomps}
	Let $\prec_l$ denote a total ordering on $\cC$, where for $C_1, C_2 \in \cC$, $C_1 \prec_l C_2$ if the leftmost neighbor of $v$ in $C_1$ is to the left of the leftmost neighbor of $v$ in $C_2$. Similarly, let $\prec_r$ denote a total ordering on $\cC$ using the rightmost neighbors of $v$ in components (note that $C_1 \prec_l C_2$ does not necessarily imply that $C_2 \prec_r C_1$, and vice versa). Let $\llc \in \light$ be the leftmost light component (according to $\prec_l$), and $\rlc \in \light$ be the rightmost light component (according to $\prec_r$). Note that any component $C$ satisfying $C \prec_l \llc$ or $C \prec_r \rlc$ must be a heavy component. 
	\\Let $\extreme \subseteq \cC$ denote all the components $C$ such that, (i) $C \prec_l \llc$, or (ii) $C \prec_r \rlc$, or
	(iii) at least one edge of $C$ crosses an edge of $\llc$ or $\rlc$.
\end{definition}

The following observation follows immediately from the definition.
\begin{observation} \label{obs:3extremebound}
	$|\extreme| \le 4\sqrt{k}$.
\end{observation}

Finally, we define the notion of star vertices, which are crucial to the result shown in this subsubsection.

\begin{definition}[Star vertices]
	we say that a vertex $p \in R_2$ is a \emph{$(v, \cC, \sigma)$-left-star vertex} (resp.~\emph{$(v, \cC, \sigma)$-right-star vertex}) if there is an edge incident to $p$ that crosses an edge of the form $vu^l \in E_{12}$, where $u^l \in \llc \cap V_2$  
	is the leftmost neighbor of $v$ in $\llc$ (resp.~$u^r \in \rlc \cap V_2$ is the rightmost neighbor of $v$ in $\rlc$). Let $\leftstar, \rightstar$ denote the sets of left-star and right-star vertices respectively. 
	When the left vs.~right distinction is not important, we refer to the vertices of $\leftstar \cup \rightstar$ as simply \emph{star vertices}. Let $\Star$ denote the resultant set of star vertices, and they are are enumerated as $\str_1, \str_2, \ldots, \str_q$ from left to right according to $\sigma_2$. See \Cref{fig:3stars} for an illustration.
\end{definition}
\begin{figure}
	\centering
	\includegraphics[scale=1,page=6]{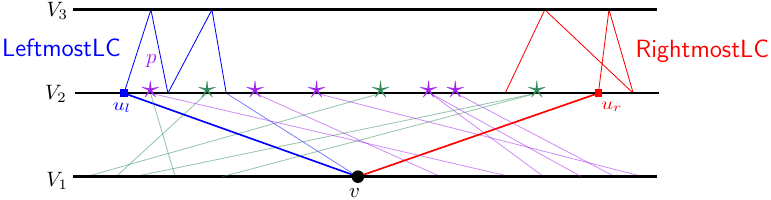}
	\caption{Illustration of components $\llc, \rlc$ and star vertices. Left-star vertices (shown in green) have at least one incident edge crossing $vu_l$, and right-star (purple) have at least one incident edge crossing $vu_r$. Note that a vertex can be both left- and right-star vertex, e.g., $p$.} \label{fig:3stars}
\end{figure}

Next, we state and prove following straightforward observations that follow immediately from the definitions.
\begin{observation} \label{obs:3star-characterization}
	No star vertex belongs to a component in $\cC$.
\end{observation}
\begin{proof}
	A star vertex $u$ has an incident edge that crosses an edge incident to $v$, and hence $u$ is adjacent to a vertex $w \in V_1$, where $w \neq v$. Hence, the connected component in $G-v$ that contains $u$, also contains $w \in V_1$, implying that $C \in \updown(v)$. 
\end{proof}

\begin{observation} \label{obs:3prelimbounds}\ 
	$|\leftstar|, |\rightstar| \le \sqrt{k}$.
\end{observation}
\begin{proof}
	A star vertex has an incident edge that crosses an edge $vu_l$ for some $u_l \in \llc$. However, since $\llc$ is a light component, it is involved in at most $\sqrt{k}$ crossings. Therefore, the number of vertices in $\leftstar$ is bounded by $\sqrt{k}$. A similar argument shows that $|S_r| \le \sqrt{k}$. 
\end{proof}

\begin{figure}
	\centering
	\includegraphics[scale=1,page=10]{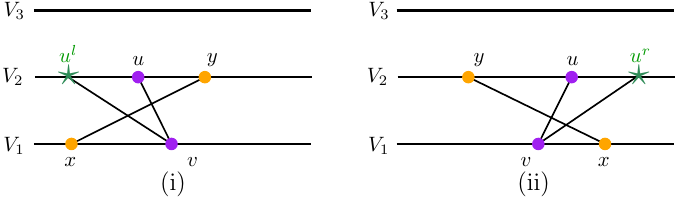}
	\caption{Two cases from the proof of \Cref{obs:3star-crossings}.} \label{fig:3star-crossings}
\end{figure}

\begin{claim} \label{obs:3star-crossings}
	Let $u, y \in V_2$ be two distinct vertices such that (1) $u$ is a neighbor of $v$, (2) there exists some neighbor $x \in V_1$ of $y$, and (3) the edges $vu$ and $xy$ cross each other. Then, $y$ must be a star vertex.
\end{claim}
\begin{proof}
	Since $vu$ and $xy$ cross, either (i) $\sigma_1(v) < \sigma_1(x)$ and $\sigma_2(y) < \sigma_2(u)$, or (ii)  $\sigma_1(v) > \sigma_1(x)$ and $\sigma_2(y) > \sigma_2(u)$ (See \Cref{fig:3star-crossings}). In case (i), note that $\sigma_2(y) < \sigma_2(u) \le \sigma_2(u_r)$. We also have that $\sigma_1(v) < \sigma_1(x)$. Therefore, $xy$ crosses $vu^r$ as well, which implies that $y$ is a star vertex. Case (ii) is symmetric -- we have that $\sigma_2(u^l) \le \sigma_2(u) < \sigma_2(y)$, and $\sigma_2(x) < \sigma_2(v)$. Therefore, $xy$ crosses $vu^l$ implying that $y$ is a star vertex.
\end{proof}

\begin{figure}
	\centering
	\includegraphics[scale=1,page=11]{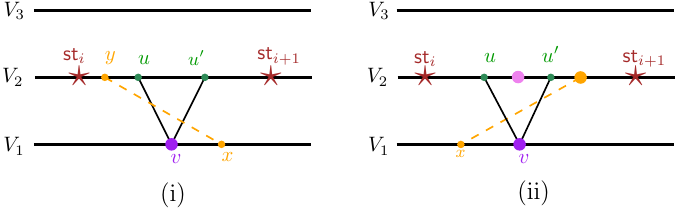}
	\caption{\small Illustration for the proof of \Cref{cl:3slotintersection}. In case (i), an edge $xy$ that crosses $vu$ must also cross $vu'$. In case (ii), there are two possible locations for the vertex $y$, shown in orange and purple, respectively. 
		In the orange case, the edge $xy$ again must cross both $vu$ and $vu'$; whereas in the purple case, the vertex $y$ must be a star vertex, contradicting that $\str^i$ and $\str^{i+1}$ are consecutive star vertices.
	} \label{fig:3samecrossing}
\end{figure}

\begin{claim} \label{cl:3slotintersection}
	Let $u, u'$ be two distinct neighbors of $v$ such that $\sigma_2(\str_i) < \sigma_2(u) < \sigma_2(u') < \sigma_2(\str_{i+1})$, where $\str_i, \str_{i+1} \in \Star$ are some pair of consecutive star vertices according to $\sigma_2$ (note that $1 \le i < |\Star|$). Then, the edges $vu$ and $vu'$ cross the same set of edges in $\sigma$.
\end{claim}
\begin{proof}
	
	Suppose for the sake of contradiction that there exists an edge $xy$ with $x \in V_1, y \in V_2$, such that $xy$ crosses $vu$ but not $vu'$ (the other case is symmetric). Since $xy$ crosses $vu$, it must be the case that (i) $\sigma_1(v) < \sigma_1(x)$ and $\sigma_2(u) > \sigma_2(y)$, or (ii) $\sigma_1(v) > \sigma_1(x)$ and $\sigma_2(u) < \sigma_2(y)$. See \Cref{fig:3samecrossing}.
	
	First, consider case (i). Since $\sigma_2(y) < \sigma_2(u) < \sigma_2(u')$ and $\sigma_1(v) < \sigma_1(x)$, it follows that $vu'$ and $xy$ cross as well. 
	
	Now consider case (ii). We have that $\sigma_2(u) < \sigma_2(u')$, $\sigma_2(u) < \sigma_2(y)$, and $\sigma_1(x) < \sigma_1(v)$. If $\sigma_2(u') < \sigma_2(y)$, then $xy$ crosses $vu'$. Otherwise, $\sigma_2(x_l) \le \sigma_2(\str_{i}) < \sigma_2(u) < \sigma_2(u') < \sigma_2(y) \le \sigma_2(\str^{i+1}) \le \sigma_2(x_r)$. However, by \Cref{obs:3star-crossings}, we obtain that $y$ must be a star vertex, which contradicts that $\str_{i+1}$ is the next consecutive star vertex after $\str_i$. 
\end{proof}

Next, we introduce the notion of components that ``surround'' a star vertex, and then bound the number of such components.

\begin{definition}
	We say that a component $C$ \emph{surrounds} a star vertex $\str$, if it contains two vertices that are placed to the either side of $\str$ in $\sigma_2$ (see \Cref{fig:3componentcrossing}). Let $\surr \subseteq \cC \setminus (\heavy \cup \extreme)$ denote the subset of components $C$ such that $C$ surrounds some $u \in \Star$. Note that no pendant component can belong to $\surr$.
\end{definition}

\begin{figure}
	\centering
	\includegraphics[scale=1,page=7]{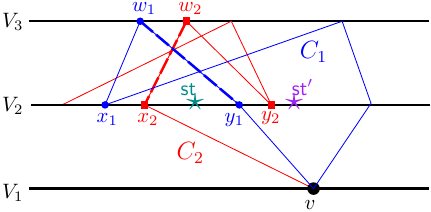}
	\caption{\small Illustration for the proof of \Cref{cl:3starcrossed}. Components $C_1$ (blue) and $C_2$ (red) surround a star vertex $\str$ ($C_1$ happens to surround $\str'$ as well). Due to the positional assumptions in the first case in the proof of \Cref{cl:3starcrossed}, edges $y_1w_1$ and $x_2w_2$ must cross (shown as thick dashed edges).} \label{fig:3componentcrossing}
\end{figure}

We show the following claim regarding the necessity of crossings between components surrounding the same star vertex.
\begin{claim} \label{cl:3starcrossed}
	Let $C_1, C_2 \in \surr$ be distinct components that surround a star vertex $\str \in \Star$. Then, there exist edges $e_1 \in E(C_1)$ and $e_2 \in E(C_2)$ that cross each other.
\end{claim}
\begin{proof}
	Since $C_1$ surrounds $\str$, it contains two vertices that are placed to the either side of $\str$. Further, since $C_1$ is a connected component in $G-v$, it must contain a vertex $w_1 \in V_3 \cap V(C_1)$ such that $w_1$ has two neighbors in $C_1$, say $x_1$ and $y_1$ such that $x_1$ is to the left of $\str$, and $y_1$ is to the right of $\str$. Similarly, define $w_2, x_2, y_2 \in V(C_2)$, respectively.
	
	Suppose first that $\sigma_3(w_1) < \sigma_3(w_2)$, i.e., $w_1$ is to the left of $w_2$. Then, since $\sigma_2(x_2) < \sigma_2(\str) < \sigma_2(y_1)$, the edges $x_2w_2$ and $y_1w_1$ cross each other (see \Cref{fig:3componentcrossing}). In the other case, when $\sigma_3(w_2) < \sigma_3(w_1)$, we can observe symmetrically that the edges $x_1w_1$ and $y_2w_2$ cross each other. This proves the claim.
\end{proof}

In the following lemma, we bound the number of components that surround some star vertex.
\begin{lemma} \label{lem:3exceptional}
	$|\surr| \le 5 \cdot (k)^{2/3}$.
\end{lemma}
\begin{proof}
	Let us number the left-star vertices from right to left (in $\sigma_2$) as $\str^l_1, \str^l_2, \ldots$, and analogously, let us number the right-star vertices from left to right as $\str^r_1, \str^r_2, \ldots $. Let $t = (k)^{1/3}$, and let $\Star' = \LR{\str^l_i \in \leftstar : i \le t} \cup \LR{\str^r_i \in \rightstar: i \le t}$. Note that $|\Star'| \le 2t = 2(k)^{2/3}$. In order to bound the size of $\surr$, we classify each $C \in \surr$ into two types: $C$ is class 1 if it surrounds some star vertex of $\Star'$; otherwise, if it only surrounds some vertex of $\Star \setminus \Star'$, then it is class 2. In the following two claims, we bound the number of components of each class separately.
	
	\begin{claim} \label{cl:3exceptional1}
		The number of class 1 components in $\surr$ is at most $4 \cdot (k)^{2/3}$.
	\end{claim}
	\begin{proof}[Proof of \Cref{cl:3exceptional1}]
		We arbitrarily number the vertices of $\Star'$ as $\str'_1, \str'_2, \ldots, \str'_{z}$, where $z = |\Star'| \le 2(k^{1/3})$. Let $n_i$ denote the number of components that surround $\str'_i$, and w.l.o.g. assume $n_i > 0$ for each $i$ -- otherwise such star vertex plays no role in the following analysis, and can be ignored. Let $m_i = n_i - 1 \ge 0$ for each $1 \le i \le z$. Hence, the total number of class 1 components is at most $\sum_{i = 1}^z n_i = \sum_{i = 1}^z m_i + z$. Since $z \le 2(k)^{2/3}$, it suffices to show that $\sum_{i = 1}^z m_i \le 2 \cdot (k)^{2/3}$. 
		
		By \Cref{cl:3starcrossed} any pair of components that surround the same $\str_i \in \Star$ must have a pair of crossing edges. Therefore, the number of crossings among the components surrounding $\str_i$ is at least $\frac{n_i(n_i-1)}{2} \ge \frac{m_i^2}{2}$. However, since there are at most $k$ crossings, it follows that $\sum_{i = 1}^z \frac{m_i^2}{2} \le k$, which implies that, $\sum_{i = 1}^z m_i^2 \le 2k$. Now, by using Cauchy–Schwarz inequality, we obtain that 
		\begin{align*}
			\lr{\sum_{i = 1}^z m_i}^2 \le \lr{\sum_{i = 1}^z m_i^2} \cdot z \le 2kz.
		\end{align*}
		This implies that, $\sum_{i = 1}^z m_i \le \sqrt{2kz} \le \sqrt{2k \cdot 2(k)^{1/3}} = 2 \cdot \sqrt{(k)^{4/3}} = 2 \cdot (k)^{2/3}$. 
	\end{proof}
	
	\begin{claim} \label{cl:3exceptional2}
		The number of class 2 components in $\surr$ is at most $(k)^{2/3}$.
	\end{claim}
	\begin{proof}[Proof of \Cref{cl:3exceptional2}.]
		We will show that each component of class 2 is involved in at least $(k)^{1/3}$ distinct relevant crossings, which immediately implies the lemma.
		
		Fix one such class 2 component $C$, and note that it surrounds a star vertex $\str \not\in \Star'$. Without loss of generality, suppose $\str$ is a left-star vertex (the right-star case is symmetric), say $\str^l_i \in \leftstar \setminus \Star'$, which implies that $i > t$. Further, since $C$ is not a class 1 component, it does not surround any $\str^l_{i'}$ with $i' \le t$. However, $C$ contains a neighbor $w_l$ of $v$, and by assumption, it also lies to the left of all of $\str^l_{i'}$ with $i' \le t$. 
		
		Since $C$ is a light component, it holds that $\llc \prec_l C$, where recall that $\llc$ is the leftmost component from \Cref{def:3leftrightmostcomps}, with $u^l \in \leftstar$  being the leftmost neighbor of $v$. Similarly, let $w^l$ denote teh leftmost neighbor of $v$ in $C$. This implies that $\sigma(u^l) < \sigma(w^l) < \sigma(\str^l_{i'})$ for all $i' \le t$. Further, since an edge $e$ incident to $\str^l_{i'}$ crosses $vu_l$, the other endpoint of $e$, say $u_b \in V_1$ lies to the left of $v$. However, this implies that $e$ also crosses the edge $vw^l$ (See \Cref{fig:3cumulativecrossing}). Since this holds for each $i \le t$, the edge $vw^l$ crosses at least $t$ distinct edges, thus accounting for at least $t$ distinct crossings.

	\end{proof}
	With \Cref{cl:3exceptional1} and \Cref{cl:3exceptional2}, we conclude the proof of \Cref{lem:3exceptional}.
	
\end{proof}

\begin{figure}
	\centering
	\includegraphics[scale=1,page=8]{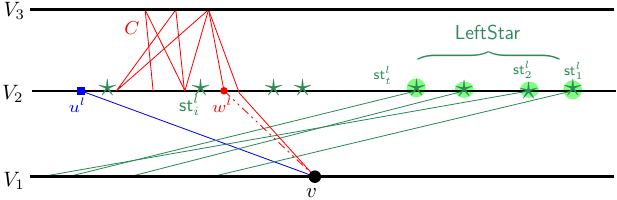}
	\caption{Illustration for \Cref{cl:3exceptional2}. Component $C$ (shown in red) surrounds star vertex $\str^l_i \not\in \Star'$. For each $i' \le k^{1/3}$, the edge $vw_l$ (dashed red) crosses all edges incident to $\str^l_{i'}$ that also cross the edge $vu^l$ (blue). Thus, $C$ accounts for at least $(k)^{1/3}$ crossings.}\label{fig:3cumulativecrossing}
\end{figure}

\subsubsection*{Component Types.}

Next, we define the notion of a type of a component.

\begin{definition}[Original and reduced type] \label{def:type}
	Let $C \in \pecunp(C) = \cC \setminus \pendant(v)$ be a component in $G-v$. Let $\loeo(C) \coloneqq |M^{\sf orig}_12| |\LR{uv: u \in C \cap V_2}|$, i.e., the number of edges between $v$ and the vertices of $C$ (including multiplicities), and $\upeo(C) \coloneqq  |E(C)|$, i.e., the number of edges in $C$. 
	\\Then, the \emph{original type} of $C$ is defined as the pair $(\loeo(C), \upeo(C))$, and the \emph{reduced type} of $C$ is defined as $(\loe(C), \upe(C)) \in \LR{1, 2, \ldots, (k)^{1/3}, \ast, L}^2$, where
	\begin{equation*}
		(\loe(C), \upe(C)) = \begin{cases}
			(\loeo(C), \upeo(C)) & \text{ if $0 < \loeo(C) < (k)^{1/3}$ and $\upeo(C) < (k)^{1/3}$}
			\\(\ast, L) & \text{ if $\upeo(C) \ge (k)^{1/3}$}
			\\(L, \ast) &  \text{ if $\upeo(C) < (k)^{1/3}$ and $\loeo \ge (k)^{1/3}$.}
		\end{cases}
	\end{equation*}
	Here, the meaning of $(\ast, L)$ type is that $|E_{23}|$ is \emph{larger} than $(k)^{1/3}$ (hence, $L$), whereas the number of neighbors of $v$ in $C$ is arbitrary (hence, $\ast$). We interpret $(L, \ast)$ similarly.
	
	For any pendant $C = \LR{u} \in \pendant(v)$. Then, with $\loeo(C)$ denote the multiplicity of the edge $vu$. Then, the original type of $C$ is defined to be $(\loeo(C), 0)$, and the reduced type of $C$ is defined as $(1, 0)$.
	
	Finally, we define $\types \coloneqq \LR{1, \ldots, (k)^{1/3}} \times \LR{1, 2, \ldots, (k)^{1/3}} \cup \LR{(\ast, L), (L, \ast), (1, 0)}$. 
\end{definition}

The following observations are immediate from the definition of original/reduced types.

\begin{observation}\label{obs:3types}\ 
	\begin{enumerate}
		\item $|\types| \le k^{2/3} + 3 = \Oh(k^{2/3})$.
		\item $C \in \pendant(v)$, iff the original type  of $C$ is $(i, 0)$ for some $i \ge 1$. Otherwise, $C \in \pecunp(v)$ iff the original type of $C$ is $(i, j)$ with $i, j > 0$. 
		\item The original/reduced type of a component entirely depends on its structural properties, and is \emph{independent} of a drawing.
	\end{enumerate}
\end{observation}

\begin{definition}[Exceptional and nice components]
	Let $\exceptional(v, \cC, \sigma) \coloneqq \heavy \cup \extreme \cup \surr$, and let $\nice(v, \cC, \sigma) \coloneqq \cC \setminus \exceptional(v, \cC, \sigma)$.
	In the drawing $\sigma$, for each $(i, j) \in \types$, the set of all \emph{nice components} of reduced type $(i, j)$  is denoted by $\cT_\sigma(i, j)$, and for each $(i', j') \in \mathbb{N}^2$, the set of all nice components of original type $(i', j')$ is denoted by $\cT'_\sigma(i', j')$. \footnote{Note that the sets $\cT_{\sigma}(\cdot, \cdot), \cT'_{\sigma}(\cdot, \cdot)$ are subsets of $\nice(v, \cC, \sigma)$ and hence depend on the drawing $\sigma$.}
\end{definition}

From \Cref{obs:3heavy}, \Cref{obs:3extremebound}, and \Cref{lem:3exceptional}, we immediately obtain the following corollary.
\begin{corollary} \label{cor:3exceptionalcomponents}
	$|\exceptional| \le 11 \cdot (k)^{2/3} = \Oh(k^{2/3})$.
\end{corollary}

\subsubsection*{Rearrangement of Nice Components}

Now we are left with analyzing the behavior of nice components. Note that by definition, each $C \in \nice(\cC)$ is a light component, and it does not surround any star vertex. We need the following definition to analyze the behavior of such nice components.
\begin{definition}[Slot] \label{def:3slots}
	Let us now order the star vertices of $\Star$ from left to right as $\str_1, \str_2, \ldots, \str_q$, and we define set of vertices $\slot_i(\sigma) \coloneqq \LR{u \in V_2 : \sigma_2(\str_i) < \sigma_2(u) < \sigma_2(\str_{i+1})}$ as an \emph{$\Star$-slot}, or simply a \emph{slot}, defined by two consecutive star vertices $\str_i, \str_{i+1}$. We say that a component $C$ \emph{belongs to} a slot $\slot_i(\sigma)$ in $\sigma$, if $V(C) \cap V_2 \subseteq \slot_i(\sigma)$. 
\end{definition}

The following observation follows from the definition of $\nice(\cC)$.

\begin{observation} \label{obs:3niceslots}
	Each $C \in \nice$ belongs to some slot $\slot_i(\sigma)$.
\end{observation}

We further define the following notion -- note that this is defined w.r.t. \emph{any} drawing $\so = (\so_1, \so_2, \so_3)$, not \emph{just} for a drawing $\sigma$ fixed at the beginning of the section.

\begin{definition}[Consecutiveness]\label{def:3consecutiveness}
	\begin{itemize}
		\item We say that a set of vertices $Q$ is \emph{consecutive} in $\so$ if for each $i \in [3]$, $Q_i \coloneqq Q \cap V_i$ satisfies the following property: (a) either $|Q_i| \le 1$, or (b) there exist two distinct vertices $u_l, u_r \in Q_i$, such that $Q_i = \LR{ u \in V_i : \so_i(u_l) \le \so_i(u) \le \so_i(u_r) }$. 
		\item We say that a collection of vertex-subsets ${\cal Q} = \LR{Q_1, Q_2, \ldots, Q_\ell}$ is \emph{consecutive} in $\so$, if (a) $\bigcup_{Q_j \in {\cal Q}} Q_j$ is consecutive in $\so$, and moreover (b) each $Q_j \in {\cal Q}$ is consecutive in $\so$. \footnote{Note that these two conditions together imply that the sets $Q_j \in {\cal Q}$ are placed ``one after the other'' in some arbitrary order, and no vertex not in any of $Q_j$'s is placed anywhere ``in between'' the sets.}
	\end{itemize}
\end{definition}

In the following definitions, we introduce the notion of a guess.

\begin{definition} \label{def:3setup}
		For each $C \in \pecunp(v)$, let $\opt(C)$ denote the optimal number of crossings required to draw $C$ on $2$ lines. Further, for a subset ${\cal S} \subseteq \pecunp(v)$, let $\opt({\cal S}) \coloneqq \sum_{C \in {\cal S}} \opt(C)$. 
	
\end{definition}

\begin{definition} \label{def:3blueprint}
	A \emph{component-blueprint} is a tuple $\cbp = (\totall, \totalr, \optl, \optr, \cE)$, where 
	\begin{itemize}
		\item $\cE \subseteq \cT(L, \ast) \cup \cT(\ast, L) \subseteq \nice(\cC)$ is a subset of components of size at most $k^{2/3}$.
		\item $\totall, \totalr: \types \to \nat$, where, for each $(i, j) \in \types$, it holds that, $\totall(i, j) + \totalr(i, j) = |\cT(i, j) \setminus \cE|$. 
		\item $\optl, \optr: \types \to \nat$, where, for each $(i, j) \in \types$, it holds that, $\optl(i, j) + \optr(i, j) = \opt(\cT(i, j) \setminus \cE)$.
	\end{itemize}
\end{definition}

Finally, we define the notion of blueprint-compliant guess.
\begin{definition} \label{def:3blueprintcompatibleguess}
	A \emph{component-guess} is a tuple $\cguess = (\leftpart, \rightpart)$ where, $\leftpart, \rightpart: \types \to \cC$ where for each $(i, j) \in \types$, $\leftpart(i, j), \rightpart(i, j)$ are disjoint subsets of $\cT(i, j)$.
	
	Let $\cguess = (\leftpart, \rightpart)$ be a component-guess, and $\cbp = (\totall, \totalr, \optl, \optr, \cE)$ be a component-blueprint. We say that the guess $\cguess$ is \emph{$\cbp$-compliant} if, for each $(i, j) \in \cT)$, 
	\begin{itemize}
		\item $|\leftpart(i, j)| = \totall(i, j)$, and $|\rightpart(i, j)| = \totalr(i, j)$, 
		\item $\opt(\leftpart(i, j)) = \optl(i, j)$, and $\opt(\rightpart(i, j)) = \optr(i, j)$.
	\end{itemize}
\end{definition}

\begin{definition}
	Let $\cG = (\leftpart, \rightpart)$ be a guess. We define a family \\$\nicecfamily(\cbp, \cguess, \sep, \sigma)$ to be a family of drawings $\sigma'$ satisfying the following properties:
	\begin{itemize}
		\item $\sigma$ and $\sigma'$, when respected to $V(G) \setminus \bigcup_{(i, j) \in \types} (\leftpart(i, j) \cup \rightpart(i, j))$ are equal, and
		\item Let $V_L \coloneqq \bigcup_{(i, j) \in \types} \leftpart(i, j)$ and $V_R \coloneqq \bigcup_{(i, j) \in \types} \rightpart(i, j)$. Then, for each $\ell \in \LR{2, 3}, \sigma_\ell(u) < \sigma_\ell(\ept(\lseptwo, \ell))$ for each $u \in V_L \cap V_\ell$, and $\sigma_\ell(w) > \sigma_\ell(\ept(\rseptwo, \ell))$ for each $w \in V_R \cap V_\ell$. 
	\end{itemize}
	
\end{definition}

Finally, we state the main lemma of our section.

\begin{lemma}\label{lem:3componentordering}
	Let $\cI$ be a $k$-minimal bottom-heavy \yin, let $\sigma$ be a drawing with $k$ crossings. Let $\sep(\sigma)$ be a separator w.r.t. $\sigma$ in (\SB, \DT) case, where $v \coloneqq \ept(\lsepone, 1) = \ept(\rsepone, 1)$. Then, there exists a blueprint $\cbp$, such that, for any $\cbp$-compliant guess $\cguess$, there exists a drawing $\sigma'$ such that 
	(i) $\sigma \in \nicecfamily(\allowbreak \cbp, \cguess, \sep, \sigma)$, and (ii) $\sigma'$ is $(\sep, \sigma)$-compatible.
\end{lemma}
\begin{proof}
	First, let us introduce some notation that will be used throughout the lemma.
	Let $\nice \coloneqq \nice(\cC) \subseteq \cC = \pecunp(v)$. Let $\compedges = \bigcup_{C \in \nice} E(C)$, and $\cedges = \bigcup_{C \in \nice} \cedges(C)$, where $\cedges(C) \coloneqq \LR{vu: u \in C \cap V_2}$. 
	
	For a drawing $\so$, we say that a component $C \in \nice$ is $\so$-left (resp.~$\so$-right), if in the drawing $\so$, for $\ell \in \LR{2, 3}$, all the vertices of $C \cap V_\ell$ are to the left of $\ept(\lseptwo, \ell)$ (resp.~to the right of $\ept(\rseptwo, \ell)$).


	\paragraph{Transforming $\sigma$.} 
	In the beginning, $\sigma'$ is initialized to $\sigma$. Throughout this proof, we will gradually transform $\sigma$ by reordering the vertices involved in $R \coloneqq \bigcup_{C \in \nice} C$, while keeping the relative ordering of the rest of the vertices unchanged. This will be done in multiple steps. Before this, however, we set up things.
	Further, during this process, we will also ensure that, $\sigma'$ remains $(\sep, \sigma)$-compatible (note that in the beginning, $\sigma' = \sigma$ and thus this property holds initially). Specifically, during each step of this process, we ensure the following: 
	\begin{enumerate}
		\item Only $\sigma'_2, \sigma'_3$ are modified by only changing the relative ordering between the vertices of $R \cap V_2$ and $R \cap V_3$, respectively, whereas $\sigma'_1$ remains unchanged.
		\item The sets of left- and right-star vertices w.r.t. $v, \cC, \sigma'$ remain unchanged throughout the procedure.
		\item The number of crossings remains the same, and moreover, the number of crossings on the left of $\lsepone, \lseptwo$, and that to the right of $\rsepone, \rseptwo$, remains the same.
		\item for any nice component $C \in \nice(\cC)$, the number of crossings involving at least one edge of $E(C)$ does not increase.
	\end{enumerate}
	At the end, when we will have achieved property 2, it follows that the final drawing
	$\sigma'$ at that point will be a $(\sep, \sigma)$-well-behaved drawing. 
	
	This gradual transformation can be broadly divided into three steps, which we explain in detail below.
	
	\begin{description}[leftmargin=*]

		\item[Step 1. Making each $C \in \nice$ consecutive.] By \Cref{obs:3niceslots}, each $C \in \nice$ belongs to a slot $S_i$.  Now, let $e = uw \in E(C) \subseteq E_{23}$ (with $u \in V_2, w \in V_3$) be an edge crosses the minimum number of edges outside $E(C)$. Let $E_{out} \subseteq E_{23} \setminus E(C)$ denote the set of edges that cross $e$. Let $p^l(C) \in V_2 \setminus C$ (resp.~$p^r(C) \in V_2 \setminus C$) be the rightmost (resp.~leftmost) vertex in $V_2 \setminus C$ according to $\sigma'_2$ that is to the left (resp.~right) of $u$. Similarly, define $q^l(C), q^r(C) \in V_3 \setminus C$ w.r.t. $w$ and $\sigma'_3$. Then, we modify $\sigma'_2$ by moving all the vertices of $C \cap V_2$ between $p^l(C)$ and $p^r(C)$, while maintaining the same relative order among themselves. Similarly, we move all the vertices of $C \cap V_3$ between $q^l(C)$ and $q^r(C)$ while maintaining the same relative order among themselves. This ensures that $C$ is consecutive in the resulting drawing $\so$.
		
		We now argue that the number of crossings does not increase after the transformation. First, note that by \Cref{obs:3niceslots}, $V(C) \cap V_2 \subseteq \slot_i$ for some slot $\slot_i$, and we only move the vertices of $C \cap V_2$ within the same slot. Therefore, via an argument similar to \Cref{cl:3slotintersection}, for each edge $e_{12} \in \cedges(C)$, the set of edges crossing $e_{12}$ does not change due to the modification. Further, since we maintain the same relative order among the vertices of $V(C) \cap V_2$ and $V(C) \cap V_3$, the set of edges of $E(C)$ that cross any edge $e_{23}$ also does not change. On the other hand, for each $e_{23} \in E(C)$, the set of edges from $E_{23} \setminus E(C)$ is now exactly equal to $E_{out}$. Due to the choice of $e$, it follows that
		the number of edges crossing any $e_{23}$ does not increase after the transformation. However, due to $k$-minimality of $\cI$, it follows that the number of crossings before and after the transformation must remain the same.
		After iteratively transforming $\sigma'$ in this manner for each $C \in \nice$, in the final drawing $\sigma'$, we ensure the following properties for each $C \in \nice$:
		\begin{enumerate}
			\item $C$ is consecutive in $\sigma'$. Further, if $C$ was $\sigma$-left (resp.~$\sigma$-right), then $C$ is $\sigma'$-left (resp.~$\sigma$-right).
			\item All the vertices in $V(C) \cap V_2$ belong to the interval $\LR{p \in V_2: \sigma'_2(p^l(C)) < \sigma'_2(p) < \sigma'2(p^r(C))}$ defined by the vertices $p^l(C), p^r(C)$,
			\item All the vertices in $V(C) \cap V_3$ belong to the interval $\LR{q \in V_2: \sigma'_2(q^l(C)) < \sigma'_2(q) < \sigma'_2(q^r(C))}$ defined by the vertices $q^l(C), q^r(C)$, and
			\item For any distinct $C, C' \in \nice$, the edges of $E(C)$ and $E(C')$ do not cross each other.
		\end{enumerate}
		These properties imply that, the final drawing $\sigma'$ obtained at the end of step 1 is $(\sep, \sigma)$-compatible, where $\sigma$ is the original drawing from the beginning of the proof. We refer to the pair of intervals defined in items 2 and 3 above as the \emph{sub-slot} defined by $p^l(C), p^r(C), q^l(C), q^r(C)$, and denote it by $\subslot(C)$. We note that, for any $C \in \nice$, since $E(C) \cup \LR{uv: u \in C}$ does not contain any edge of $\crosse \cup \mde$, it holds that 
		either (i) for each $u \in \lseptwo(C) \cap V_i$, $\sigma'_i(u) < \sigma'_i(\ept(\lseptwo, i))$, for $i \in \LR{2, 3}$ or (ii) for each $u \in \lseptwo(C) \cap V_i$, $\sigma'_i(u) > \sigma'_i(\ept(\rseptwo, i))$, for $i \in \LR{2, 3}$. Informally, either the entire subslot---and hence---is to the left of $\lsepone$, or to the right of $\rseptwo$.
		
		Further, let $E_{out}(C) \subseteq E_{23} \setminus E(C)$ to be the set of all edges that cross all of the edges of $E_{23}$ in $\sigma'$. 
		
		\item[Step 2. Constructing an Extended Component Blueprint.] 
		
		Recall that each reduced type $(i, j) \in \types \setminus \LR{(L, \ast), (\ast, L)}$ is the same as the corresponding original type; whereas the two special types $(L, \ast), (\ast, L)$ consist of multiple original types.  We first extend the notion of a component-blueprint to original types, instead of reduced types. Note that an important difference here is that, this definition does not the ``exceptional set'' $\cE$.
		
		\begin{definition}
			An \emph{extended component-blueprint} is a tuple $$\ecbp = (\totall, \totalr, \optl,\optr), $$ such that, 
			\begin{itemize}
				\item $\totall, \totalr: \mathbb{N}^2 \to \mathbb{N}$, where for each original type $(i', j') \in \mathbb{N}^2$, it holds that $\totall(i', j') + \totalr(i', j') = |\cT'(i', j')|$.
				\item $\optl, \optr: \mathbb{N}^2 \to \mathbb{N}$, where for each original type $(i', j') \in \mathbb{N}^2$, it holds that $\optl(i', j') + \optr(i', j') = \opt(\cT'(i', j'))$.
			\end{itemize} 
		\end{definition}
		We similarly define the notion of an extended component-guess, and the compliance thereof with an extended component-blueprint.
		
		Next, we will look at the current $\sigma'$ to define an $\ecbp^o = (\totall^o, \totalr^o, \optl^o, \optr^o)$ that will be used as a reference point. To this end, define for each original type $(i', j') \in \mathbb{N}^2$, 
		\begin{itemize}
			\item $(\leftpart^o(i', j'), \rightpart^o(i', j'))$ denote the partition of $\cT'(i', j')$, where $\leftpart^o(i', j')$ is the set of components that are $\sigma'$-left, and $\rightpart^o(i', j')$ is the set of components that are $\sigma'$-right.
			\item $\totall^o(i', j') \coloneqq |\leftpart^o(i', j')|$, and $\totalr^o(i', j') \coloneqq |\rightpart^o(i', j')|$.
			\item $\optl^o(i', j') \coloneqq \opt(\leftpart^o(i', j'))$, and $\optr^o(i', j') \coloneqq \opt(\rightpart^o(i', j'))$.
		\end{itemize}
		
		Now, consider any $\ecguess = (\leftpart, \rightpart)$ that is $\ecbp^o$-compliant. In this step, we will iteratively modify $\sigma'$ such that, for each $(i', j') \in \mathbb{N}^2$, all the components in $\leftpart(i', j')$ are $\sigma'$-left, and are placed consecutively, and similarly, all the components in $\rightpart(i', j')$ are $\sigma'$-right, and are placed consecutively.
		
		We iterate over all original types $(i', j') \in \mathbb{N}^2$ in an arbitrary order. Fix one such $(i', j') \in \mathbb{N}^2$, and let $C^l(i', j') \coloneqq \arg\min_{C \in \leftpart(i', j')} |\sloted(C)| \cdot i' + |E_{out}(C)| \cdot j'$, where $\sloted(C) \subseteq E_{12}$ denotes the set of edges that cross any edge of the form $vu$, where $u$ belongs to the same as $C$. Let $\subslot^l(i', j') \coloneqq \subslot(C(i', j'))$. We transform $\sigma'$ as follows: for each $C' \in \leftpart(i', j')$, we move all the vertices of $C'$ to the $\subslot^l(i', j')$ while maintaining the same relative order of vertices of $C'$. We place all such components one after the other in an arbitrary order in this subslot. We perform a similar transformation for $\rightpart(i', j')$. Note that after completely processing each original type $(i', j')$, the resultant drawing $\sigma''$ at the end of the iteration, is $(\sep, \sigma')$-compatible, where $\sigma'$ is the drawing before the start of the iteration.
		
		Note that for $(i', j') \in \LR{1, 2, \ldots, (k)^{1/3}}^2$, the original and reduced types are the same. Hence, this step makes all such $\leftpart(i', j'), \rightpart(i', j')$ consecutive. For each $(i', j')$, let $c^l_{12}(i', j') \coloneqq |\sloted(C^l(i, j))|, c^l_{23}(i', j') \coloneqq |E_{out}(C^l(i, j)|$, where $C^l(i, j)$ is the component that is used to relocate all other components. Analogously, define $c^r_{12}(i', j')$ and $c^r_{23}(i', j')$.
		
		\item[Step 3. Handling reduced type $(\ast, L)$, $(L, \ast)$.] Let $T_{orig}(L, \ast)$ be the set of all original types whose corresponding reduced type is $(L, \ast)$. Note that these consist of original types $(i', j')$, where $i < k^{1/3}$ and $j \ge (k)^{1/3}$. Also define $T_{orig}(\ast, L)$ as the set of original types $(i', j')$ with $i \ge (k)^{1/3}$. 
		
		We first consider the left side. We classify the original types $T_{orig}(L, \ast)$ into two sub-categories $(i', j') \in T_{orig}(i', j')$ is \emph{left-category 0} if $c^l_{12}(i', j') = 0$; and \emph{left-category 1} if $c^l_{12}(i', j') \ge 1$. Similarly, classify the original types $T_{orig}(\ast, L)$ into two sub-categories as follows: $(i', j') \in T_{orig}(\ast, L)$ is \emph{left-category 0} if $c^l_{23}(i', j') = 0$; and \emph{left-category 1} if $c^l_{23}(i', j') \ge 1$. Similarly, define the notion of right-category 0 and right-category 1, respectively.
		
		
		First, we show how to handle left-category 0 types from $(L, \ast)$. Let these original types be $(i_1, j_1), (i_2, j_2), \ldots, (i_q, j_q)$. By definition, for each $1 \le p \le q$, $c^l_{12}(i_p, j_p) = 0$. Let $\ell = \argmin_{1 \le p \le q} c^l_{23}(i_p, j_p)$. We move all components in $\leftpart(i_p, j_p)$ to the subslot containing all the components of $\leftpart(i_\ell, j_\ell)$, while preserving the relative order within the components of $\leftpart(i_p, j_p)$ for each $1 \le p \le q$, as well as relative order of vertices within each component. Note that, for each component $C \in \leftpart(i_p, j_p)$, the number of crossings from outside $E(C)$, before the movement is $c_{23}(i_p, j_p)$, and that after the movement is $c_{23}(i_\ell, j_\ell) \le c_{23}(i_p, j_p)$. Therefore, we do not increase the number of crossings in this transformation. 
		
		We handle left-category 0 types from $(\ast, L)$, along with reduced type $(1, 0)$ in a very similar manner, which we describe for completeness. Let $(i_1, j_1), (i_2, j_2), \ldots, (i_{q}, j_{q})$, indicate the original types corresponding to $(\ast, L)$ such that $1 \le p \le q$, $c^l_{23}(i_p, j_p) = 0$. Let $\ell = \argmin_{1 \le p \le q} c^l_{12}(i_p, j_p)$. We move all components in $\leftpart(i_p, j_p)$ to the subslot containing all the components of $\leftpart(i_\ell, j_\ell)$, while preserving the relative order within the components of $\leftpart(i_p, j_p)$ for each $1 \le p \le q$, as well as relative order of vertices within each component. Note that, for each component $C \in \leftpart(i_p, j_p)$, the number of crossings of $\cedges$ from outside, before the movement is $c^l_{12}(i_p, j_p)$, and that after the movement is $c^l_{12}(i_\ell, j_\ell) \le c^l_{12}(i_p, j_p)$. Therefore, we do not increase the number of crossings in this transformation. 
		
		We do a similar transformation for right-category 0 types from $T_{orig}(L, \ast)$ and $T_{orig}(\ast, L)$ to move the components $\rightpart(\cdot, \cdot)$ of the respective types. We omit the description.
		
		Finally, we handle left/right-category 1 types from $(\ast, L)$ and $(L, \ast)$ by showing that the total number of components belonging to such types is bounded by $(k)^{2/3}$. Formally, we show the following claim.
		
		\begin{claim} \label{cl:3category1bound}
			$$\sum_{\substack{(i, j): \max(i, j) \ge (k)^{1/3} \\ (i, j) \text{ is left/right-category 1} }} |\cT'(i, j)| \le (k)^{2/3}.$$
		\end{claim}
		\begin{proof}[Proof of \Cref{cl:3category1bound}.]
			We prove that each component $C$ of original type $(i, j)$ with $(i, j)$ being left/right-category 1, accounts for a distinct set of $(k)^{1/3}$ crossings. Since the total number of crossings is $k$, this will show the claim. 
			
			First consider $C \in \leftpart(i', j')$ with $(i', j') \in T_{orig}(L, \ast)$, such that $(i', j')$ is left-category 1 (the proof for right-category 1 is analogous). Since $(i', j')$ is left-category 1, $c^l_{12}(i, j) \ge 1$. Note that $|\cedges| \ge k^{1/3}$. By the previous step, all the components in $\leftpart'(i', j')$ belong to a single slot, and by \Cref{cl:3slotintersection}, all the edges in $\cedges$ cross the same non-empty set of edges -- where the non-emptiness follows from $c^l_{12}(i, j) \ge 1$. Thus, there are $i \ge (k)^{1/3}$ such edges, and each edge is involved in at least $(k)^{1/3}$ crossings. 
			
			Now let us consider $C \in \leftpart(i', j')$, where $(i', j') \in T_{orig}(\ast, L)$ and $(i', j')$ is left-category 1 (the proof for right-category 1 is analogous). Since $(i', j')$ is category 1, $c^l_{23}(i', j') \ge 1$. Further, all the components in $\cT'(i', j')$ belong to a single slot. Further, by step 1, the edges of different components $C, C' \in \leftpart(i', j')$ do not cross each other; and all the edges of $\bigcup_{C \in \leftpart(i', j')} E(C)$ cross the same non-empty set of edges not in $\bigcup_{C \in \leftpart(i', j')} E(C)$ -- here, the non-emptiness follows from $c^l_{23}(i', j') \ge 1$. Thus, for each $C \in \leftpart(i', j')$, each edge in $E(C)$ is crossed at least once, and there are at least $j \ge (k)^{1/3}$ edges in $E(C)$, which amounts to at least $(k)^{1/3}$ crossings for each $C \in \leftpart(i', j')$. 
		\end{proof}
		Let $\cE$ denote the set of all components $C \in \cT'(i', j')$, where $(i', j')$ is left/right-category 1. \Cref{cl:3category1bound} shows that $|\cE| \le k^{2/3}$. Now, we use this $\cE$ to obtain a component-blueprint and component-guess.
		
		Recall the extended component-blueprint $\ecbp^o = (\totall^o, \totalr^o, \optl^o, \optr^o)$ defined earlier. Now, define $\cbp^p = (\totall^p, \totalr^p, \optl^p, \optr^p, \cE)$, where for each $(i, j) \in \types \setminus \LR{(L, \ast), (\ast, L)}$, define $\totall^p(i, j) = \totall^o(i, j)$, $\totalr^p(i, j) = \totalr^o(i, j)$, $\optl^p(i, j) = \optl^o(i, j)$, and $\optr^p(i, j) = \optr^o(i, j)$. On the other hand, for $(i, j) \in \LR{(L, \ast), (\ast, L)}$, define 
		\begin{itemize}
			\item $\totall^p(i, j) = \sum_{(i', j') T_{orig}(i, j) \text{ and left-category 0}} \totall^o(i', j')$,
			\item $\totalr^p(i, j) = \sum_{(i', j') \in T_{orig}(i, j) \text{ and right-category 0}} \totalr^o(i', j')$,
			\item $\optl^p(i, j) = \sum_{(i', j') T_{orig}(i, j) \text{ and left-category 0}} \optl^o(i', j')$, 
			\item $\optr^p(i, j) = \sum_{(i', j') T_{orig}(i, j) \text{ and left-category 0}} \optr^o(i', j')$.
		\end{itemize}	
		Now, notice that, any $\cbp^p$-compliant guess $\cguess^p$ can be extended to $\ecbp^o$-compliant $\ecguess^o$, by considering any partition of the category 1 types with appropriate constraints, which is essentially a partition of components in $\cE$. Then, applying the 3-step proof as described above, we can show that, for the $\cbp^p$ as defined above, for any $\cbp^p$-compliant $\cguess$, the drawing $\sigma'$ obtained at the end belongs to $\nicecfamily(\cbp^p, \cguess, \sep, \sigma)$, and also is $(\sep, \sigma)$-compatible. This completes the proof.
	\end{description}  
\end{proof}

\subsection{Step 2: Guessing the Interfaces} \label{subsec:3guessing}
In this step, we will ``guess'' various subsets of edges and vertices that can be used in the creation of two sub-instances in the next step. Each of the guesses are made by supposing that $\cI$ is a $k$-minimal \yin with an unknown drawing $\sigma$. Further, we describe the algorithm corresponding to bottom-heavy case, i.e., when $k_{12} \ge k_{23}$, which implies that $k_{12} \ge k/2$ -- top-heavy case is symmetric by exchanging the roles of lines 1 and 3. These guesses will be made in multiple steps, with each step being dependent on all possible choices of the previous steps. Further, if any of the guesses in a subsequent step contradicts a guess done at a previous step, then such a guess is deemed invalid and is terminated. These guesses can be thought of \emph{branching} on each of the specified choices in all steps. If all the guesses are terminated, or the corresponding recursive calls report that the sub-instances are {\sc No}-instances, then we will conclude that our assumption that $\cI$ was a \yin was wrong, and thus we must have a {\sc No}-instance.

In each step, we specify at a high level the kind of object that we are trying to guess. Whenever we are trying to guess a subset of edges, we follow the following convention: all the edges of a multi-edge are treated in the same manner, such as being chosen in a guess.

\begin{description}[leftmargin=10pt]
	\item[Step A: Separator.] We start by guessing the separator $\sep = \LR{\lsepone, \rsepone, \lseptwo, \rseptwo}$ w.r.t. a hypothetical (unknown) drawing $\sigma$. Note that there are at most $m^4 = n^{\Oh(1)}$ choices. Furthermore, based on our guess, we can also conclude whether we are in (\DB, \DT), (\DB, \ST), (\SB, \DT), or (\SB, \ST) case. Next, we guess $\mde(\sep) \subseteq E_{12}$ of size at most $2\sqrt{k}$, and two sets $\gapleft, \gapright \subseteq V_2^{conn}$ of size at most $3\sqrt{k}$, as in the definition of \Cref{def:3separator}. There are at most $n^{\Oh(\sqrt{k})}$ choices for this. Overall, we have at most $n^{\Oh(\sqrt{k})}$ choices. 
	

	
	\item[Step B: Labeling $\epts(\boundary) \cup \epts(\sep) \cup \gapleft \cup \gapright$.]\ 
	\\The vertices of $\epts(\leftbone)$ and $\epts(\leftbtwo)$ receive labels $\lf$. The vertices of $\epts(\lsepone) \cup \epts(\lseptwo) \cup \gapleft$ are each given two labels $\lf, \mm$. Similarly, the vertices of $\epts(\rsepone) \cup \epts(\rseptwo) \cup \gapright$ are given two labels $\mm, \rt$. Note that the assignment of labels in this step does not overwrite previous labels. For example,  in \SB/\ST case the shared bottom (resp.~top) endpoint of $\lsepone, \lseptwo$ (resp.~$\lseptwo, \rseptwo$) gets all three labels. Similarly, if $\lsepone, \lseptwo$ share top endpoint, then the shared endpoint will also receive all three labels. Note that there is no guessing involved in this step.
	
	
	\item[Step C: $\crosse$ and labeling their endpoints.] For each edge $e \in \sep$, we guess $\crs(e)$, which is a set of at most $\sqrt{k}$ edges that cross $e$. Let $\crosse \coloneqq \bigcup_{e \in \sep} \crs(e)$. 
	
	We now describe feasible pairs for labels of the endpoints of $\crosse$; and each guess corresponds to using labeling the endpoint of each edge in $\crosse$ by one of the feasible (ordered) pairs. However, if the guess for the label is \emph{incompatible} with a previously assigned label, then we consider it as an invalid guess and proceed to the next one. Note that, if a vertex has previously been assigned multiple labels, then the incompatible label(s) is a label that has not been assigned to the said vertex, if any.
	\begin{itemize}
		\item For $e \in \crs(\lsepone) \cup \crs(\lseptwo)$, the feasible label pairs are $(\lf, \mm), (\mm, \lf), (\lf, \rt), (\rt, \lf)$.
		\item For $e \in \crs(\rsepone) \cup \crs(\rseptwo)$, the feasible label pairs are $(\rt, \mm), (\mm, \rt), (\lf, \rt), (\rt, \lf)$.
	\end{itemize}
	Overall, we have at most $\binom{m}{4\sqrt{k}}^4 \cdot 2^{\Oh(\sqrt{k})} = n^{\Oh(\sqrt{k})}$ choices. Further, let ${\sf CrossLeft, CrossMid,}$ ${\sf CrossRight} \subseteq \epts(\crosse)$ denote the set of vertices that receive labels $\lf, \mm, \rt$, respectively.

	\item[Step D: Edges incident to $\gapleft \cup \gapright$.]\ 
	\begin{itemize}
		\item Consider an edge $uv \in E_{12} \setminus \crosse$, where $v \in \gapleft$ (resp.~$v \in \gapright$). Then, we label $u$ as $\lf$ (resp.~$\rt$). 
		\item Consider an edge $uv \in E_{23} \setminus \crosse$, where $v \in \gapleft \cup \gapright$. Then, we label $u$ as $\mm$.
	\end{itemize}
	
	\item[Step E: Handling $\comps(v)$ for $v \in \epts(\sep)$.] 
	Recall that for each $v \in \epts(\sep)$, each component $C \in \comps(v)$ can be classified into $\spl(v), \predef(v), \updown(v), \pecunp(v)$, and $\pecunp(v)$ according to \Cref{def:3componenttypes}. Now we describe how to handle each of the components separately.
	
	\begin{description}[leftmargin=2pt]
		\item[$\blacktriangleright$ $\spl(v)$.] Some of the vertices in such components are already labeled. We will return to the remaining vertices in these components at a later step.
		
		\item[$\blacktriangleright$ $\predef(v)$.] Recall that each $\pi_i$ for $i \in [3]$ is comprised of $\lambda = \Oh(\log k^\star)$ total orders $\pi_i^1, \pi_i^2, \ldots, \pi_i^\lambda$ for $V_i^1, V_i^2, \ldots, V_i^\lambda$, respectively. For each $i \in [3]$, and for each $j \in [\lambda]$, we do the following. 
		\begin{itemize}
			\item We guess two vertices $u(j, l), w(j, r) \in V^j_i$ with $\pi^j_i(u(j, l)) < \pi^j_i(w(j, r))$.
			\item Any $C \in \predef(v)$ that contains a vertex $u \in V^j_i$ with $\pi^j_i(u) \le \pi^j_i(u(j, l))$, we label all the vertices of $C$ using $\lf$. 
			\item Any $C \in \predef(v)$ that contains a vertex $w \in V^j_i$ with $\pi^j_i(w) \ge \pi^j_i(w(j, r))$, we label all the vertices of $C$ using $\rt$.
		\end{itemize}  
		If any of the guesses made in this step is incompatible with a previous step, or with a prior guess made in this step, then we terminate this guess. Note that the number of guesses made in this step is at most $n^{\Oh(\log k^\star)} = n^{\Oh(k^{2/3})}$ since $k = \Omega((k^{\star})^{2/3})$. 
		
		\item[$\blacktriangleright$ $\updown(v)$.] By \Cref{lem:3type3components}, the number of such components in an \yin is bounded by $8\sqrt{k}$; otherwise we terminate this guess. Otherwise, for each such component, we guess a label from $\LR{\lf, \rt}$ and use that label for all the vertices of the component. The number of guesses is $n^{\Oh(\sqrt{k})}$.
		
		\item[$\blacktriangleright$ $\pendant(v)$ for all $v \in \epts(\sep)$.] \ 
		\begin{itemize}[leftmargin=14pt]
			\item We guess a pendant-blueprint $$\pbp = \LR{(\pendantsl(v), \pendantsm(v),  \pendantsr(v)): v \in \epts(\sep)}$$ in accordance with \Cref{def:3pendantblueprint}. Note that the number of guesses for pendant-blueprints is bounded by $n^{\Oh(1)}$.
			\item Now, fix a $\pbp$. For each $v \in \epts(\sep)$, we use a variant of 3D knapsack to find a partition (if one exists) $(\pendantl(v), \pendantm(v)$, $\pendantr(v))$ of $\pendant(v)$ for each $v \in \epts(\sep)$ in polynomial time, such that the resulting $\pguess$ is compatible with $\pbp$ (as defined in \Cref{def:3blueprintcompatibleguess}). If such a partition does not exist, then we terminate the guess.
			\item Then, for each $v \in \epts(v)$, we label each pendant in $\pendantl(v), \pendantm(v)$, and $\pendantr(v)$ with label $\lf, \mm$ and $\rt$, respectively.
		\end{itemize}
		Note that the number of guesses in this step is bounded by $n^{\Oh(1)}$, and otherwise the step takes polynomial time. 
		\\At this point, for all $v \in \epts(\sep, 2)$, we have labeled all the components in $\comps(v)$ (other $\spl(v)$ as mentioned above). Thus, we are left with handling $\pecunp(v)$ for $v \in \epts(\sep, 1) \cup \epts(\sep, 3)$. 
		
		\item[$\blacktriangleright$ $\pecunp(v)$ for all $v \in \epts(\sep)$.] We consider different cases based.
		
		(i) (\DB, \DT) case: Here, $v \in \epts(\sep, 1) \cup \epts(\sep, 3)$ can be labeled with $\lf/\rt$ as appropriate in accordance with \Cref{lem:3pecunp-cases}. 
		
		(ii) (\SB, \ST) case, \Cref{lem:3diffcases}, we can guess two subsets of  $P_1 \subseteq \pecunp(\ept(\lsepone, 1))$ and $P_2 \subseteq \pecunp(\ept(\lseptwo, 3))$, each of size at most $\sqrt{k}$. Then, we consider the following two possibilities (guesses).
		\begin{enumerate}
			\item Label all components of $P_1$ using $\lf$, and all components of $\pecunp(\ept(\lsepone, 1)) \setminus P_1$ using $\rt$. 
			\\Label all components of $P_2$ using $\rt$, and label all components of $\pecunp(\ept(\lseptwo, 3)) \setminus P_2$ using $\lf$.
			\item Label all components of $P_1$ using $\rt$, and all components of $\pecunp(\ept(\lsepone, 1)) \setminus P_1$ using $\lf$. 
			\\Label all components of $P_2$ using $\lf$, and label all components of $\pecunp(\ept(\lseptwo, 3)) \setminus P_2$ using $\rt$.
		\end{enumerate}

		(iii) (\SB, \DT) case. For the rest of the description of step E, let $v \coloneqq \ept(\lsepone, 1) = \ept(\rsepone, 1)$, and $\cC \coloneqq \pecunp(v)$. 
		\begin{itemize}[leftmargin=14pt]
			\item We guess a subset $\exceptional \subseteq \cC$ of size at most $11 \cdot k^{2/3}$ as in \Cref{lem:3exceptional}. Note that there are at most $n^{\Oh(k^{2/3})}$ guesses. Let $\cC' \coloneqq \cC \setminus \exceptional$.
			\item For each $C \in \cC'$, we use the two layer algorithm of \Cref{thm:2layertheorem} to find an optimal number of crossings, $\opt(C) \le k$ (if $\opt(C) > k$, then clearly we have a {\sc No}-instance). This takes time $|V(C)|^{\Oh(\sqrt{k})} = n^{\Oh(\sqrt{k})}$ (note that we skip the kernelization step), and obtain a drawing of $C$. 
			\item According to \Cref{def:3componenttypes}, we classify $\cC'$ into reduced types, and use $\cT(i, j) \subseteq \cC'$ to denote the set of components of type $(i, j)$. Note that the number of reduced types is at most $k^{2/3} + 2$. 
			\item Now, we guess a component-blueprint $\cbp$ according to \Cref{def:3blueprint}. Note that the number of guesses for a blueprint is bounded by $n^{\Oh(k^{2/3})}$. Fix one such $\cbp = (\totall, \totalr, \optl, \optr, \cE)$.
			\item We try to construct a $\cbp$-compliant $\cguess = (\leftpart, \rightpart)$ as follows: for each reduced type $(i, j) \in \types$, we try to find a partition of $\cT(i, j) \setminus \cE$ into \\$(\leftpart(i, j), \rightpart(i, j))$ using a version of 4D knapsack (this runs in polynomial time) that satisfies the following two conditions:
			\begin{itemize}
				\item $|\leftpart(i, j)| = \totall(i, j)$, and $|\rightpart(i, j)| = \totalr(i, j)$
				\item $\opt(\leftpart(i, j)) = \optl(i, j)$, and $\opt(\rightpart(i, j)) = \optr(i, j)$.
			\end{itemize}
			If such $\cguess$ does not exist, then we terminate the current guess. 
			\item Otherwise, for each component $C \in \bigcup_{(i, j) \in \types} \leftpart(i, j)$, we label it using $\lf$, and for each component $C \in \bigcup_{(i, j) \in \types} \rightpart(i, j)$, we label it using $\rt$.
		\end{itemize}
		
		(iv) (\DB, \ST) case: this case is handled symmetrically as (\SB, \DT) case as explained above in (iii). We omit the description.
	\end{description}

	\item[Step F: Label propagation.] At this point, some of the vertices of the graph may still be unlabeled. 
	
	Let $A$ denote the set of vertices that have already received a label, and let $A^\lf, A^\mm, A^\rt \subseteq A$ denote the subsets of vertices that have received label $\lf, \mm, \rt$, respectively (note that, if a vertex has multiple labels, then it belongs to multiple subsets). Further, let $O \subseteq A$ denote the subset of vertices that have only received one label, and again $O^\lf, O^\mm, O^\rt$ denote the corresponding partition of $O$.
	
	Next, we define three sets of vertices, as follows.
	\begin{align}
		{\sf Origin}^\lf &\coloneqq \lbnd \cup \epts(\LR{\leftbone, \leftbtwo, \lsepone, \lseptwo}) \cup \gapleft \cup {\sf CrossLeft} \nonumber
		\\{\sf Origin}^\rt &\coloneqq \rbnd \cup \epts(\LR{\rsepone, \rseptwo, \rightbone, \rightbtwo}) \cup \gapright \cup {\sf CrossRight} \nonumber
		\\{\sf Origin}^\mm &\coloneqq \epts(\sep) \cup \epts(\mde) \cup \gapleft \cup \gapright \cup {\sf CrossMid} \label{eqn:origindef}
	\end{align}

	Now, we proceed as follows.
	\begin{enumerate}
		\item For each $u \in {\sf Origin}^\lf$, and each unlabeled $w \not\in A$ that is reachable from $u$ in the graph $G - (O^\mm \cup O^\rt)$, we label it using $\lf$.
		\item For each $u \in {\sf Origin}^\rt$, and each unlabeled $w \not\in A$ that is reachable from $u$ in the graph $G - (O^\lf \cup O^\mm)$, we label it using $\rt$.
		\item For each $u \in {\sf Origin}^\mm$, and each unlabeled $w \not\in A$ that is reachable from $u$ in the graph $G - (O^\lf \cup O^\rt)$, we label it using $\mm$.
	\end{enumerate}
	During this process, we terminate the guess, if any of the following happens.
	\begin{itemize}
		\item A vertex $w \not\in A$ receives two different labels.
		\item A vertex $w \in V_1 \setminus A$ that has at least one neighbor in $V_2$ receives label $\mm$.
		\item A vertex $w \in V_2 \setminus A$ that has at least one neighbor in $V_1$ receives label $\mm$.
	\end{itemize}
	Otherwise, let $V^{\cal T}$ denote the set of vertices that have received the label ${\cal T}$, for ${\cal T} \in \LR{\lf, \mm, \rt}$, respectively. Note that these sets are not disjoint. 

	\item[Step G: Partial Order.] 
	For each $i \in [3]$, we will guess permutation $\pi^{\lambda+1}_i$ (total order) of a certain subset of vertices of $V_i$ in the drawing $\sigma$, which will be used to extend $\pi_i$ to a new partial order $\pi'_i$. 
	Let $S \coloneqq \epts(\boundary) \epts(\sep) \cup \epts(\mde) \cup \epts(\crosse) \cup \gapleft \cup \gapright$. From the bounds on the respective sets of vertices and edges, it is straightforward to check that $|S| = \Oh(\sqrt{k})$. For each $i \in [3]$, let $S_i \coloneqq V_i \cap S$. We guess permutations $\pi^{\lambda+1}_i$ of $S_i$ for each $i \in [3]$. Since $|S_i| = \Oh(\sqrt{k})$, it follows that the total number of guesses is $|S_1|! \cdot |S_2|! \cdot |S_3|! = (\Oh(\sqrt{k})!)^3 = 2^{\Oh(\sqrt{k} \log k)}$. Next, we define $\pi^{\lambda+2}_2$ as an ordering over $\bnd \cup \epts(\sep, 2) \cup \gapleft \cup \gapright $ by combining parts of $\pi^{\lambda+1}_2$, and one of the constituent total orders of $\pi^2$ of the set $\bnd$.
	Let $\pi^{\lambda+1} = (\pi^{\lambda+1}_1, \pi^{\lambda+1}_2, \pi^{\lambda+1}_3)$. We say that $\pi^{\lambda+1}$ is \emph{valid} if the following conditions are satisfied:
	\begin{itemize}
		\item For each $i \in [3]$, for any distinct $u, v \in S_i$ with $\pi^{\lambda+1}_i(u)$ $\pi_i$ defines an order between $u$ and $v$, then both $\pi_i$ and $\pi_i^{\lambda+1}$ are order $u$ and $v$ consistently. \footnote{Note that, it is possible that $u$ and $v$ are incomparable in $\pi_i$, but since they both belong to $S_i$, they are comparable in $\pi^{\lambda+1}_i$.}
		\item $\pi^{\lambda+1}$ satisfies the following properties (that must satisfied by the separator):
		\begin{itemize}
			\item $\lsepone, \rsepone$ do not cross, $\rsepone, \rseptwo$ do not cross
			\item For each $e \in \sep$, the set of edges crossing $e$ is exactly equal to $\crs(e)$,
			\item $\mde$ is exactly the set of edges that lie between $\lsepone, \rsepone$
		\end{itemize}
	\end{itemize}
	Note that the validity of $\pi^{\lambda+1}$ can be checked in polynomial time. If a guess for $\pi^{\lambda+1}_i$ is invalid, then we terminate it and move to the next guess. Otherwise, for each $i \in \LR{1, 3}$, let $\pi'_i = \pi_i \cup \pi^{\lambda+1}_i$, and $\pi'_2 = \pi_2 \cup \pi^{\lambda+1}_2 \cup \pi^{\lambda+2}_2$, and note that $\pi'_i$ is a valid extension of $\pi_i$, since $\pi^{\lambda+1}$ is valid. Henceforth, we assume that we are working with such a $\pi' = (\pi'_1, \pi'_2, \pi'_3)$ for some guess for $\pi^{\lambda+1}$. 
	
	\item[Step H: Parameter division.]
	First, let $k_{\sf special}$ denote the number of crossings (with multiplicity) of the following form according to the prior guess:
	\begin{itemize}
		\item $e_1 \in \crs(\lseptwo), e_2 \in \crs(\rseptwo)$ such that $e_1$ and $e_2$ cross according to their ordering in $\pi^{\lambda+1}$.
	\end{itemize}
	Let $k_{\sf current}$ denote the number of crossings among the edges of $\crosse \cup \mde \cup \sep$ according to $\pi^{\lambda+1}_i$, except those counted in $k_{\sf special}$. If $k_{\sf current} > k$, then we terminate the current guess. Otherwise, let $k' \coloneqq k - k_{\sf current}$. Let $k'$ denote the final value. We guess three non-negative integers $0 \le k^\lf, k^{\mm}, k^\rt$ such that $k^{\lf}, k^{\rt} \le 3k/4$, and $k^{\lf} + k^{\mm} + k^{\rt} = k'$. Further, we also guess $k^{\lf}_{12}, k^{\lf}_{23}, k^{\rt}_{12}, k^{\rt}_{23} \ge 0$ such that $k^{\lf} = k^{\lf}_{12} + k^{\lf}_{23}$, and $k^{\rt} = k^{\rt}_{12} + k^{\rt}_{23}$. Note that the number of such guesses is bounded by $k^{\Oh(1)}$.

\end{description}

From the discussion above, the following lemma follows. 
\begin{lemma} \label{lem:3guessbound}
	The total number of guesses made in steps A through H, for (i) separator $\sep$ and its associated sets of vertices and edges, (ii) labels of vertices of $V(G)$, (iii) partial orders $\pi'i$, and (iv) division of $k$ into smaller parameters, is bounded by $n^{\Oh(k^{2/3})}$.
\end{lemma}

In the subsequent steps, we assume that we are working with one such fixed guess for the separator $\sep$ and its associated objects, labels of vertices, partial orders $\pi' = (\pi'_1, \pi'_2, \pi'_3)$, and the smaller parameters. 

\subsection{Step 3: Creating sub-instances.} \label{subsec:3creating}

For each set of guesses made in the previous step (as stated in \Cref{lem:3guessbound}), we create three sub-instances (in fact, extended instances), (i) a left sub-instance $\cI^\lf$, and (ii) a right sub-instance $\cI^\rt$, and (iii) a middle sub-instance $\cI^\mm$. Among these, $\cI^\lf$ and $\cI^\rt$ are instances of \probThreeCr, whereas $\cI^\mm$ is an instance of \probTwoCr. In this step (Step 3), we describe the construction of $\cI^\lf$ and $\cI^\rt$, and in the next step (Step 4) we will recursively call our algorithm to find a drawing for each of the two sub-instances; or to conclude that either or both of them are {\sc No}-instances, in which case we terminate the guess. Then, in Step 5, we will use the drawings of $\cI^\lf, \cI^\rt$ found by the recursive algorithm to construct the middle sub-instance $\cI^\mm$ of \probTwoCr, which will be solved using the algorithm of \Cref{thm:2layertheorem}. 

Let $V^{\lf}, V^{\mm}, V^{\rt}$ denote the set of vertices that have received labels $\lf, \mm, \rt$, respectively -- note that, if a vertex has received multiple labels, then it belongs to all the respective subsets. Now, we proceed to the description of $\cI^\lf, \cI^\rt$. 

\subsubsection*{Creating Left Sub-instance}

\begin{description}[leftmargin=2pt]
	\item[$\blacktriangleright$ Base graph.] Let $G'^{\lf} = G[V^{\lf}]$, with its edge-set partitioned into regular and multi-edges, denoted by $\regulare^\lf$ and $\weightede^\lf$, respectively. 
	\item[$\blacktriangleright$ Boundary edges.] Let $\boundary^{\lf} \coloneqq \LR{\leftbone^\lf, \rightbone^\lf, \rightbone^\lf, \rightbtwo^\lf},$ where $\leftbone^\lf \coloneqq \leftbone, \leftbtwo \coloneqq \leftbtwo, \rightbone^\lf \coloneqq \lsepone, \rightbtwo^\lf \coloneqq \lseptwo$.
	\item[$\blacktriangleright$ Extended Boundary.] Let $\lbnd^\lf \coloneqq \lbnd \cap V^\lf$, $\rbnd^\lf \coloneqq \gapleft \setminus \lbnd^\lf$. Then, $\bnd^\lf = \lbnd^\lf \cup \rbnd^\lf$.
	\item[$\blacktriangleright$ Multi-edges] Next, we add a few more multi-edges to the set $\weightede^\lf$, as follows:
	\\For each (multi-)edge $uv \in \crosse$ with at least one endpoint in $V^{\lf}$, we identify a multi-edge $e_{add}(uv)$ to be added in each case. After the case analysis, we will discuss the multiplicity of this edge, and actually add it to the graph.
	\begin{itemize}[leftmargin=15pt]
		\item $uv \in E_{12}$, where $u \in V_1$ has a label $\lf$ and $v \in V_2$ has label $\mm$ or $\rt$. 
		\\Note that, since $uv \in \crosse$, the vertex $v$ cannot belong to $\gapleft$ and hence cannot have a label $\lf$ -- this is ensured while guessing the labels. 
		\\\textbf{Operation.} $e_{add}(uv) \coloneqq (u, \ept(\lsepone, 2))$.
		\item $vu \in E_{12}$, where $u \in V_2$ has label $\lf$ and $v \in V_1$ has label $\mm$ or $\rt$.
		\\\textbf{Operation.} $e_{add}(uv) \coloneqq (u, \ept(\lsepone, 1))$ to $\weightede^\lf$. 
		\item $vu \in E_{23}$, where $u \in V_3$ has label $\lf$ and $v \in V_2$ has label $\mm$ or $\rt$. 
		\\Note that the vertex $v$ may belong to $\gapleft \subseteq V^{\lf}$, and $u$ cannot be equal to $\ept(\lseptwo, 3)$.
		\\\textbf{Operation.} If $v \in \gapleft \cup \LR{\ept(\lsepone, 2)}$, then we also delete the edge $vu$ from $E(G^{\lf})$. 
		$e_{add} \coloneqq (\ept(\lseptwo, 2), u)$.  
		\item $uv \in E_{23}$, where $u \in V_2$ has label $\lf$ and $v \in V_3$ has label $\mm$ or $\rt$.
		\\Note that, since $uv \in \crosse$, the vertex $u$ cannot belong to $\gapleft$, and hence cannot have a label $\lf$ -- this is ensured while guessing the labels.
		\\\textbf{Operation.} $e_{add}(uv) \coloneqq (u, \ept(\lseptwo, 3))$.
	\end{itemize}
	Now, we proceed as follows:
	\begin{itemize}
		\item If an edge is defined as $e_{add}$ of multiple edges, say $e_1, e_2, \ldots$, then its multiplicity is defined as the sum of the respective multiplicities of $e_1, e_2, \ldots$ -- here we adopt the convention that, if an edge $e_i$ belongs to $\regulare$, then its multiplicity is $1$.
		\item Finally, we add $e_{add}$ with the combined multiplicity to $\weightede^\lf$. Here, if $e_{add}$ is already in $\regulare$, we do not remove it.
	\end{itemize}
	Let $G^\lf$ denote the resulting graph after adding the edges in this manner. 
	\item[$\blacktriangleright$ Vertex set.] For each $i \in [3]$, let $V^\lf_i \coloneqq V^\lf \cap V_i$, and let $\pi^\lf_i$ denote $\pi'_i$ projected to the vertices of $V^\lf$. 
	\item[$\blacktriangleright$ Partial order] Let $\pi^\lf = (\pi^\lf_1, \pi^\lf_2, \pi^\lf_3)$. Observe that $(\pi^\lf)^{\lambda+2}_2$ contains an ordering of $\bnd^\lf$ that satisfies the required properties.
	\item[$\blacktriangleright$ Final instance.] Let $\cI^\lf = (G^\lf, V^\lf_1, V^\lf_2, V^\lf_3, \bnd^\lf, \pi^\lf, k^\lf_1, k^\lf_2, k^\lf)$.
\end{description}
At this point, we check whether in the graph $G^\lf$, each vertex of $V(G^\lf)$ is reachable from some vertex of $\bnd^\lf \cup \epts(\boundary)$ in $G^\lf$. If not, we terminate the guess. Otherwise, we continue to the next step.

At this point, we note the following observations akin to the two-layer algorithm.
\begin{itemize}
	\item As required in \Cref{item:3partialorder}, the total number of partial orders in each $\pi_i$ remains bounded by $\Oh(\log k^\star)$, since parameter reduces by at least a constant factor in each step. 
	\item Note that $|\lbnd^\lf| \le |\lbnd| \le \alpha \sqrt{k^\star}$, and $|\lbnd^\rt| \le |\gapleft| \le 3\sqrt{k}$, and analogous bounds hold for $|\rbnd^\rt|,|\lbnd^\rt|$, respectively. Thus, $\cI^\lf, \cI^\rt$ satisfy \Cref{item:3extdconn}.
	\item Using a similar argument as in two layer case, it follows that the total sum of multiplicities of \weightede incident to each vertex in $\epts(\sep)$ remains bounded by $2\sqrt{k^\star}$, satisfying \Cref{item:3multiplicity-sum}
\end{itemize}

\subsubsection*{Creating Right Sub-instance}
The right sub-instance $\cI^\rt = (G^\rt, V^\rt_1, V^\rt_2, V^\rt_3, \bnd^\rt, \pi^\rt, k^\rt_1, k^\rt_2, k^\rt)$ is created analogously, and we omit the description. Again, we check whether there exists some vertex of $V(G^\rt)$ that is not reachable from some vertex of $\epts(\boundary^\rt) \cup \bnd^\rt$ in $G^\rt$, and terminate the guess. Otherwise, we proceed to the next step.

From the preceding constructions of $\cI^\lf, \cI^\rt$---assuming that the guesses are not terminated---it is straightforward to verify that they satisfy the properties of extended instances from \Cref{def:extendedinstance3}. Thus, we have the following claim, whose proof is omitted.

\begin{claim} \label{cl:3validlr}
	$\cI^\lf$ and $\cI^\rt$ are valid extended instances of \probThreeCr.
\end{claim}

\subsection{Step 4: Conquer left and right.} \label{subsec:3leftright}

For each guess made in Step 2, we have defined a pair of sub-instances $\cI^\lf$ and $\cI^\rt$ in Step 3, with each having parameter at most $3k/4$. For each such pair of instances, which we solve recursively in time $T_3(k^\lf) + T_3(k^\rt) \le 2 T_3(3k/4) \le 2 \cdot n^{\Oh((3k/4)^{2/3})}$. There are two possibilities. (i) At least one of the recursive calls returns that the corresponding sub-instance is a \textsc{No}-instance, in which case we proceed to the next guess. (ii) Otherwise, if both of the recursive calls return that $\cI^{\lf}$ and $\cI^{\rt}$ are \textsc{Yes}-instances along with the respective drawings $\sigma^\lf, \sigma^\rt$, then we proceed to the next step. 

\subsection{Step 5: Conquer Middle.} \label{subsec:3middle}

For each $i \in [3]$, let $M_i \coloneqq V_i \cap V^\mm \cap S$, where $S = \epts(\sep) \cup \epts(\mde) \cup \epts(\crosse) \cup \gapleft \cup \gapright$ as defined in step G. Note that $M_i$ is a subset of the set $S_i$ defined in Step H, for which we guess the permutation $\pi^{\lambda+1}_i$. 
%

Note that, the relative order among the vertices of $M_i$ for each $i \in [3]$ is already determined by $\pi^{\lambda+1}_i$. In particular, note that that $M_1 = V^\mm \cap V_1$, and $\pi^{\lambda+1}_1$ completely determines the ordering of $M_1$. Now we are left with determining the order of vertices in $V^\mm \cap V_2$ and $V^\mm \cap V_3$ that respects the partial order $\pi'$. To this end, we define an instance of \probTwoCr, as follows. 

\paragraph{Defining a two-layer instance.} 

Note that the base graph for the instance of \probTwoCr will be defined on the vertex subsets $V^\mm \cap V_2$ and $V^\mm \cap V_3$, respectively. Later, we will rename them to first and second layer, respectively.

\begin{description}[leftmargin=2pt]
	\item[$\blacktriangleright$ Base graph.] Let $G'^{\mm} = G[V^\mm \cap (V_2 \cap V_3)]$, with its edge-set partitioned into regular and multi-edges, denoted by $\regulare^\mm$ and $\weightede^\mm$, respectively. 
	\item[$\blacktriangleright$ Extended boundary.] Let $\lbnd^\mm \coloneqq \gapleft \cup \ept(\lsepone, 2) \cup \ept(\lseptwo, 2)$, and $\rbnd^\mm \coloneqq \gapright \cup \ept(\rsepone, 2) \cup \ept(\rseptwo, 2)$. Then, $\bnd^\mm \coloneqq \lbnd^\mm \cup \rbnd^\mm$.
	\item[$\blacktriangleright$ Boundary edges.] Let $\boundary^{\mm} \coloneqq \LR{\leftb^\mm, \rightb^\mm},$ where $\leftb^\mm \coloneqq \lseptwo, \rightb^\lf \coloneqq \rseptwo$.
	\item[$\blacktriangleright$ Multi-edges.] For each (multi-)edge $uv \in \crosse \cap E_{23}$ with at least one endpoint in $V^{\mm}$, we add a multi-edge to the graph $G^{\mm}$, based on the following cases.
	\begin{itemize}[leftmargin=15pt]
		\item $vu \in E_{23}$, where $u \in V_3$ has label $\lf$ and $v \in V_2$ has label $\mm$ or $\rt$. 
		\\Note that the vertex $v$ may belong to $\gapleft \subseteq V^{\lf}$, and $u$ cannot be equal to $\ept(\lseptwo, 3)$.
		\\\textbf{Operation.} First, if $v \in \gapleft \cup \LR{\ept(\lsepone, 2)}$, then we also delete the edge $vu$ from $E(G^{\lf})$. We add a multi-edge $(\ept(\lseptwo, 2), u)$ to $\weightede^\mm$.  
		\item $uv \in E_{23}$, where $u \in V_2$ has label $\lf$ and $v \in V_3$ has label $\mm$ or $\rt$.
		\\Note that, since $uv \in \crosse$, the vertex $u$ cannot belong to $\gapleft$, and hence cannot have a label $\lf$ -- this is ensured while guessing the labels.
		\\\textbf{Operation.} We add a new multi-edge $(u, \ept(\lseptwo, 3))$ to $\weightede^\mm$.
	\end{itemize}
	For each edge considered in each of the preceding cases, if the original edge $uv \in \weightede$, then the newly added edge gets multiplicity equal to $\mu(uv)$; otherwise, it gets multiplicity equal to $1$.
	\\Let $G^\mm$ denote the resulting graph after adding the edges in this manner.
	\item[$\blacktriangleright$ Vertex set and partial Order.] Let $W_1 \coloneqq V^\mm \cap V_2$ and $W_2 \coloneqq V^\mm \cap V_3$, and let $\pi^\mm_1 = \pi'_2$ projected to the vertices of $W_1$, and $\pi^\mm_2 = \pi'_3$ projected to the vertices of $W_2$. Let $\pi^\mm = (\pi^\mm_1, \pi^\mm_2)$. Note that, $\gapleft$ (resp.~$\gapright$) must contain all the vertices of $V^{conn}_2$ that are between $\ept(\lseptwo, 2)$ and $\ept(\lsepone, 2)$ (resp.~$\ept(\rsepone, 2)$ and $\ept(\rseptwo, 2)$).  Thus, if $(\pi^\mm)^{\lambda+2}_1$, which is $\pi^{\lambda+2}_2$ projected to $V_2$ if violates this condition, then we terminate the guess.
	\item[$\blacktriangleright$ Final instance.] Let $\cI^\mm = (G^\mm, W^\mm_1, V^\mm_2, \pi^\mm, k^\mm)$.
\end{description}
Finally, we check whether each vertex of $V(G^\mm)$ is reachable from some vertex of $\bnd^\mm \cup \epts(\boundary^\mm)$ in the constructed graph $G^\mm$. If not, we terminate the guess. Otherwise, it can be verified that the constructed instance satisfies all the properties of an elaborate extended instance of \probTwoCr from \Cref{def:extendedinstance3}, which we state in the following claim.

\begin{claim} \label{cl:3validmiddle}
	$\cI^\mm$ is a valid elaborate extended instance of \probTwoCr.
\end{claim}

Given this claim, we call the recursive algorithm of \Cref{thm:2layertheorem} (without kernelizing) to solve it in time $T_2(n, k) \le n^{\Oh(\sqrt{k})}$. Either, the algorithm returns that $\cI^\mm$ is a {\sc No}-instance, or it returns a drawing $\sigma^\mm$. In the former case, we terminate the guess. Otherwise, we proceed to the following final step.


\subsection{Step 6: Obtaining the Final Drawing.} \label{subsubsec:3finaldrawing}

If, for any of the non-terminated guesses, all three recursive calls reporting that the respective sub-instances are  \textsc{Yes}-instances, along with the corresponding orderings, then we can combine these orderings to obtain orderings of $V_1, V_2$ that has at most $k$ crossings, and respect the given partial order $\pi$. 

Specifically, suppose $\sigma^\lf, \sigma^\mm, \sigma^\rt$ be the drawings of $\cI^\lf, \cI^\rt$ returned by the corresponding recursive calls. Then, we obtain the final drawing of $\cI$ by defining $\so_i \coloneqq \sigma_i^\lf \circ \sigma_i^\mm \circ \sigma_i^\rt$ for $i \in [3]$. We return $\so = (\so_1, \so_2, \so_3)$ as our final drawing. In the following lemma, we show that the combined drawing is a valid drawing of $\cI$ with $k$ crossings.

\begin{lemma} \label{lem:3drawingcomb}
	Let $\cI^\lf, \cI^\rt, \cI^\mm$ be the three sub-instances defined in Steps 3 and 5 (with respective parameters $k^\lf, k^\rt, k^\mm$) corresponding to some guess made in Step 2. Further, let  $\sigma^\lf, \sigma^\rt, \sigma^\mm$ be valid drawings of the $\cI^\lf, \cI^\rt, \sigma^\mm$. Then, if we obtain a combined drawing $\sigma'$ as in Step 6, then $\sigma'$ is a valid drawing of $\cI$.
\end{lemma}
\begin{proof}\ 
	Recall that, in step H we guess the permutations $\pi^{\lambda+1}_i$ of $S_i$, and we use it in step I to determine the number of crossings $k_{\sf current}$ among the edges of $\crosse \cup \mde \cup \sep$. Then, we define $k' = k - k_{\sf current}$, and divide it into $k^\lf + k^\mm + k^\rt$, which are used as parameters for the respective sub-instances. 
	
	Next, we claim that all the $k_{\sf current}$ crossings among the edges are not counted in any of the sub-instances. To this end, first observe that the edges of $\mde$ are not involved in any of the subproblems, since the instance $\cI^\mm$ is only defined over layers 2 and 3. The proof is analogous to that of \Cref{lem:2drawingcomb}, except that we need to be careful about the following case. Consider $e_1, e_2 \in \crosse$ where $e_1 \in \crs(\lseptwo)$, $e_2 \in \crs(\rseptwo)$ and suppose $e_1, e_2$ cross according to $\pi^{\lambda+1}_i$. Then, note that the crossing between $e_1, e_2$ is a special crossing as defined in Step H, which is not counted in $k_{\sf current}$. Therefore, it implies that any crossing between $e_1, e_2 \in \crosse$, that contributes towards $k_{\sf current}$ must satisfy that $e_1, e_2$ cross the same edge of $\sep$ (corresponding to case (b) in \Cref{lem:2drawingcomb}). Then the rest of the proof follows.

	Next, we note that in $\cI^\lf$, the left boundary edges ($\leftbone, \leftbtwo$) and left extended boundary ($\lbnd$) are derived from the current instance $\cI$, and thus are placed at the appropriate locations. An analogous property holds about $\cI^\rt$. Next, we note that the orderings $\pi^\lf, \pi^\mm, \pi^\rt$ is obtained by the projecting $\pi'$ (which is an extension of $\pi$) to the respective vertex subsets. Since the respective drawings are compatible with these orders, it follows that the combined drawing is compatible with $\pi'$, and hence with $\pi$. Furthermore, since $\pi'$ also contains the orderings $\pi^{\lambda+1}$, the respective drawings also satisfy the ordering on the vertices of $\gapleft \cup \gapright \cup \bnd \cup \epts(\sep, 2)$. 
	
	Finally, since $\sigma^\lf, \sigma^\mm, \sigma^\rt$ are valid drawings for $\cI^\lf, \cI^\mm, \cI^\rt$, when we obtain the combined drawing $\sigma$ by gluing the three drawings, it is straightforward to verify that $\gapleft$ and $\gapright$ are placed consecutively at appropriate locations between the vertices of $\epts(\sep, 2)$, and all of the conditions of \Cref{def:3extdrawing} are satisfied. This shows that $\sigma'$ is a valid drawing for $\cI$.

\end{proof}

Otherwise, for each guess, if either it is terminated, or at least one sub-instance is a \textsc{No}-instance, then we return that the current extended instance is a \textsc{No}-instance.

\subsection{Proof of Correctness} \label{subsec:threelayerproof}

\begin{lemma} \label{lem:3layerlemma}
	 $\cI$ is an extended \yin of \probThreeCr for parameter $k$ iff the algorithm returns a drawing $\sigma'$.
\end{lemma}
\begin{proof}
	In the base case when $k \le c \sqrt{k^\star}$, or $\min\LR{|E_{12}|, |E_{23}|} \le c \sqrt{k^\star}$, the correctness of the algorithm follows via \Cref{lem:3layerbasecase}. Thus, we consider the recursive case. Suppose inductively that the statement is true for all values of parameter strictly smaller than $k > c \sqrt{k^{\star}}$ (with no assumption on $\min\LR{|E_{12}|, |E_{23}|}$), and we want to show that the statement holds for parameter $k$.
	
	\textbf{Forward direction.} Suppose $\cI$ is a $k$-minimal \yin, and let $\sigma$ be a drawing of $\cI$ with $k$ crossings. Consider the guess where the following objects are correctly guessed in the algorithm:
	\begin{enumerate}[leftmargin=*]
		\item Separator $\sep$ as guaranteed by \Cref{lem:3separatorlemma}, the corresponding edge-sets $\crosse(\sep)$, $\mde(\sep)$, vertex-sets $\gapleft, \gapright$, and the relative ordering of the vertices involved in these sets,
		\item For each $v \in \epts(\sep)$, by \Cref{lem:3type3components}, $|\updown(v)| \le 8\sqrt{k}$. Then, for each component $C \in \updown(v)$, if all the vertices of $C \cap V_i$ belong to the left of $\ept(\lsep, i)$ (resp. to the right of $\ept(\rsep, i)$) for each $i \in [3]$, then consider the guess where it is labeled $\lf$ (resp.~$\rt$).
		\item For each $i \in [3]$, and $j \in [\lambda]$, the vertices $v(j, l), v(j, r) \in V^j_i$ that are immediately to the left of $\ept(\lseptwo, 2)$ and to the right of $\ept(\rseptwo, 2)$, respectively. Note that for each connected component $C \in \predef(v)$ containing a vertex $u \in V^j_i$ with $\pi^j_i(u) \le \pi^j_i(v(j, l))$ (resp.~$\pi^j_i(u) \le \pi^j_i(v(j, r))$), all the vertices of $C$ must belong to $V(\sigma, \sep, \lf)$ (resp.~$V(\sigma, \sep, \rt)$). Hence, all the vertices belonging to $\predef(v)$ also get correctly labeled.
		\item Consider the blueprint $\pbp$ guaranteed by \Cref{lem:3permutation}. Since our algorithm tries all possible pendant-blueprints, consider the guess where it considers such $\pbp$. Then, it finds a $\pguess$ that is compatible with $\pbp$. Then, \Cref{lem:3permutation} implies that for any guess $\pguess$ that is $\pbp$-compliant, and in particular the $\pguess$ that is constructed in the algorithm, there exists a drawing $\sigma' \in \nicepfamily(\pbp, \pguess, \sep, \sigma)$ that is $(\sep, \sigma)$-separator-compatible. Henceforth, we rename $\sigma' \gets \sigma$ and proceed to the next step in the analysis. 
		\item Next, if we are in (\DB, \DT) or (\ST, \SB) case, consider where we correctly guess at most $\Oh(\sqrt{k})$ components of $\pecunp(v)$ for $v \in \epts(\sep, 1) \cup \epts(\sep, 3)$ in accordance with \Cref{lem:3pecunp-cases}, along with their labels.
		\\Otherwise, in (\SB, \DT) case, consider the $\cguess$ guaranteed by \Cref{lem:3componentordering}, and again, since our algorithm tries all possible pendant-blueprints, consider the guess where it considers such $\cbp$. Then, it finds a $\cguess$ that is compatible with $\cbp$. Then, \Cref{lem:3componentordering} implies that for any guess $\cguess$ that is $\cbp$-compliant, and in particular the $\cguess$ that is constructed in the algorithm, there exists a drawing $$\sigma' \in \nicecfamily(\cbp, \cguess, \sep, \sigma)$$ that is $(\sep, \sigma)$-separator-compatible. Again, we rename $\sigma \gets \sigma'$. The other case, (\DB, \ST) is handled analogously.
	\end{enumerate}
	Let $\sigma$ be the drawing obtained at the end after appropriate redefintions as described above. 
	In the following claim, we show that the rest of the labeling process correctly labels all the vertices.
	\begin{claim} \label{cl:3cllabeling}
		Let $\sigma$ be the drawing after appropriate redfinitions as described above, and consider the guess as described above. Then, after the vertex labeling process of step H corresponding to this particular guess, we have that $V^{\cal T} = V(\sigma,\sep, {\cal T})$, for each ${\cal T} \in \LR{\lf, \mm, \rt}$. Note that the definition of the sets $V(\sigma, \sep, {\cal T})$ can be found in \Cref{def:3division}.
	\end{claim}
	\begin{proof}[Proof \Cref{cl:3cllabeling}] 
		Let us consider a vertex $u \in V(\sigma, \sep, \lf)$. By \Cref{item:3extendedconn} of \Cref{def:extendedinstance3}, $u$ is reachable from some vertex $v \in \bnd \cup \epts(\boundary)$, say via a path $P(uv)$. We consider different cases based on $v$ and a path $P(u  v)$. Let $P(uv) = \langle u = w_0, w_1, w_2, \ldots, w_\ell = v \rangle$. Let $w_{i+1}$ be the first vertex (if any) along $P(u  v)$, such $w_0, w_1, \ldots, w_i \in V(\sigma, \sep, \lf)$, and $w_{i+1} \not\in V(\sigma, \sep, \lf)$. Otherwise, let $w_i = w_\ell = v$. We consider these two cases separately.

		(i) If $w_i = v$, then since $w_i \in V(\sigma, \sep, \lf)$, it cannot be the case that $w_i \in (\rbnd \cup \epts(\LR{\rightbone, \rightbtwo})) \setminus \epts(\sep)$. Thus, it implies that $w_i \in \lbnd \cup \epts(\LR{\leftbone, \leftbtwo}) \cup {\sf CrossLeft}$. 
		
		(ii) If $w_i \neq v$, then we know that $w_{i+1}$, the vertex immediately after $w_i$ in $P(uv)$ satisfies that $w_{i+1} \notin V(\sigma, \sep, \lf)$. Then, it follows that either the edge $(w_i, w_{i+1})$ crosses some edge of $\sep$--implying that $w_i \in {\sf CrossLeft}$, or that $w_{i+1} \in V_3$, and $w_i \in \gapleft$.
		
		In any case, for any vertex $u \in V(\sigma, \sep, \lf) \setminus {\sf Origin}^\lf$, and any path $P(uv)$ to a vertex $v \in \bnd \cup \epts(\boundary)$, there exists a vertex $w_i \in P(uv)$ with label $\lf$, such that $w_i \in {\sf Origin}^\lf$, and further the entire subpath $P(u  w_{i})$ does not contain any vertex in $V(\sigma, \sep, \lf) \cup V(\sigma, \sep, \rt)$. This also implies that, all the vertices of $V(\sigma, \sep, \lf)$ are reachable from some vertex $w_i \in {\sf Origin}^\lf$ in the graph $G - (O^{\mm} \cup O^{\rt})$. Thus, all the unlabeled vertices in $V(\sigma, \sep, \lf)$ get labeled $\lf$ in the label propagation phase. This shows that $V(\sigma, \sep, \lf) \subseteq V^\lf$. Similar proofs also show that $V(\sigma, \sep, \lf) \subseteq V^\mm$, and $V(\sigma, \sep, \rt) \subseteq V^\rt$. Now, consider a vertex of $u' \in (V(\sigma, \sep, \lf) \cup V(\sigma, \sep, \rt)) \setminus {\sf Origin}^\lf$. However, such a vertex $u'$ belongs to $O^\mm \cup O^\rt$, and hence is not part of the graph $G - (O^\mm \cup O^\rt)$. Thus, it cannot receive label $\lf$, which shows that $V^\lf \subseteq V(\sigma, \sep, \lf)$. By combining both containments, we obtain that $V^\lf = V(\sigma, \sep, \lf)$. Similar proofs also show that $V^\mm = V(\sigma, \sep, \lf)$, and $V^\rt = V(\sigma, \sep, \rt)$, and are therefore omitted.
	\end{proof}

	After accounting for the $k_{\sf current}$ crossings among the endpoints of $\crosse \cup \mde \cup \sep$ in $\sigma$, the rest of the $k'$ crossings are either to the left of $\lsep$ or to the right of $\rsep$, or within the edges of $E_{23}$ whose endpoints are labeled $\mm$. Let the number of these crossings be $k^\lf, k^\rt, k^\mm$, respectively, and consider the guess where these numbers are correctly guessed, and for $k^\lf, k^\rt$, also suppose that their partition between $k^\lf_{12}, k^\lf_{23}, k^\rt_{12}, k^\rt_{23}$ is correctly guessed. 
	
	Further, since we label all the vertices correctly, it follows that all the vertices of $V^\lf, V^\rt$ are reachable from ${\sf Origin}^\lf, {\sf Origin}^\rt$ respectively. Then, we observe that vertices of ${\sf CrossLeft}$ (resp.~${\sf CrossRight}$) are directly connected to one of the vertices of $\epts(\LR{\lsepone, \lseptwo})$ (resp.~$\epts(\LR{\rsepone, \rseptwo})$) in the graph $G^\lf$ (resp.~$G^\rt$), which shows that these guesses are not terminated by the connectivity check done in Step 3 after creating the instances $\cI^\lf, \cI^\rt$, respectively. This shows that these guesses are not terminated by the connectivity check done in Step 3 after creating the instances $\cI^\lf, \cI^\rt$, respectively. Then, \Cref{cl:3validlr} imply that $\cI^\lf, \cI^\rt, \cI^\mm$ are valid extended instances. Next, $\sigma$ restricted to the respective vertex sets $V^\lf, V^\rt$ gives a valid drawing of the respective instances. This implies that $\cI^\lf, \cI^\rt$ are {\sc Yes}-instances for the parameters $k^\lf, k^\rt$, respectively. Further, by using an argument similar to the one from \Cref{lem:2layerlemma}, we can also infer that $\cI^\lf, \cI^\rt$ are \emph{minimal} {\sc Yes}-instances for the respective parameters. Then, by induction (since $k^\lf, k^\rt \le 3k/4$) the respective recursive calls return valid drawings $\sigma^\lf, \sigma^\rt$ of the respective instances, implying that we proceed to Step 4 and create $\cI^\mm$.
	
	Again, note that all the vertices of $V(G^\mm)$ are labeled correctly, which implies that they are reachable from ${\sf Origin}^\mm$. Again, we observe that vertices of ${\sf CrossMid}$ are directly connected to one of the vertices of $\epts(\LR{\lseptwo, \rseptwo})$ in the graph $G^\mm$. Recall that in the elaborate extended instance $\cI^\mm$ constructed, $\lseptwo, \rseptwo$ become the boundary edges $\leftb^\mm, \rightb^\mm$, and $\gapleft, \gapright$ become $\lbnd, \rbnd$, respectively; and the preceding argument shows that each vertex in $V^\mm$ is reachable from some vertex of $\epts(\boundary^\mm) \cup \bnd^\mm$. This implies that the guess is not terminated, and then \Cref{cl:3validmiddle} implies that $\cI^\mm$ is a valid elaborate extended instance of \probTwoCr, and since the order on $\epts(\sep, 2) \cup \gapleft \cup \gapright$ is compatible with $\sigma$, it follows that $\sigma$ restricted to the vertices of $V^\mm$ gives a valid drawing of $\cI^\mm$. Therefore, $\cI^\mm$ is a \yin. Again, we can argue that $\cI^\mm$ is a \emph{minimal} \yin, and via \Cref{lem:2layerlemma}, we obtain a corresponding drawing $\sigma^\mm$. Thus, in Step 6, we combine the three drawing and return a drawing $\sigma'$.

	\textbf{Reverse direction.} The reverse direction is trivial, since the valid drawing $\sigma'$ returned by the algorithm is a witness that $\cI$ is a {\sc Yes}-instance.
\end{proof}

In the following lemma, we argue that the algorithm runs in subexponential time.

\begin{lemma} \label{lem:3runtime}
	For any input \ein $\cI$ of \probThreeCr with parameter $k$, the algorithm runs in time $T_3(k) \le n^{\Oh(k^{2/3})}$.
\end{lemma}
\begin{proof}
	The base case corresponds to when $k \le c \sqrt{k^\star}$, or when $\min\LR{|E_{12}|, |E_{23}|} \le c \sqrt{k^\star}$, for some constant $c$. Then, the running time of the algorithm is $n^{\Oh(\sqrt{k^\star})}$ via \Cref{lem:3layerbasecase}.
	
	Otherwise, consider the recursive case, and suppose the claim is true for all $k' \le k-1$, where $k > c \sqrt{k^{\star}}$ is the parameter of the given instance $\cI$. \Cref{lem:2guessbound} implies that the total number of guesses is bounded by $n^{\Oh(k^{2/3})}$, and corresponding to each guess that is not terminated, we make two recursive calls to the algorithm for \probThreeCr with parameters at most $3k/4$ each. If for one such guess, both recursive calls return a drawing, then we also solve an instance of \probTwoCr using \Cref{thm:2layertheorem} (without the kernelization step), where the parameter is at most $k$. Thus, it takes time $T_2(n, k) = n^{\Oh(\sqrt{k})}$ time.
	Thus, $T_3(n, k)$ satisfies the recurrence given in \Cref{eqn:3recurrence}, and then the claim follows by induction.
\end{proof}

By combining \Cref{obs:3normal2extended}, \Cref{lem:3layerlemma}, \Cref{lem:3runtime}, and the fact that due to linear-time kernelization algorithm of \Cref{thm:kernel}, we have $n = (k^\star)^{\Oh(1)}$, we obtain the following theorem (where we rename $k^\star$ as $k$).

\begin{theorem} \label{thm:3layertheorem}
	On an $n$ vertex graph, \probThreeCr can be solved in time $2^{\Oh(k^{2/3} \log k)} + n \cdot k^{\Oh(1)}$, where $k$ denotes the number of crossings.
\end{theorem}

\section{Kernelization for \probThreeCr}\label{sec:kernel}
The whole section is the proof of \Cref{thm:kernel},  which we restate here.

\thmkern*

In~\Cref{sec:constr}, we give the kernelization algorithm and prove its correctness. In~\Cref{sec:run-time}, we show that using some additional steps, we can implement the algorithm to run in time which is linear in $n$.  

\subsection{Construction of the kernel}\label{sec:constr}
Let $(G,V_1,V_2,V_3,k)$ where $(V_1,V_2,V_3)$ is a $3$-partition of $V(G)$ be an instance of \probThreeCr.  
We give a polynomial time algorithm that either correctly concludes that  $(G,V_1,V_2,V_3,k)$ is a no-instance, or constructs an equivalent instance $(G',V'_1,V'_2,V'_3,k)$, with $\cO(k^{8})$ vertices and edges.

By  \Cref{prop:linear-time}, one can check in time $\cO(n)$ whether $G$ admits a $3$-layer drawing respecting $(V_1,V_2,V_3)$ without crossings. If this is the case, we immediately return a trivial yes-instance and stop. We also stop and output a trivial no-instance if $k=0$ and $G$ does not admit such a drawing. From now on, we assume that $G$ has no drawing respecting $(V_1,V_2,V_3)$ without crossings and that $k\geq 1$.

We apply a sequence of reduction rules. Almost each of the rules operates as follows: Suppose $G$ contains a sufficiently large subgraph $H$ that can be (a) separated from $G$ by deleting a small (at most 6) number of vertices and (b) drawn on two or three layers without crossing. Then, we derive an equivalent but smaller instance of the problem by removing some ``irrelevant'' vertices of $H$ from $G$. Occasionally, we may need to introduce new edges, but the number of vertices in the new equivalent instance always decreases.
None of the rules changes the parameter $k$. 
The rules and the proof of their correctness (or, using the established expression in the area, their safety) depend on the number of vertices required to separate $H$ from $G$. The rules are exhaustively applied whenever possible. If, by a rule, we delete a vertex or a set of vertices, we assume that the sets $V_1,V_2,V_3$ are modified accordingly by the deletion of vertices.

\smallskip 
First, we address connected components that can be drawn without crossings. Trivially, such components are irrelevant because they can be drawn independently without creating any crossings.

\begin{reduction}\label{red:comp}
If $G$ has a connected component $H$ that admits a  $3$-layer drawing respecting $(V_1\cap V(H),V_2\cap V(H),V_3\cap V(H))$ without crossings, then set $G:=G-V(H)$.
\end{reduction}

%
%

\smallskip
For a vertex $v\in V(G)$ and $i\in\{1,2,3\}$, let $P_i(v)$ be the set of pendent vertices in $V_i$ adjacent to $v$.  
 The following reduction rule 
reduces the number of pendent neighbors of the vertices from $V_2$.  
\begin{reduction}\label{red:pend}
If there is $u\in V_2$ such that $|P_i(u)|>k+1$ for some $i\in\{1,3\}$ then select (arbitrarily) vertex $v\in P_i(u)$ and 
 set $G:=G-v$.
\end{reduction}

\begin{lemma}\label{cl:pend}
\Cref{red:pend} is safe.
\end{lemma}

\begin{proof}
Assume without loss of generality that the rule is applied for $u\in V_2$ and $v\in P_1(u)$,  and let $G'=G-v$ and $V_1'=V_1\setminus\{v\}$. Trivially, if $G$ admits  a $3$-layer drawing respecting $(V_1,V_2,V_3)$ with at most $k$ crossing then $G'$ has  a $3$-layer drawing respecting $(V_1',V_2,V_3)$ with at most $k$ crossing.  Suppose that $G'$ admits   a $3$-layer drawing 
$\sigma'=(\sigma_1',\sigma_2,\sigma_3)$ respecting $(V_1',V_2,V_3)$ with at most $k$ crossing. Because $|P_1(u)|>k+1$, $u$ is adjacent to at least $k+1$ pendent vertices in $V_1'$. Since the number of crossings is upper bounded by $k$, there is $w\in P_1(u)\setminus\{v\}$ such that $uw$ does not cross any other edge. We construct the order $\sigma_1$ of $V_1$ from $\sigma_1'$ by inserting $v$ immediately after $w$. Then $uv$ does not cross any edge of $G$. Thus, $\sigma=(\sigma_1,\sigma_2,\sigma_3)$ is $3$-layer drawing respecting $(V_1,V_2,V_3)$ with at most $k$ crossing. This concludes the proof. 
\end{proof}

\smallskip
We apply a similar rule to reduce the number of vertices in $V_2$. Let $u\in V_1$ and $v\in V_3$. We denote by $Q(u,v)$ the set of vertices of $V_2$ of degree two, that is, $u$ and $v$ are the only neighbors of any $w\in Q(u,v)$.


\begin{reduction}\label{red:degtwo}
If there are $u\in V_1$ and $v\in V_3$ such that $|Q(u,v)|>k+1$   then  select (arbitrarily) vertex
$w\in Q(u,v)$ and 
 set $G:=G-w$.
\end{reduction}

\begin{lemma}\label{cl:degtwo}
\Cref{red:pend} is safe.
\end{lemma}

\begin{proof}
Suppose that the rule is applied for $u\in V_1$, $v\in V_3$, and $w\in Q(u,v)$, and let $G'=G-w$ and $V_2'=V_2\setminus \{w\}$. As in the poof of \Cref{cl:pend}, it is sufficient to show that if  $G'$ has  a $3$-layer drawing respecting $(V_1,V_2',V_3)$ with at most $k$ crossing then  $G$ has  a $3$-layer drawing respecting $(V_1,V_2,V_3)$ with at most $k$ crossing. Assume that $\sigma'=(\sigma_1,\sigma_2',\sigma_3)$ is such a drawing of $G'$. Because $|Q(u,v)|>k+1$, there is $x\in Q(u,v)\setminus \{w\}$ such that edges $ux$ and $vx$ do not cross any edge in the drawing. We construct the order $\sigma_2$ of $V_2$ from $\sigma_2'$ by inserting $w$ immediately after $w$. Then, neither $uw$ nor $vw$  cross an edge of $G$. We obtain that $\sigma=(\sigma_1,\sigma_2,\sigma_3)$ is $3$-layer drawing respecting $(V_1,V_2,V_3)$ with at most $k$ crossing. This concludes the proof of the claim. 
\end{proof}

Let $u\in V_1\cup V_3$ be a cut vertex of $G$. 
We say that a connected component $H$ of $G-u$ is \emph{nice} 
 if 
 (i) $V(H)\subseteq V_2\cup V_3$ if $u\in V_1$ and $V(H)\subseteq V_1\cup V_2$ if $u\in V_3$, and (ii) $H$ admits a $2$-layer drawing without crossings respecting either $(V_2\cap V(H),V_3\cap V(H))$ if $u\in V_1$ or $(V_1\cap V(H),V_2\cap V(H))$ if $u\in V_3$. The next rule is used to reduce the number of nice components.

\begin{reduction}\label{red:nice}
If there is $u\in V_1\cup V_3$  such that the number of nice connected components of $G-u$ 
  is at least $2k+2$,  then select a nice component $H$ 
  and set  $G:=G-V(H)$; we select a nice component of size one if such a component $H$ exists and choose arbitrary $H$, otherwise.
\end{reduction}

\begin{lemma}\label{cl:nice}
\Cref{red:nice} is safe.
\end{lemma}

\begin{proof}
By symmetry, let us assume that the rule is applied for $u\in V_1$ and $V(H)\subseteq V_2\cup V_3$ for the deleted nice component $H$. Let $G'=G-V(H)$, $V_2'=V_2\setminus V(H)$, and $V_3'=V_3\setminus V(H)$. 
Again, it is sufficient to show that if  $G'$ has  a $3$-layer drawing respecting $(V_1,V_2',V_3')$ with at most $k$ crossing then  $G$ has  a $3$-layer drawing respecting $(V_1,V_2,V_3)$ with at most $k$ crossing. Assume that $\sigma'=(\sigma_1,\sigma_2',\sigma_3')$ is a $3$-layer drawing of $G'$ with at most $k$ crossings. Because the number of nice connected components of $G-u$ 
is at least $2k+2$, there is a nice component $C$ distinct from $H$ such that neither edges of $C$ nor edges $uv$ for $v\in V(C)$ cross other edges. Since we cannot apply \Cref{red:comp}, $C$ cannot be a connected component of $G$. Thus, there is a vertex of $C$ adjacent to $u$.
Denote by $x$ the last, with respect to $\sigma_2'$,  vertex of $V_2\cap V(C)$ adjacent to $u$. 

Suppose first that $H$ is trivial, that is, it has a single vertex. Then we construct a $3$-layer drawing $\sigma=(\sigma_1,\sigma_2,\sigma'_3)$ of $G$ respecting $(V_1,V_2,V_3)$ by modifying $\sigma_2'$ the same way as in the proof of \Cref{cl:pend}. Namely, we insert the unique vertex of $H$ immediately after $x$ in $\sigma_2'$. It does not create a new crossing because $ux$ does not cross any edge in $\sigma'$. 

From now on, we assume that $H$ has at least two vertices. By the choice of $H$, there is no nice component of size one and, therefore,  
$C$ also has at least two vertices and has vertices from both $V_2$ and $V_3$.
Let $y$ be the last vertex of $V_2\cap V(C)$ and $z$ be the last vertex of $V_3\cap V(C)$ with respect to $\sigma_2'$ and $\sigma_3'$, respectively.  It may happen that $x=y$.  Suppose that there is $v\in V_2'\setminus V(C)$ such that $\sigma_2'(x)<\sigma_2'(v)<\sigma_2'(y)$. Because $C$ is connected, we have that $v$ has no neighbors in $V_3$ and is adjacent only to some vertices 
$w\in V_1$ with $\sigma_1(u)\leq \sigma_2(w)$. This observation allows us to modify $\sigma_2'$ as follows. Let $v_1,\ldots,v_\ell$ be all the vertices of $V_2\setminus V(C)$ such that $\sigma_2'(x)<\sigma_2'(v_i)<\sigma_2'(y)$ for $i\in\{1,\ldots,\ell\}$ (the set of these vertices may be empty) and assume that 
$\sigma_2'(v_1)<\dots<\sigma_2'(v_\ell)$. We construct $\sigma_2''$ from $\sigma_2'$ by placing $v_1,\ldots,v_\ell$ following their order in $\sigma_2'$ after $y$ in the order of $V_2'$. Notice that this modification does not create any new crossing. Furthermore, for $\sigma''=(\sigma_1,\sigma_2'',\sigma_3')$, we have that either $u$, $y$, and $z$ are the last vertices in
$\sigma_1$, $\sigma_2''$, and $\sigma_3'$, respectively, or 
  $u$ separates the vertices that are after $u$, $y$, and $z$ with respect to $\sigma''$ from the vertices that are before them. More formally, 
it holds that (a) there is no edge $ab$ with $a\in V_2'$ and $b\in V_3'$ such that either $\sigma_2''(a)\leq \sigma_2''(y)$ and $\sigma_3'(b)>\sigma_3'(z)$ or $\sigma_2''(a)> \sigma_2''(y)$ and $\sigma_3'(b)\leq \sigma_3'(z)$, and (b) there is no edge $ab$ with $a\in V_1$ and $b\in V_2'$ such that either 
$\sigma_1(a)< \sigma_a(u)$ and $\sigma_2''(b)>\sigma_2''(y)$ or $\sigma_1(a)> \sigma_1(u)$ and $\sigma_2''(b)\leq \sigma_2''(y)$.  This allows us to construct a $3$-layer drawing $\sigma=(\sigma_1,\sigma_2,\sigma_3)$ of $G$ respecting $(V_1,V_2,V_3)$ without increasing the number of crossings.

Because   $H$ admits a $2$-layer drawing without crossings respecting $(V_2\cap V(H),V_3\cap V(H))$, there are the corresponding orders $\hat{\sigma}_2$ and $\hat{\sigma}_3$ of $V_2\cap V(H)$ and $V_3\cap V(H)$, respectively. We construct $\sigma=(\sigma_1,\sigma_2,\sigma_3)$ by modifying $\sigma_2''$ and $\sigma_3'$. To construct $\sigma_2$, we insert the vertices of $V_2\cap V(H)$ in $\sigma_2''$ immediately after $y$ following $\hat{\sigma}_2$. Similarly, $\sigma_3$ is constructed from $\sigma_3'$ by inserting the vertices of $V_3\cap V(H)$ immediately after $z$ following $\hat{\sigma}_3$. This concludes the proof.
\end{proof}

By the following five rules, we detect pieces of the graph with only particular drawings and reduce their size. We start with pieces that can be separated by four vertices from the remaining graph.

 \begin{figure}[ht]
\centering
\scalebox{0.7}{
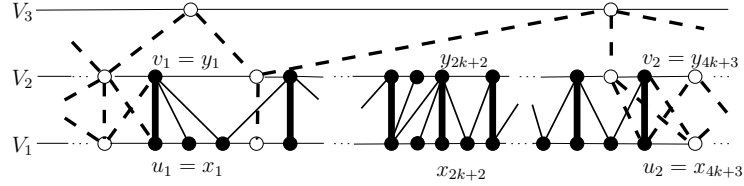}
\caption{Applying \Cref{red:matching-A}; the vertices of $A$ are black bullets and the vertices of $B\setminus A$ are white; the edges of $H$ are  solid lines and the edges outside $H$ are drawn by dashed lines; the edges of $M$ are highlighted by thick lines.}
\label{fig:matching-A}
\end{figure} 

Informally, \Cref{red:matching-A} does the following. Assume that $H$ is a subgraph of $G$ that has a  $2$-layer drawing in layers 
 $V_1$ and  $V_2$ without crossings. Moreover, assume that in such a drawing, the endpoints of the leftmost edge $u_1v_1$ and the rightmost edge $u_2v_2$ of $H$ separate it from the remaining part of $G$. Then, if $H$ contains a matching $M$ of size $4k+3$, at least one of the endpoints of the central edge in  $M$  could be safely removed. (We also must add one or two edges to maintain connectivity.) 

\begin{reduction}[4-Separation A]\label{red:matching-A}
Suppose that there is a separation $(A,B)$ of $G$ with $A\cap B=\{u_1,u_2,v_1,v_2\}$ such that (see \Cref{fig:matching-A})
\begin{itemize}
\item[(i)] $H=G[A]$ is a connected graph with  $A\subseteq V_1\cup V_2$, $u_1,u_2\in V_1$, and $v_1,v_2\in V_2$,
\item[(ii)]  $u_1v_1,u_2v_2\in E(G)$, 
\item[(iii)] $H$ has a $2$-layer drawing $\sigma_H=(\sigma_1^H,\sigma_2^H)$ without crossings respecting $(V_1\cap V(H),V_2\cap V(H))$, where $\sigma_i^H$ is an order of $V_i\cap V(H)$ for $i\in\{1,2\}$, such that 
(a) for every $w\in V_1$, $\sigma_1^H(u_1)\leq \sigma_1^H(w)\leq \sigma_1^H(u_2)$ and (b) for every $w\in V_2$, $\sigma_2^H(v_1)\leq \sigma_2^H(w)\leq \sigma_2^H(v_2)$, 
\item[(iv)] $H$ has a matching $M$ of size $4k+3$ with $u_1v_1,u_2v_2\in M$.
\end{itemize}
Then let $M=\{x_1y_1,\ldots,x_{4k+3}y_{4k+3}\}$, where $x_1=u_1$, $y_1=v_1$, $x_{4k+3}=u_2$, $y_{4k+3}=v_2$, $\sigma_1^H(x_1)<\dots<\sigma_1^H(x_{4k+3})$, and 
$\sigma_2^H(y_1)<\dots<\sigma_2^H(y_{4k+3})$, and do the following:
\begin{itemize}
\item[(a)] find the unique $z_1\in \{x_{2k+1},y_{2k+1}\}$ that has a neighbor either after $x_{2k+1}$ in $\sigma_1^H$ or after $y_{2k+1}$ in $\sigma_2^H$, respectively, and find 
  the unique $z_2\in \{x_{2k+3},y_{2k+3}\}$ that has a neighbor either before $x_{2k+1}$ in $\sigma_1^H$ or before $y_{2k+1}$ in $\sigma_2^H$, respectively, 
\item[(b)] delete every vertex $w\in V_1\cap V(H)$ with $\sigma_1^H(x_{2k+1})<\sigma_1^H(w)<\sigma_1^H(x_{2k+3})$,
\item[(c)] delete every vertex $w\in V_2\cap V(H)$ with $\sigma_2^H(y_{2k+1})<\sigma_2^H(w)<\sigma_2^H(y_{2k+3})$,
\item[(d)] if either $z_1\in V_1$ and $z_2\in V_2$ or  $z_1\in V_2$ and $z_2\in V_1$ then make $z_1$ and $z_2$ adjacent to maintain connectivity,
\item[(e)] if $z_1\in V_1$ and $z_2\in V_1$ then restore $y_{2i+2}$ and make it adjacent to $z_1$ and $z_2$,
\item[(f)] if $z_1\in V_2$ and $z_2\in V_2$ then restore $x_{2i+2}$ and make it adjacent to $z_1$ and $z_2$.
\end{itemize}
\end{reduction}

We note that we use a complicated procedure of adding edges to ensure connectivity for simplifying the running time analysis in the next subsection. In particular, we guarantee that the rule does not create new pendent vertices and ensure that the graph $H'$ obtained from $H$ does not contain a matching of the same size as $M$ if $M$ is a maximum size matching satisfying the conditions of the rule.

\begin{lemma}\label{cl:matching}
\Cref{red:matching-A} is safe.
\end{lemma}

\begin{proof}
First, we observe that because $H$ has a $2$-layer drawing $\sigma_H=(\sigma_1^H,\sigma_2^H)$ without crossings respecting $(V_1\cap V(H),V_2\cap V(H))$ such that (iii)(a) and (iii)(b) are fulfilled, we can order the endpoints of the edges of the matching $M$ to satisfy the conditions that  $\sigma_1^H(x_1)<\dots<\sigma_1^H(x_{4k+3})$, and 
$\sigma_2^H(y_1)<\dots<\sigma_2^H(y_{4k+3})$. Thus, the rule can be executed if there is a separation $(A,B)$ satisfying (i)--(iv). Note also that because at least one of the vertices  $x_{2k+2}$ and $y_{2k+2}$ is deleted by the rule, it reduces the size of the graph.

Suppose that the rule is applied for a separation $(A,B)$ and denote by $G'$ the graph obtained by applying the rule. Let also $D\subseteq V_1\cup V_2$ be the set of deleted vertices, $V_1'=V_1\setminus D$, $V_2'=V_2\setminus D$.

Assume that $\sigma=(\sigma_1,\sigma_2,\sigma_3)$ is a $3$-layer drawing of $G$ respecting $(V_1,V_2,V_3)$ with at most $k$ crossings. Notice that because the number of crossings is at most $k$, there are $i\in\{1,\ldots,2k+1\}$ and $j\in\{2k+3,\ldots,4k+3\}$ such that the edges $x_iy_i$ and $x_jy_j$ do not cross any edge. Without loss of generality, we assume that $\sigma_1(x_i)<\sigma_1(x_j)$ and $\sigma_2(y_i)<\sigma_2(y_j)$. Otherwise, we just reverse the orders in $\sigma$.
We set  
\[U_1=\{w\in V_1\cap V(H)\colon \sigma_1(x_i)<\sigma_1(w)<\sigma_1(x_j)\}\] and 
\[U_2=\{w\in V_2\cap V(H)\colon \sigma_2(y_i)<\sigma_2(w)<\sigma_2(y_j)\},\] 
and define
\[W_1=\{w\in V_1\setminus V(H)\colon \sigma_1(x_i)<\sigma_1(w)<\sigma_1(x_j)\}\] and \[W_2=\{w\in V_2\setminus V(H)\colon \sigma_2(y_i)<\sigma_2(w)<\sigma_2(y_j)\}.\]
Observe that because $H$ is connected and $x_iy_i$ and $x_jy_j$ do not cross any edge, we have that 
\[U_1=\{w\in V_1\cap V(H)\colon \sigma_1^H(x_i)<\sigma_1^H(w)<\sigma_1^H(x_j)\}\] and 
\[U_2=\{w\in V_2\cap V(H)\colon \sigma_2^H(y_i)<\sigma_2^H(w)<\sigma_2^H(y_j)\}.\]
Furthermore, if $xy\in E(G)\setminus E(H)$ with $x\in V_1$ and $y\in V_2$ such that $x\in W_1$ or $y\in W_2$, then $x\in W_1$, $y\in W_2$, and $xy$ crosses at least one edge of $H$ in the drawing of $G$.

We construct a $3$-layer drawing $\sigma'=(\sigma'_1,\sigma'_2,\sigma_3)$ of $G'$  respecting $(V_1',V_2',V_3)$ with no more crossing than $\sigma$ as follows. 
We obtain $\sigma_1'$ from $\sigma_1$ by reordering the vertices of $U_1\cap V(G')$ and $W_1$.
First, we place the vertices of $U_1\cap V(G')$  between $x_i$ and $x_j$ following their order in $\sigma^H_1$ and then place the vertices of $W_1$ between $x_{2k+1}$ and $x_{2k+3}$ following their order in $\sigma_1$. 
 Similarly,  
 $\sigma_2'$ is obtained from $\sigma_2$ by reordering the vertices of $U_2\cap V(G')$ and $W_2$: 
 we place the vertices of $U_2\cap V(G')$  between $y_i$ and $y_j$ following their order in $\sigma^H_2$ and then place the vertices of $W_2$ between 
  $y_{2k+1}$ and $y_{2k+3}$ following their order in $\sigma_2$.  Observe that an  edge $xy\in E(H)\cap E(G')$ with $\sigma'_1(x_i)\leq \sigma'_1(x)\leq \sigma'_1(x_j)$ and 
$\sigma'_2(y_i)\leq \sigma'_2(y)\leq \sigma'_1(y_j)$ does not cross any edge of $G'$. Further,  
 the edges constructed by the rules, that is, either $x_{2k+1}y_{2k+3}$, or $y_{2k+1}x_{2k+3}$, or $x_{2k+1}y_{2k+2}$ with $y_{2k+2}x_{2k+3}$, or $y_{2k+1}x_{2k+2}$ with $x_{2k+2}y_{2k+3}$, 
cross only the edges $xy\in E(G')$ with $x\in W_1$ and $y\in W_2$ and $xy$ crosses exactly one edge constructed by the rule. This implies that in 
 $\sigma'$ the number of crossings does not exceed the number of crossings in $\sigma$.

For the opposite direction, assume that there is a $3$-layer drawing $\sigma'=(\sigma'_1,\sigma'_2,\sigma_3)$ of $G'$  respecting $(V_1',V_2',V_3)$ with at most $k$ crossings. We are using almost the same arguments as for the forward implication to show that there is a $3$-layer drawing $\sigma=(\sigma_1,\sigma_2,\sigma_3)$ of $G$  respecting $(V_1,V_2,V_3)$ such that the number of crossings in $\sigma$ is at most the number of crossings in $\sigma'$.

Because the number of crossings is at most $k$, there are $i\in\{1,\ldots,2k+1\}$ and $j\in\{2k+3,\ldots,4k+3\}$ such that the edges $x_iy_i$ and $x_jy_j$ do not cross any edge with respect to $\sigma'$. We can assume that  $\sigma'_1(x_i)<\sigma'_1(x_j)$ and $\sigma'_2(y_i)<\sigma'_2(y_j)$. 
We define vertex sets
\[U_1=\{w\in V_1\cap V(H)\colon \sigma_1^H(x_i)<\sigma_1^H(w)<\sigma_1^H(x_j)\},\] \[U_2= \{w\in V_2\cap V(H)\colon \sigma_2^H(y_i)<\sigma_2^H(w)<\sigma_2^H(y_j)\}.\] 
We also  introduce 
\[U_1'=\{w\in V_1'\cap V(H)\colon \sigma'_1(x_i)<\sigma_1(w)<\sigma'_1(x_j)\}\] and \[U_2'=\{w\in V_2'\cap V(H)\colon \sigma'_2(y_i)<\sigma'_2(w)<\sigma'_2(y_j)\}.\] We further define
\[W_1'=\{w\in V_1\setminus V(H)\colon \sigma'_1(x_i)<\sigma'_1(w)<\sigma'_1(x_j)\},\] and 
\[W_2'=\{w\in V_2\setminus V(H)\colon \sigma'_2(y_i)<\sigma'_2(w)<\sigma'_2(y_j)\}.\]
We remark that graph $H'$ obtained from $H$ by deleting the vertices of $D$ and 
adding one or two edges
is connected. 
Then, since  $x_iy_i$ and $x_jy_j$ do not cross any edge, we have that 
$U_1'\subseteq U_1$ and $U_2'\subseteq U_2$.
Moreover,  if $xy\in E(G')\setminus E(H')$ with $x\in V_1$ and $y\in V_2$ such that $x\in W_1$ or $y\in W_2$, then $x\in W_1$ and $y\in W_2$ and $xy$ crosses at least one edge of $H'$ in the drawing of $G'$.

We construct $\sigma_1$ from $\sigma_1'$.
First, we replace the vertices of $U_1'$ with the vertices of $U_1$ and put them between $x_i$ and $x_j$ following their order in $\sigma^H_1$. Then we place the vertices of $W_1$ immediately after $x_{2k+2}$  following their order in $\sigma_1'$. 
 The order 
 $\sigma_2$ is constructed from $\sigma'_2$ in a similar way. We replace the vertices of $U_2'$ with the vertices of $U_2$ and put them between $y_i$ and $y_j$ following their order in $\sigma^H_2$. Then we place the vertices of $W_2$ immediately before $y_{2k+2}$ following their order in $\sigma'_2$. 
 We have that any  edge $xy\in E(H)$ with $\sigma_1(x_i)\leq \sigma_1(x)\leq \sigma_1(x_j)$ and 
$\sigma_2(y_i)\leq \sigma_2(y)\leq \sigma_1(y_j)$ that is distinct from $x_{2k+2}y_{2k+2}$ 
does not cross any edge of $G$ 
and $x_{2k+1}y_{2k+3}$ 
crosses only the edges $xy\in E(G)$ with $x\in W_1$ and $y\in W_2$. Therefore,  the number of crossings in the constructed $3$-layer drawing $\sigma$ does not exceed the number of crossings in $\sigma'$. This concludes the proof of the claim.
\end{proof}

\Cref{red:matching-A} is applied for the drawing of $H$ in the layers $V_1$ and $ V_2$. 
Of course, by symmetry, the same arguments are applicable when the vertices of $H$ are in layers  $V_2$ and $ V_3$.  
Since the statement of \Cref{red:matching-A} is already quite long and challenging to read. Thus, instead of putting it in the general form, capturing two symmetric cases (like we did in \Cref{red:pend}), it is slightly preferable to have it as two reduction rules.  
The next  \Cref{red:matching-B}, up to the symmetry, is exactly the same as \Cref{red:matching-A}. 
We state the rule (without proof of its safety) for completeness.

\begin{reduction}[4-Separation B]\label{red:matching-B}
Suppose that there is a separation $(A,B)$ of $G$ with a separator $A\cap B=\{u_1,u_2,v_1,v_2\}$ such that
\begin{itemize}
\item[(i)] $H=G[A]$ is a connected graph with  $A\subseteq V_2\cup V_3$, $u_1,u_2\in V_2$, and $v_1,v_2\in V_3$,
\item[(ii)]  $u_1v_1,u_2v_2\in E(G)$, 
\item[(iii)] $H$ has a $2$-layer drawing $\sigma_H=(\sigma_2^H,\sigma_3^H)$ without crossings respecting $(V_2\cap V(H),V_3\cap V(H))$, where $\sigma_i^H$ is an order of $V_i\cap V(H)$ for $i\in\{2,3\}$, such that 
(a) for every $w\in V_2$, $\sigma_2^H(u_1)\leq \sigma_2^H(w)\leq \sigma_2^H(u_2)$ and (b) for every $w\in V_3$, $\sigma_3^H(v_1)\leq \sigma_3^H(w)\leq \sigma_3^H(v_2)$, 
\item[(iv)] $H$ has a matching $M$ of size $4k+3$ with $u_1v_1,u_2v_2\in M$.
\end{itemize}
Then let $M=\{x_1y_1,\ldots,x_{4k+3}y_{4k+3}\}$, where $x_1=u_1$, $y_1=v_1$, $x_{4k+3}=u_2$, $y_{4k+3}=v_2$, $\sigma_2^H(x_1)<\dots<\sigma_2^H(x_{4k+3})$, and 
$\sigma_3^H(y_1)<\dots<\sigma_3^H(y_{4k+3})$, and do the following:
\begin{itemize}
\item[(a)] find the unique $z_1\in \{x_{2k+1},y_{2k+1}\}$ that has a neighbor either after $x_{2k+1}$ in $\sigma_2^H$ or after $y_{2k+1}$ in $\sigma_3^H$, respectively, and find 
  the unique $z_2\in \{x_{2k+3},y_{2k+3}\}$ that has a neighbor either before $x_{2k+1}$ in $\sigma_2^H$ or before $y_{2k+1}$ in $\sigma_3^H$, respectively, 
\item[(b)] delete every vertex $w\in V_1\cap V(H)$ with $\sigma_2^H(x_{2k+1})<\sigma_2^H(w)<\sigma_2^H(x_{2k+3})$,
\item[(c)] delete every vertex $w\in V_2\cap V(H)$ with $\sigma_3^H(y_{2k+1})<\sigma_3^H(w)<\sigma_3^H(y_{2k+3})$,
\item[(d)] if either $z_1\in V_2$ and $z_2\in V_3$ or  $z_1\in V_3$ and $z_2\in V_2$ then make $z_1$ and $z_2$ adjacent to maintain connectivity,
\item[(e)] if $z_1\in V_2$ and $z_2\in V_2$ then restore $y_{2i+2}$ and make it adjacent to $z_1$ and $z_2$,
\item[(f)] if $z_1\in V_3$ and $z_2\in V_3$ then restore $x_{2i+2}$ and make it adjacent to $z_1$ and $z_2$.
\end{itemize}
\end{reduction}

In 
4-Separation Reduction Rules, 
we considered subgraphs with the vertices in two sets of the $3$-partition. Now, we turn to the subgraphs with vertices in three sets. First is the case of a subgraph with one vertex in $V_1$ or $V_3$.

 \begin{figure}[ht]
\centering
\scalebox{0.7}{
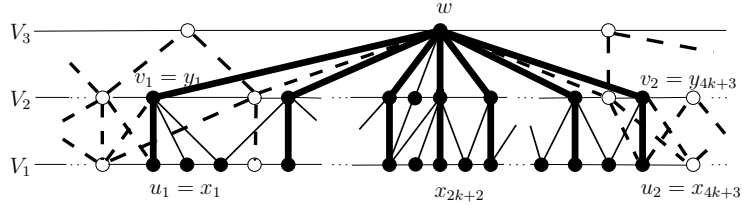}
\caption{Applying \Cref{red:paths-A}; the vertices of $A$ are shown by black bullets and the vertices of $B\setminus A$ are white; the edges of $H$ are shown by solid lines and the edges outside $H$ are drawn by dashed lines; the edges of $M$ and the edges $wy_i$ for $i\in\{1,\ldots,4k+3\}$ are highlighted by thick lines.}
\label{fig:matching-A}
\end{figure}

\begin{reduction}[5-Separation A]\label{red:paths-A}
Suppose that there is a separation $(A,B)$ of $G$ with a separator $A\cap B=\{u_1,u_2,v_1,v_2,w\}$ such that (see~\Cref{fig:matching-A})
\begin{itemize}
\item[(i)] $G[A]$ is a graph with  $A\cap V_3=\{w\}$, $u_1,u_2\in V_1$, and $v_1,v_2\in V_2$,
\item[(ii)] $H=G[A]-w$ is connected, 
\item[(iii)]  $u_1v_1,u_2v_2,v_1w,v_2w\in E(G)$, 
\item[(iv)] $H$ has a $2$-layer drawing $\sigma_H=(\sigma_1^H,\sigma_2^H)$ without crossings respecting $(V_1\cap V(H),V_2\cap V(H))$, where $\sigma_i^H$ is an order of $V_i\cap V(H)$ for $i\in\{1,2\}$, such that 
(a) for every $x\in V_1\cap V(H)$, $\sigma_1^H(u_1)\leq \sigma_1^H(x)\leq \sigma_1^H(u_2)$ and (b) for every $x\in V_2\cap V(H)$, $\sigma_2^H(v_1)\leq \sigma_2^H(x)\leq \sigma_2^H(v_2)$, 
\item[(v)] $H$ has a matching $M=\{x_1y_1,\ldots,x_{4k+3}y_{4k+3}\}$ of size $4k+3$ with $u_1v_1,u_2v_2\in M$ and $y_1,\ldots,y_{4k+3}\in N_G(w)$.
\end{itemize}
Then assume that $x_1=u_1$, $y_1=v_1$, $x_{4k+3}=u_2$, $y_{4k+3}=v_2$, $\sigma_1^H(x_1)<\dots<\sigma_1^H(x_{4k+3})$, and 
$\sigma_2^H(y_1)<\dots<\sigma_2^H(y_{4k+3})$, and do the following:
\begin{itemize}
\item[(a)] delete every vertex $z\in V_1\cap V(H)$ with $\sigma_1^H(x_{2k+1})<\sigma_1^H(z)<\sigma_1^H(x_{2k+3})$,
\item[(b)] delete every vertex $z\in V_2\cap V(H)$ with $\sigma_2^H(y_{2k+1})<\sigma_2^H(z)<\sigma_2^H(y_{2k+3})$ except $y_{2k+2}$ and delete the edges incident to $y_{2k+2}$,
\item[(c)]  make $x_{2k+1}$ and $x_{2k+3}$ adjacent to $y_{2k+2}$ to maintain connectivity.  
\end{itemize}
\end{reduction}

We note that the deleted vertices and edges ensure that we do not create new paths of length two with their end-vertices in $V_1$ and $V_3$.

\begin{lemma}\label{cl:paths-A}
\Cref{red:paths-A} is safe.
\end{lemma}

\begin{proof}
The proof follows the same lines as the proof of \Cref{cl:matching} and is, in fact, simpler. 
Because $H$ has a $2$-layer drawing $\sigma_H=(\sigma_1^H,\sigma_2^H)$ without crossings respecting $(V_1\cap V(H),V_2\cap V(H))$ such that (iv)(a) and (iv)(b) are fulfilled, the endpoints  of the edges of the matching $M$ can be ordered to satisfy the conditions that  $\sigma_1^H(x_1)<\dots<\sigma_1^H(x_{4k+3})$, and 
$\sigma_2^H(y_1)<\dots<\sigma_2^H(y_{4k+3})$. Therefore, the rule is feasible if there is a separation $(A,B)$ satisfying (i)--(v). Also, because $x_{2k+2}$ 
is deleted by the rule, it reduces the size of the graph.

Suppose that the rule is applied for a separation $(A,B)$ and denote by $G'$ the graph obtained by the application of the rule. Let also $D\subseteq V_1\cup V_2$ be the set of deleted vertices, $V_1'=V_1\setminus D$, $V_2'=V_2\setminus D$. 

Let  $\sigma=(\sigma_1,\sigma_2,\sigma_3)$ be a $3$-layer drawing of $G$ respecting $(V_1,V_2,V_3)$ with at most $k$ crossings. Because the number of crossings is at most $k$, there are $i\in\{1,\ldots,2k+1\}$ and $j\in\{2k+3,\ldots,4k+3\}$ such that the edges $x_iy_i$, $y_iw$, $x_jy_j$, and $y_jw$ do not cross any edge. Without loss of generality, we assume that $\sigma_1(x_i)<\sigma_1(x_j)$ and $\sigma_2(y_i)<\sigma_2(y_j)$. 
We define \[U_1=\{z\in V_1\cap V(H)\colon \sigma_1(x_i)<\sigma_1(z)<\sigma_1(x_j)\}\] and \[U_2=\{z\in V_2\cap V(H)\colon \sigma_2(y_i)<\sigma_2(z)<\sigma_2(y_j)\}.\]
 Because $H$ is connected and $x_iy_i$, $y_iw$, $x_jy_j$, and $y_jw$ do not cross any edge, we have that 
\[U_1=\{z\in V_1\cap V(H)\colon \sigma_1^H(x_i)<\sigma_1^H(z)<\sigma_1^H(x_j)\}\] and \[U_2=\{z\in V_2\cap V(H)\colon \sigma_2^H(y_i)<\sigma_2^H(z)<\sigma_2^H(y_j)\}.\]
Moreover, there is no vertex $z\in (V_1\cup V_2)\setminus (U_1\cup U_2)$  such that 
$\sigma_1(x_i)< \sigma_1(z)< \sigma_1(x_j)$ if $z\in V_1$ and $\sigma_2(y_i)< \sigma_2(z)< \sigma_2(y_j)$ if $z\in V_2$.

We  construct the $3$-layer drawing $\sigma'=(\sigma'_1,\sigma'_2,\sigma_3)$ of $G'$  respecting $(V_1',V_2',V_3)$ that has no more crossing than $\sigma$ by modifying $\sigma_1$ and $\sigma_2$. 
We place the vertices of $U_1\cap V(G')$  between $x_i$ and $x_j$ following their order in $\sigma^H_1$ and 
we place the vertices of $U_2\cap V(G')$  between $y_i$ and $y_j$ following their order in $\sigma^H_2$.  
Since there are no crossing edges with their endpoints in $U_1$ or $U_2$, the number of crossings in  $\sigma'$  does not exceed the number of crossings in $\sigma$.

For the opposite direction, the proof is almost the same. 
Assume that there is a $3$-layer drawing $\sigma'=(\sigma'_1,\sigma'_2,\sigma_3)$ of $G'$  respecting $(V_1',V_2',V_3)$ with at most $k$ crossings. We construct  a $3$-layer drawing $\sigma=(\sigma_1,\sigma_2,\sigma_3)$ of $G$  respecting $(V_1,V_2,V_3)$ such that the number of crossings in $\sigma$ is at most the number of crossings in $\sigma'$.
For this, we note that because the number of crossings is at most $k$, there are $i\in\{1,\ldots,2k+1\}$ and $j\in\{2k+3,\ldots,4k+3\}$ such that the edges 
$x_iy_i$, $y_iw$, $x_jy_j$, and $y_jw$  do not cross any edge with respect to $\sigma'$. We can assume that  $\sigma'_1(x_i)<\sigma'_1(x_j)$ and $\sigma'_2(y_i)<\sigma'_2(y_j)$. 
We set  
\[U_1=\{z\in V_1\cap V(H)\colon \sigma_1^H(x_i)<\sigma_1^H(z)<\sigma_1^H(x_j)\}\] and \[U_2= \{z\in V_2\cap V(H)\colon \sigma_2^H(y_i)<\sigma_2^H(z)<\sigma_2^H(y_j)\},\] and set 
\[U_1'=\{z\in V_1'\cap V(H)\colon \sigma'_1(x_i)<\sigma_1(z)<\sigma'_1(x_j)\}\] and \[U_2'=\{z\in V_2'\cap V(H)\colon \sigma'_2(y_i)<\sigma'_2(z)<\sigma'_2(y_j)\}.\] 
Because  the graph $H'$ obtained from $H$ by deleting the vertices of $D$ and making $x_{2k+1}$ and $x_{2k+3}$ adjacent to $y_{2k+2}$ is connected and 
$x_iy_i$, $y_iw$, $x_jy_j$, and $y_jw$  do not cross any edge,  
$U_1'\subseteq U_1$ and $U_2'\subseteq U_2$.
Also, there is no vertex $z\in (V_1\cup V_2)\setminus (U_1'\cup U_2')$  such that 
$\sigma'_1(x_i)\leq \sigma'_1(z)\leq \sigma'_1(x_j)$ if $z\in V_1$ and $\sigma'_2(y_i)\leq \sigma'_2(z)\leq \sigma'_2(y_j)$ if $z\in V_2$.

We construct $\sigma_1$ from $\sigma_1'$ by replacing the vertices of $U_1'$ with the vertices of $U_1$ and putting them between $x_i$ and $x_j$ following their order in $\sigma^H_1$. Similarly,  $\sigma_2$ is constructed from $\sigma'_2$ by replacing the vertices of $U_2'$ with the vertices of $U_2$ and putting them between $y_i$ and $y_j$ following their order in $\sigma^H_2$.  This concludes the proof.
\end{proof}

Similarly to 
4-Separation Rules, 
 there is a symmetric version of \Cref{red:paths-A},  which we state for completeness.

\begin{reduction}[5-Separation B]\label{red:paths-B}
Suppose that there is a separation $(A,B)$ of $G$ with a separator $A\cap B=\{u_1,u_2,v_1,v_2,w\}$ such that
\begin{itemize}
\item[(i)] $G[A]$ is a graph with  $A\cap V_1=\{w\}$, $u_1,u_2\in V_2$, and $v_1,v_2\in V_3$,
\item[(ii)] $H=G[A]-w$ is connected, 
\item[(iii)]  $u_1v_1,u_2v_2,v_1w,v_2w\in E(G)$, 
\item[(iv)] $H$ has a $2$-layer drawing $\sigma_H=(\sigma_2^H,\sigma_3^H)$ without crossings respecting $(V_2\cap V(H),V_3\cap V(H))$, where $\sigma_i^H$ is an order of $V_i\cap V(H)$ for $i\in\{2,3\}$, such that 
(a) for every $x\in V_2\cap V(H)$, $\sigma_2^H(u_1)\leq \sigma_2^H(x)\leq \sigma_2^H(u_2)$ and (b) for every $x\in V_3\cap V(H)$, $\sigma_3^H(v_1)\leq \sigma_3^H(x)\leq \sigma_3^H(v_2)$, 
\item[(v)] $H$ has a matching $M=\{x_1y_1,\ldots,x_{4k+3}y_{4k+3}\}$ of size $4k+3$ with $u_1v_1,u_2v_2\in M$ and $x_1,\ldots,x_{4k+3}\in N_G(w)$.
\end{itemize}
Then assume that $x_1=u_1$, $y_1=v_1$, $x_{4k+1}=u_2$, $y_{4k+1}=v_2$, $\sigma_2^H(x_1)<\dots<\sigma_2^H(x_{4k+3})$, and 
$\sigma_3^H(y_1)<\dots<\sigma_3^H(y_{4k+3})$, and do the following:
\begin{itemize}
\item[(a)] delete every vertex $z\in V_2\cap V(H)$ with $\sigma_2^H(x_{2k+1})<\sigma_2^H(z)<\sigma_2^H(x_{2k+3})$ except $x_{2k+2}$ and delete the edges incident to $x_{2k+2}$,
\item[(b)] delete every vertex $z\in V_3\cap V(H)$ with $\sigma_3^H(y_{2k+1})<\sigma_3^H(z)<\sigma_3^H(y_{2k+3})$,
\item[(c)]  make $x_{2k+2}$ adjacent to $y_{2k+1}$ and $y_{2k+3}$ to maintain connectivity.  
\end{itemize}
\end{reduction}

In the following rule, we consider subgraphs separated by six vertices from two paths of length two. The general idea for this rule is very similar to the 4-Separation Rule. The main difference is that instead of a large matching in graph $H$, we have many  disjoint paths passing from level $V_1$ to $V_3$.

 \begin{figure}[ht]
\centering
\scalebox{0.7}{
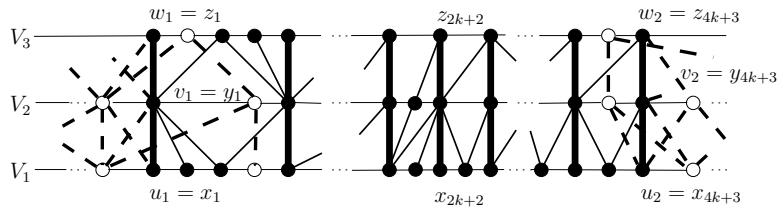}
\caption{Applying \Cref{red:paths-C}; the vertices of $A$ are   black bullets and the vertices of $B\setminus A$ are white; the edges of $H$ are   solid lines and the edges outside $H$ are drawn by dashed lines; the edges of the paths in $P$ are highlighted by thick lines.}
\label{fig:path-C}
\end{figure}

\begin{reduction}[6-Separation]\label{red:paths-C}
Suppose that there is a separation $(A,B)$ of $G$ with a separator $A\cap B=\{u_1,v_1,w_1,u_2,v_2,w_2\}$ such that (see~\Cref{fig:path-C})
\begin{itemize}
\item[(i)] $H=G[A]$ is a connected graph, $u_1,u_2\in V_1$,  $v_1,v_2\in V_2$,  and $w_1,w_3\in V_3$
\item[(ii)]  $u_1v_1,v_1w_1,u_2v_2,v_2w_2\in E(G)$, that is, $u_1v_1w_1$ and $u_2v_2w_2$ are paths of length two, 
\item[(iii)] $H$ has a $3$-layer drawing $\sigma_H=(\sigma_1^H,\sigma_2^H,\sigma_2^H)$ without crossings respecting $(V_1\cap V(H),V_2\cap V(H),V_3\cap V(H))$, where $\sigma_i^H$ is an order of $V_i\cap V(H)$ for $i\in\{1,2,3\}$, such that 
(a) for every $x\in V_1\cap V(H)$, $\sigma_1^H(u_1)\leq \sigma_1^H(x)\leq \sigma_1^H(u_2)$, 
(b) for every $x\in V_2\cap V(H)$, $\sigma_2^H(v_1)\leq \sigma_2^H(x)\leq \sigma_2^H(v_2)$, and
(c) for every $x\in V_3\cap V(H)$, $\sigma_3^H(w_1)\leq \sigma_3^H(x)\leq \sigma_3^H(w_2)$,
\item[(iv)] $H$ has a family of vertex disjoint path $P=\{x_1y_1z_1,\ldots,x_{4k+3}y_{4k+3}z_{4k+3}\}$ of size $4k+3$ with $u_1v_1w_1,u_2v_2w_2\in P$.
\end{itemize}
Then assume that $x_1=u_1$, $y_1=v_1$, $z_1=w_1$, $x_{4k+3}=u_2$, $y_{4k+3}=v_2$, $z_{4k+3}=w_2$,
$\sigma_1^H(x_1)<\dots<\sigma_1^H(x_{4k+3})$, 
$\sigma_2^H(y_1)<\dots<\sigma_2^H(y_{4k+3})$, and 
$\sigma_3^H(z_1)<\dots<\sigma_3^H(z_{4k+3})$
 and do the following:
\begin{itemize}
\item[(a)] delete every vertex $s \in V_1\cap V(H)$ with $\sigma_1^H(x_{2k+1})<\sigma_1^H(s)<\sigma_1^H(x_{2k+3})$,
\item[(b)] delete every vertex $s\in V_2\cap V(H)$ with $\sigma_2^H(y_{2k+1})<\sigma_2^H(s)<\sigma_2^H(y_{2k+3})$ except $y_{2k+2}$ and delete the edges incident to $y_{2k+2}$,
\item[(c)]  delete every vertex $s\in V_3\cap V(H)$ with $\sigma_3^H(z_{2k+1})<\sigma_3^H(s)<\sigma_3^H(z_{2k+3})$,
\item[(d)] make $x_{2k+1}$ and $x_{2k+3}$ adjacent to $y_{2k+2}$ to maintain connectivity. 
\end{itemize}
\end{reduction}

We again note that deleting of vertices and edges is done in such a way that we cannot create new path of length two with the end-vertices in $V_1$ and $V_3$.

\begin{lemma}\label{cl:paths-C}
\Cref{red:paths-C} is safe.
\end{lemma}

\begin{proof}
Again, we exploit the same ideas as in the proof of \Cref{cl:matching}; now, we use the family of paths $P$ instead of the matching $M$.  The proof is very similar to the proof of \Cref{cl:paths-A}.
Because $H$ has a $3$-layer drawing $\sigma_H=(\sigma_1^H,\sigma_2^H,\sigma_3^H)$ without crossings respecting $(V_1\cap V(H),V_2\cap V(H),V_3\cap V(H))$ such that (iii)(a)--(c)  are fulfilled, the vertices of the paths of $P$ can be ordered 
 to satisfy the conditions that  $\sigma_1^H(x_1)<\dots<\sigma_1^H(x_{4k+3})$, $\sigma_2^H(y_1)<\dots<\sigma_2^H(y_{4k+3})$, and
  $\sigma_3^H(y_1)<\dots<\sigma_3^H(y_{4k+3})$. Thus, the rule is feasible. Also, the rule deletes at least two vertices.   
 Suppose that the rule is applied for a separation $(A,B)$ and denote by $G'$ the graph obtained by applying the rule. Let also $D\subseteq V(H)$ be the set of deleted vertices, $V_1'=V_1\setminus D$, $V_2'=V_2\setminus D$, and $V_3'=V_3\setminus D$.

Assume that   $\sigma=(\sigma_1,\sigma_2,\sigma_3)$ is a $3$-layer drawing of $G$ respecting $(V_1,V_2,V_3)$ with at most $k$ crossings. Because the number of crossings is at most $k$, there are $i\in\{1,\ldots,2k+1\}$ and $j\in\{2k+3,\ldots,4k+3\}$ such that the edges $x_iy_i$, $y_iz_i$, $x_jy_j$, and $y_jz_j$ do not cross any edge. Without loss of generality, we assume that $\sigma_1(x_i)<\sigma_1(x_j)$, $\sigma_2(y_i)<\sigma_2(y_j)$, and $\sigma_3(z_i)<\sigma_3(z_j)$. 
We define \[U_1=\{s\in V_1\cap V(H)\colon \sigma_1(x_i)<\sigma_1(s)<\sigma_1(x_j)\}\] \[U_2=\{s\in V_2\cap V(H)\colon \sigma_2(y_i)<\sigma_2(s)<\sigma_2(y_j)\},\]  and 
\[U_3=\{s\in V_3\cap V(H)\colon \sigma_3(z_i)<\sigma_2(s)<\sigma_3(z_j)\}.\]
 Because $H$ is connected and $x_iy_i$, $y_iz_i$, $x_jy_j$, and $y_jz_j$ do not cross any edge, we have that 
\[U_1=\{s\in V_1\cap V(H)\colon \sigma_1^H(x_i)<\sigma_1^H(s)<\sigma_1^H(x_j)\},\] \[U_2=\{s\in V_2\cap V(H)\colon \sigma_2^H(y_i)<\sigma_2^H(s)<\sigma_2^H(y_j)\},\] and
\[U_3=\{s\in V_3\cap V(H)\colon \sigma_3^H(z_i)<\sigma_3^H(s)<\sigma_3^H(z_j)\}.\]
Also, there is no vertex $s\in V(G)\setminus (U_1\cup U_2\cup U_3)$
such that $\sigma_1(x_i)< \sigma_1(s)< \sigma_1(x_j)$ if $s\in V_1$, $\sigma_2(y_i)<\sigma_2(s)< \sigma_2(y_j)$ if $s\in V_2$, and
$\sigma_3(z_i)< \sigma_3(s)< \sigma_3(z_j)$ if $s\in V_3$.

We construct a  $3$-layer drawing $\sigma'=(\sigma'_1,\sigma'_2,\sigma'_3)$ of $G'$  respecting $(V_1',V_2',V'_3)$ by modifying $\sigma$ as follows.
We place the vertices of $U_1\cap V(G')$  between $x_i$ and $x_j$ following their order in $\sigma^H_1$,
we place the vertices of $U_2\cap V(G')$  between $y_i$ and $y_j$ following their order in $\sigma^H_2$, and
we place the vertices of $U_3\cap V(G')$  between $z_i$ and $z_j$ following their order in $\sigma^H_2$.  
Because there are no crossing edges with their endpoints in $U_1$, $U_2$, or $U_3$, the number of crossings in  $\sigma'$  is at most the number of crossings in $\sigma$.

For the opposite direction, assume that there is a $3$-layer drawing $\sigma'=(\sigma'_1,\sigma'_2,\sigma'_3)$ of $G'$  respecting $(V_1',V_2',V'_3)$ with at most $k$ crossings. We construct  a $3$-layer drawing $\sigma=(\sigma_1,\sigma_2,\sigma_3)$ of $G$  respecting $(V_1,V_2,V_3)$ such that the number of crossings in $\sigma$ is at most the number of crossings in $\sigma'$. Because the number of crossings is at most $k$, there are $i\in\{1,\ldots,2k+1\}$ and $j\in\{2k+3,\ldots,4k+3\}$ such that the edges 
$x_iy_i$, $y_iz_i$, $x_jy_j$, and $y_jz_i$  do not cross any edge with respect to $\sigma'$. We can assume that  $\sigma'_1(x_i)<\sigma'_1(x_j)$, $\sigma'_2(y_i)<\sigma'_2(y_j)$,  and
$\sigma'_3(z_i)<\sigma'_3(z_j)$.
Let   
\[U_1=\{s\in V_1\cap V(H)\colon \sigma_1^H(x_i)<\sigma_1^H(s)<\sigma_1^H(x_j)\},\]   \[U_2= \{s\in V_2\cap V(H)\colon \sigma_2^H(y_i)<\sigma_2^H(s)<\sigma_2^H(y_j)\},\] 
and 
\[U_3= \{s\in V_3\cap V(H)\colon \sigma_3^H(z_i)<\sigma_3^H(s)<\sigma_3^H(z_j)\}.\] We set  
\[U_1'=\{s\in V_1'\cap V(H)\colon \sigma'_1(x_i)<\sigma_1(s)<\sigma'_1(x_j)\},\] \[U_2'=\{s\in V_2'\cap V(H)\colon \sigma'_2(y_i)<\sigma'_2(s)<\sigma'_2(y_j)\},\] and
\[U_3'=\{s\in V_3'\cap V(H)\colon \sigma'_3(z_i)<\sigma'_3(s)<\sigma'_3(z_j)\}.\]
Because  the graph $H'$ obtained from $H$ by deleting the vertices of $D$ and making $x_{2k+1}$ and $x_{2k+3}$ adjacent to $y_{2k+2}$ is connected and 
$x_iy_i$, $y_iz_i$, $x_jy_j$, and $y_jz_j$  do not cross any edges,  
$U_1'\subseteq U_1$, $U_2'\subseteq U_2$,  and $U_3'\subseteq U_3$.
Furthermore, there is no vertex $s\in V(G)\setminus (U'_1\cup U'_2\cup U'_3)$  such that 
$\sigma'_1(x_i)< \sigma'_1(s)< \sigma'_1(x_j)$ if $s\in V_1$, $\sigma'_2(y_i)< \sigma'_2(s)< \sigma'_2(y_j)$ if $s\in V_2$, and
 $\sigma'_3(z_i)< \sigma'_3(s)< \sigma'_3(y_j)$ if $s\in V_3$.

We construct $\sigma_1$ from $\sigma_1'$ by replacing the vertices of $U_1'$ with the vertices of $U_1$ and putting them between $x_i$ and $x_j$ following their order in $\sigma^H_1$. Similarly,  $\sigma_2$ is constructed from $\sigma'_2$ by replacing the vertices of $U_2'$ with the vertices of $U_2$ and putting them between $y_i$ and $y_j$ following their order in $\sigma^H_2$, and $\sigma_3$ is constructed from $\sigma'_3$ by replacing the vertices of $U_3'$  by the vertices of $U_3$ and putting them between $z_i$ and $z_j$ following their order in $\sigma^H_3$.  This concludes the proof of the claim.
\end{proof}
 
 Recall that the reduction rules are applied exhaustively. We show that if no rule could be applied then for yes-instances, the size of the graph is bounded. 
 
Suppose that  $(G,V_1,V_2,V_3,k)$ is a yes-instance of \probThreeCr. Then there is a $3$-layer drawing $\sigma=(\sigma_1,\sigma_2,\sigma_3)$  of $G$ respecting $(V_1,V_2,V_3)$ with at most $k$ crossings. We denote by $D$ the endpoints of crossing edges.  Then $|D|\leq 4k$, $|D\cap V_2|\leq 2k$, and $|D\cap (V_1\cup V_3)|\leq 2k$.
First, we make the following observation.

\begin{observation}\label{obs:deg-V-two} 
Suppose that  \Cref{red:pend} cannot be applied for $(G,V_1,V_2,V_3,k)$. Then
for every $u\in V_2$ and $i\in\{1,3\}$, $|N_G(u)\cap (V_i\setminus D)|\leq k+3$. 
\end{observation} 
 
 \begin{proof}
 Let $u\in V_2$. We observe that the vertices of $N_G(u)\cap (V_i\setminus D)$ for $i\in\{1,3\}$ should be consecutive in $\sigma_i$ because the edges $uv$ with $v\in N_G(u)\cap (V_i\setminus D)$ do not cross any edge. Furthermore, at most two of them, namely the first and the last vertex of $N_G(u)\cap (V_i\setminus D)$ with respect to $\sigma_i$, may have neighbors distinct from $u$.  This implies that $|N_G(u)\cap (V_i\setminus D)|\leq k+3$ because, otherwise,  \Cref{red:pend} would be applied. 
 This concludes the proof. 
 \end{proof}
 
 For two (not necessarily distinct) vertices $u,v\in V_2$ with $\sigma_2(u)\leq \sigma_2(v)$, we say 
that $[u,v]=\{w\in V_2\colon \sigma_2(u)\leq \sigma_2(w)\leq \sigma_2(v)\}$ is an \emph{interval} of $V_2$. For interval $[u,v]$, we use  
$|[u,v]|$ to denote its size, that is, the number of vertices in $[u,v]$.
 We say that an interval $[u,v]$ is  \emph{safe}  if $[u,v]$ contains no vertex incident with a crossing edge; $[u,v]$ is a \emph{maximal} safe interval if 
 $[u,v]$ is an inclusion maximal safe interval. 
  Note that since the number of crossings is at most $k$, at most $2k$ vertices of $V_2$ are incident with crossing edges. Therefore, the number of maximal safe intervals is at most $2k+1$. 
 
\begin{lemma}\label{lem:safe-one}
Suppose that \Cref{red:comp}--\Cref{red:matching-B} cannot be applied for  $(G,V_1,V_2,V_3,k)$ and 
let  $[x,y]\subseteq V_2$ be a safe interval such that any $u\in[x,y]$ has neighbors either only in $V_1$ or only in $V_3$. 
Then $|[x,y]|\leq 4(2k+3)(4k+2)$.
\end{lemma}

\begin{proof}
For the sake of contradiction, assume that  there is $[x,y]$ with $|[x,y]|> 4(2k+3)(4k+2)$. 
Then there is  $X\subseteq [x,y]$ with $|X|>2(2k+3)(4k+2)$ such that either every vertices of $X$ has neighbors only in $V_1$ or every vertex of $X$ has neighbors only in $V_3$. By symmetry, we assume without loss of generality that $N_G(X)\subseteq V_1$.   

We observe that the vertices of $X$ are included in at most two connected components of $G$. Otherwise, if there are three vertices $a,b,c\in X$ in distinct connected components of $G$ such that $\sigma_2(a)<\sigma_2(b)<\sigma_2(c)$, the connected component $H$ containing $b$ has no crossing edges. This contradicts the assumption that \Cref{red:comp} cannot be applied. Thus, there is $Y\subseteq X$ such that  $|Y|>(2k+3)(4k+2)$ and the vertices of $Y$ are in the same connected component. 
 
Consider a vertex $w\in V_1$ that has neighbors in $Y$. Notice that at most two neighbors of $w$ in $Y$, namely, the first and the last neighbor with respect to $\sigma_2$, may be adjacent to vertices distinct from $w$ because the edges $wy$ with $y\in Y$ do not cross any edge and every $y\in Y$ has neighbors only in $V_1$. 
Then because \Cref{red:nice} cannot be applied, $|N_G(w)\cap Y|\leq 2k+3$. Furthermore, if $|N_G(w)\cap Y|= 2k+3$ then at least two neighbors of $w$ have other (and distinct) neighbors in $V_1$.
As $|Y|>(2k+3)(4k+2)$, we have that at least $4k+3$ vertices of $V_2$ have distinct neighbors in $Y$, that is, there a matching $M$ of size $4k+3$ with endpoints of the edges in $Y$ and $V_1$. The edges of $M$ do not cross any edge. Let $u_1$ and $u_2$ be the first and the last endpoint, respectively, in $V_1$ with respect to $\sigma_1$, and let $v_1$ and $v_2$ be the first and the last endpoint, respectively, with respect to $\sigma_2$.  Let 
$A_1=\{w\in V_1\colon \sigma_1(u_1)\leq \sigma_1(w)\leq \sigma_1(u_2)\}$ and 
$A_2=\{w\in [x,y]\colon \sigma_2(v_1)\leq \sigma_2(w)\leq \sigma_2(v_2)\text{ and }N_G(w)\cap V_3=\emptyset\}$. 
We obtain that \Cref{red:matching-A} can be applied for $A=A_1\cup A_2$ and $u_1,u_2,v_1,v_2$. However, this contradicts the assumption that the reduction rules are not applicable.
We conclude that $|[x,y]|\leq 4(2k+3)(4k+2)$.
\end{proof} 
 
 \begin{lemma}\label{lem:safe-two}
 Suppose that \Cref{red:comp}--\Cref{red:paths-B} cannot be applied for  $(G,V_1,V_2,V_3,k)$ and 
let  $[u,v]\subseteq V_2$ be a safe interval. Let also $Z\subseteq [u,v]$ be the set of vertices having neighbors in both $V_1$ and $V_3$.
Then for any $w\in V_1\cup V_3$, $w$ has at most $(k+3)(4k+2)(2k+3)$ neighbors in $Z$.
 \end{lemma}
 
 \begin{proof}
The proof is by contradiction. Suppose that there is $w\in V_1\cup V_3$ such that for $U=N_G(w)\cap Z$,
it holds that $|U|> (k+3)(4k+2)(2k+3)$. By symmetry, we assume without loss of generality that $w\in V_1$.

For each vertex $x\in V_3$, 
 at most two neighbors of $x$ in $U$, namely, the first and the last neighbor with respect to $\sigma_2$, are adjacent to vertices distinct from $x$ and $w$. This is because the edges $xy$ and $wy$ with $y\in U$ do not cross any edges. 
Since \Cref{red:degtwo} is not applicable, we have that $|N_G(x)\cap X|\leq k+3$.  

Therefore, there is $W\subseteq U$ of size at least $(4k+2)(2k+3)+1$ such the vertices of $W$ have distinct neighbors in $V_3$. Suppose that there is a connected component of $G-w$ that contains at least $4k+3$ vertices $x_1,\ldots,x_{4k+3}\in W$. Let us assume that $\sigma_2(x_1)<\dots <\sigma_2(x_{4k+3})$. Denote by $y_1,\ldots,y_{4k+3}$ their respective distinct neighbors in $V_3$.  Note that $\sigma_3(y_1)<\dots <\sigma_2(y_{4k+3})$. Let $u_1=x_1$, $u_2=x_{4k+3}$, $v_1=y_1$, and $v_2=y_{4k+3}$. Also, set
\[A_1=\{x\in V_2\colon \sigma_2(x_1)\leq \sigma_2(x)\leq \sigma_2(x_{4k+3})\}\] and
\[A_2=\{y\in V_3\colon \sigma_3(y_1)\leq \sigma_3(y)\leq \sigma_3(y_{4k+3})\}.\] 
Then \Cref{red:paths-B} can be applied for $A=A_1\cup A_2$ and $w,u_1,u_2,v_1,v_2$; which is a contradiction. Therefore,  each connected component of $G-w$ containing vertices of $W$ has at most $4k+2$ such vertices. Thus, there are at least $2k+4$ connected components of $G-w$ with vertices of $W$. At most two such components, namely the components containing the first and the last vertices of $W$ with respect to $\sigma_2$, are not nice.
 Therefore, we have at least $2k+2$ connected components of $G-w$ that are nice. 
 In this case, we can apply  \Cref{red:nice}, thus contradicting the assumption that the reduction rules are not applicable. It yields that the 
  assumption   $|W|>(4k+2)(2k+3)$ was wrong. This concludes the proof. 
 \end{proof}
 
\begin{lemma}\label{lem:safe-three} 
 Suppose that \Cref{red:comp}--\Cref{red:paths-B} cannot be applied for  $(G,V_1,V_2,V_3,k)$ and 
let  $[u,v]\subseteq V_2$ be a safe interval. Let also $R=\{y_1,\ldots,y_\ell\}\subseteq [u,v]$ be an inclusion maximal set of vertices such that each vertex of $R$ has neighbors in both $V_1$ and $V_3$ and all these neighbors are distinct, and assume that $\sigma_2(y_1)<\dots<\sigma_2(y_\ell)$.
Then 
\begin{equation*}
|[u,v]|\leq  4(\ell+1)(k+3)(4k+2)(2k+3)[4(2k+3)(4k+2)+1],
\end{equation*} 
and if $\ell\geq 2$ then
\begin{equation*}
|[y_1,y_\ell]|\leq 4(\ell-1)(k+3)(4k+2)(2k+3)[4(2k+3)(4k+2)+1].
\end{equation*}
\end{lemma}

\begin{proof}
For $i\in \{1,\ldots,\ell\}$, denote by $x_i$ and $z_i$ neighbors of $y_i$ in $V_1$ and $V_3$, respectively, chosen in such a way that $x_1,\ldots,x_\ell$ and all $z_1,\ldots,z_\ell$ are distinct. 
We set $y_0=u$ and $y_{\ell+1}=v$ (it may happen that $y_0=$ or $y_{\ell+1}=y_\ell$). Let also $x_0$ and $x_{\ell+1}$ be the first and the last vertex, respectively, of $V_1$ with respect to $\sigma_1$ that 
is adjacent to a vertex of $[u,v]$. Similarly, let $z_0$ and $z_{\ell+1}$ be the first and the last vertex of $V_3$ with respect to $\sigma_3$ that is adjacent to a vertex of $[u,v]$. It may happen that 
$x_0=x_1$ or $x_{\ell+1}=x_\ell$ or $z_0=z_1$ or $z_{\ell+1}=z_\ell$.
For $i\in\{0,\ldots,\ell\}$, denote $I_i=[y_i,y_{i+1}]$.
 We prove the following claim.
  
\begin{claim}\label{cl:size-Ii}
 For every $i\in\{0,\ldots,\ell\}$, $|I_i|\leq 4(k+3)(4k+2)(2k+3)[4(2k+3)(4k+2)+1]$. 
\end{claim} 

\begin{proof}[Proof of \Cref{cl:size-Ii}]
Consider arbitrary  $i\in\{0,\ldots,\ell\}$. By \Cref{lem:safe-one}, for a subinterval $[y',y'']\subseteq I_i$ such that any $y\in[y',y'']$ has neighbors either only in $V_1$ or only in $V_3$, it holds that $|[y',y'']|\leq 4(2k+3)(4k+2)$. Then $I_i$ contains at least  $\frac{|I_i|}{4(2k+3)(4k+2)+1}$ vertices having neighbors in both $V_1$ and $V_3$. 
Observe that if $y\in I_i\setminus\{y_i,y_{i+1}\}$ is adjacent to $x\in V_1$ and $z\in V_3$ then either $x\in\{x_i,x_{i+1}\}$ or $z\in\{z_i,z_{i+1}\}$ by the maximality of $R$. 
By \Cref{lem:safe-two}, each of the vertices $x_i,x_{i+1},y_i,y_{i+1}$ is adjacent to at most $(k+3)(4k+2)(2k+3)$ vertices of $[u,v]$ that have neighbors in both $V_1$ and $V_2$. 
Thus, the total number of such neighbors of $x_i,x_{i+1},y_i,y_{i+1}$ is at most
$4(k+3)(4k+2)(2k+3)\geq \frac{|I_i|}{4(2k+3)(4k+2)+1}$. Therefore, $|I_i|\leq 4(k+3)(4k+2)(2k+3)[4(2k+3)(4k+2)+1]$. This concludes the proof.
\end{proof}

By \Cref{cl:size-Ii}, we have that 
\begin{equation*}
|[u,v]|\leq\sum_{i=0}^{\ell}|I_i|\leq  4(\ell+1)(k+3)(4k+2)(2k+3)[4(2k+3)(4k+2)+1],
\end{equation*}
and if $\ell\geq 2$ then 
\begin{equation*}
|[y_1,y_\ell]|\leq\sum_{i=1}^{\ell-1}|I_i|\leq  4(\ell-1)(k+3)(4k+2)(2k+3)[4(2k+3)(4k+2)+1],
\end{equation*}
This concludes the proof.
\end{proof}

Now we are ready to prove the crucial lemma.

 \begin{lemma}\label{cl:size}
 Suppose that \Cref{red:comp}--\Cref{red:paths-C} are not applicable for a yes-instance $(G,V_1,V_2,V_3,k)$ of \probThreeCr. Then 
 $|V(G)|\leq 2^{16}(k+4)^{11}$.
  \end{lemma}

\begin{proof}
Suppose that
 $(G,V_1,V_2,V_3,k)$ is a yes-instance of \probThreeCr such that \Cref{red:comp}--\Cref{red:paths-C} are not applicable. Recall that for a $3$-layer drawing $\sigma=(\sigma_1,\sigma_2,\sigma_3)$  of $G$ respecting $(V_1,V_2,V_3)$ with at most $k$ crossings, we use $D$ to denote the endpoints of crossing edges.  Recall also that 
 $|D|\leq 4k$, $|D\cap V_2|\leq 2k$, and $|D\cap (V_1\cup V_3)|\leq 2k$.  
 Notice that because \Cref{red:comp} is not applicable, $G$ has no isolated vertices. That is, every vertex has at least one neighbor. 
 
We claim that 
 \begin{equation}\label{eq:V-two}
|V_2|\leq 4(8k+6)(2k+1)(k+3)(4k+2)(2k+3)[4(2k+3)(4k+2)+1]+2k. 
 \end{equation} 
 
To prove it, we upper bound the size of maximal safe intervals.  Let $[u,v]$ be a maximal safe interval. 
Let also $R=\{y_1,\ldots,y_\ell\}\subseteq [u,v]$ be a set of maximum size  such that each vertex of $R$ has neighbors in both $V_1$ and $V_3$ and all these neighbors are distinct, and assume that $\sigma_2(y_1)<\dots<\sigma_2(y_\ell)$. We claim that $\ell\leq 8k+5$.

For the sake of contradiction, assume that $\ell\geq 8k+6$. Notice that the vertices of $[u,v]$ are included in at most two connected components of $G$.
Otherwise, if there are three vertices $a,b,c\in [u,v]$ in distinct connected components of $G$ such that $\sigma_2(a)<\sigma_2(b)<\sigma_2(c)$, the connected component $C$ containing $b$ has no crossing edges contradicting the assumption that \Cref{red:comp} cannot be applied. 
This implies that the first $4k+3$ or the last $4k+3$ vertices of $R$ are in the same connected component. By symmetry, we can assume without loss of generality that 
$y_1,\ldots,y_{4k+3}$ are in the same connected component. 
By the definition of $R$, for each $i\in \{1,\ldots,4k+3\}$ there are neighbors  $x_i$ and $z_i$  of $y_i$ in $V_1$ and $V_3$ suchthat $x_1,\ldots,x_{4k+3}$ and all $z_1,\ldots,z_{4k+3}$ are distinct. 
Note that $\sigma_1(x_1)<\dots<\sigma_1(x_{4k+3})$ and $\sigma_3(z_1)<\dots<\sigma_3(z_{4k+3})$.
However, then \Cref{red:paths-C} is applicable for $H=G[A_1\cup A_2\cup A_3]$ where 
\[A_1=\{x\in V_1\colon \sigma_1(x_1)\leq \sigma_2(x)\leq \sigma_1(x_{4k+3})\}\] 
\[A_2=\{y\in V_2\colon \sigma_2(y_1)\leq \sigma_2(y)\leq \sigma_2(y_{4k+3})\},\] and
\[A_3=\{z\in V_3\colon \sigma_3(z_1)\leq \sigma_3(z)\leq \sigma_3(s_{4k+3})\}\] 
with $u_1=x_1$, $v_1=y_1$, $w_1=z_1$, $u_2=x_{4k+3}$, $v_2=y_{4k+3}$, and $w_2=z_{4k+3}$.
This contradicts the assumption that the rule is not applicable. Thus, $\ell\leq 8k+5$.

By \Cref{lem:safe-three}, we obtain  that 
\begin{equation*}
|[u,v]|\leq 4(8k+6)(k+3)(4k+2)(2k+3)[4(2k+3)(4k+2)+1]. 
 \end{equation*}

Because the number of maximal safe intervals is at most $2k+1$, we obtain that at most  $4(8k+6)(2k+1)(k+3)(4k+2)(2k+3)[4(2k+3)(4k+2)+1]$ vertices of $V_2$ are in the maximal safe intervals. 
Since the only vertices that are not in a safe intervals are the vertices of $D\cap V_2$ and $|D\cap V_2|\leq 2k$, we conclude that (\ref{eq:V-two}) holds.

To conclude the proof of the lemma, recall that for every $u\in V_2$, $|N_G(u)\cap (V_i\setminus D)|\leq k+3$ for $i\in\{1,3\}$ by \Cref{obs:deg-V-two}. Therefore, 
 $|(V_1\cup V_2)\setminus D|\leq 2(k+3)|V_2|$. Recall that also that $|D\cap (V_1\cup V_3)|\leq 2k$. Then
 \begin{align*}
 |V(G)|\leq &(2k+7)|V_2|+2k\\ \leq &(2k+7)[4(8k+6)(2k+1)(k+3)(4k+2)(2k+3)[4(2k+3)(4k+2)+1]+2k]+2k\\ \leq &2^{15}(k+2)^8. 
 \end{align*}
\end{proof}

Now we state the final rule.

\begin{reduction}\label{red:final}
If $|V(G)|>2^{15}(k+2)^{8}$ or $|E(G)|>2^{16}(k+2)^{8}+k$ then return a trivial no-instance and stop.
\end{reduction}

\begin{lemma}\label{cl:final}
\Cref{red:paths-C} is safe.
\end{lemma}

\begin{proof}
If $|V(G)|>2^{15}(k+2)^{8}$ then $(G,V_1,V_2,V_3,k)$ is a no-instance of \probThreeCr by \Cref{cl:size}. Suppose that $|V(G)|\leq 2^{15}(k+2)^{8}$.
If  $(G,V_1,V_2,V_3,k)$ is a yes-instance then there is a set $S$ of at most $k$ edges such that the graph $G'$ obtained from $G$ by the deletion of the edges of $S$ is planar. Because $G'$ is bipartite with the bipartition $(V_2,V_1\cup V_3)$, we have that 
$|E(G')|\leq 2|V(G')|-4\leq 2^{16}(k+2)^{8}-4$. Thus, if $|E(G)|>2^{16}(k+2)^{8}+k$ then $(G_1,G_2,G_3,k)$ is a no-instance.
\end{proof}

This completes the description of the kernelization algorithm. The correctness follows from \Cref{cl:pend}--\Cref{cl:final}.

\subsection{Running time evaluation}\label{sec:run-time}
It is easy to show that each of the reduction rules can be executed in polynomial time implying that the total running time of the kernelization algorithm is polynomial. However,
a naive implementation may give a huge dependence on $n$. For example, for the direct application of  \Cref{red:paths-C}, we can consider all 6-tuples  $(u_1,v_1,w_1,u_2,v_2,w_2)$ of vertices and then consider all separations $(A,B)$ with $A\cap B=\{u_1,v_1,w_1,u_2,v_2,w_2\}$ where $G[A]$ is connected and conditions (i)--(iv) are fulfilled. It is straightforward to verify (i) and (ii), then condition (iii) is verified by making use of \Cref{prop:linear-time} in linear time, and to verify (iv) and find paths, we can use the standard maximum flow algorithm of Ford and Fulkerson~\cite{FordF56}.  However, this leads to $\Oh(n^8)$ time. We show that with some additional work, we can achieve running time which is linear in $n$.
To obtain this, we show that we can replace the exhaustive application of each rule whenever possible by a restricted number of calls in such a way that the total running time used by the rule is linear in $n$. Also, we implement some additional steps in the algorithm allowing us to rule out no-instances. We give them here instead of \Cref{sec:constr} because they do not influence the kernel size and their only purpose is to speed up computations.

Recall that at the beginning of the algorithm, we use \Cref{prop:linear-time}, to check in time $\cO(n)$ whether $G$ admits a $3$-layer drawing respecting $(V_1,V_2,V_3)$ without crossings. If this is the case, we immediately return a trivial yes-instance and stop. We also stop and output a trivial no-instance if $k=0$ and $G$ does not admit such a drawing. 
Further, we note that we can bound the number of the edges in the input instance   $(G,V_1,V_2,V_3,k)$. In the same way as in the proof of \Cref{cl:final}, we observe that 
if  $(G,V_1,V_2,V_3,k)$ is a yes-instance then there is a set $S$ of at most $k$ edges such that the graph $G'$ obtained from $G$ by the deletion of the edges of $S$ is planar. Then because  $G'$ is bipartite,  $|E(G')|\leq 2|V(G')|-4$. Thus, if $|E(G)|\geq |V(G)|+k-3$, we conclude that  $(G,V_1,V_2,V_3,k)$ is a no-instance and return a trivial no-instance. 
If $n\leq 2^{15}(k+2)^{8}$ then we can simply return the input instance. 
This means that it can be assumed that $G$ has no drawing respecting $(V_1,V_2,V_3)$ without crossings, $k\geq 1$, $|E(G)|\leq n+k-3$, and $n> 2^{15}(k+2)^{8}$. 
In particular, we have that $|E(G)|\leq 2n$, that is, the size of the input graph is linear in $n$.

To apply~\Cref{red:comp}, we use standard tools~\cite{HopcroftT73} to compute the connected components of $G$ in $\Oh(n)$ time. Then for each connected component $H$, we use \Cref{prop:linear-time} to decide whether $H$ a  $3$-layer drawing respecting $(V_1\cap V(H),V_2\cap V(H),V_3\cap V(H))$ without crossings and delete the vertices of $H$ if this holds. By \Cref{prop:linear-time}, we have that this can be done in $\Oh(n)$ time. Notice that the subsequent reduction rules do not create new connected components and do not turn any component into a graph that could be drawn without crossings. Thus, \Cref{red:comp} will not be called again afterward.

For~\Cref{red:pend}, we observe that the inclusion maximal sets of pendent vertices $P\subseteq V_1$ or $P\subseteq V_3$ that have the same neighbors in $V_2$ can be found in linear time. Then \Cref{red:pend} can be applied in the total $\Oh(n)$ time for all such subsets. Similarly, to apply \Cref{red:degtwo}, we find inclusion maximal subsets of vertices $Q\subseteq V_2$ of degree two that have the same unique neighbors in $V_1$ and $V_2$. This can be done in $\Oh(n)$ time either directly or by observing that these sets compose modules of $G$ and the modules can be found in linear time~\cite{corneil2024recursive}. Note that \Cref{red:degtwo}--\Cref{red:paths-B} do not create new pendent vertices. However, new pendent vertices may appear after applying  \Cref{red:paths-C}.
This possibility will be considered in the analysis of  \Cref{red:paths-C}.
Similarly, to apply \Cref{red:degtwo}, we find inclusion maximal subsets of vertices $Q\subseteq V_2$ of degree two that have the same unique neighbors in $V_1$ and $V_2$. This can be done in $\Oh(n)$ time either directly or by observing that these sets compose modules of $G$ and the modules can be found in linear time~\cite{corneil2024recursive}. Notice that \Cref{red:paths-A}--\Cref{red:paths-C} can create new vertices in $V_2$ of degree two with the neighbors in $V_1$ and $V_3$. 
To remedy this for \Cref{red:paths-A} and \Cref{red:paths-B},  we simply repeat the described procedure applying \Cref{red:degtwo} after the exhaustive application of \Cref{red:matching-A} and
 \Cref{red:paths-A} and \Cref{red:paths-B}. Since this is done two times, the total running time for \Cref{red:degtwo} is still linear.
 For  \Cref{red:paths-C}, we incorporate the application of \Cref{red:degtwo} in the analysis of the rule.

The analysis of the subsequent rules is more complicated and we do it in separate statements.

\begin{lemma}\label{lem:nice-rt}
\Cref{red:nice} can be exhaustively applied in $\Oh(n)$ time. 
\end{lemma}

\begin{proof}
We show how to apply \Cref{red:nice} for nice connected components $H$ of $G-u$ for vertices $u\in V_1$. The case $u\in V_3$ is symmetric. 
Recall that a connected component $H$ of $G-u$ is nice if 
 (i) $V(H)\subseteq V_2\cup V_3$  and (ii) $H$ admits a $2$-layer drawing without crossings respecting either $(V_2\cap V(H),V_3\cap V(H))$.
 Because we are interested in $H$ with $V(H)\subseteq V_2\cup V_3$, we consider the graph $G'=G[V_2\cup V_3]$ and list all connected components by standard algorithms~\cite{HopcroftT73}. Simultaneously, we compute the number of vertices in each component. This can be done in $\Oh(n)$ time. 
 Then among all these components, we find the components $H$ such that (i) $H$ has a $2$-layer drawing, that is, $H$ is a caterpillar by \Cref{obs:caterpillar}, and (ii) $H$ has 
 a unique neighbor in $V_1$; when we find such a neighbor $u$ of $H$, we assign $H$ to $u$. Finally, for all $u\in V_1$, if we have  $s\geq 2k+2$ components $H$ assigned to the same vertex $u$ then we delete $s-(2k+1)$ components starting with trivial ones.  This also can be done in linear time and we conclude that the total running time for  
 \Cref{red:nice} is $\Oh(n)$. 
 This concludes the proof.
 \end{proof}

We note that \Cref{red:matching-A}--\Cref{red:paths-B}  cannot create new nice components for some $u\in V_1\cup V_3$ but this may happen after applying \Cref{red:paths-C}. We deal with this issue through the analysis of \Cref{red:paths-C}.
 
The following easy observation is useful for the next rules.

\begin{observation}\label{obs:cycles}
If either $G[V_1\cup V_2]$ or $G[V_2\cup V_3]$ contains a cycle with at least $2k+4$ vertices then $(G,V_1,V_2,V_3,k)$ is a no-instance of \probThreeCr.
\end{observation}

\begin{proof}
Assume that $C$ is a cycle in $G[V_1\cup V_2]$ with at least  $2k+4$ vertices 
and consider arbitrary $3$-layer drawing $\sigma=(\sigma_1,\sigma_2,\sigma_3)$ of $G$.

Suppose that $k$ is even.  Let $u$ be the first vertex of $V(C)\cap V_1$ in $\sigma_1$ and let $v$ be the last vertex of $V(C)\cap V_2$ in $\sigma_2$. Because $C$ is a cycle, there are two internally disjoint $(u,v)$-paths $P$ and $Q$ such that one of them, say, $P$ has at least $k+3$ edges. Then $Q$ crosses every edge of $P$ except, possibly, the first and the last edge. Thus, the total number of crossings is at least $k+1$. 

If $k$ is odd, we let $u$ and $v$ be the first and the last vertices, respectively, of $V(C)\cap V_1$ in $\sigma_1$. We again have internally disjoint $(u,v)$-path $P$ and $Q$ is $C$ such that $P$ has at least  $k+3$ edges. Then $Q$ crosses at least $k+1$ edges of $P$.

In both cases, the number of crossings is at least $k+1$ implying that $(G,V_1,V_2,V_3,k)$ is a no-instance. This concludes the proof. 
\end{proof}

\begin{lemma}\label{lem:four-rt}
There is an algorithm running in $\Oh(n)$ time that either exhaustively applies 
\Cref{red:matching-A} and \Cref{red:matching-B} or correctly concludes that $(G,V_1,V_2,V_3,k)$ is a no-instance.
\end{lemma}

\begin{proof}
We show the lemma for \Cref{red:matching-A} as the proof for  \Cref{red:matching-B} is symmetric. The idea is to find maximal subgraphs $H$ satisfying conditions (i)--(iv) of the rule and apply the rule to them simultaneously. 

Suppose that  \Cref{red:matching-A} can be applied for $H=G[A]$ with the corresponding vertices $u_1,u_2,v_1,v_2$. Because $H$ admits a $2$-layer drawing $\sigma_H=(\sigma_1^H,\sigma_2^H)$ without crossings respecting $(V_1\cap V(H),V_2\cap V(H))$ such that (a) for every $w\in V_1$, $\sigma_1^H(u_1)\leq \sigma_1^H(w)\leq \sigma_1^H(u_2)$ and (b) for every $w\in V_2$, $\sigma_2^H(v_1)\leq \sigma_2^H(w)\leq \sigma_2^H(v_2)$, we have that $H$ is a caterpillar by \Cref{obs:caterpillar} with  a backbone $P$ that has one end-vertex in $\{u_1,v_1\}$, the other in $\{u_2,v_2\}$, and does not contain $u_1v_1,u_2v_2$; we call such a backbone \emph{proper}. 
Since $H$ has a matching of size $4k+3$, $|V(P)|\geq 4k+3$. Notice  that each inner vertex $w$ of $P$ has no neighbors in $V_3$ and is adjacent to at most two vertices with non-pendent neighbors.   

We use these properties to identify the backbones of caterpillars for which the rule will be applied. More precisely, 
we aim to find inclusion maximal proper backbones and  we say that $H$ is \emph{backbone maximal} if the proper backbone of $H$ is inclusion maximal. Notice 
that $G$ can contain distinct $H$ with the same proper backbone $P$. In particular, while one endpoint of $u_1v_1$ and $u_2v_2$ is in $P$, we can choose the second endpoint in a different way. However, for any possible choice of these vertices, the edges $u_1v_1$ and $u_2v_2$ are incident to the same vertices of $P$. This implies the following property for any 
two graphs $H$ and $H'$ with the same proper backbone $P$ with boundary vertices $u_1,u_2,v_1,v_2$ and $u_1',u_2',v_1',v_2'$, respectively: 
$M$ is a matching in $H$ with 
$u_1v_1,u_2v_2\in M$ if and only if $M'=(M\setminus\{u_1v_1,u_2v_2\})\cup\{u_1'v_1',u_2'v_2'\}$ is a matching in $H'$ with $u_1'v_1',u_2'v_2'\in M'$.
Thus, it is sufficient to find the inclusion maximal proper backbones and restrict the application of \Cref{red:matching-A} to the graphs $H$ constructed for them.

We find the sets $U_i\subseteq V_i$ of the vertices that have two non-pendent neighbors for $i\in\{1,2\}$ where each vertex of $U_2$ has no neighbors in $V_3$. Then we construct $G[U_1\cup U_2]$ and find all connected components of size at least $4k+1$ in linear time using the depth-first search~\cite{HopcroftT73}. Each connected component is either a cycle or a path. If there is a component $C$ that is a cycle with at least $4k+1$ vertices, $C$ has at least $4k+2\geq 2k+4$ vertices by bipartiteness. By \Cref{obs:cycles}, we conclude that 
$(G,V_1,V_2,V_3,k)$ is a no-instance and stop. From now on, we assume that this is not the case and the  family $\mathcal{P}$ 
of  connected components of $G[U_1\cup U_2]$ consists of paths  with at least $4k+1$ vertices.

For both end-vertices of each path $P\in\mathcal{P}$, we do the following. By the construction, each end-vertex of $P$ has a neighbor of degree two that is not in $P$. 
If these neighbors are the same we obtain that $G[V_1\cup V_2]$ has a cycle with at least $4k+2\geq 2k+4$ vertices. Then by \Cref{obs:cycles}, $(G,V_1,V_2,V_3,k)$ is a no-instance and we stop. Otherwise, we construct $P'$ from $P$ by appending these neighbors. 
Thus, we obtain the family $\mathcal{P}'$ of all potential edge-disjoint 
backbones of caterpillars $H$ for which 
\Cref{red:matching-A} can be applied and, furthermore, the construction of $\mathcal{P}'$ can be done in linear time. Notice that the paths in $\mathcal{P}'$ are not necessarily proper backbones because the first and the last edges may be not in a proper backbone. However, each inclusion maximal proper backbone is included in some path of the family.

For each $P\in\mathcal{P}'$, we construct $H_P$ together with  vertices $u_1,u_2\in V_1$ and $v_1,v_2\in V_2$ such that $P$ is a backbone of $H_P$. Initially, we set $H_P:=P$. Let $x$ and $y$ be the end-vertices of $P$.
\begin{itemize}
\item[(i)] For each $z\in V(P)\setminus\{x,y\}$, find the pendent vertices $z'\in V_1\cup V_2$ adjacent to $z$ and add $z'$ and $zz'$ to $H_P$.
\item[(ii)] If $x$ is adjacent in $G$ to a vertex $z\in (V_1\cup V_2)\setminus V(P)$ then add $z$ and $xz$ to $H_P$. Otherwise, select $z$ to be the neighbor of $x$ in $P$ and set $P:=P-x$. Then set $\{u_1,v_1\}:=\{x,z\}$.
\item[(iii)] If $y$ is adjacent in $G$ to a vertex $z\in (V_1\cup V_2)\setminus V(P)$ then add $z$ and $yz$ to $H_P$. Otherwise, select $z$ to be the neighbor of $y$ in $P$ and set $P:=P-y$. Then set $\{u_2,v_2\}:=\{y,z\}$.
\end{itemize}
Notice that it may happen that $\{u_1,v_1\}\cap \{u_2,v_2\}\neq \emptyset$ if either $z$ selected in (ii) equals $y$,  or  $z$ selected in (iii) equals $x$, or the vertices $z$ selected in (ii) and (iii) are the same. However, in all these cases, we obtain that $G[V_1\cup V_2]$ has a cycle with at least $4k+2\geq 2k+4$ vertices. Then we conclude that 
$(G,V_1,V_2,V_3,k)$ is a no-instance and we stop. Otherwise, $H_P$ with $u_1,u_2,v_1,v_2$ is a backbone maximal subgraph of $G$ for an inclusion maximal backbone $P$ with the properties that 
\begin{itemize}
\item[(i)] $H_P$ is a connected graph such that for $A=V(H_P)\subseteq V_1\cup V_2$ and $B=(V(G)\setminus V(H_P))\cup\{u_1,u_2,v_1,v_2\}$, $(A,B)$ is a separation of $G$ with $A\cap B=\{u_1,u_2,v_1,v_2\}$, and  $u_1,u_2\in V_1$, and  $v_1,v_2\in V_2$,
\item[(ii)]  $u_1v_1,u_2v_2\in E(G)$, 
\item[(iii)] $H_P$ has a $2$-layer drawing $\sigma_H=(\sigma_1^H,\sigma_2^H)$ without crossings respecting $(V_1\cap V(H_P),V_2\cap V(H_P))$, where $\sigma_i^H$ is an order of $V_i\cap V(H_P)$ for $i\in\{1,2\}$, such that 
(a) for every $w\in V_1$, $\sigma_1^H(u_1)\leq \sigma_1^H(w)\leq \sigma_1^H(u_2)$ and (b) for every $w\in V_2$, $\sigma_2^H(v_1)\leq \sigma_2^H(w)\leq \sigma_2^H(v_2)$, 
\item[(iv)] $|V(H_P)|\geq 4k+1$. 
\end{itemize}
Furthermore, the construction of $H_P$ for all $P\in\mathcal{P}'$ can be done in linear time.


Let $P\in\mathcal{P}'$. We check whether $H_P$ has a matching $M$ with at least $4k+3$ edges containing $u_1v_1$ and $u_2v_2$ and if such a matching exists we pick $M$ in such a way that would allow us to delete the maximum number of vertices. We use the fact that $H_P$ is a caterpillar and the $2$-layer drawing of $H_P$ is unique by \Cref{obs:caterpillar} up to the reversals of the orders and permutation of pendent vertices in the orders. We find the $2$-layer drawing $\sigma_H=(\sigma_1^H,\sigma_2^H)$ without crossings with $\sigma_1^H(u_1)<\sigma_1^H(u_2)$ and $\sigma_2^H(v_1)<\sigma_2^H(v_2)$. Notice that the vertices of $V(P)\cap V_1$ and $V(P)\cap V_2$ are ordered in the path order and for every matching $M$, at least one of the endpoints of of each edge in $P$. We find $M$ by the following greedy procedure. We set $M:=\{x_1y_1\}$ where $x_1=u_1$ and $y_1=v_1$.
Then we try to add $2k+1$ edges $x_iy_i$ with $x_i\in V_1\cap V(H_P)$ and $y_i\in V_2\cap V(H_P)$ for $i\in\{2,\ldots,2k+2\}$. We greedily select first with respect to $\sigma_1^H$ vertex $x_i$ that is not incident to any edge of the already constructed matching that has a neighbor that is not saturated in $M$. Then among the unsaturated neighbors of $x_i$, we select $y_i$ which is first with respect to $\sigma_2^H$. If we succeed with the first $2k+2$ edges we proceed with selecting $2k+1$ from the other side. We add $x_{4k+3}y_{4k+3}$ with $x_{4k+3}=u_2$ and $y_{4k+3}=v_2$ if possible. Then we add $2k$ edges $x_{4k+3-i}y_{4k+3-i}$ for $i\in\{1,\ldots2k\}$ following the reversals of $\sigma_1^H$ and $\sigma_2^H$. If this greedy choice gives a matching $M$ of size $4k+3$ we apply \Cref{red:matching-A}. By the choice of $M$, we have that the rule would be applied at most once for each $H_P$. Because $\sigma_H$ can be constructed in linear time by \Cref{obs:caterpillar} and the greedy procedure can be done in linear time, we conclude that the overall running time is $\Oh(n)$. 

Finally, we note that because of the choice of the edges added by the rule, the rule cannot be applied a second time for the graph obtained from each $H_P$.
This completes the proof.
\end{proof}

We use similar arguments to deal with \Cref{red:paths-A} and \Cref{red:paths-B}.

\begin{lemma}\label{lem:five-rt}
There is an algorithm running in $\Oh(n)$ time that either exhaustively applies 
\Cref{red:paths-A} and \Cref{red:paths-B} or correctly concludes that $(G,V_1,V_2,V_3,k)$ is a no-instance.
\end{lemma}

\begin{proof}
By symmetry, it is sufficient to prove the lemma for \Cref{red:paths-A}. As in the previous lemma, we find inclusion maximal subgraphs $H$ satisfying conditions (i)--(v) of the rule and apply the rule for them simultaneously.

Suppose that  \Cref{red:paths-A} can be applied for $H=G[A]$ with the corresponding vertices $u_1,u_2,v_1,v_2,w$. Since $H$ admits a $2$-layer drawing $\sigma_H=(\sigma_1^H,\sigma_2^H)$ without crossings respecting $(V_1\cap V(H),V_2\cap V(H))$ such that (a) for every $x\in V_1$, $\sigma_1^H(u_1)\leq \sigma_1^H(x)\leq \sigma_1^H(u_2)$ and (b) for every $x\in V_2$, $\sigma_2^H(v_1)\leq \sigma_2^H(x)\leq \sigma_2^H(v_2)$,  $H$ is a caterpillar by \Cref{obs:caterpillar} with  a backbone $P$ that has one end-vertex in $\{u_1,v_1\}$, the other in $\{u_2,v_2\}$, and does not contain $u_1v_1,u_2v_2$. As before, we call such a backbone \emph{proper}.
We find the potential inclusion maximal backbones and construct $H$ with them. Again, we say that
 $H$ is \emph{backbone maximal} if the proper backbone of $H$ is inclusion maximal.
Note that since $H$ has a matching of size $4k+3$, $|V(P)|\geq 4k+3$. 

Consider $G'=G[V_1\cup V_2]$. For $i\in\{1,2\}$, we find the set $U_i\subseteq V_i$ of vertices that have two non-pendent neighbors in $G'$. 
We construct $G[U_1\cup U_2]$ and find all connected components of size at least $4k+1$ in linear time using the depth-first search~\cite{HopcroftT73}. Each connected component is either a cycle or a path. If there is a component $C$ that is a cycle with at least $4k+1$ vertices, then because $G'$ is bipartite, $C$ has at least  $4k+2\geq 2k+4$ vertices and 
$(G,V_1,V_2,V_3,k)$ is a no-instance by~\Cref{obs:cycles}. In this case, we stop. Otherwise, each connected component is a path. For such a path $P$, each end-vertex of $P$ has a neighbor of degree two that is not in $P$ by definition. If these neighbors are the same, we conclude that  $(G,V_1,V_2,V_3,k)$ is a no-instance by  \Cref{obs:cycles} and stop. 
Assume that this is not the case. We modify each $P$ by appending the corresponding neighbors of the end-vertices. This way we obtain in $\Oh(n)$ time the family of edge-disjoint path $\mathcal{P}$ with the property that every inclusion maximal proper backbone is a subpath of one of the paths of the family. 

Up to now, our analysis did not include the vertices of $V_3$. Recall that we are looking for $H$ together with $u_1,u_2,v_1,v_2,w$ such that $H$ is separated by 
$\{u_1,u_2,v_1,v_2,w\}$. To find inclusion maximal backbones, we consider subpaths of the paths of $\mathcal{P}$.

Consider $P\in\mathcal{P}$. For a vertex $v\in V(P)$, let $R(v)$ be the sets of pendent neighbors of $v\notin V(P)$ in $G'$. Note that these sets can be constructed in linear time. 
Let $\{s_0,\ldots,s_\ell\}=V(P)\cap V_2$ and assume that the vertices are numbered in the path order. For $i\in\{1,\ldots,\ell\}$, denote by $t_i$ the common neighbor of $s_{i-1}$ and $s_i$ in $P$. We find inclusion maximal $(s_i,s_j)$-subpaths $Q_P$ of $P$ for $0\leq i\leq j\leq\ell$ such that there is $w\in V_3$ that
\begin{itemize}
\item[(i)] there is $z\in \{s_{i}\}\cup R(t_{i+1})$ such that $w$ is the unique neighbor of $z$ in $V_3$,
\item[(ii)] there is  $z\in \{s_{j}\}\cup R(t_{j})$ such that $w$ is the unique neighbor of $z$ in $V_3$,
\item[(iii)] for any $z\in\{s_i,\ldots,s_j\}\cup\bigcup_{p=i+1}^jR(t_p)$, $z$ either has no neighbors in $V_3$ or is adjacent to $w$. 
\end{itemize}
All these paths $Q_P$ for all $P\in\mathcal{P}$ can be constructed in $\Oh(n)$ time by tracing each path $P$ in the path order. 
For each $P\in\mathcal{P}$, we find the paths $Q_P$ with at least $4k+1$ vertices each and denote by $\mathcal{Q}_P$ the family of all these paths. 
We use the path from $\mathcal{Q}_P$ for finding the inclusion maximal proper backbones.
However, in some cases, we have to these paths.

Consider a path $Q\in\mathcal{Q}_P$ for some $P\in \mathcal{P}$ and assume that $Q$ is the $(s_i,s_j)$-path for $0\leq i<j\leq\ell$. We construct $Q'$ together with special vertices $v_1$ and $v_2$ as follows. 

Suppose that $i>0$ and there is $i'\in \{0,\ldots,i-1\}$ such that there exists $v_1\in \{s_{i'}\}\cup R(t_{i'+1})$ adjacent to $w$ (and some other vertex of $V_3$) but for any $i''\in\{i'+1,i-1\}$, the vertices of $\{s_{i''}\}\cup R(t_{i''+1})$ are not adjacent to any vertex of $V_3$. Then we extend $Q$ by adding $v_1$, $v_1t_{i'+1}$ and the $(t_{i'+1},s_i)$ subpath of $P$. 
Assume now that this is not the case and for each $i'\in\{0,i-1\}$, the vertices of $\{s_{i'}\}\cup R(t_{i'+1})$ are not adjacent to any vertex of $V_3$. 
Then if the first vertex $z$ of $P$ is in $V_1$ and there is $v_1\in V_2\setminus V(P)$ adjacent to both $w$ and $z$, we add $v_1$, $v_1z$, and the $(z,s_i)$-subpath.
In all other cases,   we set $v_1$ to be the vertex of  $\{s_{i}\}\cup R(t_{i+1})$ adjacent to $w$, delete $s_i$, and add $v_1$ and $v_1t_{i+1}$.

Symmetrically, if $j<\ell$ and there is $j'\in \{j+1,\ldots,\ell\}$ such that there exists $v_2\in \{s_{j'}\}\cup R(t_{j'})$ adjacent to $w$ (and some other vertex of $V_3$) but for any $j''\in\{j+1,j'-1\}$, the vertices of $\{s_{j'}\}\cup R(t_{j''})$ are not adjacent to any vertex of $V_3$ then we extend $Q$ by adding $v_2$, $v_2t_{j''}$ and the $(t_{i''},s_j)$ subpath of $P$. 
Furthermore, if this is not the case but for each $j'\in\{j+1,\ell\}$, the vertices of $\{s_{j'}\}\cup R(t_{j'})$ are not adjacent to any vertex of $V_3$ and it holds that the last vertex $z$ of $P$ is in $V_1$ and there is $v_2\in V_2\setminus V(P)$ adjacent to both $w$ and $z$, we add $v_2$, $v_2z$, and the $(s_j,z)$-subpath.
In all other cases, we set $v_2$ to be the vertex of  $\{s_{j}\}\cup R(t_{j})$ adjacent to $w$, delete $s_j$, and add $v_2$ and $v_2t_{j}$.

Observe that by the above construction, it may happen that $v_1=v_2$. However, in this case, we obtain a cycle in $G'$ with at least $4k+2$ vertices and by \Cref{obs:cycles}, conclude that $(G,V_1,V_2,V_3,k)$ is a no-instance. From now on, we assume that this is not the case and we obtained a path. We denote it by  $Q'$. Observe that all such paths can be constructed from the paths $Q\in \mathcal{Q}_P$ in linear time by tracing $P$.

For given $Q$ and $Q'$ together with two special vertices $v_1$ and $v_2$, we construct $H_Q$ together with $u_1,u_2,v_1,v_2$ as follows. Initially, $H_Q:=Q'$.
\begin{itemize}
\item For each internal vertex $z$ of $Q'$ which is an internal vertex of $Q$, we add $z$ and the vertices of $R(z)$ together with incident edges.
\item If there is a vertex of $V_1\setminus V(Q')$ adjacent to $v_1$ then select $u_1$ be such a vertex. Otherwise, set $u_1$ be the neighbor of $v_1$ in $Q'$.
\item If there is a vertex of $V_1\setminus V(Q')$ adjacent to $v_2$ then select $u_2$ be such a vertex. Otherwise, set $u_2$ be the neighbor of $v_2$ in $Q'$.
\end{itemize}

 If we obtain that $u_1=u_2$ then by by \Cref{obs:cycles}, $(G,V_1,V_2,V_3,k)$ is a no-instance and we stop. Otherwise, $H_Q$ is a caterpillar with the properties: 
\begin{itemize}
\item[(i)] $H_Q$ is a connected graph such that for $A=V(H_Q)\cup\{w\}$ and $B=(V(G)\setminus A)\cup\{u_1,u_2,v_1,v_2,w\}$, $(A,B)$ is a separation of $G$ with $A\cap B=\{u_1,u_2,v_1,v_2,w\}$, $A\cap V_3=\{w_3\}$, $u_1,u_2\in V_1$, and $v_1,v_2\in V_2$.
\item[(ii)]  $u_1v_1,u_2v_2,v_1w,v_2w\in E(G)$, 
\item[(iii)] $H_Q$ has a $2$-layer drawing $\sigma_H=(\sigma_1^H,\sigma_2^H)$ without crossings respecting $(V_1\cap V(H_Q),V_2\cap V(H_Q))$, where $\sigma_i^H$ is an order of $V_i\cap V(H_Q)$ for $i\in\{1,2\}$, such that 
(a) for every $x\in V_1\cap V(H_Q)$, $\sigma_1^H(u_1)\leq \sigma_1^H(x)\leq \sigma_1^H(u_2)$ and (b) for every $x\in V_2\cap V(H)$, $\sigma_2^H(v_1)\leq \sigma_2^H(x)\leq \sigma_2^H(v_2)$, 
\item[(iv)] $|V(H_Q)|\geq 4k+1$.
\end{itemize}
Furthermore, the construction of all $H_Q$ for $Q\in \mathcal{Q}_P$ and $P\in\mathcal{P}$ can be done in $\Oh(n)$.

We apply \Cref{red:paths-A} for the constructed $H_Q$ given together with the vertices $u_1,u_2,v_1,v_2,w$. Similarly to the proof of \Cref{lem:four-rt}, 
we greedily find a matching $M$ with at least $4k+3$ to ensure that the rule would delete the maximum number of vertices if the required matching exists. 
We use  \Cref{obs:caterpillar} to find the unique (up to the permutations of pendent vertices)  $2$-layer drawing $\sigma_H=(\sigma_1^H,\sigma_2^H)$ without crossings with $\sigma_1^H(u_1)<\sigma_1^H(u_2)$ and $\sigma_2^H(v_1)<\sigma_2^H(v_2)$. 
Recall that $w$ is adjacent to $v_1$ and $v_2$. We set $M:=\{x_1y_1\}$ where $x_1=u_1$ and $y_1=v_1$. 
Then we try to add $2k+1$ edges $x_iy_i$ with $x_i\in V_1\cap V(H_Q)$ and $y_i\in V_2\cap V(H_Q)\cap N_G(w)$ for $i\in\{2,\ldots,2k+2\}$. We greedily select first with respect to $\sigma_2^H$ vertex $y_i\in N_G(w)$ that is not incident to any edge of the already constructed matching that has a neighbor that is not saturated in $M$. Then among the unsaturated neighbors of $y_i$, we select $x_i$ which is first with respect to $\sigma_1^H$. If we succeed with the first $2k+2$ edges we proceed with selecting $2k+1$ from the other side. We add $x_{4k+3}y_{4k+3}$ with $x_{4k+3}=u_2$ and $y_{4k+3}=v_2$ if possible. Then we add $2k$ edges $x_{4k+3-i}y_{4k+3-i}$ for $i\in\{1,\ldots2k\}$ following the reversals of $\sigma_1^H$ and $\sigma_2^H$. If this greedy choice gives a matching $M$ of size $4k+3$ we apply \Cref{red:paths-A}. By the choice of $M$, we have that the rule would be applied at most once for each $H_Q$. Because $\sigma_H$ can be constructed in linear time by \Cref{obs:caterpillar} and the greedy procedure can be done in linear time, we conclude that the overall running time is $\Oh(n)$. 

Notice that the rule does not create new paths of length two with their end-vertices in $V_1$ and $V_3$. Therefore, we cannot get a possibility to apply \Cref{red:paths-A} again. Also, after applying \Cref{red:paths-B} we do not get new possibilities to apply \Cref{red:paths-A}. 
This completes the proof. 
\end{proof}


We proved that \Cref{red:comp}--\Cref{red:paths-B} can be exhaustively applied in $\Oh(n)$ time. However, for \Cref{red:matching-A}--\Cref{red:paths-B}, our analysis heavily relied on \Cref{obs:caterpillar}. This cannot be done for \Cref{red:paths-C} and with this rule, we show that it can be implemented in $k^{\Oh(1)}\cdot n$ time. We use the following auxiliary results.

Suppose that \Cref{red:paths-C} can be applied for $H$ with the corresponding vertices $u_1,v_1,w_1,u_2,v_2,w_2$ satisfying properties (i)--(iv) of the rule. In particular,
$H$ has a family of vertex disjoint paths $P=\{x_1y_1z_1,\ldots,x_{4k+3}y_{4k+3}z_{4k+3}\}$ of size $4k+3$ with $u_1v_1w_1,u_2v_2w_2\in P$ and 
it holds  that $x_1=u_1$, $y_1=v_1$, $z_1=w_1$, $x_{4k+3}=u_2$, $y_{4k+3}=v_2$, $z_{4k+3}=w_2$,
$\sigma_1^H(x_1)<\dots<\sigma_1^H(x_{4k+3})$, 
$\sigma_2^H(y_1)<\dots<\sigma_2^H(y_{4k+3})$, and 
$\sigma_3^H(z_1)<\dots<\sigma_3^H(z_{4k+3})$. 
We say that the vertex $y_{2k+2}$ is a \emph{center} of $H$.

\begin{lemma}\label{lem:struct-paths-C}
Suppose that \Cref{red:comp}--\Cref{red:paths-B} are not applicable for $(G,V_1,V_2,V_3,k)$ and suppose that \Cref{red:paths-C} can be applied for $H$ centered in $v \in V_2$ where $H$ is an inclusion minimal subgraph with this property.  Then $|V(H)|\leq 2^{13}(k+2)^7$ and all vertices of $H$ are at distance at most $2^{11}(k+2)^6$ from $v$.
\end{lemma}

\begin{proof}
Suppose that \Cref{red:paths-C} can be applied for $H$  centered in $v$ where $H$ is chosen to be inclusion minimal. 
Consider the corresponding vertices $u_1,v_1,w_1$ and $u_2,v_2,w_2$, the $3$-layer drawing  $\sigma_H=(\sigma_1^H,\sigma_2^H,\sigma_2^H)$ without crossings respecting $(V_1\cap V(H),V_2\cap V(H),V_3\cap V(H))$, and the 
family $P=\{x_1y_1z_1,\ldots,x_{4k+3}y_{4k+3}z_{4k+3}\}$ of vertex disjoint paths of size $4k+3$ with $u_1v_1w_1,u_2v_2w_2\in P$ and 
such   that $x_1=u_1$, $y_1=v_1$, $z_1=w_1$, $x_{4k+3}=u_2$, $y_{4k+3}=v_2$, $z_{4k+3}=w_2$,
$\sigma_1^H(x_1)<\dots<\sigma_1^H(x_{4k+3})$, 
$\sigma_2^H(y_1)<\dots<\sigma_2^H(y_{4k+3})$, and 
$\sigma_3^H(z_1)<\dots<\sigma_3^H(z_{4k+3})$. Recall that $v=y_{2k+2}$. 

Applying \Cref{lem:safe-three} for $H$ and $\sigma$, we obtain that  
\begin{equation*}
|[y_1,y_\ell]|\leq 4(4k+2)(k+3)(4k+2)(2k+3)[4(2k+3)(4k+2)+1].
\end{equation*}
Using \Cref{obs:deg-V-two}, we obtain that 
\begin{align*}
 |V(H)|\leq& |V(H)\cap V_2|+ 2(k+3)|V(H)\cap V_2|=(2k+7)|V(H)\cap V_2|\\ \leq& 
 4(2k+7)(4k+2)(k+3)(4k+2)(2k+3)[4(2k+3)(4k+2)+1]\\
 \leq& 2^{13}(k+2)^7.
  \end{align*}
  Similarly, we obtain that
  \begin{equation*} 
  [y_1,y_{2k+2}],[y_{2k+2},y_{4k+3}]\leq 4(2k+1)(k+3)(4k+2)(2k+3)[4(2k+3)(4k+2)+1]  
   \end{equation*}
  and, therefore, each vertex of $H$ is at the distance at most 
  \begin{equation*}
  8(2k+1)(k+3)(4k+2)(2k+3)[4(2k+3)(4k+2)+1]\leq 2^{11}(k+2)^6 
  \end{equation*}
   from $v$. This concludes the proof.
 \end{proof}
 
 \Cref{lem:struct-paths-C} gives the following idea for the implementation of \Cref{red:paths-C}. We consider vertices $v$ of $V_2$ that have neighbors in both $V_1$ and $V_2$, that is, they potentially could be  centers for $H$. Then we perform the breadth-first search from $v$ on depth $\Oh(k^7)$ and apply the rule in the ball. However, for this, we have to guarantee that the ball has bounded size. We obtain such a bound in the next lemma. 
 
\begin{lemma}\label{lem:dist} 
 Suppose that \Cref{red:comp}--\Cref{red:paths-B} are not applicable for a yes-instance $(G,V_1,V_2,V_3,k)$. Then for any $v\in V_2$ and any $\ell\geq 1$, the number of vertices of $G$ at distance at most $\ell$ from $v$ is at most $2^{12}(\ell+2k+2)(k+2)^6$.  
 \end{lemma}
  
 \begin{proof}
 Let $\sigma=(\sigma_1,\sigma_2,\sigma_3)$ be a $3$-layer drawing of $G$ with at most $k$ crossings. Let $D$ be the set of endpoints of crossing edges. Recall that $|D\cap V_2|\leq 2k$ and 
 $|D\cap(V_1\cup V_3)|\leq 2k$. 
 Consider $v\in V_2$ and let $r\in V_2$ be the last vertex with respect to $\sigma_2$ at distance at most $\ell$ from $v$. We prove the following claim.
 
\begin{claim}\label{cl:rightmost}
$|[v,r]|\leq 4(\ell+2k+2)(k+3)(4k+2)(2k+3)[4(2k+3)(4k+2)+1]+2k$
\end{claim} 

\begin{proof}
Let $R=\{y_1,\ldots,y_s\}\subseteq [v,r]$ be an inclusion maximal set of vertices is safe subintervals of $[v,r]$  such that each vertex of $R$ has neighbors in both $V_1$ and $V_3$ and all these neighbors are distinct, and assume that $\sigma_2(y_1)<\dots<\sigma_2(y_s)$. Notice that $s\leq \ell+1$. Let $I_1,\ldots,I_p$ be the maximal safe subintervals of $[v,r]$ containing the vertices of $R$ and let $s_i$ be the number of vertices of $R$ in $I_i$ for $i\in\{1,\ldots,p\}$. Note that $p\leq 2k+1$. 
By \Cref{lem:safe-three},
\begin{equation*}
|I_i|\leq  4(s_i+1)(k+3)(4k+2)(2k+3)[4(2k+3)(4k+2)+1],
\end{equation*} 
for every $i\in\{1,\ldots,p\}$
and 
\begin{align*}
\sum_{i=1}^p|I_i|\leq & 4(s+p)(k+3)(4k+2)(2k+3)[4(2k+3)(4k+2)+1]\\ \leq &4(\ell+2k+2)(k+3)(4k+2)(2k+3)[4(2k+3)(4k+2)+1].
\end{align*} 
Then 
\begin{equation*}
|[v,r]|\leq 4(\ell+2k+2)(k+3)(4k+2)(2k+3)[4(2k+3)(4k+2)+1]+2k 
\end{equation*}
as $|D\cap V_2|\leq 2k$. 
This proves the claim.
\end{proof}

Now let $r'\in V_2$ be the first vertex with respect to $\sigma_2$ at distance at most $\ell$ from $v$. By symmetry, we have that 
$|[r',v]|]\leq 4(\ell+2k+2)(k+3)(4k+2)(2k+3)[4(2k+3)(4k+2)+1]+2k$ and conclude that 
\begin{equation*}
|[r',r]|]\leq 8(\ell+2k+2)(k+3)(4k+2)(2k+3)[4(2k+3)(4k+2)+1]+2k
\end{equation*}
To prove the lemma, we observe that every vertex at distance at most $\ell$ from $v$ is in $N_G[[r',r]]$.
By \Cref{obs:deg-V-two} and because $|D\cap (V_1\cup V_2)|\leq k$,
\begin{align*}
|N_G[[r',r]]|\leq & |[r',r]|+2(k+3)|[r',r]|+2k=(2k+7)|[r',r]|+2k\\
\leq & (2k+7)[8(\ell+2k+2)(k+3)(4k+2)(2k+3)[4(2k+3)(4k+2)+1]+2k]+2k\\
\leq & 2^{12}(\ell+2k+2)(k+2)^6.
\end{align*}   
This proves the lemma.   
 \end{proof} 
 
 \begin{lemma}\label{lem:six-rt}
There is an algorithm running in $k^{\Oh(1)}\cdot n$ time that either exhaustively applies 
\Cref{red:paths-C} or correctly concludes that $(G,V_1,V_2,V_3,k)$ is a no-instance. Furthermore,
the algorithms simultaneously apply \Cref{red:pend}--\Cref{red:nice} and \Cref{red:paths-A} and \Cref{red:paths-B} if they could be used after applying \Cref{red:paths-C}.
\end{lemma} 

\begin{proof}
Let $(G,V_1,V_2,V_3,k)$ be an instance of \probThreeCr such that \Cref{red:comp}--\Cref{red:paths-B} cannot be applied. We consider all vertices $v\in V_2$ having neighbors in both $V_1$ and $V_3$ and try to apply \Cref{red:paths-C} for $H$ centered in $v$. For this, we do the following.

We perform BFS starting from $v$ for the depth upper bounded by  $2^{11}(k+2)^6$. If the implementation of BFS, we count the number of vertices visited by the algorithm. If BFS visits more than 
$2^{12}(2^{11}(k+2)^6+2k+2)(k+2)^6$ vertices in some step then we conclude that $(G,V_1,V_2,V_3,k)$ is a no-instance by \Cref{lem:dist}. We stop BFS and return the answer. Suppose that this is not the case and BFS produces the set of all vertices $U$ at distance at most  $2^{11}(k+2)^6$ from $v$ whose size is at most $2^{12}(2^{11}(k+2)^6+2k+2)(k+2)^6$.

By \Cref{lem:struct-paths-C}, if \Cref{red:paths-C} can be applied for some $H$ centered in $v$ then such a graph can be found in $G'=G[W]$. Because $|W|\leq 2^{12}(2^{11}(k+2)^6+2k+2)(k+2)^6$, we can in $k^{\Oh(1)}$ time find $H$ together with the vertices $u_1,v_2,w_1,u_2,v_2,w_2$, the $3$-layer drawing $\sigma_H=(\sigma_1^H,\sigma_2^H,\sigma_3^H)$ of $H$, and the family of vertex disjoint path $P=\{x_1y_1z_1,\ldots,x_{4k+3}y_{4k+3}z_{4k+3}\}$ of size $4k+3$ satisfying conditions (i)--(iv) of the rule with $v=y_{2k+2}$ if such a graph $H$ exists. 
For example, we can do it in a naive way as follows.\footnote{This can be done slightly faster using \Cref{obs:deg-V-two} to guess  $u_1,w_1,u_2,w_2$ for given $v_1$ and $v_2$ and then using the fact that the pathwidth of $H$ is at most $3$ if $H$ admits a $3$-layer drawing.} We can consider all 6-tuples  $(u_1,v_1,w_1,u_2,v_2,w_2)$ of vertices and then consider all separations $(A,B)$ $G'$ with $A\cap B=\{u_1,v_1,w_1,u_2,v_2,w_2\}$ where $G[A]$ is connected and conditions (i)--(iv) are fulfilled. It is straightforward to verify (i) and (ii), condition(iii) is verified by making use of \Cref{prop:linear-time} in linear time, and to verify (iv), we can additionally guess $x_{2i+2}$ and $z_{2k+2}$ and use the standard maximum flow algorithm of Ford and Fulkerson~\cite{FordF56} to find other paths in $P$. 
If the algorithm finds $H$, we apply \Cref{red:paths-C}. 

Suppose that the rule was applied and let $H'$ be the graph obtained from $H$. Then it can happen that some previous rules would be applicable for the obtained instance. However, applying  these rules is 
restricted by $H'$ and can be performed locally.

First, note that $z_{2k+1}$ or $z_{2k+3}$ may become pendent and, potentially, \Cref{red:pend} could be applied for them and their neighbors $y_{2k+1}$ and $y_{2k+3}$. 
Since the neighbors of $y_{2k+1}$ and $y_{2k+3}$ are in $H'$ and  \Cref{red:pend} could not be applied for them before, we can apply the rule in $\Oh(k)$ time.
Similarly, for $i\in\{2k+1,2k+3\}$, the vertex $y_i$ may become a vertex of degree two adjacent only to $x_{i}$ and $z_{i}$. 
Then \Cref{red:degtwo}
may be applicable and we can apply it in $\Oh(k)$ time. 
Further, we can obtain a new nice component for $u=x_{2k+1}$ or $u=x_{2k+3}$. 
Observe that all nice components of $G-u$ have their vertices in $H$. 
More precisely, for $u=x_{2k+1}$, it holds that 
for each $y\in V_2$ in a nice component, $\sigma_2(y_{2k})<\sigma_2(x)<\sigma_2(x_{2k+2})$ and 
for each $z\in V_3$ in a nice component, $\sigma_2(z_{2k})<\sigma_3(z)<\sigma_3(3_{2k+2})$. 
For $u=x_{2k+3}$, the structure is symmetric.
Thus, \Cref{red:nice} can be applied in polynomial in $k$ time in $H'$ (in fact, in $\Oh(k^2)$ time).
Notice that because \Cref{red:paths-C} does not create paths of length two with their end-vertices in $V_1$ and $V_2$, the rule does not create new possibilities to apply \Cref{red:paths-A} and \Cref{red:paths-B}. Therefore, the total clean-up time is polynomial in $k$.
  
Notice that for each $v\in V_2$, we execute the above procedure only once. If we modify the graph and, in particular, make $v$ nonadjacent to any vertex of $V_3$ 
then we make it impossible to apply the rule for $v$. If we decide that $v$  cannot serve as a center then   because \Cref{red:paths-C} does not create new paths of length two with their end-vertices in $V_1$ and $V_2$, $v$ cannot be a center of some $H$ obtained after applying the rule for some other centers.
This means that the total running time is $k^{\Oh(1)}\cdot n$ and concludes the proof.
\end{proof}

The algorithms from \Cref{lem:four-rt}, \Cref{lem:five-rt}, and \Cref{lem:six-rt} either exhaustively apply the corresponding rules or correctly report that the input instance is a no-instance of \probThreeCr. In the latter case, we immediately stop and return a trivial no-instance. Otherwise, summarizing the time of initial preprocessing, our observations about \Cref{red:comp}--\Cref{red:degtwo}, 
\Cref{lem:nice-rt}--\Cref{lem:five-rt}, and \Cref{lem:six-rt}, we conclude that the total running time of our kernelization algorithm is $k^{\Oh(1)}\cdot n$. This completes the proof of \Cref{thm:kernel}.
	

\section{Lower bounds}\label{sec:lb}
In this section, we show computational lower bounds stated in  \Cref{thm:nokern} and \Cref{thm:ETHlb}. Both bounds are proved via constructing a polynomial parameter transformation of \probDF introduced by Bodlaender, Thomass{\'{e}}, and Yeo~\cite{BodlaenderTY11}.

Given a string $s$ over an alphabet $\Sigma$, a \emph{factor} is a nontrivial substring of $s$ whose first and last symbols are the same. Then \probDF asks, given a string $s$ over an alphabet $\Sigma$ of size $k$, whether there are $k$ disjoint factors with distinct first (and last) symbols. It was proved in~\cite{BodlaenderTY11} that \probDF has no polynomial kernel when parameterized by $k$ unless  $\classNP\subseteq\classCoNP/{\rm poly}$. Besides this,  the authors proved that \probDF is \classNP-complete. The proof was done by a reduction from the 3-\textsc{SAT} problem. We stress that, given an instance of 3-\textsc{SAT} with $n$ variables and $m$ clauses, the reduction constructs an instance of \probDF with the alphabet of size $\Oh(n+m)$. Thus, the \classNP-completeness proof implies a computational  lower bound up to ETH which we need. We summarize these results in the following proposition. 

\begin{proposition}[\cite{BodlaenderTY11}]\label{prop:DF}
 \probDF does not admit a polynomial kernel when parameterized by $k$ unless  $\classNP\subseteq\classCoNP/{\rm poly}$. Furthermore, the problem cannot be solved in $2^{o(k)}\cdot n^{\Oh(1)}$ time on strings of length $n$ unless ETH fails.
\end{proposition}

Our lower bounds for \probHCr are proved by almost the same reductions. However, note that  
to prove \Cref{thm:nokern}, it is sufficient to obtain any polynomial parameter transformation. For the ETH lower bound, we have to be more careful and avoid blowing up the parameter, and to do it, we use an extra layer in the drawing.  That is the reason why the kernalization lower bound is established for \probFourCr and the ETH bound is obtained for \probFiveCr.

We need the following observation.

\begin{observation}\label{obs:encoding}
Given an alphabet $\Sigma$ of size $k$, the symbols of $\Sigma$ can be encoded by distinct binary strings of length $2\lceil \log k\rceil +2$ such that each binary string  $s$ encoding a symbol contains an equal number of zeros and ones and the reversal of $s$ is distinct from any other string. Furthermore, such an encoding can be constructed in polynomial time. 
\end{observation}

\begin{proof}
To see this, note that $\binom{2\lceil \log k\rceil}{\lceil \log k\rceil}\geq 2^{\lceil \log k\rceil}$, that is, $k$ symbols can be encoded by distinct strings of length $2\lceil \log k\rceil$ with exactly $\lceil \log k\rceil$ zeros and the same number of ones. Then to prevent the  reversal of a string to be equal to another string, we can add zero in the beginning and append one in the end.
\end{proof}

Suppose that $(s,k)$ is an instance of \probDF where $s$ is a string over $\Sigma=\{s_1,\ldots,s_k\}$ of size $k\geq 2$.
Using \Cref{obs:encoding}, we assume that the symbols of $\Sigma$ are encoded by  binary strings of length $2\ell=2\lceil \log k\rceil +2$ such that the strings and their reversals are distinct and each symbol is encoded by 
a string containing equal number $\ell$ of zeros and ones. 
In the following subsection, we will describe the basic gadgets used in our reductions. 

\subsection{Basic gadgets}\label{sec:basic}
First, we construct gadgets that are used to encode zeros and ones in the encoding of the symbols. The vertices of these gadgets are placed in the sets $V_1,V_2,V_3$ of the $4$ or $5$-partite graphs constructed by the reductions.

\begin{figure}[ht]
\centering
\scalebox{0.6}{
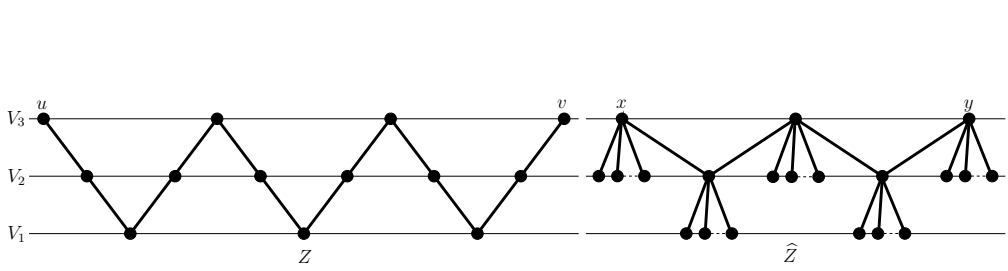}
\caption{Construction of $Z$ and $\widehat{Z}$.}
\label{fig:Z}
\end{figure}

For zero, we construct the graph $Z$ and the ``complementing'' graph $\widehat{Z}$ as follows (see~\Cref{fig:Z}).
\begin{itemize}
\item Construction of $Z$:
\begin{itemize}
\item construct two \emph{root} vertices $u$ and $v$ and join them by a $(u,v)$-path $w_0\cdots w_{12}$,
\item put  $w_2,w_6,w_{10}$ in $V_1$, $w_1,w_3,w_5,w_7,w_9,w_{11}$ in $V_2$, and $w_0,w_4,w_8,w_{12}$ in $V_3$.
\end{itemize}
\item Construction of $\widehat{Z}$:
\begin{itemize}
\item construct two \emph{root} vertices $x$ and $y$ and join them by a $(x,y)$-path $w_0w_1w_2w_3w_{4}$,
\item put  $w_1,w_{3}$ in $V_2$ and  $w_0,w_2,w_4$ in $V_3$, 
\item for $z\in \{w_1,w_3\}$, construct a set of $p=40k^3$ vertices, make them adjacent to $z$, and put then in $V_1$,
\item for $z\in \{w_0,w_2,w_4\}$, construct a set of $p$ vertices, make them adjacent to $z$, and put then in $V_2$.
\end{itemize}
\end{itemize}

\begin{figure}[ht]
\centering
\scalebox{0.6}{
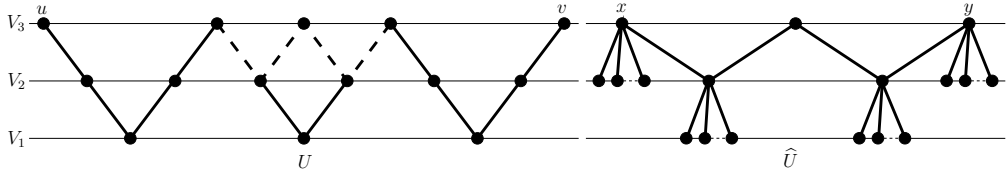}
\caption{Construction of $U$ and $\widehat{U}$; the $V_3$-no-gap path is shown by dashed lines.}
\label{fig:U}
\end{figure}

For one, the graph $U$ and the ``complementing'' graph $\widehat{U}$ are constructed as follows (see~\Cref{fig:U}).
\begin{itemize}
\item Construction of $U$:
\begin{itemize}
\item construct two \emph{root} vertices $u$ and $v$ and join them by a $(u,v)$-path $w_0\cdots w_{12}$,
\item put  $w_2,w_6,w_{10}$ in $V_1$, $w_1,w_3,w_5,w_7,w_9,w_{11}$ in $V_2$, and $w_0,w_4,w_8,w_{12}$ in $V_3$,
\item construct a vertex $w$, make it adjacent to $w_5$ and $w_7$, and put it in $V_3$.
\end{itemize}
\item Construction of $\widehat{U}$:
\begin{itemize}
\item construct two \emph{root} vertices $x$ and $y$ and join them by a $(x,y)$-path $w_0w_1w_2w_3w_{4}$,
\item put  $w_1,w_{3}$ in $V_2$ and  $w_0,w_2,w_4$ in $V_3$, 
\item for $z\in \{w_1,w_3\}$, construct a set of $p$ vertices, make them adjacent to $z$, and put then in $V_1$,
\item for $z\in \{w_0,w_4\}$, construct a set of $p$ vertices, make them adjacent to $z$, and put then in $V_2$.
\end{itemize}
\end{itemize}

\begin{figure}[ht]
\centering
\scalebox{0.6}{
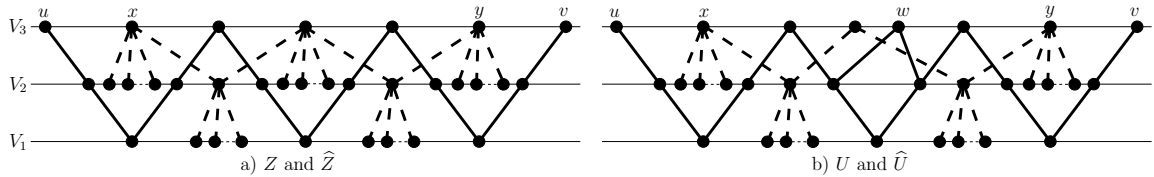}
\caption{$3$-layer drawing of (a) $Z$ and $\widehat{Z}$ and (b) $U$ and $\widehat{U}$; the edges of $\widehat{Z}$ and $\widehat{U}$ are shown by dashed lines.}
\label{fig:UZ}
\end{figure}

The choice of $p$ will be explained later; now we just observe that $p$ is bigger than the value of the parameter in the constructed instances of \probFourCr and  \probFiveCr. 
The unique (up to the reversals of the orders  and permutations of pendent vertices) $3$-layer drawings of $Z$, $\widehat{Z}$, $U$, and $\widehat{U}$ without crossings are shown in \Cref{fig:Z} and \Cref{fig:U}; we call these drawings \emph{canonical}.
To give an intuition of how these gadgets are going to be used in the reductions, we show in \Cref{fig:UZ} types of $3$-layer drawings that will occur.

Our next aim is the construction of gadgets for encoding the symbols $s_1,\ldots,s_k$ of $\Sigma$. They are constructed by making use of the above gadgets for zeros and ones and two auxiliary gadgets $F$ and $C$ that are constructed as follows.

\begin{figure}[ht]
\centering
\scalebox{0.7}{
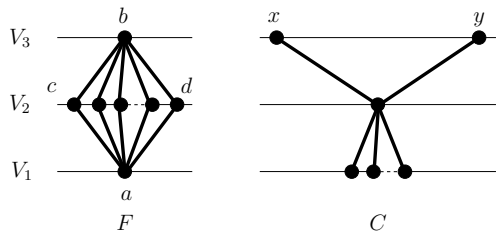}
\caption{Construction of $F$ and $C$.}
\label{fig:FC}
\end{figure}

\begin{itemize}
\item Construction of the \emph{filling} gadget $F$:
\begin{itemize}
\item construct two vertices $a$ and $b$, and put $a$ and $b$ in $V_1$ and $V_3$, respectively.  
\item construct a set $X$ of $p$ vertices, select two vertices $c,d\in X$, make the vertices of $X$ adjacent to $a$ and $b$, and include $X$ in $V_2$,
\end{itemize}
\item Construction of the \emph{connection} gadget  $C$:
\begin{itemize}
\item construct two \emph{root} vertices $x$ and $y$ and put them in $V_3$,
\item construct a vertex $z$, make it adjacent to $x$ and $y$, and put $z$ in $V_2$,
\item construct a set of $p$ vertices, make them adjacent to $z$, and put them in $V_1$.
\end{itemize}
\end{itemize}

To encode $s_i$ for $i\in\{1,\ldots,k\}$, we assume that $s_i=\delta_1\dots\delta_{2\ell}$ where $\delta_j\in\{0,1\}$ for each $j\in \{1,\ldots,2\ell\}$, that is, we identify the symbols and their binary encodings. Then we construct two graphs $S_i$ and $\widehat{S}_i$ by making using the gadgets constructed above. Informally, $S_i$ is constructed by concatenating copies of $Z$ and $U$ constructed for each $j\in\{1,\ldots,2\ell\}$ and adding two ``boundary'' copies of $F$. Similarly, $\widehat{S}_i$ is constructed by connecting copies of $\widehat{Z}$ and $\widehat{U}$ but these copies are connected via copies of $C$. The construction of $S_i$ and $\widehat{S}_i$ for $s_i=0101$ is shown in \Cref{fig:Si}.

\begin{figure}[ht]
\centering
\scalebox{0.6}{
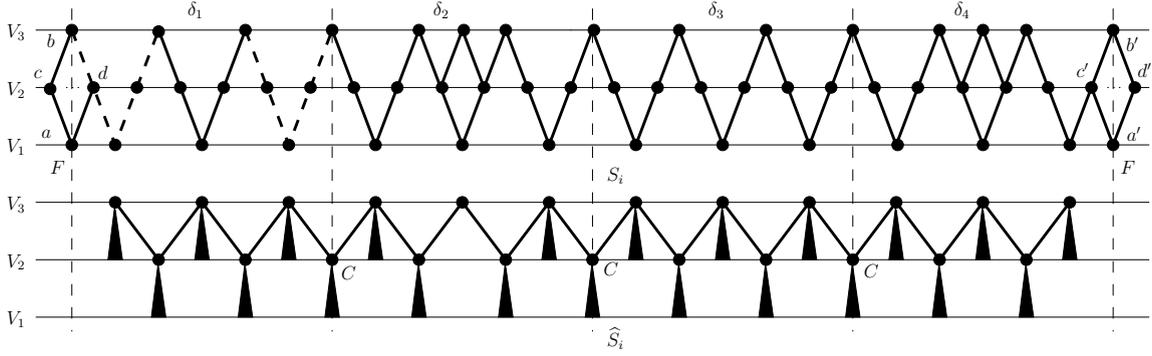}
\caption{Construction of $S_i$ and $\widehat{S}_i$ for $s_i=0101$; black triangles are used to show sets of pendent vertices of size $p$ in $\widehat{S}_i$, and examples of  $V_3$-gap paths in $S_i$ are shown by thick dashed lines.}
\label{fig:Si}
\end{figure}

\begin{itemize}
\item Construction of $S_i$:
\begin{itemize}
\item construct a copy of $F$,
\item if $\delta_1=0$ then construct a copy of $Z$ and construct a copy of $U$ if $\delta_1=1$, then identify the root $u$ of the gadget with the vertex $b$ of $F$ and the vertex adjacent to $u$ with $d$,  
\item for each $j\in\{2,\ldots,2\ell\}$, construct a copy of $Z$ if $\delta=0$ and construct a copy of $U$ if $\delta_1=1$, then identify the root $u$ of the gadget with the root $v$ of the gadget constructed for $\delta_{j-1}$,
\item construct a copy of $F$,  then identify the vertex $b$ with the root $v$ of the gadget constructed for $\delta_{2\ell}$ and identify the vertex $c$ with the vertex adjacent to $v$.
\end{itemize}
\item Construction of $\widehat{S}_i$:
\begin{itemize}
\item if $\delta_1=0$ then construct a copy of $\widehat{Z}$ and construct a copy of $\widehat{U}$ if $\delta_1=1$,
\item for each $j\in\{2,\ldots,2\ell\}$, construct a copy of $C$ and then construct a copy of $Z$ if $\delta=0$ and construct a copy of $U$ if $\delta_1=1$,  identify the root $x$ of the copy of $C$ with the root $y$ of the gadget constructed for $\delta_{j-1}$ and identify the root $y$ of $C$ with the root of the new copy of either  $\widehat{Z}$ or $\widehat{U}$.
\end{itemize}
\end{itemize} 

The unique (up to the reversals of the orders and permutations of twin vertices) $3$-layer drawings of each $S_i$ and  $\widehat{S}_i$  without crossings are shown in \Cref{fig:Si}; these drawings are obtained from the canonical drawing of the gadgets $Z$, $\widehat{Z}$, $U$, and $\widehat{U}$, and we call such drawings of $S_i$ and $\widehat{S}_i$ \emph{canonical}. We use $a'$, $b'$, $c'$, and $d'$ to denote the vertices $a$, $b$, $c$, and $d$, respectively, of the second copy of $F$ in $S_i$. 

For $i,j\in \{1,\ldots,2\ell\}$, we say that a $3$-layer drawing $\sigma=(\sigma_1,\sigma_2,\sigma_3)$ of the disjoint union of $S_i$ and $\widehat{S}_j$ is  \emph{restrictive} if for every vertex $w\in V_3\cap V(\widehat{S}_j)$, $\sigma_3(b)<\sigma_3(w)<\sigma_3(b')$ where $b$ and $b'$, respectively, are vertices of the copies of $F$ in $V_3\cap V(S_i)$.   The properties of the restrictive drawings are summarized in the following lemma. 

\begin{lemma}\label{lem:restr}
Let $i,j\in\{1,\ldots,2\ell\}$. Then the following holds for the disjoint union $H$ of  $S_i$ and $\widehat{S}_j$:
\begin{itemize}
\item[(i)] Every restrictive $3$-layer drawing of $H$ has at least $14\ell-2$ crossings. 
\item[(ii)] A restrictive $3$-layer drawing of $H$ with exactly $14\ell-2$ crossings exists if and only if $i=j$.
\item[(iii)] If $\sigma=(\sigma_1,\sigma_2,\sigma_3)$ is a restrictive $3$-layer drawing of $H$ with $14\ell-2$ crossings
 then the induced drawings of $S_i$ and $\widehat{S}_j$ are canonical, and, furthermore, if $\sigma_{3}(b)<\sigma_3(b')$ then
 at least one vertex of $V_3\cap V(\widehat{S}_j)$ is drawn immediately after $b$ and at least one vertex is drawn immediately before $b'$, and if 
 $\sigma_{3}(b')<\sigma_3(b)$ then
 at least one vertex of $V_3\cap V(\widehat{S}_j)$ is drawn immediately after $b'$ and at least one vertex is drawn immediately before $b$.
\end{itemize}
\end{lemma} 

\begin{proof}
Let $i,j\in\{1,\ldots,2\ell\}$ and consider $H=S_i\uplus \widehat{S}_j$. 

\begin{figure}[ht]
\centering
\scalebox{0.6}{
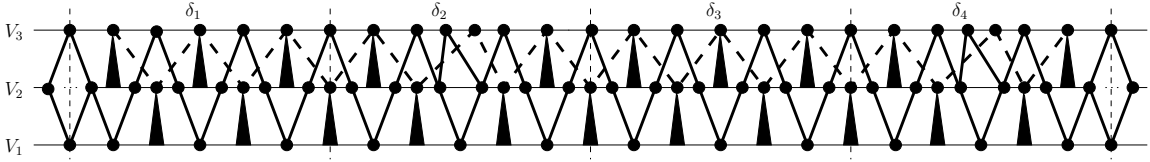}
\caption{The restrictive $3$-layer drawing of $H=S_i\uplus \widehat{S}_i$ constructed for $s_i=0101$; black triangles are used to show sets of pendent vertices of size $p$ in $\widehat{S}_i$, the edges of $\widehat{S}_i$ are shown by thick dashed lines.}
\label{fig:Sir}
\end{figure}

First, we observe that if $i=j$ then a restrictive $3$-layer drawing with exactly $14\ell-2$ crossings exists. We consider the canonical $3$-layer drawing $\sigma=(\sigma_1,\sigma_2,\sigma_3)$ of $S_i$ and then insert the vertices of $\widehat{S}_i$ between the vertices of $S_i$ as it is shown in \Cref{fig:Sir} following the idea demonstrated in \Cref{fig:UZ}. By the direct computations, we have that the number of crossings for this drawing is exactly $14\ell-2$ because $s_i$ has length $2\ell$ and contains exactly $\ell$ ones. 

Now we prove that for any restricted $3$-layer drawing of $H$, the number of crossings is at least $14\ell-2$ and the number of crossings is at least $14\ell-1$ if $i\neq j$. 

Consider  a restrictive $3$-layer drawing $\sigma=(\sigma_1,\sigma_2,\sigma_3)$ of $H$ whose number of crossing is minimum. If the number of crossings is at least 
$p=40k^3>14\ell-2$ then the claim holds. Assume that the number of crossings is at most $p-1$. 

By the definition of restrictive drawings, for every vertex $w\in V_3\cap V(\widehat{S}_j)$, $\sigma_3(b)<\sigma_3(w)<\sigma_3(b')$. Because $F$ has $p$ $(a,b)$-paths of length two, we additionally obtain that   for every vertex $w\in V_1\cap V(\widehat{S}_j)$, $\sigma_1(a)<\sigma_1(w)<\sigma_1(a')$. Otherwise, because $\widehat{S}_j$ is connected, each of the paths in one of the copies of $F$ in $S_i$ would have a crossing edge. 

For $t\in\{2,3\}$, we say that the vertices of $\widehat{S}_j$ in $V_t$ incident to sets of $p$ pendent vertices are \emph{$V_t$-big} and we say that a vertex of degree two in $V_3\cap V(\widehat{S}_j)$ is \emph{$V_3$-small}. We need the following observation.

\begin{claim}\label{cl:crosspend}
Let $u$ be a pendent vertex adjacent to a $V_t$-big vertex $v$ for some $t\in\{2,3\}$. Then $uv$ does not cross any edge in $\sigma$.
\end{claim}
 
\begin{proof}[Proof of~\Cref{cl:crosspend}]
Because $v$ is a $V_t$-big vertex, $v$ is a adjacent to $p$ pendent vertices $U$. Since the number of crossings is at most $p-1$, there is $w\in U$ such that $wv$ does not cross any edge. If $uv$ is crossing some other edge, we can modify the drawing by putting $u$ immediately after $w$. This would decrease the number of crossings contradicting the choice of $\sigma$. This proves the claim.
\end{proof}

Given an induced path $P=w_1w_2w_3w_4w_5$ on five vertices in $S_i$, we say that $P$ is a \emph{$V_3$-gap path}  if (i)~$w_1,w_5\in V_3$, (ii)~$w_3\in V_{1}$, and 
(iii)~$d_{S_i}(w_2)=2$ or $d_{S_i}(w_4)=2$ (see~\Cref{fig:Si}). 
Let $\sigma'=(\sigma_1,\sigma_2',\sigma_3')$ be the $3$-layer drawing of $S_i$ induced by $\sigma$.
For a $V_3$-gap path $P=w_1w_2w_3w_4w_5$,  it is said that $P$ is drawn \emph{properly} if (i)~either $\sigma_3'(w_1)<\sigma_3'(w_5)$ and $\sigma_2'(w_2)<\sigma_2'(w_4)$ or, symmetrically,  $\sigma_3'(w_5)<\sigma_3'(w_1)$ and $\sigma_2'(w_4)<\sigma_2'(w_2)$, and (ii) one may insert two adjacent additional vertices $u\in V_2$ and $v\in V_3$  in such a way that either  $u$ is between $w_2$ and $w_4$ in $\sigma_2'$ and $v$ is between $w_1$ and $w_5'$ in $\sigma_3'$ and $uv$ does not cross any edge in the drawing of $S_i$ with respect to $\sigma'$. 

\Cref{cl:crosspend} and the construction of $S_i$ imply that for every $V_3$-big vertex $v$ there is a properly drawn $V_3$-gap path $P=w_1w_2w_3w_4w_5$ such that either 
$\sigma_3(w_1)<\sigma_3(v)<\sigma_3(v_5)$ or, symmetrically, $\sigma_3(w_5)<\sigma_3(v)<\sigma_3(v_1)$. We say that $v$ is \emph{drawn in the $P$-gap}. 
It is also easy to make the following observation.

\begin{claim}\label{cl:two-cross-big}
If $v$ is a $V_3$-big vertex of $\widehat{S}_j$ then each edge $vx$ where $x$ is not pendent crosses at least one edge of $S_i$.
\end{claim}

\begin{proof}[Proof of~\Cref{cl:two-cross-big}]
We have that $v$ is drawn in the  $P$-gap for some properly drawn $V_3$-gap path $P=w_1w_2w_3w_4w_5$. We assume without loss of generality that 
$\sigma_3(w_1)<\sigma_3(v)<\sigma_3(w_5)$ and $\sigma_2(w_2)<\sigma_2(w_4)$.
Notice that $x$ is a $V_2$-big vertex. Then by \Cref{cl:crosspend}, either $\sigma_2(x)<\sigma_2(w_2)$ or $\sigma_2(w_4)<\sigma_2(x)$. In both cases, $vx$ crosses at least one edge of $S_i$. This proves the claim.
\end{proof} 

Now we show that $V_3$-big vertices are arranged in $\sigma$ in a specific way.

\begin{claim}\label{cl:gap}
There are no two distinct $V_3$-big vertices $u$ and $v$ that are drawn in the same $P$-gap. 
\end{claim}

\begin{proof}[Proof of~\Cref{cl:gap}]
Assume that $V_3$-big vertices $u$ and $v$ are drawn in the same $P$-gap for a $V_3$-gap path $P=w_1w_2w_3w_4w_5$ with $\sigma_3(w_1)<\sigma_3(u)<\sigma_3(v)<\sigma_3(w_5)$. Notice that $\widehat{S}_j$ has a $(u,v)$-path $Q$ that contains a $V_2$-big vertex $w$. By \Cref{cl:crosspend}, we have that either $\sigma_2(w)<\sigma_2(w_2)$ or $\sigma_2(w_4)<\sigma_2(w)$. In the first case, the edges of $Q$ would cross each edge $ux$ for a pendent vertex $x$ adjacent to $u$, and $Q$ would cross every  edge $vx$ for a pendent vertex $x$ adjacent to $v$. In both cases, we obtain a contradiction with \Cref{cl:crosspend}. This proves the claim.
\end{proof}

Notice that $\widehat{S}_j$ has exactly $5\ell$ $V_3$-big vertices and $S_i$ has exactly the same number of $V_3$-gap paths. Then \Cref{cl:gap} implies every $V_3$-gap path is drawn properly and for each  $V_3$-gap path $P$, there is a $V_3$-big vertex drawn in the $P$-gap. 

We say that three distinct vertices $u_1,u_2,u_3\in V_3\cap V(\widehat{S}_j)$ are \emph{consecutive} if the unique $(u_1,u_3)$-path in $\widehat{S}_j$ contains $u_2$.

\begin{claim}\label{cl:cons}
Suppose that $u_1,u_2,u_3\in V_3\cap V(\widehat{S}_j)$ are consecutive vertices and $u_1$ and $u_3$ are $V_3$-big. Then either $\sigma_3(u_1)<\sigma_3(u_2)<\sigma_3(u_3)$ or, symmetrically, $\sigma_3(u_3)<\sigma_3(u_2)<\sigma_3(u_1)$.
\end{claim}

\begin{proof}[Proof of~\Cref{cl:cons}]
Assume without loss of generality that $\sigma_3(u_1)<\sigma_3(u_3)$. We show that  $\sigma_3(u_1)<\sigma_3(u_2)<\sigma_3(u_3)$. For the sake of contradiction, assume that $\sigma_3(u_2)<\sigma_3(u_1)$. Then consider the unique $(u_2,u_3)$-path $Q$ in $\widehat{S}_j$. Because $u_1,u_2,u_3$ are consecutive, $Q$ avoids $u_1$. Then we have that edges $Q$ cross every $u_1x$ for a pendent vertex $x$ incident to $u_1$ contradicting  \Cref{cl:crosspend}. The case $\sigma_3(u_3)<\sigma_3(u_2)$ is symmetric. This concludes the proof. 
\end{proof}

By \Cref{cl:cons}, we immediately obtain that the $3$-layer drawing of $\widehat{S}_j$ induced by $\sigma$ is canonical.

Now we consider $V_3$-small vertices, that is, vertices of $\widehat{S}_j$ of degree two. Following the previous notation, we say that $v$ is \emph{drawn in the $P$-gap} for 
a $V_3$-gap 
$P=w_1w_2w_3w_4w_5$ if either $\sigma_3(w_1)<\sigma_3(v)<\sigma_3(w_5)$ or  $\sigma_3(w_5)<\sigma_3(v)<\sigma_3(w_1)$.
In exactly the same way as in \Cref{cl:two-cross-big}, we have the following claim.

\begin{claim}\label{cl:two-cross-small-one}
If $v$ is a $V_3$-small vertex of $\widehat{S}_j$ drawn in the $P$-gap for a $V_3$-gap path $P$ then each edge $vx$ where $x$ is not pendent crosses at least one edge of $S_i$.
\end{claim}

We say that $P=w_1w_2w_3w_4w_5$ is a \emph{$V_3$-no-gap path} if $P$ is a induced path in one of the gadgets $U$ such that 
(i)~$w_1,w_3, w_5\in V_3$ and (ii) $w_2,w_4\in V_2$ (see~\Cref{fig:U}). A $V_3$-small vertex $v$ is in the $P$-gap for a $V_3$-no-gap path $P$ if there are $u_1,u_2\in \{w_1,w_3,w_5\}$ such that $\sigma_3(u_1)<\sigma_3(v)<\sigma_3(u_2)$. 
We show the following bound on the number of crossings. Notice that by the construction, for a $V_3$-small vertex $v$, there are two distinct $V_3$-big vertices $x,y\in V_3\cap V(\widehat{S}_j)$ at distance two from $v$.

\begin{claim}\label{cl:two-cross-small-two}
Let $v$ be a $V_3$-small vertex of $\widehat{S}_j$ drawn in the $P$-gap for a $V_3$-no gap-path $P$ and let $x,y\in V_3\cap V(\widehat{S}_j)$ be two distinct $V_3$-big vertices  at distance two from $v$. Then the edges of the $(x,y)$-path in $\widehat{S}_j$ cross at least $6$ edges of $S_i$. 
\end{claim}

\begin{proof}[Proof of \Cref{cl:two-cross-small-two}]
By \Cref{cl:cons}, we can assume without loss of generality that $\sigma_3(x)<\sigma_3(v)<\sigma_3(y)$. Let $Q=w_1w_2w_3w_4w_5$ and $Q'=w_1'w_2'w_3'w_4'w_4'$ be properly drawn $V_3$-gap paths such that $x$ and $y$ are drawn in the $Q$ and $Q'$-gaps, respectively. We assume without loss of generality that 
$\sigma_3(w_1)<\sigma_3(x)<\sigma_3(w_5)$ and $\sigma_3(w_1')<\sigma_3(y)<\sigma_3(w_5')$. By \Cref{cl:gap}, $Q$ and $Q'$ are distinct $V_3$-gap paths. 
Because of \Cref{cl:cons} and the observation that every $V_3$-gap path is drawn properly and for each  $V_3$-gap path $P$, there is a $V_3$-big vertex drawn in the $P$-gap, we obtain that $P=u_1u_2u_3u_4u_5$ for $u_1=w_5$ and $u_5=w_1'$. Let $R=w_4u_1u_2u_3u_4u_5w_2'$. This path has $6$ edges and each edge of $R$ crosses at least one edge of the $(x,y)$-path in $\widehat{S}_j$. This proves the claim.
\end{proof}

Now we obtain the lower bound for the number of crossings if every $V_3$-small vertex is drawn in the $P$-gap for a $V_3$-no-gap path.

\begin{claim}\label{cl:lb-canonical}
If every $V_3$-small vertex is drawn in the $P$-gap for a $V_3$-no-gap path $P$ then the total number of crossings is at least $14\ell-2$. Furthermore, if the number of crossings is exactly $14\ell-2$ then the induced drawings of $S_i$ and $\widehat{S}_j$ are canonical. 
\end{claim}

\begin{proof}[Proof of~\Cref{cl:lb-canonical}]
Recall that $\widehat{S}_j$ has exactly $5\ell$ $V_3$-big vertices and $S_i$ has exactly the same number of $V_3$-gap paths. Also, $\widehat{S}_j$ has exactly $\ell$ $V_3$-small vertices and $S_i$ has $\ell$ $V_3$-no-gap paths. Notice that $v_3$-small vertices are at a distance at least $6$ from each other. 
Consider the edges of $(x_v,y_v)$-paths $R_v$ in $\widehat{S}_j$ for $V_3$-small vertices $v$ such that $x_v,y_v\in V_3$ are distinct vertices at distance two from $v$. 
 By \Cref{cl:two-cross-small-two}, these edges cross at least $6\ell$ edges of $S_i$. Furthermore, $\widehat{S}_j$ has $8\ell-2$ edges $xy$ not included in any path $R_v$ such that 
 $x\in V_3$ and $y\in V_2$ is not pendent. By \Cref{cl:two-cross-big}, these edges cross at least $8\ell-2$ edges of $S_i$. Summarizing, we obtain that the edges of $\widehat{S}_j$ cross at least $14\ell-2$ edges of $S_i$. We already noticed that  the $3$-layer drawing of $\widehat{S}_j$ induced by $\sigma$ is canonical by \Cref{cl:cons}. 
 For $S_i$, we observe that If the number of crossings is exactly $14\ell-2$ then there are no crossings between the edges of $S_i$. Thus, the induced drawing of $S_i$ is canonical.
\end{proof}

Next, we get the lower bound if there is a $V_3$-small vertex that is not drawn in the $P$-gap for a $V_3$-no-gap path.

\begin{claim}\label{cl:lb-no-gap}
If there is $V_3$-small vertex $v$ that is not drawn in the $P$-gap for a $V_3$-no-gap path $P$ then the total number of crossings is at least $14\ell$. 
\end{claim}

\begin{proof}[Proof of~\Cref{cl:lb-no-gap}]
Denote by $t\geq 1$ the number of $V_3$-small vertices that are not drawn in the $P$-gap for a $V_3$-no-gap paths $P$. Consider a $V_3$-no-gap path $P=u_1u_2u_3u_4u_5$ such that no $V_3$-small vertex is drawn in the $P$-gap. Let $Q=w_1w_2w_3w_4w_5$ and $Q'=w_1'w_2'w_3'w_4'w_5'$ be $V_3$-gap paths such that $u_1=w_5$ and $u_5=w_1'$.
Recall that \Cref{cl:gap}, there are two $V_3$-big vertices $u$ and $u'$ that are drawn in the $Q$ and $Q'$-gaps, respectively. Also, we can assume that 
$\sigma_3(w_1)<\sigma_3(w_5),\sigma_3(u_3),\sigma_3(w_1')$ and $\sigma_3(w_5),\sigma_3(u_3),\sigma_3(w_1')<\sigma_3(w_5')$. We consider two cases. 

First, assume that $u$ and $u'$ are at distance two in $\widehat{S}_j$. Then there is a $V_2$-big vertex $v$ that is adjacent to $u$ and $u'$. By~\Cref{cl:crosspend}, we have that 
either $\sigma_2(w_4)<\sigma_2(v)<\sigma_2(u_2),\sigma_2(u_4),\sigma_2(w_2')$ or, symmetrically, $\sigma_2(w_4),\sigma_2(u_2),\sigma_2(u_4)<\sigma_2(v)<\sigma_2(w_4')$. 
These cases are symmetric and we assume that $\sigma_2(w_4)<\sigma_2(v)<\sigma_2(u_2),\sigma_2(u_4),\sigma_2(w_2')$. Then the edge $vu'$ crosses every edge of the path $u_1u_2u_3u_4u_5w_2'$. Thus, $vu'$ crosses at least $5$ edges. Notice also that $u'$ and $v$ are at a distance of at least two from any small $V_2$-vertex of $\widehat{S}_j$.

Next, we suppose that $u$ and $u'$ are at a distance of at least four. Then by~\Cref{cl:gap} and \Cref{cl:cons}, $u$ and $u'$ are at distance four and there is a $V_3$-small vertex $x$ at distance two from both $u$ and $u'$. We have that $x$ is not drawn in the $P$-gap. Using \Cref{cl:cons}, we conclude that $x$ is either drawn in the $Q$-gap or $Q'$-gap. By symmetry, we assume that  $x$ is drawn in the $Q$-gap. Then $\sigma_3(u)<\sigma_3(x)<\sigma_3(u_1),\sigma_3(u_3),\sigma_3(u_5)$. 
By the construction of $\widehat{S}_j$,  there is a $V_2$-big vertex $v$ that is adjacent to $x$ and $u'$. By~\Cref{cl:crosspend}, we have that 
either $\sigma_2(w_4)<\sigma_2(v)<\sigma_2(u_2),\sigma_2(u_4),\sigma_2(w_2')$ or $\sigma_2(w_4),\sigma_2(u_2),\sigma_2(u_4)<\sigma_2(v)<\sigma_2(w_4')$. 
In the first case, we have that $vu'$ crosses every edge of the path $u_1u_2u_3u_4u_5w_2'$. Notice that  $u'$ and $v$ are at a distance at least two from any small $V_2$-vertex of $\widehat{S}_j$ distinct from $x$. If $\sigma_2(w_4),\sigma_2(u_2),\sigma_2(u_4)<\sigma_2(v)<\sigma_2(w_4')$ then, by same arguments, we have that $xv$ crosses every edge of $w_4u_1u_2u_3u_4u_5$ and $x$ and $v$  are at distance at least two from any small $V_2$-vertex of $\widehat{S}_j$ distinct from $x$.

We conclude that in both cases, there is an edge $xy\in E(\widehat{S}_j)$ with $x\in V_2$ and $y\in V_3$ such that $xy$ crosses at least $5$ edges of $S_i$ and 
both $x$ and $y$ are at distance at least two from any small $V_3$-vertex that is drawn in the $R$-gap for a $V_3$-no-gap path $R$. 

For a $V_3$-small vertex $v$ that is drawn in the $P$-gap for some $V_3$-gap path $P$, denote by $R_v$ the  $(x_v,y_v)$-paths in $\widehat{S}_j$ such that $x_v,y_v\in V_3$ are distinct vertices at distance two from $v$. There are $\ell-t$ such paths for $v$ that are not drawn in the $Q$-gap for a $V_3$-no-gap path $Q$. 
By \Cref{cl:two-cross-small-two}, these edges cross at least $6(\ell-t)$ edges of $S_i$. Further, there a set $L$ of $8\ell+4t-2$ edges $xy$
of $\widehat{S}_j$ that not included in any path $R_v$ such that  $x\in V_3$ and $y\in V_2$ is not pendent. By \Cref{cl:two-cross-big} and 
\Cref{cl:two-cross-small-one} each of the edges of $A$ crosses at least one edge of $S_i$. Also, we proved that there are  $t$ edges $xy\in A$ such that $xy$ crosses at least $5$ edges of $S_i$. Then the total number of crossings is at least $6(\ell-t)+8\ell+4t-2+4t=14\ell+2t-2\geq 14\ell$. This concludes the proof.
\end{proof}

Combining \Cref{cl:lb-canonical} and \Cref{cl:lb-no-gap}, we obtain that the number of crossing for the optimum restrictive $3$-layer drawing $\sigma$ of $H=S_i\uplus \widehat{S}_j$ is at least $14\ell-2$. This proves claim (i) of the lemma.

To show (ii), we have to argue that  the number of crossings is $14\ell-2$ if and only if $i=j$.
If $i=j$ then recall that we already demonstrated a restrictive $3$-layer drawing with $14\ell-2$ crossings. 
Thus, it remains to prove that if the number of crossings is $14\ell-2$ then  $i=j$.

Assume that the number of crossings for $\sigma$ is $14\ell-2$. Because of \Cref{cl:lb-no-gap}, every $V_3$-small vertex is drawn in the $P$-gap for some $V_3$-no-gap path $P$. By \Cref{cl:lb-canonical}, we have that  the induced drawings of $S_i$ and $\widehat{S}_j$ are canonical.  Recall that $S_i$ was constructed for the binary string encoding of the symbol $s_i$ and $\widehat{S}_j$ was constructed for the encoding of $s_j$.
By \Cref{cl:gap} and \Cref{cl:cons} and the observation that $\widehat{S}_j$ has exactly $5\ell$ $V_3$-big vertices and $S_i$ has exactly the same number of $V_3$-gap paths, we conclude that either the encodings of $s_i$ and $s_j$ are the same or the reversal of the encoding of $s_i$ coincides with the encoding of $s_j$. However, the considered binary encoding of $\Sigma$ has the properties that the strings encoding distinct symbols are distinct, and the reversal of each string is distinct from any other strings in the encoding. Therefore, $s_i=s_j$ and  $i=j$. 

To obtain (iii), we combine \Cref{cl:gap}, \Cref{cl:lb-canonical}, \Cref{cl:lb-no-gap}, and the observation that  for each  $V_3$-gap path $P$, there is a $V_3$-big vertex drawn in the $P$-gap.  This concludes the proof of the lemma.
\end{proof}

\subsection{Main step}\label{sec:ms}
Given the gadgets $S_i$ and $\widehat{S}_i$ constructed for  the symbols $s_1,\ldots,s_k$ of $\Sigma$ and their properties summarized in 
 \Cref{lem:restr}, we proceed with reducing the instance $(s,k)$ of  \probDF. In this subsection, we explain the main step. Here, we deal with $4$-layer embeddings. 
We remind that $\ell=\lceil \log k\rceil +1$ and $p=40k^3$.  
 
 Let $s=r_1\dots r_n$ where $r_i\in \Sigma$. We construct the graphs $R$ and $\widehat{R}$ encoding the instance $(s,k)$ of  \probDF. 

 \begin{figure}[ht]
\centering
\scalebox{0.7}{
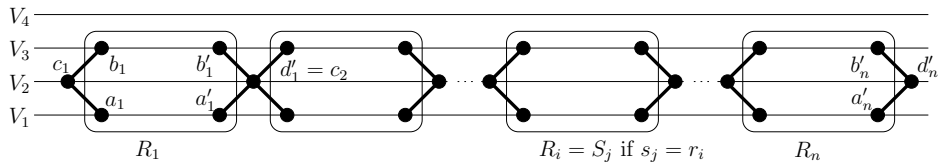}
\caption{Construction of $R$.}
\label{fig:R}
\end{figure}

\begin{itemize}
\item Construction of $R$ (see~\Cref{fig:R}):
\begin{itemize}
\item for every $i\in\{1,\ldots,n\}$, construct a copy $R_i$ of $S_j$ for $s_j=r_i$ and denote by  $a_i,b_i,c_i,d_i,a_i',b_i',c_i',d_i'$ the respective vertices $a,b,c,d,a',b',c',d'$ of $S_j$.
\item concatenate these copies by merging the vertex $d_{i-1}'$ of $R_{i-1}$ with the vertex $c_i$ of $R_i$ for $i\in\{2,\ldots,n\}$.
\end{itemize}
\item Construction of $\widehat{R}$ (see):
\begin{itemize}
\item for every $i\in\{1,\ldots,k\}$, construct two copies $\widehat{S}_i^1$ and $\widehat{S}_i^2$ of $\widehat{S}_i$,
\item for every $i\in\{1,\ldots,k\}$, construct a vertex $v_i$ and make it adjacent to every vertex of $V_3\cap (V(\widehat{S}_i^1)\cup V(\widehat{S}_i^2))$,
\item put the vertices $v_1,\ldots,v_k$ in $V_4$.   
\end{itemize}
\end{itemize} 
Notice that $V_4\cap V(R)=\emptyset$. 

 \begin{figure}[ht]
\centering
\scalebox{0.7}{
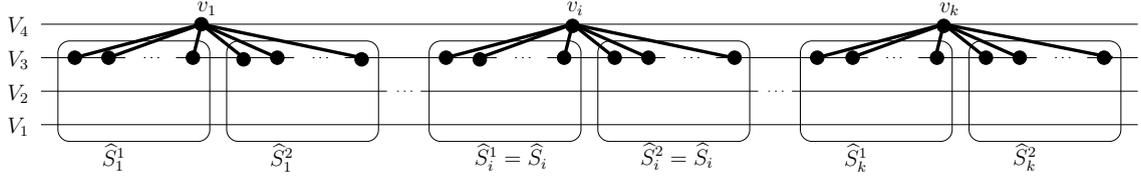}
\caption{Construction of $\widehat{R}$.}
\label{fig:hR}
\end{figure} 
 
We say that a $4$-layer drawing $\sigma=(\sigma_1,\sigma_2,\sigma_3,\sigma_4)$ of the disjoint union of $R$ and $\widehat{R}$ is  \emph{restrictive} if for every vertex $w\in V_3\cap V(\widehat{R})$, $\sigma_3(b_1)<\sigma_3(w)<\sigma_3(b_n')$. We prove the following lemma that is crucial for the reduction.

\begin{lemma}\label{lem:crucial}
The instance $(s,k)$ of \probDF is a yes-instance of the problem if and only if there is a restrictive $4$-layer drawing of  the disjoint union of $R$ and $\widehat{R}$ with at most $2k(14\ell-2)$ crossings.
\end{lemma}

\begin{proof}
Suppose that $(s,k)$ is a yes-instance of \probDF. Then the string $s=r_1\ldots r_n$ has $k$ disjoints nontrivial substrings $r_{i_1}\ldots r_{j_1},\dots, r_{i_k}\ldots r_{j_k}$ such that 
the symbols $r_{i_t}$ and $r_{j_t}$ are the same for $t\in\{1,\ldots,k\}$. We assume without loss of generality that $r_{i_t}=r_{j_t}=s_t$ for each $t\in\{1,\ldots,k\}$ and it holds 
that $j_{t-1}<i_t$ for all $t\in\{2,\ldots,k\}$.
We construct the $4$-layer drawing $\sigma=(\sigma_1,\sigma_2,\sigma_3,\sigma_4)$ of $H=R\uplus \widehat{R}$ as follows.
\begin{itemize}
\item Construct canonical $3$-layer drawings of each $R_i$ for $i\in \{1,\ldots,n\}$ such that $c_i$ is drawn before $d_i'$.  
\item Concatenate these drawings by merging the vertex $d_{i-1}'$ of $R_{i-1}$ with the vertex $c_i$ of $R_i$ for $i\in\{2,\ldots,n\}$ as it is shown in~\Cref{fig:R}.
\item For each $t\in\{1,\ldots,k\}$, draw $\widehat{S}_t^1$ in such a way that the vertices of $V_1\cap V(\widehat{S}_t^1)$ are drawn between $a_{i_t}$ and $a_{i_t}'$,  the vertices of $V_2\cap V(\widehat{S}_t^1)$ are drawn between $c_{i_t}$ and $d'_{i_t}$,  the vertices of $V_3\cap V(\widehat{S}_t^1)$ are drawn between $b_{i_t}$ and $b'_{i_t}$, and the $3$-layer drawing of $R_{i_t}\uplus \widehat{S}_t^1$ has $14\ell-2$ crossings using \Cref{lem:restr}.
\item For each $t\in\{1,\ldots,k\}$, draw $\widehat{S}_t^2$ in such a way that the vertices of $V_1\cap V(\widehat{S}_t^2)$ are drawn between $a_{j_t}$ and $a_{j_t}'$,  the vertices of $V_2\cap V(\widehat{S}_t^2)$ are drawn between $c_{j_t}$ and $d'_{j_t}$,  the vertices of $V_3\cap V(\widehat{S}_t^2)$ are drawn between $b_{j_t}$ and $b'_{j_t}$, and the $3$-layer drawing of $R_{j_t}\uplus \widehat{S}_t^2$ has $14\ell-2$ crossings using \Cref{lem:restr}.
\item Define $\sigma_4$ by setting $\sigma_4(v_1)<\dots<\sigma_4(v_k)$.
\end{itemize}
By the assumptions that  $r_{i_t}=r_{j_t}=s_t$ for each $t\in\{1,\ldots,k\}$ and $j_{t-1}<i_t$ for all $t\in\{2,\ldots,k\}$, we have that the edges $xy$ with $x\in V_3$ and $y\in V_4$ have no crossings. Then the total number of crossings is 
$2k(14\ell-2)$.

For the opposite direction, assume that $H=R\uplus \widehat{R}$ admits a restrictive $4$-layer drawing $\sigma=(\sigma_1,\sigma_2,\sigma_3,\sigma_4)$ with at most $2k(14\ell-2)$ crossings.

Consider arbitrary $i\in\{1,\ldots,n\}$. Recall that $R_i$, that is a copy of $S_j$ for $s_j=r_i$, contains two disjoint copies of filling gadgets $F$.   We denote the copy containing $a_i,b_i$ by $F_i$ and the copy containing $a_i',b_i'$ by $F_i'$. We have that $a_i\in V_1$, $b_i\in V_3$, and  $F_i$ contains $p$ $(a_i,b_i)$-paths of length two. Because $p=40k^3>2k(14(k+1)-2)\geq \ell$, 
there is a $(a_i,b_i)$-path $P_i$ in $F_i$ such that the edges of $P_i$ do not cross any edge. By the same arguments, there is a $(a_i',b_i')$ path $P_i'$ of length two such that   
the edges of $P_i'$ do not cross any edge. 

Because $R$ is connected, we conclude that either 
$\sigma_1(a_1)<\sigma_1(a_1')<\sigma_1(a_2)<\dots \sigma_1(a_n)<\sigma_1(a_n')$ and $\sigma_3(b_1)<\sigma_3(b_1')<\sigma_3(b_2)<\dots \sigma_3(b_n)<\sigma_3(b_n')$ or
$\sigma_1(a_n')<\sigma_1(a_n)<\sigma_1(a_{n-1}')<\dots \sigma_1(a_1')<\sigma_1(a_1)$ and $\sigma_3(b_n')<\sigma_3(b_n)<\sigma_3(b_{n-1}')<\dots \sigma_3(b_1')<\sigma_3(b_1)$. Because $\sigma$ is restrictive, 
$\sigma_1(a_1)<\sigma_1(a_1')<\sigma_1(a_2)<\dots \sigma_1(a_n)<\sigma_1(a_n')$ and $\sigma_3(b_1)<\sigma_3(b_1')<\sigma_3(b_2)<\dots \sigma_3(b_n)<\sigma_3(b_n')$.

Furthermore,  because all the paths $P_i$ and $P_i'$ have no crossing edges and the gadgets $\widehat{S}_t$ are connected, we have that for every $t\in \{1,\ldots,k\}$ and every $q\in\{1,2\}$, it holds that 
\begin{itemize}
\item either there is $i\in\{1,\ldots,n\}$ such that $\sigma_3(b_i)<\sigma_3(w)<\sigma_3(b_i')$ for each $w\in V_3\cap V(\widehat{S}_t^q)$,
\item or there is $i\in\{2,\ldots,n\}$ such that  $\sigma_3(b_i')<\sigma_3(w)<\sigma_3(b_{i-1})$ for each $w\in V_3\cap V(\widehat{S}_t^q)$.
\end{itemize} 
However, in the second case, the paths $a_{i-1}'d_{i-1}a_i$ and $b_{i-1}'d_{i-1}b_i$ would cross every edge of  $\widehat{S}_t^q$ contradicting that the number of crossings is at most 
$2k(14\ell-2)$ because $\widehat{S}_t^q$ contains more than $p$ edges. Thus, for every $t\in \{1,\ldots,k\}$ and every $q\in\{1,2\}$, we have that  
there is $i\in\{1,\ldots,n\}$ such that $\sigma_3(b_i)<\sigma_3(w)<\sigma_3(b_i')$ for each $w\in V_3\cap V(\widehat{S}_t^q)$. 
Then by  \Cref{lem:restr}(i), we have that the induced drawing of $R_i\uplus \widehat{S}_t^q$ has at least $14\ell-2$ crossings. Because the total number of crossings is at most 
$2k(14\ell-2)$, we conclude that the induced drawing of $R_i\uplus \widehat{S}_t^q$ has exactly $14\ell-2$ crossings. By \Cref{lem:restr}(ii), this means that $R_i$ is a copy of $S_t$. Note also that the edges $xy\in E(\widehat{R})$ with $x\in V_3$ and $y\in V_4$ do not cross each other.

Suppose that there is $i\in\{1,\ldots,n\}$ such that $R_i$ is a copy of $S_j$ and it holds that $\sigma_3(b_i)<\sigma_3(w)<\sigma_3(b_i')$ for all 
$w\in V_3\cap(V(\widehat{S}_j^1)\cup V(\widehat{S}_j^2)$. Then by \Cref{lem:restr}(iii), we have that there are crossings of edges of $\widehat{S}_j^1$ and $\widehat{S}_j^2$.  This means that the total number of crossings is at least $2k(14\ell-2)+1$. Therefore, for each $t\in\{1,\ldots,k\}$, there are distinct $i_t,j_t\in\{1,\ldots,n\}$ such that $R_{i_t}$ and $R_{j_t}$ are copies of $S_t$ and  it holds that (i) $\sigma_3(b_{i_t})<\sigma_3(w)<\sigma_3(b_{i_t}')$ for all 
$w\in V_3\cap V(\widehat{S}_t^1)$ and (ii) $\sigma_3(b_{j_t})<\sigma_3(w)<\sigma_3(b_{j_t}')$ for all 
$w\in V_3\cap V(\widehat{S}_t^2)$. By symmetry, we assume that $i_t<j_t$ for each $t\in\{1,\ldots,k\}$.

Consider the substrings $r_{i_t}\dots r_{j_t}$ of $s$ for $t\in \{1,\ldots,k\}$. Because  $R_{i_t}$ and $R_{j_t}$ are copies of $S_t$, we have that $r_{i_t}=r_{j_t}=s_t$, that is, each 
$r_{i_t}\dots r_{j_t}$ is a factor starting and ending with the symbol $s_t$. 

Consider  distinct $t,t'\in \{1,\ldots,k\}$  and  assume, by symmetry, that $\sigma_4(v_t)<\sigma_4(v_{t'})$. For each $x\in  V_3\cap V(\widehat{S}_t^2)$ and each 
$y\in  V_3\cap V(\widehat{S}_t^1)$, the edges $v_tx$ and $v_{t'}y$ do not cross. Thus, $j_t<i_{t'}$. This means that 
$r_{i_t}\dots r_{j_t}$ and $r_{i_{t'}}\dots r_{j_{t'}}$. Then the substrings $r_{i_t}\dots r_{j_t}$ of $s$ for $t\in \{1,\ldots,k\}$ are disjoint factors of $s$ starting with distinct symbols. 
We conclude that $(s,k)$ is a yes-instance of \probDF. This finishes the proof.
\end{proof}

\subsection{Proof of Theorems~\ref{thm:nokern} and \ref{thm:ETHlb}}\label{sec:lb-fin}
 Now we are ready to prove \Cref{thm:nokern} and \Cref{thm:ETHlb}. Recall that in our base reduction in \Cref{sec:ms}, we use restrictive $4$-layer drawings of the disjoint union of $R$ and $\widehat{R}$ (see~\Cref{lem:crucial}). To prove our main results, we have to get rid of this constraint. 
 We again remind that $\ell=\lceil \log k\rceil +1$ and $p=40k^3$.  
 We start with  \Cref{thm:nokern} which we restate.
 
\thmnokern*  
 
 \begin{proof}
 Clearly, it is sufficient to prove the theorem for $h=4$. 
 The theorem is proved by constructing a polynomial parameter transformation of \probDF. Let $(s,k)$ be an instance of \probDF. We construct the instance $(G,V_1,V_2,V_3,V_4,k')$ of 
 \probFourCr. The graph $G$ and $V_1,V_2,V_3,V_4$ are constructed as follows (see~\Cref{fig:Gone}).

 \begin{figure}[ht]
\centering
\scalebox{0.7}{
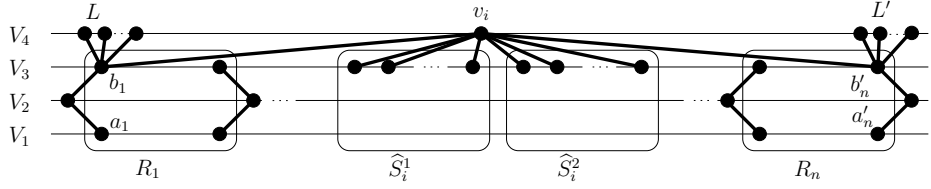}
\caption{Construction of $G$.}
\label{fig:Gone}
\end{figure} 

 \begin{itemize}
 \item Construct  $R$ and $\widehat{R}$ as explained in~\Cref{sec:ms}.
 \item Construct a set $L$ of $p$ vertices, put them in $V_4$, and make them adjacent to the vertex $b_1$ of $R$.
 \item Construct a set $L'$ of $p$ vertices, put them in $V_4$, and make them adjacent to the vertex $b_n'$ of $R$.
 \item For each $i\in \{1,\ldots,k\}$, make the vertex $v_i$ adjacent to $b_1$ and $b_n'$.
 \end{itemize}
We set $k'=2k(14\ell-2)+\frac{1}{2}k(k-1)+k(k-1)(12\ell+1)$.
 
We claim that $(s,k)$ is a yes-instance of \probDF if and only if $(G,V_1,V_2,V_3,V_4,k')$ is a yes-instance of  \probFourCr.

Suppose that $(s,k)$ is a yes-instance of \probDF.  Then by~\Cref{lem:crucial}, there is a restrictive $4$-layer drawing $\sigma=(\sigma_1,\sigma_2,\sigma_3,\sigma_4)$
of  the disjoint union of $R$ and $\widehat{R}$ with at most $2k(14\ell-2)$ crossings. We construct the $4$-layer drawing $\sigma'=(\sigma_1,\sigma_2,\sigma_3,\sigma_4')$ of $G$ by the following modification of $\sigma_4$: we insert the vertices of $L$ before the vertices $v_1,\ldots,v_k$ and we put the vertices of $L'$ after the vertices $v_1,\ldots,v_k$.
Notice that the paths $b_1v_ib_n'$ and $b_1v_jb_n'$ for distinct $i,j\in\{1,\ldots,k\}$ have exactly one pair of crossing edges. Therefore, the total number of crossings of such paths is $\binom{k}{2}$. Also, for all distinct $i,j\in\{1,\ldots,k\}$, the path $b_1v_ib_n'$ crosses every edge $v_jx$ for $x\in V_3\cap(V(\widehat{S}_j^1)\cup V(\widehat{S}_j^2))$.
The edges $b_1y$ for $y\in L$ and $b_n'z$ for $z\in L'$ do not cross any edge. Therefore, the total number of crossings in $\sigma'$ is at most $k'=2k(14\ell-2)+\frac{1}{2}k(k-1)+k(k-1)(12\ell+1)$.

For the opposite direction, assume that $(G,V_1,V_2,V_3,V_4,k')$ is a yes-instance of  \probFourCr, that is, there is a $4$-layer drawing $\sigma=(\sigma_1,\sigma_2,\sigma_3,\sigma_4)$ with at most $k'$ crossings. Consider $(a_1,x)$-paths of length three in $G$ for $x\in L$. By the construction, $G$ has $p$ edge-disjoint paths of this type. We have that $p=40k^3>2k(14(k+1)-2)+\frac{1}{2}k(k-1)+k(k-1)(12(k+1)+1)\geq k'$ if $k\geq 2$.\footnote{The value of $p=40k^3$ was chosen to be big enough to satisfy this inequality.} Therefore, there is $x\in L$ such that the edges of the  $(a_1,x)$-path $P$ of length three in $G$ do not cross any edge in $\sigma$. By the same arguments, there is $x'\in L'$ such that  the edges of the  $(a_n',x')$-path $P'$ of length three in $G$ have no crossings. By symmetry, we can assume without loss of generality that 
$\sigma_1(a_1)<\sigma_1(a_n')$ and $\sigma_3(b_1)<\sigma_3(b_n')$. Then because $G$ is connected, we have that for every vertex $w\in V_3\cap V(\widehat{R})$, 
 $\sigma_3(b_1)<\sigma_3(w)<\sigma_3(b_n)'$. Thus, $\sigma$ induces a  restrictive $4$-layer drawing of the disjoint union of $R$ and $\widehat{R}$.
 Observe that for any choice of $\sigma_4$, the paths $b_1v_ib_n'$ and $b_1v_jb_n'$ for distinct $i,j\in\{1,\ldots,k\}$ have exactly one pair of crossing edges and the total number of crossings with the edges of these paths is  $\binom{k}{2}$.  Because $\sigma$ is restrictive, we also have that  for all distinct $i,j\in\{1,\ldots,k\}$, the path $b_1v_ib_n'$ crosses every edge $v_jx$ for $x\in V_3\cap(V(\widehat{S}_j^1)\cup V(\widehat{S}_j^2))$. Therefore, the total number of crossings of this type is at least $k(k-1)(12\ell+1)$. It follows that
 the disjoint union of $R$ and $\widehat{R}$ has at most $k'-\big(\frac{1}{2}k(k-1)+k(k-1)(12\ell+1)\big)\leq 2k(14\ell-2)$ crossings. Since the drawing induced by $\sigma$ is restrictive, \Cref{lem:crucial} implies 
 that $(s,k)$ is a yes-instance of \probDF.

It is straightforward to verify that given $(s,k)$, the construction of $G$ can be done in polynomial time. We also have that $k'=\Oh(k^2\log k)$. Therefore, the reduction is  a polynomial parameter transformation.  Then using \Cref{prop:DF}, we obtain that \probFourCr  does not admit a polynomial kernel unless $\classNP\subseteq\classCoNP/{\rm poly}$. This concludes the proof.
\end{proof}

In the proof of \Cref{thm:nokern}, we constructed a  polynomial parameter transformation of \probDF where the parameter in the obtained instance of \probFourCr is $k'=\Oh(k^2\log k)$.
This is sufficient for a kernelization lower bound but $k'$ is too big to give a reasonable lower bound for the running time up to ETH. Hence, we use a slightly different reduction in the proof of \Cref{thm:ETHlb} which we restate. 

\thmeth*

\begin{proof}
Is sufficient to prove the theorem for $h=5$. We again reduce \probDF and for an instance $(s,k)$ of \probDF, we construct the instance $(G,V_1,V_2,V_3,V_4,V_5,k')$ of 
\probFiveCr. The graph $G$ and $V_1,V_2,V_3,V_4,V_5$ are constructed as follows (see~\Cref{fig:Gtwo}).

 \begin{figure}[ht]
\centering
\scalebox{0.7}{
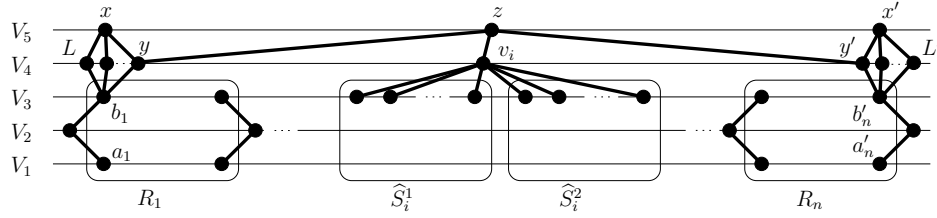}
\caption{Construction of $G$.}
\label{fig:Gtwo}
\end{figure} 

\begin{itemize}
\item Construct  $R$ and $\widehat{R}$ as explained in~\Cref{sec:ms}.
\item Construct two vertices $x$ and $x'$ and put them in $V_5$.
\item Construct a set $L$ of $p$ vertices, put them in $V_4$, and make them adjacent to $x$ and the vertex $b_1$ of $R$.
\item Construct a set $L'$ of $p$ vertices, put them in $V_4$, and make them adjacent to $x'$ and the vertex $b_n'$ of $R$.
\item Select vertices $y\in L$ and $y'\in L'$.
\item Construct a vertex $z$, put it in $V_5$, and make it adjacent to $y$, $y'$, and  $v_i\in V(\widehat{R})$ for all $i\in \{1,\ldots,k\}$.
\end{itemize}
We set $k'=2k(14\ell-2)$.
 
 We prove that $(s,k)$ is a yes-instance of \probDF if and only if $(G,V_1,V_2,V_3,V_4,V_5,k')$ is a yes-instance of  \probFiveCr.

Let  $(s,k)$ be a yes-instance of \probDF.  Then by~\Cref{lem:crucial}, there is a restrictive $4$-layer drawing $\sigma=(\sigma_1,\sigma_2,\sigma_3,\sigma_4)$
of  the disjoint union of $R$ and $\widehat{R}$ with at most $2k(14\ell-2)$ crossings.
We construct the $5$-layer drawing $\sigma'=(\sigma_1,\sigma_2,\sigma_3,\sigma_4',\sigma_5)$ of $G$:
\begin{itemize}
\item set $\sigma_5(x)<\sigma_5(z)<\sigma_5(x')$,
\item construct $\sigma_4'$ by modifying $\sigma_4$:
\begin{itemize}
\item insert the vertices of $L$ before  $v_1,\ldots,v_k$ in such a way that $y$ is the last vertex of $L$ in the order,
\item insert the vertices of $L$ after $v_1,\ldots,v_k$ in such a way that $y'$ is the first vertex of $L$ in the order.
\end{itemize}
\end{itemize}
Notice that in the obtained $5$-layer drawing all crossings are in the disjoint union of $R$ and $\widehat{R}$. Thus, the total number of crossings is at most  $2k(14\ell-2)$ and 
$(G,V_1,V_2,V_3,V_4,V_5,k')$ is a yes-instance of  \probFiveCr.

For the opposite direction, assume that $(G,V_1,V_2,V_3,V_4,k')$ is a yes-instance of  \probFiveCr, that is, there is a $5$-layer drawing 
$\sigma=(\sigma_1,\sigma_2,\sigma_3,\sigma_4,\sigma_5)$ with at most $k'$ crossings. Consider $(a_1,x)$-paths of length four in $G$ and note that  $G$ has $p$ edge 
disjoint paths of this type. Since $p=40k^4>2k(14(k+1)-2)\geq 2k(14\ell-2)$, there is  an $(a_1,x)$-path $P$ of length four in $G$ such that its edges do not cross any edge. Symmetrically, there is an $(a_n',x')$-path of length four in $G$ with edges without crossings. Then we can assume without loss of generality that 
$\sigma_1(a_1)<\sigma_1(a_n')$ and $\sigma_3(b_1)<\sigma_3(b_n')$. Since $G$ is connected,  for every vertex $w\in V_3\cap V(\widehat{R})$, 
 $\sigma_3(b_1)<\sigma_3(w)<\sigma_3(b_n')$. Therefore,  $\sigma$ induces a restrictive $4$-layer drawing of the disjoint union of $R$ and $\widehat{R}$.
 Because the total number of crossings is at most $2k(14\ell-2)$, \Cref{lem:crucial} implies 
that $(s,k)$ is a yes-instance of \probDF.

Suppose that there is an algorithm $\mathcal{A}$ solving \probFiveCr in $2^{o(k/\log k)}\cdot n^{\Oh(1)}$ time on $n$-vertex graphs. Then we can solve \probDF using our reduction. Given an instance $(s,k)$ of \probDF, we use the reduction to construct an equivalent instance of  \probFiveCr. This can be done in polynomial time and the value of the parameter is $\Oh(k\log k)$. Then we apply $\mathcal{A}$ to the obtained instance. The total running time is $2^{o(k\log k/\log(k\log k))}\cdot |s|^{\Oh(1)}$. Since 
$\log k\leq \log(k\log k)$, the algorithm runs in $2^{o(k)}\cdot |s|^{\Oh(1)}$ time. However, by \Cref{prop:DF}, this refutes ETH. This completes the proof.
\end{proof}


\section{Conclusion}\label{sec:conclusion}
We conclude by stating some open problems. We proved in \Cref{thm:subexp} and \Cref{thm:subexpTHREE} that \probHCr can be solved in $2^{\Oh(k^{\alpha} \log k)}\cdot n^{\Oh(1)}$ time if $h=2,3$ with $\alpha = 1/2, 2/3$ respectively, and we demonstrated in 
\Cref{thm:ETHlb} that it is unlikely that for any fixed the problem could be solved in $2^{o(k/\log k)}\cdot n^{\Oh(1)}$ time for any $h\geq 5$. The most intriguing question is whether \probFourCr can be solved in $2^{\Oh(k^{1-\varepsilon} \text{poly}\log k)}\cdot n^{\Oh(1)}$ time for some positive $\varepsilon<1$.  Note also that even for $h\geq 5$, our computation lower bound does not exclude the existence of an algorithm with a subexponential in $k$ running time. Can \probHCr be solved in $2^{\Oh(k/\log k)}\cdot n^{\Oh(1)}$ time for $h\geq 4$?
In fact, we are not aware even of whether there is a $2^{\Oh(k)}\cdot n^{\Oh(1)}$ algorithm for $h\geq 4$. In particular, the algorithm of Dujmovi{\'c} et al.~\cite{dujmovic2008parameterized} is not single-exponential. 
Finally, is it possible to improve our results for $h=2,3$? Our subexponential algorithms for $h = 2, 3$ run in   $2^{\Oh(k^{1/2} \log k)}\cdot n^{\Oh(1)}$ and $2^{\Oh(k^{2/3} \log k)}\cdot n^{\Oh(1)}$ time, respectively. Is it possible to shave off the logarithmic factor from the exponent of the FPT term? More strongly, for $h = 3$, can we obtain an $2^{\Oh(k^{1/2} \log k)} \cdot n^{\Oh(1)}$ time algorithm? The size of the kernel \Cref{thm:kernel} for $h=3$ is $\Oh(k^{8})$ and the constant hidden in the big $\Oh$-notation is huge. We believe that by a more accurate analysis, we can decrease the constant. However, it seems that to substantially decrease the dependence on $k$, more sophisticated algorithmic tools are needed. We note that the quadratic kernel for \probTwoCr by Kobayashi et al.~\cite{KobayashiMNT14} heavily relies on the structure of graphs admitting $2$-layer drawings without crossings (see~\Cref{obs:caterpillar}), and for $h=3$, the structure becomes a great deal more complicated.  Is it possible to obtain a polynomial kernel of ``moderate'' size for $h=3$, say, quadratic or cubic?

	\bibliographystyle{siam}
	\bibliography{Crossing,book_pc}

\begin{thebibliography}{10}

\bibitem{alon2009fast}
{\sc N.~Alon, D.~Lokshtanov, and S.~Saurabh}, {\em Fast {FAST}}, in Proceedings
  of the 36th International Colloquium of Automata, Languages and Programming
  (ICALP), vol.~5555 of Lecture Notes in Comput. Sci., Springer, 2009,
  pp.~49--58.

\bibitem{AshokFKSZ18}
{\sc P.~Ashok, F.~V. Fomin, S.~Kolay, S.~Saurabh, and M.~Zehavi}, {\em Exact
  algorithms for terrain guarding}, {ACM} Trans. Algorithms, 14 (2018),
  pp.~25:1--25:20.

\bibitem{bla2024constrained}
{\sc V.~Blazej, B.~Klemz, F.~Klesen, M.~D. Sieper, A.~Wolff, and J.~Zink}, {\em
  Constrained and ordered level planarity parameterized by the number of
  levels}, in 40th International Symposium on Computational Geometry, SoCG
  2024, June 11-14, 2024, Athens, Greece, W.~Mulzer and J.~M. Phillips, eds.,
  vol.~293 of LIPIcs, Schloss Dagstuhl - Leibniz-Zentrum f{\"{u}}r Informatik,
  2024, pp.~20:1--20:16.

\bibitem{BodlaenderTY11}
{\sc H.~L. Bodlaender, S.~Thomass{\'{e}}, and A.~Yeo}, {\em Kernel bounds for
  disjoint cycles and disjoint paths}, Theor. Comput. Sci., 412 (2011),
  pp.~4570--4578.

\bibitem{bruckner2017partial}
{\sc G.~Br{\"u}ckner and I.~Rutter}, {\em Partial and constrained level
  planarity}, in Proceedings of the 28h Annual ACM-SIAM Symposium on Discrete
  Algorithms (SODA), SIAM, 2017, pp.~2000--2011.

\bibitem{carpano1980automatic}
{\sc M.-J. Carpano}, {\em Automatic display of hierarchized graphs for
  computer-aided decision analysis}, IEEE Transactions on Systems, Man, and
  Cybernetics, 10 (1980), pp.~705--715.

\bibitem{catarci1995assignment}
{\sc T.~Catarci}, {\em The assignment heuristic for crossing reduction}, IEEE
  Transactions on Systems, Man, and Cybernetics, 25 (1995), pp.~515--521.

\bibitem{chuzhoy2022subpolynomial}
{\sc J.~Chuzhoy and Z.~Tan}, {\em A subpolynomial approximation algorithm for
  graph crossing number in low-degree graphs}, in Proceedings of the 54th
  Annual ACM Symposium on Theory of Computing, STOC, 2022, p.~TBA.

\bibitem{corneil2024recursive}
{\sc D.~Corneil, M.~Habib, C.~Paul, and M.~Tedder}, {\em A recursive linear
  time modular decomposition algorithm via lexbfs}, 2024.

\bibitem{CyganFKLMPPS15}
{\sc M.~Cygan, F.~V. Fomin, L.~Kowalik, D.~Lokshtanov, D.~Marx, M.~Pilipczuk,
  M.~Pilipczuk, and S.~Saurabh}, {\em Parameterized Algorithms}, Springer,
  2015.

\bibitem{di1988hierarchies}
{\sc G.~Di~Battista and E.~Nardelli}, {\em Hierarchies and planarity theory},
  IEEE Transactions on systems, man, and cybernetics, 18 (1988),
  pp.~1035--1046.

\bibitem{Diestel12}
{\sc R.~Diestel}, {\em Graph Theory, 4th Edition}, vol.~173 of Graduate texts
  in mathematics, Springer, 2012.

\bibitem{10.1007/s00453-005-1181-y}
{\sc V.~Dujmovi{\'c}, M.~Fellows, M.~Hallett, M.~Kitching, G.~Liotta,
  C.~Mccartin, N.~Nishimura, P.~Ragde, F.~Rosamond, M.~Suderman, S.~Whitesides,
  and D.~R. Wood}, {\em A fixed-parameter approach to 2-layer planarization},
  Algorithmica, 45 (2006), pp.~159--182.

\bibitem{dujmovic2008parameterized}
{\sc V.~Dujmovi{\'c}, M.~R. Fellows, M.~Kitching, G.~Liotta, C.~McCartin,
  N.~Nishimura, P.~Ragde, F.~Rosamond, S.~Whitesides, and D.~R. Wood}, {\em On
  the parameterized complexity of layered graph drawing}, Algorithmica, 52
  (2008), pp.~267--292.

\bibitem{dujmovic2008fixed}
{\sc V.~Dujmovi{\'c}, H.~Fernau, and M.~Kaufmann}, {\em Fixed parameter
  algorithms for one-sided crossing minimization revisited}, Journal of
  Discrete Algorithms, 6 (2008), pp.~313--323.

\bibitem{dujmovic2004efficient}
{\sc V.~Dujmovi{\'c} and S.~Whitesides}, {\em An efficient fixed parameter
  tractable algorithm for 1-sided crossing minimization}, Algorithmica, 40
  (2004), pp.~15--31.

\bibitem{EadesMCW85}
{\sc P.~Eades, B.~D. McKay, and N.~C. Wormald}, {\em On an edge crossing
  problem}, in Proceedings of the 9th Australian Computer Science Conference,
  Australian National University, 1986, pp.~327--334.

\bibitem{eades1994edge}
{\sc P.~Eades and N.~C. Wormald}, {\em Edge crossings in drawings of bipartite
  graphs}, Algorithmica, 11 (1994), pp.~379--403.

\bibitem{Feige09}
{\sc U.~Feige}, {\em Faster {FAST} ({F}eedback {A}rc {S}et in {T}ournaments)},
  CoRR, abs/0911.5094 (2009).

\bibitem{DBLP:journals/jgaa/Fernau05}
{\sc H.~Fernau}, {\em Two-layer planarization: Improving on parameterized
  algorithmics}, Journal of Graph Algorithms and Applications, 9 (2005),
  pp.~205--238.

\bibitem{fernau2014social}
{\sc H.~Fernau, F.~V. Fomin, D.~Lokshtanov, M.~Mnich, G.~Philip, and
  S.~Saurabh}, {\em Social choice meets graph drawing: How to get
  subexponential time algorithms for ranking and drawing problems}, Tsinghua
  Science and Technology, 19 (2014), pp.~374--386.

\bibitem{fomin2019kernelization}
{\sc F.~V. Fomin, D.~Lokshtanov, S.~Saurabh, and M.~Zehavi}, {\em
  Kernelization: theory of parameterized preprocessing}, Cambridge University
  Press, 2019.

\bibitem{FominP13}
{\sc F.~V. Fomin and M.~Pilipczuk}, {\em Subexponential parameterized algorithm
  for computing the cutwidth of a semi-complete digraph}, in Proceedings of the
  21st Annual European Symposium on Algorithms (ESA), LNCS, vol.~8125 of
  Lecture Notes in Comput. Sci., Springer, 2013, pp.~505--516.

\bibitem{FordF56}
{\sc L.~R. Ford, Jr. and D.~R. Fulkerson}, {\em Maximal flow through a
  network}, Canadian J. Math., 8 (1956), pp.~399--404.

\bibitem{fulek2013hanani}
{\sc R.~Fulek, M.~J. Pelsmajer, M.~Schaefer, and D.~{\v{S}}tefankovi{\v{c}}},
  {\em Hanani--{T}utte, monotone drawings, and level-planarity}, Thirty essays
  on geometric graph theory,  (2013), pp.~263--287.

\bibitem{DBGanianMOPR24}
{\sc R.~Ganian, H.~M{\"{u}}ller, S.~Ordyniak, G.~Paesani, and M.~Rychlicki},
  {\em A tight subexponential-time algorithm for two-page book embedding}, in
  51st International Colloquium on Automata, Languages, and Programming,
  {ICALP} 2024, July 8-12, 2024, Tallinn, Estonia, K.~Bringmann, M.~Grohe,
  G.~Puppis, and O.~Svensson, eds., vol.~297 of LIPIcs, Schloss Dagstuhl -
  Leibniz-Zentrum f{\"{u}}r Informatik, 2024, pp.~68:1--68:18.

\bibitem{garey1983crossing}
{\sc M.~R. Garey and D.~S. Johnson}, {\em Crossing number is {NP}-complete},
  SIAM Journal on Algebraic Discrete Methods, 4 (1983), pp.~312--316.

\bibitem{grohe2004computing}
{\sc M.~Grohe}, {\em Computing crossing numbers in quadratic time}, Journal of
  Computer and System Sciences, 68 (2004), pp.~285--302.

\bibitem{HammH22}
{\sc T.~Hamm and P.~Hlinen{\'{y}}}, {\em Parameterised partially-predrawn
  crossing number}, in 38th International Symposium on Computational Geometry,
  SoCG 2022, June 7-10, 2022, Berlin, Germany, X.~Goaoc and M.~Kerber, eds.,
  vol.~224 of LIPIcs, Schloss Dagstuhl - Leibniz-Zentrum f{\"{u}}r Informatik,
  2022, pp.~46:1--46:15.

\bibitem{harary1971trees}
{\sc F.~Harary and A.~Schwenk}, {\em Trees with hamiltonian square},
  Mathematika, 18 (1971), pp.~138--140.

\bibitem{heath1992laying}
{\sc L.~S. Heath and A.~L. Rosenberg}, {\em Laying out graphs using queues},
  SIAM Journal on Computing, 21 (1992), pp.~927--958.

\bibitem{HopcroftT73}
{\sc J.~E. Hopcroft and R.~E. Tarjan}, {\em Efficient algorithms for graph
  manipulation {[H]} (algorithm 447)}, Commun. {ACM}, 16 (1973), pp.~372--378.

\bibitem{ImpagliazzoP99}
{\sc R.~Impagliazzo and R.~Paturi}, {\em Complexity of k-sat}, in Proceedings
  of the 14th Annual {IEEE} Conference on Computational Complexity, Atlanta,
  Georgia, USA, May 4-6, 1999, {IEEE} Computer Society, 1999, pp.~237--240.

\bibitem{ImpagliazzoPZ01}
{\sc R.~Impagliazzo, R.~Paturi, and F.~Zane}, {\em Which problems have strongly
  exponential complexity?}, J. Comput. Syst. Sci., 63 (2001), pp.~512--530.

\bibitem{junger1998level}
{\sc M.~J{\"u}nger, S.~Leipert, and P.~Mutzel}, {\em Level planarity testing in
  linear time}, in In Proceedings of the 6th International Symposium Graph
  Drawing (GD), Springer, 1998, pp.~224--237.

\bibitem{JungerLM98}
{\sc M.~J{\"{u}}nger, S.~Leipert, and P.~Mutzel}, {\em Level planarity testing
  in linear time}, in Graph Drawing, 6th International Symposium, GD'98,
  Montr{\'{e}}al, Canada, August 1998, Proceedings, S.~Whitesides, ed.,
  vol.~1547 of Lecture Notes in Computer Science, Springer, 1998, pp.~224--237.

\bibitem{junger20022}
{\sc M.~J{\"u}nger and P.~Mutzel}, {\em 2-layer straightline crossing
  minimization: Performance of exact and heuristic algorithms}, in Graph
  algorithms and applications i, World Scientific, 2002, pp.~3--27.

\bibitem{junger2012graph}
\leavevmode\vrule height 2pt depth -1.6pt width 23pt, {\em Graph drawing
  software}, Springer Science \& Business Media, 2012.

\bibitem{KarpinskiS10}
{\sc M.~Karpinski and W.~Schudy}, {\em Faster algorithms for {F}eedback {A}rc
  {S}et {T}ournament, {K}emeny {R}ank {A}ggregation and {B}etweenness
  {T}ournament}, in Proceedings of the 21st International Symposium on
  Algorithms and Computation (ISAAC), vol.~6506 of Lecture Notes in Comput.
  Sci., Springer, 2010, pp.~3--14.

\bibitem{Kawarabayashi09}
{\sc K.~Kawarabayashi}, {\em Planarity allowing few error vertices in linear
  time}, in Proceedings of the 50th Annual Symposium on Foundations of Computer
  Science (FOCS), IEEE, 2009, pp.~639--648.

\bibitem{KleinM14}
{\sc P.~N. Klein and D.~Marx}, {\em A subexponential parameterized algorithm
  for subset {TSP} on planar graphs}, in Proceedings of the Twenty-Fifth Annual
  {ACM-SIAM} Symposium on Discrete Algorithms, {SODA} 2014, Portland, Oregon,
  USA, January 5-7, 2014, C.~Chekuri, ed., {SIAM}, 2014, pp.~1812--1830.

\bibitem{klemz2019ordered}
{\sc B.~Klemz and G.~Rote}, {\em Ordered level planarity and its relationship
  to geodesic planarity, bi-monotonicity, and variations of level planarity},
  ACM Transactions on Algorithms (TALG), 15 (2019), pp.~1--25.

\bibitem{KobayashiMNT14}
{\sc Y.~Kobayashi, H.~Maruta, Y.~Nakae, and H.~Tamaki}, {\em A linear edge
  kernel for two-layer crossing minimization}, Theor. Comput. Sci., 554 (2014),
  pp.~74--81.

\bibitem{DBLP:journals/algorithmica/KobayashiT15}
{\sc Y.~Kobayashi and H.~Tamaki}, {\em A fast and simple subexponential fixed
  parameter algorithm for one-sided crossing minimization}, Algorithmica, 72
  (2015), pp.~778--790.

\bibitem{kobayashi2016faster}
{\sc Y.~Kobayashi and H.~Tamaki}, {\em A faster fixed parameter algorithm for
  two-layer crossing minimization}, Information Processing Letters, 116 (2016),
  pp.~547--549.

\bibitem{LokshtanovP0S0Z25}
{\sc D.~Lokshtanov, F.~Panolan, S.~Saurabh, R.~Sharma, J.~Xue, and M.~Zehavi},
  {\em Crossing number in slightly superexponential time (extended abstract)},
  in Proceedings of the 2025 Annual {ACM-SIAM} Symposium on Discrete
  Algorithms, {SODA} 2025, New Orleans, LA, USA, January 12-15, 2025, Y.~Azar
  and D.~Panigrahi, eds., {SIAM}, 2025, pp.~1412--1424.

\bibitem{pach2000thirteen}
{\sc J.~Pach and G.~T{\'o}th}, {\em Thirteen problems on crossing numbers},
  Geombinatorics, 9 (2000), pp.~194--207.

\bibitem{schaefer2012graph}
{\sc M.~Schaefer}, {\em The graph crossing number and its variants: A survey},
  The Electronic Journal of Combinatorics,  (2012), pp.~DS21--Sep.

\bibitem{schaefer2018crossing}
\leavevmode\vrule height 2pt depth -1.6pt width 23pt, {\em Crossing numbers of
  graphs}, CRC Press, 2018.

\bibitem{Shahrokhi01}
{\sc F.~Shahrokhi, O.~S\'{y}kora, L.~A. Sz\'{e}kely, and I.~Vr\v{t}o}, {\em On
  bipartite drawings and the linear arrangement problem}, SIAM J. Comput., 30
  (2001), pp.~1773--1789.

\bibitem{sugiyama1981methods}
{\sc K.~Sugiyama, S.~Tagawa, and M.~Toda}, {\em Methods for visual
  understanding of hierarchical system structures}, IEEE Transactions on
  Systems, Man, and Cybernetics, 11 (1981), pp.~109--125.

\bibitem{tamassia2013handbook}
{\sc R.~Tamassia}, {\em Handbook of graph drawing and visualization}, CRC
  press, 2013.

\bibitem{tantau2013graph}
{\sc T.~Tantau}, {\em Graph drawing in tikz}, Journal of Graph Algorithms and
  Applications, 17 (2013), pp.~495--513.

\bibitem{Zehavi22}
{\sc M.~Zehavi}, {\em Parameterized analysis and crossing minimization
  problems}, Comput. Sci. Rev., 45 (2022), p.~100490.

\end{thebibliography}
	
\end{document}

\end{document}